\documentclass[11pt,aps,pra,notitlepage,tightenlines,nofootinbib,superscriptaddress]{revtex4-2}

\usepackage{fancyhdr}
\fancypagestyle{header}{
  \fancyhf{}%
  \fancyhead[R]{\hyperref[toc]{\thepage}}%
}

\makeatletter

\newcommand{\apptocfile}{atoc}

\let\apptoc@orig@appendix\appendix
\renewcommand{\appendix}{%
  \apptoc@orig@appendix
  \let\apptoc@orig@addtocontents\addtocontents
  \long\def\addtocontents##1##2{%
    \def\apptoc@ext{##1}%
    \def\apptoc@toc{toc}%
    \ifx\apptoc@ext\apptoc@toc
      \apptoc@orig@addtocontents{\apptocfile}{##2}%
    \else
      \apptoc@orig@addtocontents{##1}{##2}%
    \fi
  }%
}

\newcommand{\appendixtableofcontents}{%
  \begingroup
    \setcounter{tocdepth}{3}%
    \phantomsection
    \let\addcontentsline\@gobblethree
    \section*{Appendices}%
    \pdfbookmark[1]{Appendices}{apxcontents}%
    \@starttoc{\apptocfile}%
  \endgroup
}

\makeatother

\usepackage{newpxtext,newpxmath}

\let\coloneqq\relax

\usepackage[latin1]{inputenc}
\usepackage{amsthm}
\usepackage{amssymb}
\usepackage{amsmath}
\usepackage{bbold}
\usepackage{bbm}
\usepackage[pdftex, backref=page]{hyperref}
\usepackage{braket}
\usepackage{dsfont}
\usepackage{mathdots}
\usepackage{mathtools}
\usepackage{enumerate}
\usepackage[shortlabels]{enumitem}
\usepackage{csquotes}
\usepackage{stmaryrd}
\usepackage[cal=boondox]{mathalfa}
\usepackage{graphicx}
\usepackage{stackengine}
\usepackage{scalerel}
\usepackage{tensor}       
\usepackage{xcolor}
\usepackage{array}
\usepackage{makecell}
\newcolumntype{x}[1]{>{\centering\arraybackslash}p{#1}}
\usepackage{tikz}
\usepackage{pgfplots}
\usetikzlibrary{shapes.geometric, shapes.misc, positioning, arrows, arrows.meta, decorations.pathreplacing, decorations.pathmorphing, patterns, angles, quotes, calc}
\usepackage{booktabs}
\usepackage{xfrac}
\usepackage{siunitx}
\usepackage{centernot}
\usepackage{comment}
\usepackage{chngcntr}
\usepackage[normalem]{ulem}

\usepackage{bm}

\newtheorem{thm}{Theorem}
\newtheorem*{thm*}{Theorem}

\newtheorem*{prop*}{Proposition}
\newtheorem{lemma}[thm]{Lemma}
\newtheorem*{lemma*}{Lemma}

\newtheorem*{algorithm*}{Algorithm}

\newtheorem*{corollary*}{Corollary}
\newtheorem{cor}[thm]{Corollary}
\newtheorem*{cor*}{Corollary}

\newtheorem*{cj*}{Conjecture}
\newtheorem{Def}[thm]{Definition}
\newtheorem*{Def*}{Definition}

\newtheorem*{question*}{Question}
\newtheorem{problem}[thm]{Problem}
\newtheorem*{problem*}{Problem}

\newtheorem*{claim*}{Claim}

\makeatletter
\def\thmhead@plain#1#2#3{%
  \thmname{#1}\thmnumber{\@ifnotempty{#1}{ }\@upn{#2}}%
  \thmnote{ {\the\thm@notefont#3}}}
\let\thmhead\thmhead@plain
\makeatother

\theoremstyle{definition}

\newcommand{\bb}{\begin{equation}\begin{aligned}\hspace{0pt}}
\newcommand{\bbb}{\begin{equation*}\begin{aligned}}
\newcommand{\ee}{\end{aligned}\end{equation}}
\newcommand{\eee}{\end{aligned}\end{equation*}}
\newcommand*{\coloneqq}{\mathrel{\vcenter{\baselineskip0.5ex \lineskiplimit0pt \hbox{\scriptsize.}\hbox{\scriptsize.}}} =}

\newcommand{\ketbra}[1]{\ket{#1}\!\!\bra{#1}}

\newcommand{\ketbraa}[2]{\ket{#1}\!\!\bra{#2}}

\renewcommand{\epsilon}{\varepsilon}

\DeclareMathOperator{\Tr}{Tr}

\DeclareMathAlphabet{\pazocal}{OMS}{zplm}{m}{n}

\DeclareMathOperator{\diag}{diag}

\newcommand{\HH}{\pazocal{H}}

\newcommand{\NN}{\pazocal{N}}

\newcommand{\lsmatrix}{\left(\begin{smallmatrix}}
\newcommand{\rsmatrix}{\end{smallmatrix}\right)}

\stackMath

\stackMath

\makeatletter
\newcommand*\rel@kern[1]{\kern#1\dimexpr\macc@kerna}
\newcommand*\widebar[1]{%
  \begingroup
  \def\mathaccent##1##2{%
    \rel@kern{0.8}%
    \overline{\rel@kern{-0.8}\macc@nucleus\rel@kern{0.2}}%
    \rel@kern{-0.2}%
  }%
  \macc@depth\@ne
  \let\math@bgroup\@empty \let\math@egroup\macc@set@skewchar
  \mathsurround\z@ \frozen@everymath{\mathgroup\macc@group\relax}%
  \macc@set@skewchar\relax
  \let\mathaccentV\macc@nested@a
  \macc@nested@a\relax111{#1}%
  \endgroup
}

\counterwithin*{equation}{part}
\counterwithin*{thm}{part}
\counterwithin*{figure}{part}

\tikzset{meter/.append style={draw, inner sep=10, rectangle, font=\vphantom{A}, minimum width=30, line width=.8, path picture={\draw[black] ([shift={(.1,.3)}]path picture bounding box.south west) to[bend left=50] ([shift={(-.1,.3)}]path picture bounding box.south east);\draw[black,-latex] ([shift={(0,.1)}]path picture bounding box.south) -- ([shift={(.3,-.1)}]path picture bounding box.north);}}}
\tikzset{roundnode/.append style={circle, draw=black, fill=gray!20, thick, minimum size=10mm}}
\tikzset{squarenode/.style={rectangle, draw=black, fill=none, thick, minimum size=10mm}}

\definecolor{Blues5seq1}{RGB}{239,243,255}
\definecolor{Blues5seq2}{RGB}{189,215,231}
\definecolor{Blues5seq3}{RGB}{107,174,214}
\definecolor{Blues5seq4}{RGB}{49,130,189}
\definecolor{Blues5seq5}{RGB}{8,81,156}

\definecolor{Greens5seq1}{RGB}{237,248,233}
\definecolor{Greens5seq2}{RGB}{186,228,179}
\definecolor{Greens5seq3}{RGB}{116,196,118}
\definecolor{Greens5seq4}{RGB}{49,163,84}
\definecolor{Greens5seq5}{RGB}{0,109,44}

\definecolor{Reds5seq1}{RGB}{254,229,217}
\definecolor{Reds5seq2}{RGB}{252,174,145}
\definecolor{Reds5seq3}{RGB}{251,106,74}
\definecolor{Reds5seq4}{RGB}{222,45,38}
\definecolor{Reds5seq5}{RGB}{165,15,21}

\allowdisplaybreaks

\usepackage{cleveref}
\crefname{thm}{Theorem}{Theorems}
\crefname{figure}{Figure}{Figures}

\theoremstyle{plain}
\newtheorem{mainproblem}{Problem}
\newtheorem{maintheorem}{Theorem}

\usepackage{algorithm}
\usepackage{algpseudocode}
\usepackage{mdframed}
\usepackage{setspace}
\setdisplayskipstretch{} 
\usepackage{xcolor}
\definecolor{blueviolet}{rgb}{0.2, 0.2, 0.6}
\definecolor{webgreen}{rgb}{0,.5,0}
\definecolor{webbrown}{rgb}{.6,0,0}
\usepackage{hyperref}
\hypersetup{
    colorlinks,
    linkcolor=blueviolet,
    citecolor=webgreen,
    urlcolor=webbrown
}
\usepackage{cleveref}

\begin{document}
\pagestyle{header}

\title{Exponential separation in sensing continuous signals via squeezing} 

\author{Francesco Anna Mele}
\thanks{These authors contributed equally to this work.}
\affiliation{California Institute of Technology, Pasadena, California 91125, USA}

\author{Nadine Meister}
\thanks{These authors contributed equally to this work.}
\affiliation{California Institute of Technology, Pasadena, California 91125, USA}

\author{Haimeng~Zhao}
\affiliation{California Institute of Technology, Pasadena, California 91125, USA}

\author{Senrui Chen}
\affiliation{California Institute of Technology, Pasadena, California 91125, USA}

\author{Hsin-Yuan Huang}
\affiliation{Oratomic, Pasadena, California 91125, USA}
\affiliation{California Institute of Technology, Pasadena, California 91125, USA}

\begin{abstract} 
\begin{spacing}{1.2}
Quantum sensing traditionally focuses on using quantum resources such as squeezing and entanglement to improve precision for sensing fixed signals. However, in many applications such as gravitational-wave detection and electromagnetic-field sensing, the signal evolves continuously and varies through time. Additionally, the learner is free to prepare, control, and measure the sensor at \emph{arbitrary} times, possibly chosen adaptively. In this work, we establish exponential separation in sensing time for continuously evolving signals due to the available squeezing. The signals we study are characterized by a pattern size~\(T\), and we find that a sensor with squeezing at least $\omega(\sqrt{\log T})$ offers a $\mathrm{poly}(T)$ sensing time, whereas those with squeezing at most $o(\sqrt{\log T})$ must use an exponential sensing time in $T$.
\end{spacing}
\end{abstract}

\maketitle

\section{Introduction}

A central goal of quantum information is to understand and harness the advantages quantum technologies offer. Beyond advantages in solving computational problems \cite{Feynman1982,Shor,grover97,dalzellQuantumAlgorithmsSurvey2023,liu2021rigorous,zhao2026massive}, quantum technologies also enable us to learn about the physical world more efficiently and facilitate scientific discoveries \cite{caves1981quantum,giovannetti2004quantum,abbott2016observation}.
An example of such a learning task is to sense an unknown signal in nature, which might be an electromagnetic signal, an exerted force, or some spacetime strain from gravitational waves. To do so, one prepares a probe sensor in a specific quantum state, strategically controls its evolution under the signal, measures it, and repeats to accumulate experimental data and infer the desired properties of the signal~\cite{degen2017quantum}.

It is well known that when sensing a static signal, quantum resources such as entanglement and squeezing enable quadratic improvement in sensitivity \cite{giovannetti2004quantum,giovannetti2011advances,pezze2018quantum}.
This quadratic advantage achieves the ultimate limit known as the Heisenberg limit and cannot be further improved asymptotically for static signals. 
This raises a natural question:
\begin{center}
    \emph{Can quantum resources lead to exponential separations in sensing time-varying signals?}
\end{center}

A fruitful theoretical abstraction to establish exponential separations in learning about the physical world is to study the query or sample complexity of learning properties of quantum channels and states.
Recent works have established exponential separations with various quantum resources in a growing range of state and channel learning tasks \cite{huang2022quantum,Liu_2025,huang2021information,aharonov2022quantum,chen2021exponential,chen2021hierarchy,Chen_2024,chen2022pauli,seif2024entanglement,chen2025efficientpaulichannelestimation,kim2025fundamental,oh2024entanglement,cotler2026quantum,coroi2025exponential,kannan2026exponential,prabhu2026exponential,king2024exponential,chen2024adaptivity,king2025triply,caro2024learningquantumprocesses,bravoprieto2026quantum}, some of which have been experimentally demonstrated \cite{huang2022quantum,Liu_2025,kannan2026exponential}. We discuss the works most closely related to sensing signals and their relation to our setting in Appendix~\ref{sec:related_work}.  
Nevertheless, it is important to note that querying the effective channel is only a particular way of sensing the underlying signal, and therefore separations in channel query complexity may not survive once continuous-time control is allowed.
A simple example is estimating the magnitude of a sinusoidal signal with known frequency. It can be easily solved by applying spin echo control~\cite{degen2017quantum}, while forcing one to query the effective channel of a full period makes the task impossible.

\begin{figure}[t]
    \centering
    \includegraphics[width=\linewidth]{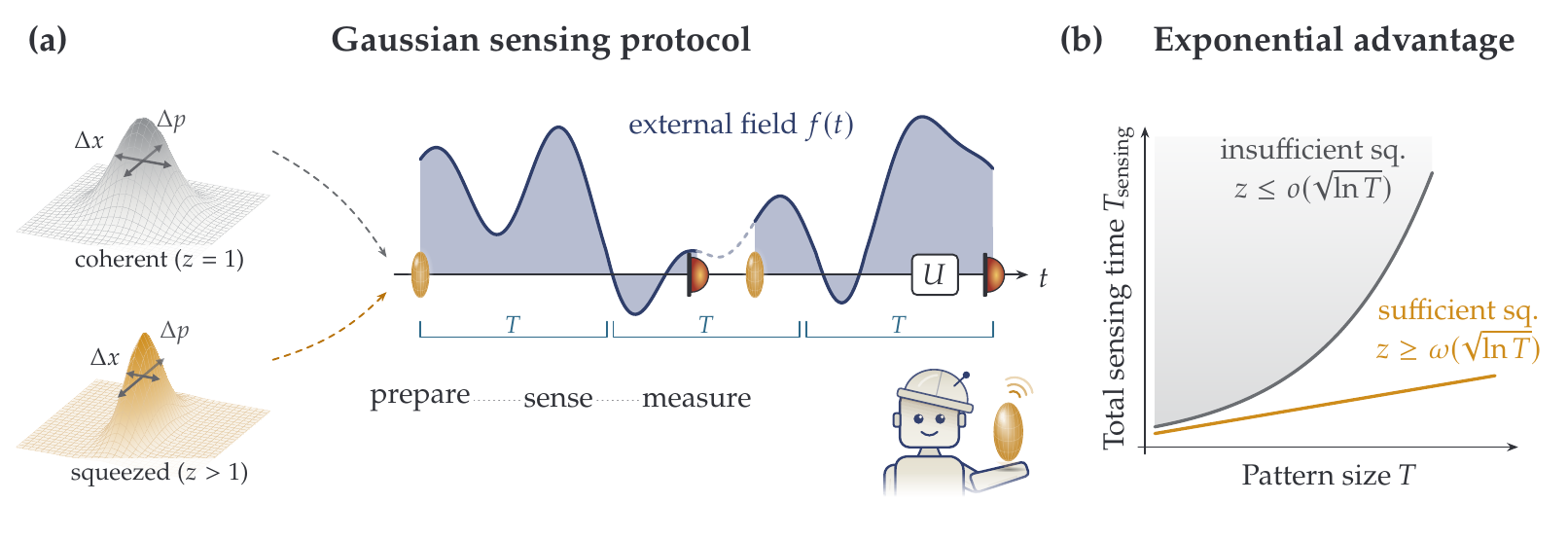}
    \caption{(a) Within our single-mode sensing framework, we consider Gaussian sensing protocols with bounded squeezing. These protocols allow the learner to prepare the sensor in single-mode Gaussian states with squeezing at most $z$, let the sensor evolve for any choice of time $t$ while performing arbitrary squeezing-free Gaussian unitaries $U$ and destructive Gaussian measurements at arbitrary times, and perform unrestricted classical processing of the measurement outcomes. States with no squeezing have equal uncertainty in both phase quadratures $\Delta x = \Delta p$, whereas states with higher squeezing $(z>1)$ have reduced uncertainty in one of the quadratures. (b) For the sensing problems studied in this work, sufficient squeezing, i.e.~$z = \omega(\sqrt{\log T})$, enables polynomial total sensing time. In contrast, for $z = o(\sqrt{\log T})$, the sensing time is fundamentally required to be exponential in $T$. Here, $T$ denotes the duration of the time-correlation pattern hidden in the time-varying signal, as defined in the sensing tasks in Subsection~\ref{subsec_def_prob}. }
    \label{fig:fig1}
\end{figure}

Thus, to obtain rigorous end-to-end separations beyond the channel query model, it is important to take continuous-time control into account.
Ref.~\cite{allen2025quantumcomputingenhancedsensing} made progress by designing a quantum computing enhanced sensing protocol for detecting weak oscillating signals based on Grover's algorithm. 
Their quantum algorithm exhibits polynomial advantage over conventional strategies. Yet, it remains unclear whether simple quantum resources and control can lead to \emph{exponential} separations in sensing time-varying signals.\footnote{Note that Refs.~\cite{cotler2026quantum,kannan2026exponential} establish exponential advantages for sensing time-varying signals in \emph{channel-query models}, where access to the signal is discretized into repeated queries to an effective input-output channel. However, in their model, the learner is not allowed to interleave control operations with the signal evolution within each query.}

In this work, we show how the simple quantum resource of squeezing can lead to end-to-end exponential separations in sensing time-varying signals. We introduce sensing tasks for which sufficient squeezing enables a single-mode continuous-variable (CV) sensor to solve the tasks efficiently (Theorem~\ref{thm:main}). In contrast, we rigorously prove that without sufficient squeezing, any adaptive Gaussian sensing protocol---even with control and measurements at arbitrary times---requires exponentially longer sensing time.

Our results show that time correlation patterns hidden in a time-varying signal can give rise to exponential separations in sensing with squeezing in a single mode. The same mechanism also allows us to convert other methods for increasing the sensitivity of the sensor into exponential separations, such as increasing the coupling strength. The benefits of increasing sensitivity and squeezing in general are self-evident (i.e.~they effectively reduce the measurement noise relative to the signal), but rigorously proving that this can lead to exponential separations in sensing time-varying signals requires a careful and somewhat puzzling mathematical analysis, because the learner is free to choose not only its operations but also their \emph{timing} adaptively.
While no-go results against adaptive operations have been extensively studied in quantum learning settings involving states and channels~\cite{chen2021exponential,huang2022quantum,chen2022complexitynisq,Chen_2024}, allowing the learner to also choose \emph{when} these operations are performed is specific to sensing time-varying signals and introduces an additional layer of technical difficulty.
To our knowledge, this freedom has been addressed rigorously only in Ref.~\cite{allen2025quantumcomputingenhancedsensing} in a rather different setting, using very different techniques.
To account for this, we develop a mathematical framework, which we call the \emph{hierarchical tree representation}, that captures this additional freedom and allows us to establish the exponential separation.

\section{Main Results}

Continuous variable systems such as bosonic modes are often used as quantum sensors, where the signals are encoded as displacements in their phase space. The displacements can then be read out by homodyne detection or more generally Gaussian measurements \cite{BUCCO,braunstein2005quantum,weedbrook2012gaussian,adesso2014continuous,andersen2010continuous}. 
For example, in laser-interferometric gravitational-wave detectors such as LIGO \cite{abbott2016observation,tse2019quantumenhanced}, a gravitational wave effectively induces a small quadrature displacement in the outgoing optical mode, which is then measured through a readout scheme functioning as homodyne detection \cite{danilishin2012quantum,tse2022squeezed}. 
Similarly, microwave signals can be sensed in circuit-QED platforms as displacements of the bosonic modes in superconducting resonators \cite{blais2021circuit}. 

A simple model of such CV sensors is therefore a bosonic mode evolving under the Hamiltonian
\bb
    \hat H(t)=f(t)\,\hat p,
\ee
where $f(t)$ is the time-varying signal whose property we wish to learn and $\hat{x}, \hat{p}$ are the position and momentum operators of the mode satisfying the canonical commutation relation $[\hat x, \hat p]=i$.
We provide an introduction to CV systems in \Cref{sec_prel}, and discuss how this sensing model arises in existing experimental platforms in \Cref{sec:exp_platforms}.
After evolving for time $t$, in the absence of control, the mode is displaced to the position
\begin{equation}
	\hat x(t) = \hat x(0) + \int_0^t f(t')\,\mathrm{d}t'
\end{equation}
in the Heisenberg picture.
One can then learn properties of the signal $f(t)$ by strategically interleaving the evolution with control and measurements.

In these bosonic platforms~\cite{abbott2016observation,tse2019quantumenhanced,blais2021circuit}, the sensor is typically prepared in a squeezed Gaussian state~\cite{BUCCO,weedbrook2012gaussian}, and the precision in measuring a displacement is fundamentally limited by the available squeezing due to the Heisenberg uncertainty principle. With quantum technology, one can squeeze the mode further to increase the precision in one quadrature at the expense of its conjugate. The amount of squeezing can be characterized by a parameter $z\in[1,\infty)$, where $z=1$ corresponds to no squeezing (and hence equal fluctuations in both conjugate quadratures), while the standard deviation when measuring the squeezed quadrature is lower bounded by order $1/z$, with this scaling attained by pure squeezed states and vanishing in the ideal limit $z\to\infty$. See \Cref{fig:fig1}(a) for a pictorial representation.

Thus, motivated by these experimental settings~\cite{abbott2016observation,tse2019quantumenhanced,blais2021circuit}, a natural class of sensing strategies is given by \emph{Gaussian protocols with squeezing bounded by $z$}, defined as follows. Apart from the time evolution driven by the Hamiltonian $\hat H(t)=f(t)\,\hat p$, these protocols allow:
\begin{itemize}[(1)]
\item[(1)] preparation of single-mode Gaussian input states with squeezing at most $z$;
\item[(2)] application of single-mode Gaussian unitaries without further squeezing during sensing, including phase-space rotations and displacements;
\item[(3)] application of arbitrary (destructive) Gaussian measurements, including homodyne and heterodyne detection;
\item[(4)] arbitrary classical post-processing of the measurement outcomes.
\end{itemize} 
Thus, these protocols consist of preparing the sensor, strategically applying squeezing-free Gaussian unitaries during the signal evolution, measuring it (destructively) with a Gaussian measurement, and then repeating the procedure. Importantly, the times at which these operations are performed are completely arbitrary and may even be chosen adaptively based on previous measurement outcomes. The higher the available squeezing $z$, the more powerful the resulting protocol becomes for sensing the unknown signal. In the following, we rigorously show that sufficient squeezing enables a single-mode CV sensor to solve certain sensing tasks of time-varying signals exponentially faster.

\subsection{Learning hidden patterns in time-varying signals}\label{subsec_def_prob}

We consider sensing tasks with the goal of learning hidden temporal correlations in a time-varying signal.
To make these correlations explicit, we focus on signals built from a sequence of consecutive unit-duration pulses (we fix a time scale $\Delta t=1$).
The $i$-th pulse is associated with a real coefficient $m_i$.
We group these pulses into independent blocks of $T$ pulses, and call $T\in\mathbb{N}^+$ the \emph{pattern size}.
The sensing problems below encode the hidden temporal structure in the correlations among the $T$ pulse coefficients $m_i$ within each block.

More precisely, we consider continuous signals of the form
\begin{equation}\label{eq_def_signal}
    f(t)=\sum_{i\ge1}m_i\,\varphi(t-i+1),
\end{equation}
where $\varphi:\mathbb{R}\to\mathbb{R}$ is a fixed continuous pulse profile that vanishes outside $[0,1]$, is positive on $(0,1)$, and is normalized so that $\int_0^1\varphi(t)\,\mathrm{d}t=1$.
Thus, during the time window $t\in[i-1,i]$, the signal consists of the $i$-th pulse, whose shape is fixed by $\varphi$ and whose overall scale is set by $m_i$. The normalization ensures that $\int_{i-1}^{i}f(t)\,\mathrm{d}t=m_i$, so that $m_i$ is precisely the integrated signal over the $i$-th pulse. The resulting signal is illustrated in \Cref{fig:fig2}.

\begin{figure} 
    \centering
    \includegraphics[width=\linewidth]{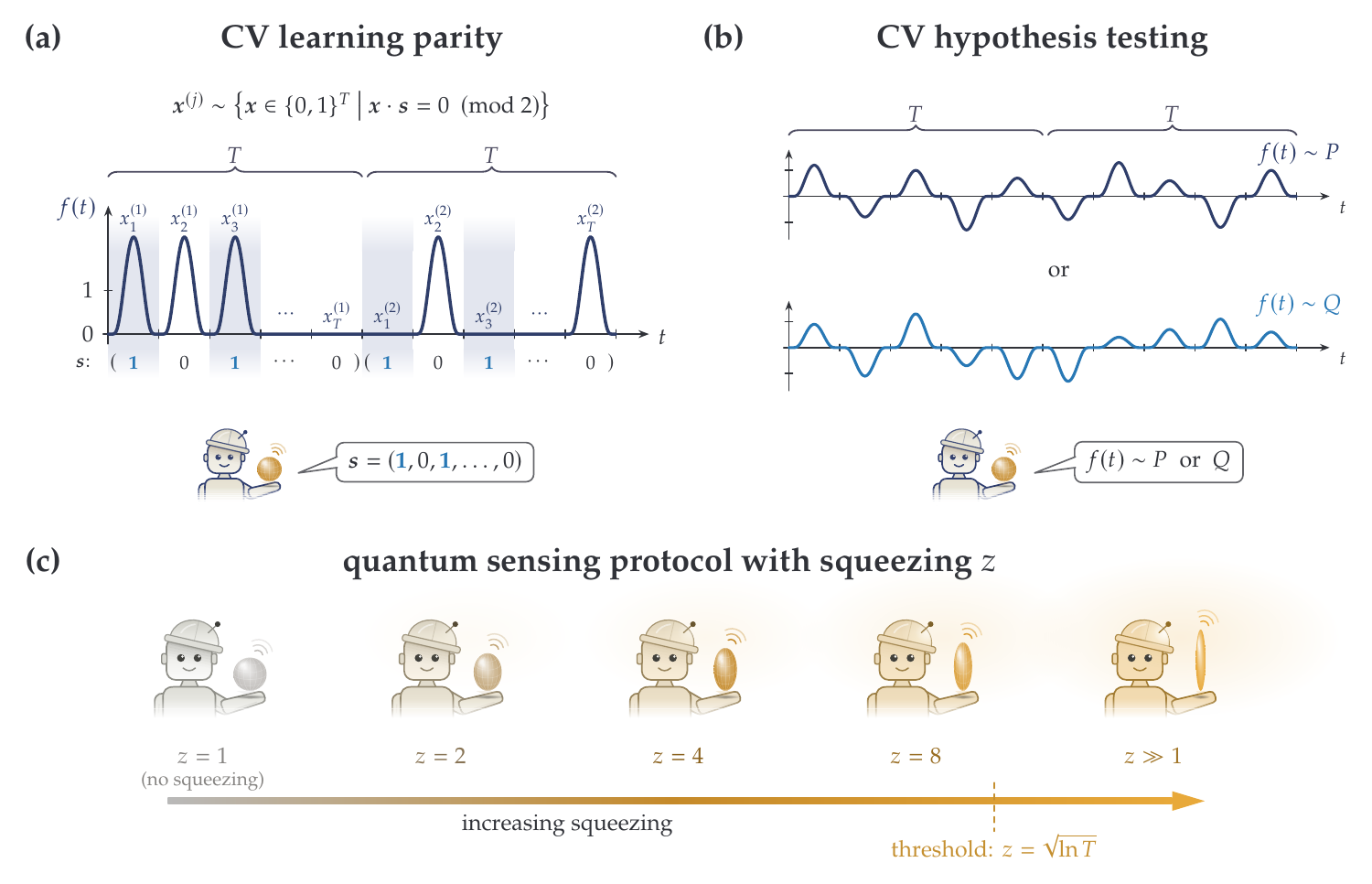}
    \caption{The two sensing problems considered in this work are (a) the CV Learning Parity Problem and (b) the CV Hypothesis Testing Problem, whose signals are depicted here. In (a), the coefficients within each block of duration $T$ form a random bit string $\mathbf{x}$ satisfying the parity constraint $\mathbf{x}\cdot\mathbf{s}\equiv0\pmod 2$, and the goal is to learn $\mathbf{s}$. In (b), the signs of the coefficients within each block are sampled according to one of two distributions, $P$ or $Q$, and the goal is to determine which source generates the signal. For both problems, the minimum sensing time scales exponentially in $T$ for $z=o(\sqrt{\log T})$, while squeezing above the threshold $z=\omega(\sqrt{\log T})$ enables polynomial sensing time.}
    \label{fig:fig2}
\end{figure}

\subsubsection{CV Learning Parity Problem}

The first sensing task we consider involves learning a simple pattern encoded in the temporal correlations of the signal $f(t)$. Specifically, we divide the coefficients $m_i$ into consecutive blocks of $T$ time intervals. Within each block, the coefficients $m_i$ are random bits whose joint distribution is constrained by a \emph{hidden parity relation}, determined by an unknown bitstring $\mathbf{s}\in\{0,1\}^T$, which is precisely the hidden pattern we wish to learn. Thus, the information about $\mathbf{s}$ is encoded in the joint statistics of the pulse coefficients $m_i$ within each block of duration $T$. The sensing task is to learn this hidden string $\mathbf{s}$ from the resulting time-varying signal. We state the task precisely as follows.

\begin{mainproblem}[(CV Learning Parity Problem)]
Let $T\in\mathbb{N}^+$ be the pattern size and let $\mathbf{s}\in\{0,1\}^T$ be an unknown bitstring. For each $j\geq1$, independently sample a bitstring
$\mathbf{x}^{(j)}=(x_1^{(j)},\ldots,x_T^{(j)})$
uniformly from
\[
    \left\{
    \mathbf{x}\in\{0,1\}^T
    \,\middle|\,
    \sum_{i=1}^T x_i s_i\equiv0\;(\mathrm{mod}\,2)
    \right\}.
\]
The coefficients of the time-varying signal $f(t)$ in Eq.~\eqref{eq_def_signal} are then given by
$(m_{(j-1)T+1},\ldots,m_{jT})=\mathbf{x}^{(j)}$.

We are given access to a single-mode sensor evolving under the Hamiltonian $\hat H(t)=f(t)\hat p$ from $t=0$ to $t=T_{\mathrm{sensing}}$, and we can control and measure it using a Gaussian sensing protocol with squeezing bounded by $z$. The goal is to output a guess $\widehat{\mathbf{s}}\in\{0,1\}^T$ such that $\widehat{\mathbf{s}}=\mathbf{s}$ with probability $\ge 2/3$, using sensing time $T_{\mathrm{sensing}}$ as short as possible.

\end{mainproblem}

The CV Learning Parity Problem is illustrated pictorially in \Cref{fig:fig2}(a).

\subsubsection{CV Hypothesis Testing Problem}
The second sensing task we consider is to distinguish between two possible sources of the signal $f(t)$. More precisely, in each block of duration $T$, the signs of the coefficients $m_i$ are drawn from one of two known distributions, $P$ or $Q$, and the same source generates the entire signal. The goal is to determine whether the observed time-varying signal is generated by $P$ or by $Q$.

We are interested in the case where the two sources look identical when only a small part of the signal is observed, so that they can be distinguished only through correlations involving almost the entire block. At the same time, we require that the two distributions are separated by at least an inverse polynomial in $T$, so that the difficulty of the distinguishing task does not simply arise from $P$ and $Q$ being extremely close to each other. We state the task precisely as follows.

\begin{mainproblem}[(CV Hypothesis Testing Problem)]
Let $T\in\mathbb{N}^+$ be the pattern size and let $c\in[T]$ satisfy $c=O(T^{0.99})$.
Let $P$ and $Q$ be two known probability distributions over $\{0,1\}^T$ whose marginals are identical on any choice of at most $T-c$ bits.\footnote{We make this assumption so that observing any such subset of bits does not allow one to distinguish between the two distributions.}
At the same time, assume that their total variation distance satisfies $\operatorname{TV}(P,Q)\ge 1/\operatorname{poly}(T)$.\footnote{We make this assumption so that the distinguishing problem is not hard simply because $P$ and $Q$ are extremely close.}

Let $D\in\{P,Q\}$ be the unknown source generating the signal. For each $j\geq1$, independently sample a sign pattern $(x_1^{(j)},\ldots,x_T^{(j)})\sim D$,
together with independent magnitudes
$b_1^{(j)},\ldots,b_T^{(j)}\sim\mathrm{Unif}([1/2,3/2])$.
The coefficients of the time-varying signal $f(t)$ in Eq.~\eqref{eq_def_signal} are then given by
\[
(m_{(j-1)T+1},\ldots,m_{jT})
=
\bigl(
(-1)^{x_1^{(j)}}b_1^{(j)},
\ldots,
(-1)^{x_T^{(j)}}b_T^{(j)}
\bigr).
\]
We are given access to a single-mode sensor evolving under the Hamiltonian $\hat H(t)=f(t)\hat p$ from $t=0$ to $t=T_{\mathrm{sensing}}$, and we can control and measure it using a Gaussian sensing protocol with squeezing bounded by $z$. The goal is to output a decision $\widehat D\in\{P,Q\}$ such that $\widehat D=D$ with probability $\ge 2/3$, using sensing time $T_{\mathrm{sensing}}$ as short as possible.
\end{mainproblem}
The CV Hypothesis Testing Problem is illustrated pictorially in \Cref{fig:fig2}(b).

\subsection{Exponential separations with a single squeezed mode}

For both sensing tasks described above, we rigorously establish exponential separations in sensing time by tuning only the amount of squeezing available for preparing the single-mode sensor.
Specifically, for squeezing $z=\omega(\sqrt{\log T})$, there exists a very simple Gaussian sensing protocol that solves these problems in polynomial sensing time, whereas for $z=o(\sqrt{\log T})$, any Gaussian sensing protocol---no matter how sophisticated---must use exponential sensing time in $T$. We summarize these results in the following theorem:

\begin{maintheorem}[(Exponential separations in sensing time-varying signals with a single squeezed mode)]
\label{thm:main}
For both the CV Learning Parity Problem and the CV Hypothesis Testing Problem with pattern size $T$, the following hold:
\begin{enumerate}
    \item if $z=\omega(\sqrt{\log T})$, there exists a Gaussian sensing protocol with squeezing bounded by $z$ that solves the problem with sensing time $\operatorname{poly}(T)$; and
    \item if $z=o(\sqrt{\log T})$, any Gaussian sensing protocol with squeezing bounded by $z$ that solves either problem must use sensing time $\operatorname{exp}(T)$.\footnote{More precisely, if $z=o(\sqrt{\log T})$, the sensing time must scale at least as $\exp(T^{1-o(1)})$. More generally, for arbitrary squeezing $z$, any such Gaussian sensing protocol must use sensing time at least $\exp\left(\Omega\left(T e^{-O(z^2)}\right)\right)$, while there exists a Gaussian sensing protocol with sensing time $\exp\left(O\left(T e^{-\Omega(z^2)}\right)\right)$, up to polynomial factors. See the Appendix for explicit bounds without asymptotic notation.}
\end{enumerate}
\end{maintheorem}
The full proof of Theorem~\ref{thm:main} is given in the Appendix; see Corollary~\ref{cor:cv_parity_advantage} for the CV Learning Parity Problem and Corollary~\ref{cor:cv_hypothesis_testing_advantage} for the CV Hypothesis Testing Problem. Here, we outline the main ideas.
\subsubsection{Proof outline}
\textbf{\textit{A simple upper bound.}} To prove the polynomial sensing-time scaling for $z=\omega(\sqrt{\log T})$ in Theorem~\ref{thm:main}, it suffices to consider perhaps the simplest Gaussian sensing protocol one can think of.
At the beginning of each pulse (i.e.~at each integer time), prepare the sensor in a zero-mean pure Gaussian state with squeezing $z$ along the $x$ quadrature, let it evolve for one unit of time, and perform an $x$-homodyne measurement.
Since $\int_{i-1}^{i}f(t)\,\mathrm{d}t=m_i$, the measurement outcome is simply the corresponding pulse coefficient $m_i$ plus Gaussian noise with standard deviation of order $1/z$.
Hence, the corresponding bit can be recovered with error probability $e^{-\Omega(z^2)}$.
When $z=\omega(\sqrt{\log T})$, this error is small enough that polynomially many blocks can be read reliably, after which classical post-processing solves either sensing task, giving total sensing time $\operatorname{poly}(T)$. For the full proofs of these upper bounds, see Theorem~\ref{thm_upp1} for the CV Learning Parity Problem and Theorem~\ref{thm_upper_bound_hyp_testing} for the CV Hypothesis Testing Problem.

Note, however, that the simple protocol above, which reads each pulse separately, need not be optimal. In fact, one can provably do better at $z=1$ for the CV Learning Parity Problem: let the sensor evolve coherently across an entire block of duration $T$, interleaving the evolution only with phase-space inversions, and performing a single homodyne measurement at the end of the block. With this strategy, one obtains a provable strict improvement; see \Cref{sec:coherent_parity_readout}. Thus, although the above ``pulse-by-pulse readout'' protocol is perhaps the most natural one, more sophisticated protocols can probe the signal's temporal correlations more effectively. To prove that low squeezing nevertheless forces exponential sensing time, one must therefore rule out not only pulse-by-pulse readout but every coherent and adaptive sensing strategy. In the next paragraph we show how to prove such a lower bound in this fully continuous-time setting.

\medskip

\textbf{\textit{Continuous-time lower bound.}} The main difficulty in proving Theorem~\ref{thm:main} is therefore establishing the sensing-time lower bound in the regime $z=o(\sqrt{\log T})$ for every Gaussian sensing protocol with that squeezing bound. A general Gaussian sensing protocol is a sequence of prepare-evolve-measure experiments. In each experiment, the learner prepares the sensor, lets it evolve under the signal while interleaving arbitrary squeezing-free Gaussian unitaries, and finally performs a Gaussian measurement. The learner may choose when to prepare the state, how long to let it interact with the signal, and when to apply Gaussian unitaries and measurements, with all these choices made adaptively.

The first part of this hardness proof is conceptually standard~\cite{chen2021exponential,huang2022quantum,chen2022complexitynisq,Chen_2024,chen2025efficientpaulichannelestimation,allen2025quantumcomputingenhancedsensing}. We start by reducing both problems to a distinguishing task. For the CV Learning Parity Problem, this amounts to distinguishing whether the hidden pattern is $\mathbf{0}$ or a random nonzero string $\bar{\mathbf{s}}$, revealed only after the sensing procedure, while the CV Hypothesis Testing Problem is already a distinguishing task.
An adaptive sensing protocol can then be represented by a \emph{learning tree}~\cite{chen2021exponential,huang2022quantum}: each node of the tree records the \emph{transcript} accumulated up to that point, namely the previous measurement outcomes and adaptive choices, and specifies the next experiment chosen by the learner, while its branches correspond to the possible measurement outcomes. Under the two hypotheses, the protocol induces two different distributions over the full transcripts, and successful discrimination requires their TVD to be bounded below by a constant. Thus, it is enough to show that the TVD remains small for any allowed adaptive protocol unless the sensing time is sufficiently long.

We first perform a change of clock that maps the continuous pulsed signal to a piecewise-constant one, without changing the number of $T$-duration blocks reached by the protocol.
We then show that any single prepare-evolve-measure experiment can be simulated by sampling from a two-dimensional Gaussian distribution whose mean is linear in the signal coefficients and whose covariance is fixed by $z$. It therefore suffices to establish hardness in this stronger statistical model.

The remaining complication is the timing: each experiment's starting and ending times are arbitrary and can depend adaptively on all previous measurement outcomes.
Thus, a single experiment may lie entirely inside a pulse, span several pulses, or even cross several blocks of duration $T$.
To handle this continuous-time freedom, we establish a useful reduction: any experiment can be simulated by at most five canonical subexperiments whose starting and ending times belong to a restricted set, at the cost of only a constant-factor increase in the effective squeezing.
These subexperiments naturally fall into three different temporal scales, which leads us to introduce a suitable learning-tree framework for this continuous-time setting that we call the \emph{hierarchical tree representation}, illustrated in \Cref{fig:hierarchical_tree_main}.

At the coarsest scale, a \emph{full-time tree} describes experiments whose durations are multiples of $T$.
A one-block experiment, highlighted in yellow in \Cref{fig:hierarchical_tree_main}, is resolved by a finer \emph{$T$-time tree}, where experiments begin and end at integer times within the block. A one-unit-time experiment, highlighted in blue in \Cref{fig:hierarchical_tree_main}, is resolved further by a \emph{unit-time tree}, where experiments may begin and end at arbitrary times within that interval.

The key difficulty is to control how the TVD between the two transcript distributions grows as the learner adapts across these three temporal scales.
We first control the finest, unit-time trees, showing that multiple measurements, even chosen adaptively at arbitrary times within a single interval, cannot outperform a single unit-time experiment. We then use this to analyze the $T$-time tree, combining the contributions of the different experiments along each branch to show that the TVD between the corresponding $T$-time tree transcript distributions (averaged over the random nonzero parity string for the CV Learning Parity Problem), denoted by $C$, is exponentially small in $T$ for low squeezing.
Finally, by analyzing the full-time tree, we show that the TVD between the full transcript distributions is at most the number of blocks sensed times the one-block bound $C$. Since successful discrimination requires the full TVD to be constant, this implies an exponential number of blocks when $z=o(\sqrt{\log T})$. For the full proofs of these lower bounds, see \Cref{sec:lower_bound_minimum_sensing_time}, and specifically Theorem~\ref{thm_lower_bound_learning_parity_gaussian} for the CV Learning Parity Problem and Theorem~\ref{thm_lower_bound_ht_gaussian} for the CV Hypothesis Testing Problem.
\begin{figure} 
    \centering
    \includegraphics[width=0.8\linewidth]{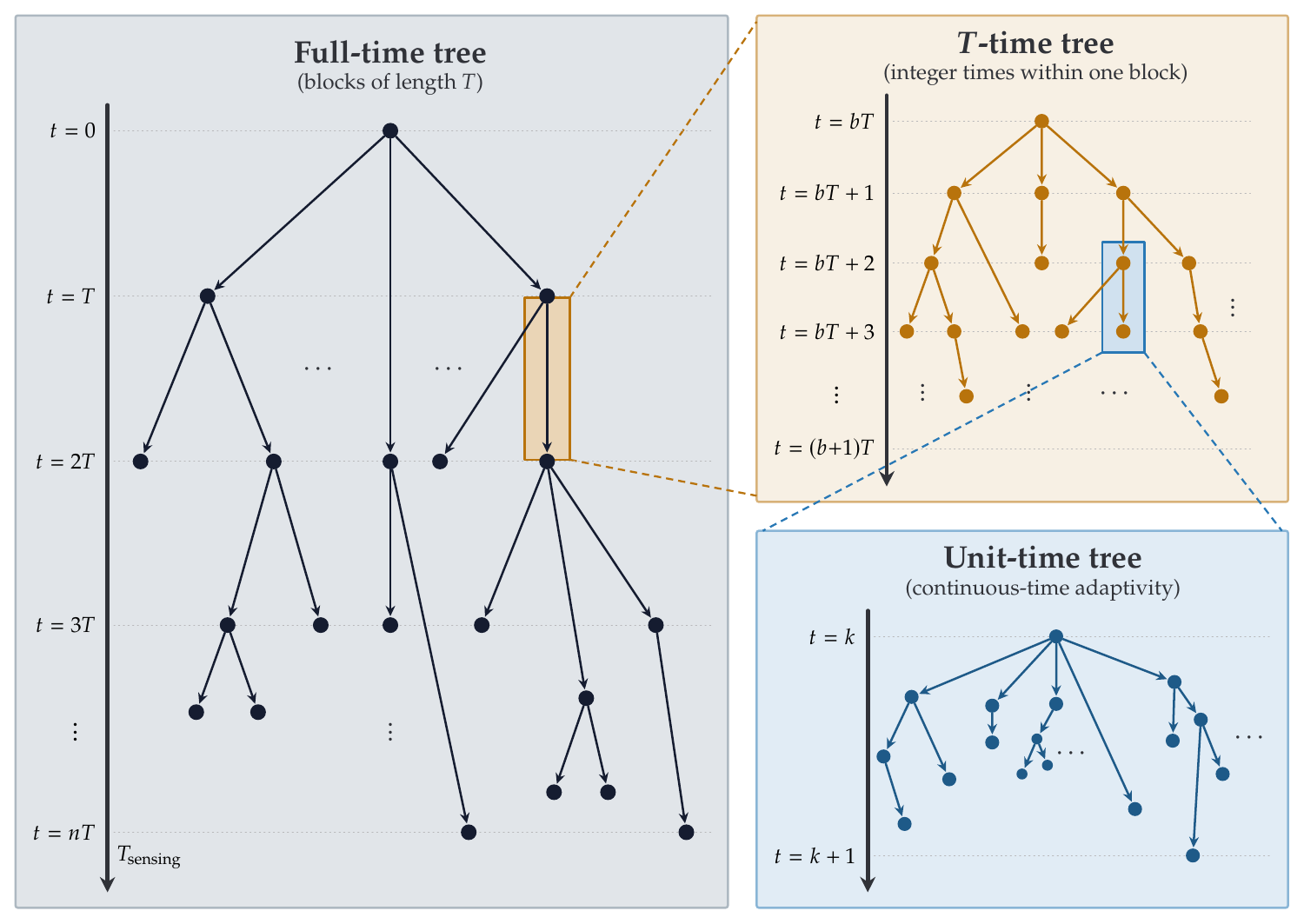}
    \caption{Hierarchical tree representation of an adaptive Gaussian sensing protocol. Any such protocol can be reduced to canonical subexperiments that naturally fall into three temporal scales. The \emph{full-time tree} describes the protocol at the scale of blocks of duration $T$, with experiments beginning and ending at multiples of $T$ and possibly spanning several blocks. A one-block experiment, highlighted in yellow, is resolved by a finer \emph{$T$-time tree}, where experiments begin and end at integer times within the block. A one-unit-time experiment, highlighted in blue, is further resolved by a \emph{unit-time tree}, which captures the remaining continuous-time adaptivity. The three trees thus describe the same adaptive protocol at progressively finer temporal scales.}
    \label{fig:hierarchical_tree_main}
\end{figure}

\section{Discussion}
In this work, we show how the simple quantum resource of squeezing can lead to exponential separations in sensing time-varying signals.
Crucially, this separation persists for fully adaptive Gaussian sensing protocols, even when the learner is free to prepare, control, and measure the sensor at arbitrary times, possibly chosen adaptively, during the signal evolution.
To handle this puzzling continuous-time setting, we introduce the \emph{hierarchical tree representation}, which allows us to lower-bound the minimum sensing time.
More specifically, in the high-squeezing regime $z=\omega(\sqrt{\log T})$, the sensing tasks defined above can be solved in polynomial sensing time in $T$, whereas in the low-squeezing regime $z=o(\sqrt{\log T})$, any Gaussian sensing protocol must require exponential sensing time in $T$.

One can easily observe that the same separation can also be phrased in terms of the amplitude of the signal instead of the squeezing of the sensor. Indeed, multiplying $f(t)$ by a factor $B$ effectively amounts to increasing the squeezing by the same factor $B$ in our original problem. Thus, at fixed constant squeezing, signals with amplitude $B=o(\sqrt{\log T})$ still require exponential sensing time in $T$, whereas for $B=\omega(\sqrt{\log T})$ the problems can be solved in polynomial sensing time. This provides an equivalent way of interpreting our results in terms of the sensitivity of the sensor rather than in terms of squeezing.

The exponential separations we prove stem from the hardness of learning patterns hidden in the temporal correlations of the signal, a feature unique to time-varying signals. This mechanism for exponential separations may apply more broadly than the specific sensing tasks considered here. For instance, sensing multiple modes in parallel can reduce the measurement noise in the same way as increasing the squeezing of a single mode, and may therefore lead to a similar exponential separation based on the same mechanism.
In fact, suppose that $n$ independent modes, each with squeezing $z$, interact simultaneously and identically with the same external signal. Measuring the modes independently along the squeezed quadrature gives Gaussian outcomes with standard deviation of order $1/z$, and averaging the $n$ outcomes reduces the standard deviation to order $1/(\sqrt{n}z)$. This is equivalent to replacing $z$ by $\sqrt{n}z$.
Therefore, for constant squeezing, $n=\omega(\log T)$ independent modes are sufficient to reach the high-squeezing regime, whereas $n=o(\log T)$ corresponds to the low-squeezing regime within this parallel sensing strategy.
This suggests a possible unified treatment in which both squeezing and the number of modes are regarded as sensing resources, and further raises interesting open questions regarding possible advantages from entanglement between multiple modes. Note, though, that our lower-bound proof does not yet extend to multiple modes, and specifically to the possibility of adaptivity between separate modes' measurements.

The Gaussian sensing protocols considered here have a prepare-evolve-measure structure: after each Gaussian measurement, the sensor state is discarded and a fresh Gaussian state is prepared. A natural direction is therefore to go beyond this setting by allowing non-destructive Gaussian measurements implementable using squeezing-free Gaussian operations, with the possibility of acting adaptively on the resulting post-measurement state. It would also be interesting to understand how the sensing-time requirements could change when other non-Gaussian yet experimentally feasible measurements, such as photon-counting, are allowed.

Beyond these theoretical extensions, for the exponential separations to be meaningful in practice, the sensing setting should remain as close as possible to realistic experimental conditions.
In this work, we focus on a single mode with squeezing, since in some scenarios engineering many identical bosonic modes may be difficult or undesirable.
For example, LIGO operates as a single mode, suppressing unwanted higher order modes that would otherwise degrade the measurement~\cite{ligo2015advanced,fricke2012dc}. Finally, the sensing tasks considered here are deliberately stylized. A natural open question is to start from experimentally relevant signals---for instance, gravitational-wave signals~\cite{abbott2016observation}---and determine whether one can establish analogous exponential separations.

\phantomsection
\section*{Acknowledgments} 
\addcontentsline{toc}{section}{Acknowledgments}
We thank Francesco Albarelli, Richard Allen, Yanbei Chen, Soonwon Choi, Su Durecki, Vittorio Giovannetti, Ludovico Lami, and John Preskill for helpful discussions. 
We are especially grateful to Amir Safavi-Naeini for initial discussions that led to this work. We are also particularly grateful to Jordan Cotler, Ishaan Kannan, and Ruohan Shen for providing extensive feedback on a final version of this manuscript. F.A.M.\ acknowledges financial support from the European Union (ERC StG ETQO, Grant Agreement no.\ 101165230). N.M. acknowledges financial support by the Air Force Office of Scientific Research under award number FA9550-23-F-0014. F.A.M.\ thanks the California Institute of Technology for its hospitality on several research visits during his PhD at Scuola Normale Superiore, and N.M.\ thanks Scuola Normale Superiore for its hospitality, part of this work having been carried out during these research visits. H.H. acknowledges support from the Broadcom Innovation Fund. H.H. acknowledges support from the U.S. Department of Energy, Office of Science, National Quantum Information Science Research Centers, Quantum Systems Accelerator. F.A.M., H.Z., S.C., and H.H. acknowledge funding provided by the Institute for Quantum Information and Matter, an NSF Physics Frontiers Center (NSF Grant PHY-2317110).

\medskip
\noindent\textit{AI statement.}--- We did the main research part of this paper during summer 2025 without the use of AI. ChatGPT and Gemini were used later on in 2026 to make minor refinements to our human-made results at the end of the project. Most notably, ChatGPT 5.6 and 6 Pro were used to extend the range of the parameter \(c\) from \(c=O(1)\) to \(c=O(T^{0.99})\) in the CV Hypothesis Testing Problem (by adapting our previous human-made proof), and to design the improved coherent protocol for \(z=1\). Claude was used to help make the figures. The authors take full responsibility for the content of the work.
\bibliography{biblio}
\bibliographystyle{unsrtnat}

\clearpage
\appendix

\appendixtableofcontents
\label{toc}

\vspace{1em}
The appendix of this paper is organized as follows. Section~\ref{sec:related_work} discusses the relation of our results to previous work. Section~\ref{sec_prel} then gives preliminaries on CV systems, followed by relevant properties of Gaussian probability distributions that we use in our proofs (Section~\ref{sec:prelim_gaussian_prob}). To conclude the background, Section~\ref{sec:exp_platforms} details how different experimental platforms represent our theoretical model of a time-dependent signal $f(t)$ coupling to a target quadrature.

With this foundation established, Section~\ref{sec:gaussian_sensing_model} defines Gaussian sensing protocols with bounded squeezing $z$. We then define the two central tasks of this work: the CV Learning Parity problem (Section~\ref{sec_cv_parity}) and the CV Hypothesis Testing problem (Section~\ref{sec_hyp_test}). Finally, our main technical results are presented in Sections~\ref{sec:lower_bound_minimum_sensing_time}, ~\ref{sec:upper_bound}, and~\ref{sec_cv_prob_hyph}. Specifically, in Sections~\ref{sec:lower_bound_minimum_sensing_time} we derive the lower bound on the minimum sensing time, and in Section~\ref{sec:upper_bound}-~\ref{sec_cv_prob_hyph}, we provide explicit sensing protocols that solve these sensing problems.

\newpage

\section{Relation to previous work}
\label{sec:related_work}

Beyond the quantum metrology literature on parameter~\cite{giovannetti2011advances,pezze2018quantum} and function~\cite{kura2020standard} estimation, a recent line of work has begun to connect exponential quantum learning advantages with sensing of signals~\cite{oh2024entanglement,cotler2026quantum,kannan2026exponential,prabhu2026exponential}.

Ref.~\cite{oh2024entanglement} proved an exponential entanglement-assisted advantage for learning bosonic random displacement channels, a setting closely related to sensing since signals can naturally induce phase-space displacements. This advantage was subsequently demonstrated on a scalable photonic platform~\cite{Liu_2025}.

More directly related to sensing, Cotler, Danielson, and Kannan~\cite{cotler2026quantum} studied quantum advantages for learning properties of signals. Among their results, they established exponential separations for learning signal properties, including a separation based on squeezing alone. The latter is particularly close to our work, although it is formulated in a prescribed interrogation setting rather than the fully continuous-time setting considered here.

More recently, Kannan et al.~\cite{kannan2026exponential} showed that coupling a bosonic sensor to a single controllable qubit can give exponential advantages for learning properties of time-varying signals, and demonstrated such advantages experimentally on a superconducting platform. In a related direction, Prabhu et al.~\cite{prabhu2026exponential} showed exponential advantages from entangled sensors for learning correlations among stochastic parameters.

These works show that quantum sensing can exhibit advantages far beyond the usual improvements in parameter-estimation precision. A key distinction is that our result is not a channel-query separation. Rather, our work asks whether such an exponential separation survives when the learner is free to interact with the sensor at arbitrary times while the signal continuously evolves. This distinction can matter. Indeed, Prabhu et al.~\cite[Appendix~D.5]{prabhu2026exponential} give an example in which optimizing the interrogation time removes an exponential separation in the number of experimental repetitions. Here, we show that for our sensing tasks, the exponential separation persists even after allowing arbitrary adaptive choices of when to prepare, control, and measure the sensor.

\newpage
\section{Preliminaries on continuous variable systems}\label{sec_prel}
In this section, we provide a concise overview of quantum information with \emph{continuous variable} (CV) systems; for further details, we refer to Ref.~\cite{BUCCO,adesso2014continuous}.

\subsection{The $n$-mode Hilbert space and its quadrature operators}
Let us start by offering an intuitive explanation of what a continuous variable system is from a quantum information perspective. In quantum information, a familiar concept is that of a \emph{qu$d$it} --- a quantum system with a $d$-dimensional Hilbert space. Specifically, the Hilbert space of a qu$d$it is defined as the span of the computational basis state vectors $\{\ket{0},\ket{1},\ldots,\ket{d-1}\}$:
\bb
    \HH_{d}\coloneqq\mathrm{Span}\{\ket{0},\ket{1},\ldots,\ket{d-1}\}\,.
\ee
In continuous-variable systems, the notion of a qu$d$it is replaced by the concept of a \emph{mode}. A mode can be thought of as a qu$d$it in the limit $d \to \infty$, that is, a quantum system with an infinite-dimensional Hilbert space of the form
\bb
    \HH_{\infty}\coloneqq\mathrm{Span}\left\{\ket{n}:\,\, n\in\mathbb{N}\right\}\,.
\ee
By definition, a \emph{continuous-variable system} is a quantum system that consists of $n$ modes and is described by the tensor product space $\HH_{\infty}^{\otimes n}$. In this sense, a continuous variable system of $n$ modes can be intuitively viewed as $n$ qu$d$its with $d=\infty$. The computational basis state vector $\ket{n}$ of the Hilbert space of a single mode is called the \emph{$n$-th Fock state}, and in quantum optics it represents the quantum state of light with $n$ photons. In particular, the state vector $\ket{0}$ is the \emph{vacuum}, i.e.~the state with zero photons. 

The \emph{annihilation operator}, denoted by $a$, is defined in the Fock basis as
\bb
    a\coloneqq\sum_{m=1}^\infty \sqrt{m}\,\ketbraa{m-1}{m}\,.
\ee
Its adjoint $a^\dagger$ is known as the \emph{creation operator}. These two operators satisfy the well-known \emph{canonical commutation relation}
\bb
    [a,a^\dagger]=\hat{\mathbb{1}}\,,
\ee
where $[A,B]\coloneqq AB-BA$ denotes the commutator and $\hat{\mathbb{1}}$ is the identity operator. Using the orthonormality of the Fock states, one can show that the $m$-th Fock state can be obtained by applying the creation operator $m$ times to the vacuum state vector:
\bb
    \ket{m}\coloneqq \frac{(a^\dagger)^m}{\sqrt{m!}}\ket{0}.
\ee
The operator $a^\dagger a$ is known as the \emph{single-mode photon number operator}, and it is diagonal in the Fock basis
\bb
    a^\dagger a=\sum_{m=0}^\infty m \ketbra{m}\,.
\ee
In this sense, $\ket{m}$ represents the state with $m$ photons.

The well-known \emph{position} and \emph{momentum} operators~\cite{BUCCO}, denoted by $\hat{x}$ and $\hat{p}$ respectively, can be defined in terms of the annihilation and creation operators as
\bb
    \hat{x}&\coloneqq \frac{a+a^\dagger}{\sqrt{2}}\,,\quad\,
    \hat{p}\coloneqq \frac{a-a^\dagger}{\sqrt{2}i}\,.
\ee
In terms of these operators, the canonical commutation relation takes the form $[\hat{x},\hat{p}]=i\hat{\mathbb{1}}$.  

From a mathematical standpoint, the Hilbert space associated with a single mode is identified with the space $L^2(\mathbb{R})$~\cite{BUCCO} which comprises all square-integrable complex-valued functions over $\mathbb{R}$. More generally, the Hilbert space of an $n$-mode continuous-variable system is given by $L^2(\mathbb{R}^n)$ --- that is, the space of square-integrable, complex-valued functions over $\mathbb{R}^n$. The set of quantum states of such a system is denoted by $\pazocal{D}(L^2(\mathbb{R}^n))$. Notably, it holds that $L^2(\mathbb{R}^n) = L^2(\mathbb{R})^{\otimes n}$, so the Hilbert space of an $n$-mode system can be naturally understood as the $n$-fold tensor product of the Hilbert space of a single mode. On this Hilbert space, one can define the position and momentum operators of the $j$-th mode, denoted respectively by $\hat{x}_j$ and $\hat{p}_j$. 
By introducing the vector of operators
\bb
    \hat{\textbf{R}}\coloneqq (\hat{x}_1,\hat{p}_1,\dots,\hat{x}_n,\hat{p}_n)^{\intercal}=(\hat{R}_1,\hat{R}_2,\dots,\hat{R}_{2n-1},\hat{R}_{2n})^{\intercal}\,
\ee
dubbed the \emph{quadrature operator vector}, the canonical commutation relations can be expressed as 
\bb
    [\hat{R}_k,\hat{R}_l]=i\,(\Omega_n)_{k,l}\hat{\mathbb{1}}\qquad\forall\,k,l\in[2n]\,,
    \label{comm_rel_quadrature}
\ee
where $\Omega_n\coloneqq \bigoplus_{i=1}^n\Omega_1$ with $\Omega_1\coloneqq \left(\begin{matrix}0&1\\-1&0\end{matrix}\right)$ being the so-called \emph{symplectic form}, and $\hat{\mathbb{1}}$ is the identity operator over $L^2(\mathbb R^n)$. The relation in~\eqref{comm_rel_quadrature} is usually expressed in the continuous-variable literature~\cite{BUCCO} in vectorial notation as $[\hat{\textbf{R}},\hat{\textbf{R}}^{\intercal}]=i\,\Omega_n\hat{\mathbb{1}}$.

\subsection{First moments, covariance matrices, and symplectic matrices}
First moments and covariance matrices play a crucial role in characterising CV systems. The \emph{first moment} $\mathbf{m}(\rho)$ of an $n$-mode quantum state $\rho$ is a $2n$-dimensional vector defined as 
\bb
    \mathbf{m}(\rho)\coloneqq\left(\Tr\!\left[\hat{R}_1\,\rho\right],\Tr\!\left[\hat{R}_2\,\rho\right],\ldots,\Tr\!\left[\hat{R}_{2n}\,\rho\right]\right)\,,
\ee
which can be written in vectorial notation as $\mathbf{m}(\rho)=\Tr\!\left[{\hat{\textbf{R}}}\,\rho\right]$.
Additionally, the \emph{covariance matrix} of $\rho$ is a $2n\times 2n$ matrix $V\!(\rho)$ with elements
\bb
	[V\!(\rho)]_{k,l}\coloneqq \Tr\!\left[\left\{{\hat{R}_k-m_k(\rho)\hat{\mathbb{1}},\hat{R}_l-m_l(\rho)\hat{\mathbb{1}}}\right\}\rho\right]= \Tr\!\left[\left\{\hat{R}_k,\hat{R}_l\right\}\rho\right]-2m_k(\rho)m_l(\rho) \, ,
\ee
for each $k,l\in[2n]$, where $\{\hat{A},\hat{B}\}\coloneqq \hat{A}\hat{B}+\hat{B}\hat{A}$ is the anti-commutator. In vectorial notation, this reads as
\bb
	V\!(\rho)&=\Tr\!\left[\left\{(\hat{\textbf{R}}-\textbf{m}(\rho)\hat{\mathbb{1}}),(\hat{\textbf{R}}-\textbf{m}(\rho)\hat{\mathbb{1}})^{\intercal}\right\}\rho\right]=  \Tr\!\left[\left\{{\hat{\textbf{R}}},{\hat{\textbf{R}}}^{\intercal}\right\}\rho\right]-2\textbf{m}(\rho)\textbf{m}(\rho)^\intercal \, .
\ee   
For example, note that $V_{11}(\rho)=2\mathrm{Var}_\rho(\hat{x}_1)$, where $\mathrm{Var}_\rho(\hat{x}_1)\coloneqq \Tr[\hat{x}_1^2\rho]-(\Tr[\hat{x}_1\rho])^2$ is the variance of the position operator on $\rho$. Analogously, it holds that $V_{22}(\rho)=2\mathrm{Var}_\rho(\hat{p}_1)$, where $\mathrm{Var}_\rho(\hat{p}_1)\coloneqq \Tr[\hat{p}_1^2\rho]-(\Tr[\hat{p}_1\rho])^2$ is the variance of the momentum operator on $\rho$. Notably, any covariance matrix $V\!(\rho)$ satisfies the matrix inequality~\cite{BUCCO}
\bb
V\!(\rho)+i\Omega_n\ge0\,,
\ee
known as the \emph{uncertainty relation}. As a consequence, since $\Omega_n$ is skew-symmetric, any covariance matrix $V\!(\rho)$ is positive semi-definite on $\mathbb{R}^{2n}$. Conversely, for any symmetric $W\in\mathbb{R}^{2n,2n}$ such that $W+i\Omega_n\ge0$ there exists an $n$-mode state $\rho$ with covariance matrix $V\!(\rho)=W$~\cite{BUCCO}.

Any covariance matrix $V$ can be written in the so-called \emph{Williamson decomposition} as follows~\cite{BUCCO}: there exists a \emph{symplectic} matrix $S$ and $n$ unique real numbers $d_1\ge d_2\ge \ldots\ge d_n\ge 1$ --- called the \emph{symplectic eigenvalues} of $V$ --- such that
\bb\label{will_eq_def}
V=SDS^{\intercal},
\ee 
where $D\coloneqq \operatorname{diag}(d_1, d_1, \ldots, d_n, d_n)$. By definition, a matrix $S$ is said to be \emph{symplectic} if it satisfies $S \Omega_n S^\intercal =\Omega_n$. Symplectic matrices satisfy the following properties:
\begin{itemize}
    \item The product of symplectic matrices is symplectic;
    \item The transpose of a symplectic matrix is symplectic;
    \item The inverse of a symplectic matrix is symplectic;
    \item The determinant of a symplectic matrix is equal to one;
    \item The singular value decomposition of a symplectic matrix $S$ can be written as follows: there exist suitable symplectic orthogonal matrices $O_1,O_2$, and suitable real numbers $z_1,\ldots,z_n\ge 1$ such that
    \bb\label{Euler_dec2}
        S=O_1ZO_2\,,
    \ee
    where $Z$ is a diagonal matrix --- known as the \emph{squeezing matrix} --- defined by
    \bb\label{sq_matttt}
    Z\coloneqq \bigoplus_{j=1}^n\left(\begin{matrix}z_j^{-1}&0\\0&z_j\end{matrix}\right)\,. 
    \ee
    The decomposition in Eq.~\eqref{Euler_dec2} is usually referred to as the \emph{Euler decomposition}.
\end{itemize} 

\subsection{Gaussian unitaries}
Among all CV systems, \emph{Gaussian systems}~\cite{BUCCO} are arguably among the most important ones. Indeed, on the one hand Gaussian systems are the most common class of systems arising both in Nature and in quantum optics laboratories, while on the other they are relatively simple to analyse mathematically. In this subsection we are going to define Gaussian unitaries, while in the forthcoming two subsections we will define Gaussian states and Gaussian measurements.

Let us start with the definition of Gaussian unitaries. By definition, a \emph{Gaussian unitary} $G$ is a unitary that can be prepared by means of evolutions induced by quadratic Hamiltonians in the quadrature operator vector $\hat{\textbf{R}}$. The most general form of a Gaussian unitary is as follows~\cite{BUCCO}:  
\bb\label{gauss_unitary}
    G =  {D}_{\textbf{r}}U_S\,,  
\ee  
where 
\bb\label{disp_real}
{D}_{\textbf{r}}\coloneqq e^{-i\,\textbf{r}^\intercal\Omega_n \hat{\textbf{R}}}
\ee
is the so-called \emph{displacement operator}~\cite{BUCCO}, characterised by the displacement vector \( \textbf{r} \in \mathbb{R}^{2n} \), while $U_S$ corresponds to the \emph{symplectic Gaussian unitary}~\cite{BUCCO}, associated with a \emph{symplectic matrix} $S$. The displacement operator \( {D}_{\textbf{r}} \), associated with a vector \( \textbf{r} \in \mathbb{R}^{2n} \), acts on the quadrature operator vector as
\bb  
    {D}_{\textbf{r}}^\dagger \hat{\textbf{R}} {D}_{\textbf{r}} = \hat{\textbf{R}} + \textbf{r} \hat{\mathbb{1}} \,,  
\ee  
while it acts at the level of first moments and covariance matrices as
\bb
    \textbf{m}({D}_{\textbf{r}}\rho {D}_{\textbf{r}}^\dagger)&=\textbf{m}(\rho)+\textbf{r}\,,\\
    V({D}_{\textbf{r}}\rho {D}_{\textbf{r}}^\dagger)&=V(\rho)\,,
\ee
for any state $\rho$. For a given symplectic matrix \( S \), the symplectic Gaussian unitary $U_S$ acts on the quadrature operator vector as 
\bb\label{laaa}
    U_S^\dagger \hat{\textbf{R}} U_S = S \hat{\textbf{R}} \,,
\ee  
while it acts at the level of first moments and covariance matrices as
\bb\label{action_gauss_moment}
    \textbf{m}(U_S\rho U_S^\dagger)&=S\textbf{m}(\rho)\,,\\
    V(U_S\rho U_S^\dagger)&=SV(\rho)S^\intercal\,,
\ee
for any state $\rho$. Moreover, given two symplectic matrices $S_1$ and $S_2$, the associated symplectic Gaussian unitary satisfies $U_{S_1S_2}=U_{S_1}U_{S_2}$, up to an irrelevant global phase. Hence, the Euler decomposition reported in Eq.~\eqref{Euler_dec2} establishes that any symplectic Gaussian unitary $U_S$ can be decomposed as
\bb\label{Gauss_eul}
    U_S=U_{O_1}U_ZU_{O_2}\,,
\ee
where $O_1,O_2$ are orthogonal symplectic matrices and $Z$ is the squeezing matrix in Eq.~\eqref{sq_matttt}. Given a squeezing matrix $Z=\diag(z_1^{-1},z_1,\ldots,z_n^{-1},z_n)$ as in Eq.~\eqref{sq_matttt}, the Gaussian unitary operation $U_Z$ is called a \emph{squeezing unitary}. Moreover, this allows us to define the squeezing of a Gaussian unitary $G$.
\begin{Def}[(Squeezing of a Gaussian unitary)]
    The squeezing of a Gaussian unitary $G$ is defined as $z\coloneqq \|S\|_\infty=\|Z\|_\infty$.
\end{Def}

\subsubsection{Single-mode Gaussian unitaries}
\label{sec:gaussian_unitaries}
Thanks to Eq.~\eqref{Gauss_eul}, the most general \emph{single-mode} Gaussian unitary $G$ is of the form
\bb\label{most_general_Gaussian_un}
    G= D_{\textbf{r}} R_{\phi} S_{z} R_{\theta}\,,
\ee
where $\textbf{r}\in\mathbb{R}^2$, $\theta\in[0,2\pi)$, $\phi\in[0,2\pi)$, $z\ge 1$, and
\begin{itemize}
    \item $D_{\textbf{r}}$ is the single-mode displacement operator with $\textbf{r}$ being the 2-dimensional displacement vector;
    \item $R_\theta\coloneqq U_{O_\theta}$ is the so-called \emph{phase-space rotation} of angle $\theta$, with $O_\theta$ being the $2\times 2$ rotation matrix of angle $\theta$:
\bb\label{eq_o_theta}
    O_\theta\coloneqq \left(\begin{matrix}\cos\theta&-\sin\theta\\ \sin\theta&\cos\theta\end{matrix}\right)\,.
\ee
    \item $S_{z}$ is the single-mode squeezing unitary with squeezing $z$, defined as $S_z\coloneqq U_{Z}$ with $Z\coloneqq \left(\begin{matrix}z^{-1}&0\\ 0&z\end{matrix}\right)\,;$
\end{itemize} 

\subsection{Gaussian states}
Among all CV quantum states, \emph{Gaussian states}~\cite{BUCCO} are arguably among the most important ones. Indeed, on the one hand Gaussian states are the most common class of states arising in Nature, while on the other they are relatively simple to analyse mathematically.

Let us now proceed with the mathematical definition of Gaussian states. By definition, an $n$-mode state $\rho$ is said to be a \emph{Gaussian state} if it can be written as a Gibbs state or a ground state of a quadratic Hamiltonian $\hat{H}$ in the quadrature operator $\hat{\textbf{R}}$~\cite{BUCCO}. Notably, it turns out that any \emph{pure} Gaussian state ${\psi}$ can be prepared by applying a Gaussian unitary $ G$ to the vacuum, i.e.~$\ket{\psi}= G\ket{0}^{\otimes n}$. A fundamental example of a Gaussian state is the single-mode vacuum state $\ket{0}$.

Gaussian states can be seen as a quantum generalisation of Gaussian probability distributions. Just as Gaussian probability distributions, any Gaussian state is uniquely identified by its first moment $\textbf{m}$ and its covariance matrix $V$~\cite{BUCCO}. That is, for any $\textbf{m}\in\mathbb{R}^{2n}$ and any $V\in\mathbb{R}^{2n\times 2n}$ such that $V+i\Omega_n\ge 0$ there exists a Gaussian state with first moment $\textbf{m}$ and covariance matrix $V$~\cite{BUCCO}. Moreover, any Gaussian state satisfies the following decomposition --- known as the \emph{normal mode decomposition}~\cite{BUCCO}.
\begin{lemma}[(Normal mode decomposition of a Gaussian state~\cite{BUCCO})]\label{lem_normal_mode}
    Let $\rho$ be an $n$-mode Gaussian state with first moment $\textbf{m}(\rho)$ and covariance matrix $V(\rho)$. Let 
   \bb\label{will_form}
        V(\rho)=SDS^{\intercal}
    \ee
    be the Williamson decomposition of $V$ as in Eq.~\eqref{will_eq_def}, where $D\coloneqq\operatorname{diag}(d_1,d_1,\ldots,d_n,d_n)$ is the diagonal matrix of symplectic eigenvalues. Then, the Gaussian state $\rho$ is unitarily equivalent --- via Gaussian unitaries --- to a tensor product of thermal states, i.e.~
   \bb
        \rho= {D}_{\textbf{m}(\rho)}U_S\left(\tau_{\frac{d_1-1}{2}}\otimes\ldots \tau_{\frac{d_n-1}{2}}\right)U_S^{\dagger}{D}_{\textbf{m}(\rho)}^\dagger\,,
    \ee
    where:
    \begin{itemize}
        \item ${D}_{\textbf{m}(\rho)}$ is the displacement operator with amplitude equal to the first moment $\textbf{m}(\rho)$;
        \item $U_S$ is the symplectic Gaussian unitary associated with the symplectic matrix $S$ that puts the covariance matrix $V(\rho)$ in Williamson decomposition (as per \eqref{will_form});
        \item For any $\nu\ge0$, the state $\tau_\nu$ is the single-mode thermal state of mean photon number $\nu$, defined as
       \bb\label{mixed_gaussian_state}
            \tau_\nu\coloneqq\frac{1}{\nu+1}\sum_{n=0}^\infty\left(\frac{\nu}{\nu+1}\right)^n\ketbra{n}\,,
        \ee
        with $\ket{n}$ being the $n$th Fock state.
    \end{itemize}
\end{lemma}
As a consequence of the above decomposition, a Gaussian state is pure if and only if its symplectic eigenvalues are all equal to one. This is so because the thermal state with zero mean photon number corresponds to the vacuum state, i.e.~$\tau_0=\ketbra{0}$. In particular, any pure Gaussian state $\psi$ can be written as 
\bb\label{pure_gaussian_state}
\ket{\psi}=D_{\textbf{m}(\psi)}U_S\ket{0}^{\otimes n}\,,
\ee
where $D_{\textbf{m}(\psi)}$ is the displacement operator with displacement vector given by the first moment of the state and $S$ is a symplectic matrix such that $V(\psi)=SS^\intercal$.

Moreover, let us define the \emph{squeezing} of a Gaussian state as follows:
\begin{Def}[(Squeezing of a Gaussian state)]
    Given a Gaussian state $\rho$, its squeezing $z$ is defined by
    \bb
        z=\|S\|_\infty\,,
    \ee
    where $\|S\|_\infty$ denotes the operator norm (maximum singular value) of the symplectic matrix $S$ that puts the covariance matrix of $\rho$ in Williamson decomposition, as in Eq.~\eqref{will_form}. In other words, the squeezing of a Gaussian state is the squeezing of the Gaussian unitary that creates the Gaussian state from the thermal state.
    \end{Def}
Note that, even if $S$ is not unique, the squeezing is uniquely defined because $\|S\|_\infty=\sqrt{\|SS^\intercal\|_\infty}$ and $SS^\intercal$ is uniquely determined by $V(\rho)$ (indeed, $SS^\intercal$ can be expressed as $SS^\intercal=V\# \big( \Omega V^{-1}\Omega^\intercal\big)$~\cite{Lami16}, where $\#$ is the \emph{matrix geometric mean}~\cite{BHATIA}). Moreover, note that $\|S\|_\infty=\|Z\|_\infty$, where $Z$ is the squeezing matrix which appears in the Euler decomposition of $S$ (see Eq.~\eqref{Euler_dec2}).

\subsubsection{Single-mode Gaussian states}\label{single_mode_sq}
Thanks to Williamson's decomposition in Eq.~\eqref{will_eq_def}, the Euler decomposition in Eq.~\eqref{Euler_dec2}, and the fact that the only orthogonal symplectic $2 \times 2$ matrices are rotations, it follows that the most general covariance matrix of a single-mode state $\rho$ can be parametrised by three real parameters: the squeezing $z \ge 1$, a rotation angle $\theta \in [0, 2\pi)$, and a mean photon number $\nu \ge 0$. In this form, the covariance matrix reads
\begin{equation}\label{eq_covsss}
V(\rho) = (2\nu + 1)O_\theta
\begin{pmatrix}
z^{-2} & 0 \\
0 & z^2
\end{pmatrix}
O_\theta^\intercal,
\end{equation}
where $O_\theta$ is the rotation matrix of angle $\theta$ defined in Eq.~\eqref{eq_o_theta}.

In summary, a single-mode Gaussian state is fully characterised by its first moment --- a real two-dimensional vector $\mathbf{r}$ --- and its covariance matrix, which is determined by the three parameters $(z, \theta, \nu)$ as given in Eq.~\eqref{eq_covsss}.

Thanks to Lemma~\ref{lem_normal_mode}, at the level of the Hilbert space, the most general Gaussian state $\rho$ can be written as the output state of the composition of a squeezing unitary $S_z$, a phase-space rotation $R_\theta$, and a displacement operator $D_{\textbf{r}}$ with input a thermal state $\tau_\nu$:
\bb
    \rho= {D}_{\textbf{r}} R_\theta S_z\,\tau_\nu\, S_z^\dagger R_\theta^\dagger {D}_{\textbf{r}}^\dagger\,.
\ee
Let us give some paradigmatic examples of single-mode Gaussian states:
\begin{itemize}
    \item For any $\textbf{r}\in\mathbb{R}^2$, the \emph{coherent state} $\ket{\textbf{r}}$ with amplitude $\textbf{r}$ is given by $\ket{\textbf{r}}\coloneqq{D}_{\textbf{r}}\ket0$, whose covariance matrix and first moment are given by $V(\ketbra{\textbf{r}})= \mathbb{1}$ and $\textbf{m}(\ketbra{\textbf{r}})=\textbf{r}$, respectively.
    \item For any $z\ge1$, the \emph{squeezed state} $\ket{z}$ with squeezing $z$ is given by $\ket{z}\coloneqq S_z\ket0$, whose covariance matrix and first moment are given by $V(\ketbra{z})= \left(\begin{matrix}z^{-2}&0\\ 0&z^{2}\end{matrix}\right)$ and $\textbf{m}(\ketbra{z})=0$, respectively. As $z$ grows, the variance of the position operator decreases as $z^{-2}$, while the variance of the momentum operator increases as $z^2$. In contrast, the state $R_{\frac{\pi}{2}}\ket{z}$ shows the opposite behaviour: its covariance matrix is given by $V(R_{\frac{\pi}{2}}\ketbra{z}R_{\frac{\pi}{2}}^\dagger)= \left(\begin{matrix}z^{2}&0\\ 0&z^{-2}\end{matrix}\right)$ and so, as $z$ grows, the variance of the position operator increases as $z^{2}$, while the variance of the momentum operator decreases as $z^{-2}$.
    \item For any $\nu\ge0$, the \emph{thermal state} $\tau_\nu$ with mean photon number $\nu$ has zero first moment and covariance matrix given by $V(\tau_\nu)= (2\nu+1)\mathbb{1}$.
    \item The \emph{vacuum state} $\ket{0}$ has zero first moment and covariance matrix given by $V(\ketbra{0})= \mathbb{1}$.
    \item For any $x\in\mathbb{R}$ and $z\ge 1$, the Gaussian state $D_{(x,0)} S_z \ket{0}$ has first moment equal to $(x,0)$ and covariance matrix given by $\left(\begin{matrix}z^{-2}&0\\ 0&z^{2}\end{matrix}\right)$. In the limit of infinite squeezing $z \to \infty$, such a state converges to the (generalised) eigenstate $\ket{x}$ of the position operator $\hat{x}$ with eigenvalue $x$.
    \item For any $p\in\mathbb{R}$ and $z\ge 1$, the Gaussian state $D_{(0,p)} R_{\frac{\pi}{2}}S_z \ket{0}$ has first moment equal to $(0,p)$ and covariance matrix given by $\left(\begin{matrix}z^{2}&0\\ 0&z^{-2}\end{matrix}\right)$. In the limit of infinite squeezing $z \to \infty$, such a state converges to the (generalised) eigenstate $\ket{p}$ of the momentum operator $\hat{p}$ with eigenvalue $p$.
\end{itemize}

\subsection{Squeezing as a quantum resource}
In the above section, we have introduced the definition of a single mode squeezed Gaussian state. In the following section, we detail why squeezing is a quantum resource.

To measure a signal, a sensing protocol is limited by standard measurement precision bounds imposed by the Heisenberg uncertainty principle.
To quantify the \emph{quantumness} of a given protocol, we rely on the quantum optics definitions of classicality \cite{glauber1963coherent} and adopt the framework of quantum metrology, which distinguishes classical and quantum strategies based on their ability to surpass classical measurement bounds \cite{caves1981quantum, giovannetti2004quantum}. The capacity to circumvent these measurement bounds depends on the intrinsic quantum noise of the mode's initial state, characterized by the squeezing parameter, $z$.

In Glauber's seminal 1963 work~\cite{glauber1963coherent}, he precisely defined the mathematical boundary between classical and quantum light, establishing as a key consequence that the coherent state satisfies the strict mathematical definition of a 'classical' signal.
A coherent state (a state with no squeezing, $z=1$) is a displacement of the vacuum, so its quantum noise originates purely from vacuum fluctuations and is distributed equally across conjugate phase-space quadratures $(\Delta x = \Delta p$, see Figure~\ref{fig:fig1}). By the Heisenberg uncertainty principle, $\Delta x \Delta p \geq \frac{1}{2} \implies \Delta x^2 \geq \frac{1}{2}$. Thus, the variance in the measured quadrature is strictly bounded by this irreducible vacuum noise. 
Sensing protocols restricted to $z=1$ are therefore fundamentally classical, as no single quadrature measurement can resolve a signal finer than this vacuum limit \cite{caves1981quantum}.

In contrast, squeezed states ($z > 1$) introduce a non-classical resource by allowing the variance of a chosen quadrature to go below the intrinsic vacuum noise limit. Squeezing redistributes the initial quantum fluctuations, shrinking the variance of the target quadrature at the expense of amplifying the variance in the conjugate quadrature. Therefore, squeezed states are quantum because they surpass the measurement limit that classical states are subject to~\cite{caves1981quantum}.

This theoretical divide is sharply reflected in the experimental complexity required to prepare these states.
States with no squeezing $(z=1)$, such as coherent states, can be generated trivially by turning on a laser~\cite{dowling2002quantumtechnologysecondquantum}. In contrast, generating highly squeezed states demands advanced quantum technologies, specifically nonlinear optical crystals that mediate the parametric downconversion of  photons from a pump laser into entangled pairs of lower-energy photons~\cite{BUCCO,braunstein2005quantum,weedbrook2012gaussian,adesso2014continuous,andersen2010continuous}.


Additionally, one can illustrate the inherently quantum nature of squeezing also by directly linking it to quantum entanglement.
Any multi-mode Gaussian state with no squeezing is guaranteed to be unentangled (completely separable)~\cite{BUCCO}. For example, interfering two ``classical" coherent states through a 50:50 beam splitter yields unentangled (product) states. However, interfering two single-mode squeezed states with appropriately chosen relative phases on the same beam splitter generates a two-mode entangled state. Beam splitters are linear optical elements and cannot create entanglement from classical input states on their own. Thus, in this Gaussian setting, squeezing is necessary to generate entanglement using passive linear optics.

Ultimately, the squeezing parameter $z$ quantifies the degree of \emph{quantumness} in a protocol: protocols that rely on highly squeezed states are genuinely quantum, whereas those based on weak or no squeezing remain effectively classical.

%

\subsection{Gaussian measurements}\label{sec_gauss_meas}
Let us proceed with the definition of Gaussian measurements. A
\emph{Gaussian measurement with a Gaussian seed} is a POVM of the form
$\{\frac{1}{(2\pi)^n}D_{\mathbf{r}}\sigma D_{\mathbf{r}}^\dagger\}_{\mathbf{r}\in\mathbb{R}^{2n}}$,
where $\sigma$ is a Gaussian state, dubbed the \emph{Gaussian seed}, with zero
first moment. Note that this is a proper POVM because it holds that~\cite{BUCCO}
\bb
    \frac{1}{(2\pi)^n} \int_{\mathbb{R}^{2n}}\mathrm{d}^{2n}\textbf{r}\,D_\textbf{r}\sigma D_\textbf{r}^\dagger =\hat{\mathbb{1}}\,.
\ee
Moreover, the probability distribution $p_{\rho,\sigma}(\textbf{r})$ of the outcome $\textbf{r}$ of the Gaussian measurement with Gaussian seed $\sigma$ performed on a state $\rho$ is given by 
\bb
    p_{\rho,\sigma}(\textbf{r})=\frac{1}{(2\pi)^n} \Tr[D_\textbf{r}\sigma D_\textbf{r}^\dagger \rho ]\,.
\ee
Additionally, if $\rho$ is Gaussian, one can prove that such a probability distribution is a Gaussian probability distribution with mean equal to the first moment of $\rho$ and covariance matrix given by the arithmetic mean of the covariance matrix of $\rho$ and that of the Gaussian seed $\sigma$~\cite{BUCCO}. Specifically, by denoting as 
\bb\label{gauss_prob}
    \NN[\textbf{m},V](\textbf{r})&\coloneqq\frac{e^{-\frac12 (\textbf{r}-\textbf{m})^\intercal V^{-1}  (\textbf{r}-\textbf{m})}}{(2\pi)^{n}\sqrt{\det V }}\,
\ee
the $2n$-dimensional Gaussian probability distribution with first moment $\textbf{m}$ and covariance matrix $V$, the probability distribution $p_{\rho,\sigma}(\textbf{r})$ of the outcome $\textbf{r}$ of a Gaussian measurement with Gaussian seed $\sigma$ performed on a Gaussian state $\rho$ reads
\bb\label{dist_gauss_meas}
    p_{\rho,\sigma}(\textbf{r})=\NN\!\left[\textbf{m}(\rho),\frac{V(\rho)+V(\sigma)}{2}\right](\textbf{r})\,.
\ee 
Let us give some paradigmatic examples of Gaussian measurements. We will start by introducing the heterodyne measurement, and then we will define the homodyne measurement.

\subsubsection{Heterodyne measurement}
The heterodyne measurement is the Gaussian measurement with Gaussian seed being the vacuum state, i.e.~$\sigma=\ketbra{0}$. Hence, the POVM of the heterodyne measurement is of the form $\{\frac{1}{(2\pi)^n} \ketbra{\textbf{r}} \}_{\textbf{r}\in\mathbb{R}^{2n}}$, where $\ket{\textbf{r}}$ denotes the coherent state with amplitude $\textbf{r}$. Moreover, the probability distribution $p_{\mathrm{het},\rho}$ of the outcome $\textbf{r}$ of the heterodyne measurement performed on a Gaussian state $\rho$ is a Gaussian probability distribution with the same mean as $\rho$ and covariance matrix given by $\frac{V(\rho)+\mathbb{1}}{2}$:
    \bb
        p_{\mathrm{het},\rho}(\textbf{r})\coloneqq p_{\rho,\ketbra{0}}(\textbf{r})=\NN\left[\textbf{m}(\rho),\frac{V(\rho)+\mathbb{1}}{2}\right](\textbf{r})\,.
    \ee
    Note that, e.g.~in the single-mode case, the outcome of the heterodyne measurement is a 2-dimensional vector $\textbf{r}=(x,p)$ and that the expectation value of $x$ is given by the expectation value of the position operator, i.e.~$\mathbb{E}[x]=\Tr[\rho\hat{x}]$, while its variance is given by $\mathrm{Var}[x]=\mathrm{Var}_\rho[\hat{x}]+\frac12$, where $\mathrm{Var}_\rho[\hat{x}]\coloneqq  \Tr[\rho\hat{x}^2]-(\Tr[\rho\hat{x}])^2 $ is the variance of the position operator measured on $\rho$. Analogously, the expectation value of $p$ is given by the expectation value of the momentum operator $\Tr[\rho\hat{p}]$ and its variance is given by $\mathrm{Var}[p]=\mathrm{Var}_\rho[\hat{p}]+\frac12$, where  $\mathrm{Var}_\rho[\hat{p}]\coloneqq  \Tr[\rho\hat{p}^2]-(\Tr[\rho\hat{p}])^2 $ is the variance of the momentum operator measured on $\rho$. The additional factor $\frac12$ present in the variances $\mathrm{Var}[x]$ and $\mathrm{Var}[p]$ can be interpreted, via the lens of Heisenberg's uncertainty principle, as the impossibility of measuring the position and momentum operators with just one measurement setting.

\subsubsection{Homodyne measurements}
    Before giving the definition of homodyne measurements, let us recall that \eqref{laaa} implies that for any symplectic matrix $S$ it holds that $U_S^\dagger \hat{\textbf{R}} U_S=S\hat{\textbf{R}}$. In particular, it follows that $\{\sum_{j=1}^{2n} S_{ij}\hat{\textbf{R}}_j\}_{i\in\{1,3,5,\ldots,2n-1\}}$ is a set of $n$ commuting observables. By definition, a \emph{homodyne measurement} is a projective measurement with respect to the common (generalised) eigenbasis of such $n$ commuting observables. For example, the projective measurement with respect to the common (generalised) eigenbasis of the position operators $\hat{x}_1,\hat{x}_2,\ldots,\hat{x}_n$ constitutes a homodyne measurement (corresponding to the choice $S=\mathbb{1}$). Also, the measurement with respect to the common (generalised) eigenbasis of the momentum operators $\hat{p}_1,\hat{p}_2,\ldots,\hat{p}_n$ is a homodyne measurement (corresponding to the choice $S=O_{\frac{\pi}{2}}$). More generally, for any $\theta_1,\ldots,\theta_n\in[0,2\pi)$, the measurement with respect to the common (generalised) eigenbasis of the commuting observables $\cos\!\theta_i\,\hat{x}_i+\sin\!\theta_i\,\hat{p}_i$, for $i\in[n]$, is a homodyne measurement (corresponding to the choice $S=\bigoplus_{i=1}^nO_{\theta_i}$).

    In general, the outcome of a homodyne measurement performed on an $n$-mode quantum state is an $n$-dimensional real vector. However, above we saw that the outcome of a Gaussian measurement is a $2n$-dimensional real vector. So, in which sense is a homodyne measurement a Gaussian measurement? Strictly speaking, a homodyne measurement is a \emph{limit} of Gaussian measurements. Specifically, it can be proved that the homodyne measurement associated with a symplectic matrix $S$ corresponds to the Gaussian measurement with Gaussian seed $\sigma =U_S^\dagger \ketbra{z}^{\otimes n}U_S$ in the limit of infinite squeezing $z\rightarrow\infty$, where $\ket{z}$ denotes the single-mode squeezed state with squeezing $z$~\cite{BUCCO}. Moreover, the probability distribution $p(x)$ of getting the outcome $x\in\mathbb{R}^n$ by performing the homodyne measurement associated with the symplectic matrix $S$ on a state $\rho$ is given by $p(x)\coloneqq \bra{x}U_S\rho U_S^\dagger \ket{x}$, where $\ket{x}=\ket{x_1}\otimes \ket{x_2}\ldots\otimes \ket{x_n}$ with each $\ket{x_i}$ being the (generalised) eigenvector associated with the eigenvalue $x_i$ of the position operator $\hat{x}_i$ of the $i$th mode. In particular, the mean value $\textbf{m}$ of $p(x)$ is given by
    \bb
        m_i\coloneqq \mathbb{E}[x_i]=m_{2i-1}(U_S\rho U_S^\dagger)\qquad\forall i\in[n]\,,
    \ee
    and the covariance matrix $V$ of $p(x)$ is given by
    \bb
        V_{i,j}\coloneqq \mathbb{E}[ (x_i-\mathbb{E}[x_i])(x_j-\mathbb{E}[x_j])]=\frac12\bigl[V(U_S\rho U_S^\dagger)\bigr]_{2i-1,2j-1}\qquad\forall i,j\in[n]\,.
    \ee
    If $\rho$ is a Gaussian state, it can be shown that $p(x)$ is the Gaussian probability distribution with mean value $\textbf{m}$ and covariance matrix $V$~\cite{BUCCO}:
    \bb
        p(x)=\NN[\textbf{m},V](x)\qquad \forall x\in\mathbb{R}^n\,.
    \ee
    Throughout this paper, we use the term \emph{Gaussian measurement} in the
extended sense that includes both Gaussian measurements with a Gaussian seed
and ideal homodyne measurements.

    To be more concrete, let us focus on single-mode systems. The homodyne measurement along the $x$-axis is the projective measurement with respect to the (generalised) eigenbasis $\{\ketbra{x}\}_{x\in\mathbb{R}}$ of the position operator $\hat{x}$, which can be written in the (generalised) spectral decomposition as $\hat{x}=\int_{\mathbb{R}}\mathrm{d}x\,x\ketbra{x}$. The probability distribution $p(x)$ of getting the outcome $x$ by performing such a measurement on a state $\rho$ is given by $p(x)\coloneqq \bra{x}\rho\ket{x}$. In particular, the mean value $m$ and the variance $V$ of the outcome $x$ are given by
    \bb
        m\coloneqq \int_{\mathbb{R}}x p(x)=\Tr[\hat{x}\rho]=m_1(\rho)\,,
    \ee
    and 
    \bb
        V\coloneqq \int_{\mathbb{R}} p(x)\left(x-\mathbb{E}[x]\right)^2=\mathrm{Var}_\rho(\hat{x})=\frac12V_{11}(\rho)\,.
    \ee 
    where we denoted $\mathrm{Var}_\rho(\hat{x})\coloneqq \Tr\left[\rho(\hat{x}-\Tr[\rho\hat{x}]\hat{\mathbb{1}})^2\right]$. If $\rho$ is a Gaussian state, it can be shown that $p(x)$ is the Gaussian probability distribution with mean $\Tr[\hat{x}\rho]$ and variance $\mathrm{Var}_\rho(\hat{x})$~\cite{BUCCO}, or alternatively:
    \bb
        p(x)=\NN\left[m_1(\rho),\frac{V_{11}(\rho)}{2}\right](x)\qquad \forall x\in\mathbb{R}\,.
    \ee

\section{Preliminaries on Gaussian probability distributions \label{sec:prelim_gaussian_prob}}
In this section, we recall and prove some preliminary results on Gaussian probability distributions. Let us begin with a few useful remarks about the \emph{convolution} between probability distributions. The convolution, denoted by $\ast$, between two probability distributions $p:\mathbb{R}^n\longmapsto \mathbb{R}$ and $q:\mathbb{R}^n\longmapsto \mathbb{R}$, is defined as
\bb
    p\ast q(\textbf{r})\coloneqq \int_{\mathbb{R}^n}\mathrm{d}^n\textbf{r}'\,\,p(\textbf{r}-\textbf{r'})\,\, q(\textbf{r}')\,,
\ee 
and yields a well-defined probability distribution. Specifically, the following holds. 
\begin{lemma}[(Probability interpretation of the convolution)]\label{interp_conv}
Given two independent random variables
\bb
    X&\sim p\,,\quad Y&\sim q\,,
\ee
the probability distribution of their sum $X+Y$ is exactly given by the convolution $p\ast q$:
\bb
    X+Y\sim p\ast q\,.
\ee
\end{lemma}
\begin{proof}
    It follows from the definition of convolution.
\end{proof}
Given $n\in\mathbb{N}$, a positive matrix $V\in\mathbb{R}^{n\times n}$, and a real vector $\textbf{m}\in\mathbb{R}^n$, the \emph{Gaussian probability distribution}  $\NN[\textbf{m},V]:\mathbb{R}^n\mapsto\mathbb{R}$ is a probability distribution defined as
\bb\label{gauss_probdef}
    \NN[\textbf{m},V](\textbf{r})&\coloneqq\frac{e^{-\frac12 (\textbf{r}-\textbf{m})^\intercal V^{-1}  (\textbf{r}-\textbf{m})}}{(2\pi)^{n/2}\sqrt{\det V }}\quad\forall\textbf{r}\in\mathbb{R}^n\,
\ee
where $\textbf{m}$ is the first moment and $V$ is the covariance matrix.  Notably, the convolution between two Gaussian probability distributions $\mathcal{N}[\textbf{m}_1, V_1]$ and $\mathcal{N}[\textbf{m}_2, V_2]$ yields another Gaussian probability distribution with mean equal to the sum of the means, i.e.~$\textbf{m}_1 + \textbf{m}_2$, and covariance matrix equal to the sums of the covariance matrices, i.e.~$V_1 + V_2$. 
\begin{lemma}[(Convolution between Gaussian distributions)]\label{conv_gauss}
Let $V_1,V_2\in\mathbb{R}^{n\times n}$ positive semidefinite matrices and let $\textbf{m}_1,\textbf{m}_2\in\mathbb{R}^n$. Then, it holds that
\bb\label{prop_conv}
    \mathcal{N}[\textbf{m}_1,V_1]\ast \mathcal{N}[\textbf{m}_2,V_2]= \mathcal{N}[\textbf{m}_1+\textbf{m}_2,V_1+V_2]\,,
\ee
$\mathcal{N}[\textbf{m},V]$ denotes a Gaussian probability distribution with mean $\textbf{m}$ and covariance matrix $V$, as defined in Eq.~\eqref{gauss_prob}. In particular, given two independent Gaussian random variables
\bb
    \textbf{r}_1&\sim \mathcal{N}[\textbf{m}_1,V_1]\,,\\
    \textbf{r}_2&\sim \mathcal{N}[\textbf{m}_2,V_2]\,,
\ee 
their sum is another Gaussian probability distribution satisfying
\bb
    \textbf{r}_1+\textbf{r}_2\sim \mathcal{N}[\textbf{m}_1+\textbf{m}_2,V_1+V_2]\,.
\ee
\end{lemma}
\begin{proof}
    Eq.~\eqref{prop_conv} is a well known result in probability theory. The quickest way to prove it is to take the Fourier transform of both
sides of Eq.~\eqref{prop_conv} and use the fact that the Fourier
transform of a convolution is the product of the Fourier transforms.
\end{proof} 
This result is particularly useful, as it enables us to derive the following lemma, which will be applied repeatedly throughout the paper.
\begin{lemma}[(Reduction via convolution)]\label{reduction_via_convolution}
Let $V,W \in \mathbb{R}^{n \times n}$ be two positive semidefinite matrices such that $W\ge V$, and let $\textbf{m} \in \mathbb{R}^n$. Then, there exists an algorithm, which is independent of $\textbf{m}$, that takes as input a sample $\textbf{r} \sim \mathcal{N}[\textbf{m}, V]$ and outputs a random variable $\textbf{r}'\sim \mathcal{N}[\textbf{m}, W]$. 
\end{lemma}
\begin{proof}
    Lemma~\ref{conv_gauss} implies that
    \bb
        \mathcal{N}[\textbf{m},W]=\mathcal{N}[0,W-V]\ast\mathcal{N}[\textbf{m},V]\,.
    \ee
    Hence, by sampling a Gaussian random variable $\bar{r}\sim \mathcal{N}[0,W-V]$, we have that $    \mathbf{r}'
    \coloneqq
    \mathbf{r}+\bar{\mathbf{r}}
    \sim
    \mathcal{N}[\mathbf{m},W]$. This concludes the proof.
\end{proof}
Another useful fact that will be extensively used is the following.
\begin{lemma}[(Reduction via change of variable)]\label{reduction_via_change}
Let $\textbf{m} \in \mathbb{R}^n$ and let $V \in \mathbb{R}^{n \times n}$ be a positive semidefinite matrix. Let $A\in\mathbb{R}^{n\times n}$ be an invertible matrix and let $\textbf{r}_0\in\mathbb{R}^n$. Given a sample $
\textbf{r} \sim \mathcal{N}[\textbf{m}, V]$, the vector $\textbf{r}'\coloneqq A\textbf{r}+\textbf{r}_0$ satisfies $\textbf{r}'\sim\mathcal{N}[A\textbf{m}+\textbf{r}_0,AVA^\intercal]$.
\end{lemma}
\begin{proof}
The proof follows directly from the definition of a Gaussian
probability distribution.
\end{proof}
Finally, we will also use the following lemma.
\begin{lemma}[(Reduction via marginalisation)]\label{reduction_via_marg}
    Let $\textbf{m}_1,\ldots,\textbf{m}_M \in \mathbb{R}^n$ be $M$ vectors and let $V \in \mathbb{R}^{n \times n}$ be a positive semidefinite matrix. Given independent samples
\bb
    \textbf{r}_i &\sim \mathcal{N}\!\left[\textbf{m}_i, \frac{V}{M}\right]\qquad\forall\,i\in[M]\,,
\ee
the vector $\mathbf{r}'\coloneqq \sum_{i=1}^M\textbf{r}_i$ satisfies 
\bb
    \textbf{r}'\sim \mathcal{N}\!\left[\sum_{i=1}^M\textbf{m}_i, V\right]\,.
\ee
\end{lemma}
\begin{proof}
    Lemma~\ref{conv_gauss} implies that
    \bb
        \mathcal{N}\!\left[\textbf{m}_1,\frac{V}{M}\right]\ast\mathcal{N}\!\left[\textbf{m}_2,\frac{V}{M}\right]\ast\ldots\ast \mathcal{N}\!\left[\textbf{m}_M,\frac{V}{M}\right]=\mathcal{N}\!\left[\sum_{i=1}^M\textbf{m}_i, V\right]\,.
    \ee
    Consequently, Lemma~\ref{interp_conv} concludes the proof.
\end{proof}

Now, we are going to derive an explicit expression for the total variation distance between two one-dimensional Gaussian probability distributions with the same variance.
\begin{lemma}[(Total variation distance between one-dimensional Gaussian probability distributions)]\label{tvd_one_g}
    Let $V\in\mathbb{R}^+$ and let $m_1,m_2\in\mathbb{R}$. Then, the total variation distance between $\mathcal{N}[m_1,V]$ and $\mathcal{N}[m_2,V]$ is given by
    \bb\label{eq_0_l}
        \frac12\left\|\mathcal{N}[m_1,V]-\mathcal{N}[m_2,V]\right\|_1= 2 \int_{0}^{\frac{|m_1-m_2|}{2\sqrt{V}}} \mathrm{d}x\, \mathcal{N}[0,1](x)\,.
    \ee
    In particular, by setting $t\coloneqq \frac{|m_1-m_2|}{2\sqrt{V}}$, the following upper bound holds:
    \bb\label{eq_first_upppp}
        \frac12\left\|\mathcal{N}[m_1,V]-\mathcal{N}[m_2,V]\right\|_1\le \frac{|m_1-m_2|}{\sqrt{2\pi V}}=\sqrt{\frac{2}{\pi}}t\,,
    \ee
    which is tight for $t\rightarrow 0$. Moreover, it also holds that 
    \bb\label{eq_2_l}
        \frac12\left\|\mathcal{N}[m_1,V]-\mathcal{N}[m_2,V]\right\|_1\le 1-\left(\frac{1}{6}e^{-t^2}+\sqrt{\frac{2}{\pi}}\frac{e^{-t^2/2}}{t+1}\right)\,,
    \ee
    which is tight for $t\rightarrow\infty$.
\end{lemma}
\begin{proof}
    Without loss of generality we can assume that $m_1\le m_2$. Then, by symmetry, it easily follows that $\mathcal{N}[m_1,V](x)\ge \mathcal{N}[m_2,V](x)$ if and only if $x\le \frac{m_1+m_2}{2}$. Consequently, we obtain that 
    \bb
        \frac12\left\|\mathcal{N}[m_1,V]-\mathcal{N}[m_2,V]\right\|_1&=\frac12\int_{-\infty}^\infty \mathrm{d}y\, \left|\mathcal{N}[m_1,V](y)-\mathcal{N}[m_2,V](y)\right|\\
        &= \int_{-\infty}^{\frac{m_1+m_2}{2}} \mathrm{d}y\,\left( \mathcal{N}[m_1,V](y)-\mathcal{N}[m_2,V](y)\right)\\
        &= \int_{-\infty}^{\frac{m_1+m_2}{2}} \mathrm{d}y\, \mathcal{N}[m_1,V](y)-\int_{-\infty}^{\frac{m_1+m_2}{2}} \mathrm{d}y\, \mathcal{N}[m_2,V](y) \\
        &= \int_{-\infty}^{\frac{m_2-m_1}{2\sqrt{V}}} \mathrm{d}y\, \mathcal{N}[0,1](y)-\int_{-\infty}^{\frac{-(m_2-m_1)}{2\sqrt{V}}} \mathrm{d}y\, \mathcal{N}[0,1](y) \\
        &=2 \int_{0}^{\frac{m_2-m_1}{2\sqrt{V}}} \mathrm{d}y\, \mathcal{N}[0,1](y)\\
        &=2 \int_{0}^{\frac{|m_2-m_1|}{2\sqrt{V}}} \mathrm{d}y\, \mathcal{N}[0,1](y)\,.
    \ee
    To prove the upper bound in \eqref{eq_first_upppp}, note that the function $g(x)\coloneqq \mathcal{N}[0,1](0) x- \int_0^x \mathrm{d}y\,\mathcal{N}[0,1](y)$ satisfies $g(0)=0$ and $g'(x)=\mathcal{N}[0,1](0)-\mathcal{N}[0,1](x)\ge 0$ for all $x\in\mathbb{R}$. In particular, it holds that $g(x)\ge0$ for all $x\ge 0$, which implies that
    \bb
        \int_0^x \mathrm{d}y\,\mathcal{N}[0,1](y)\le \frac{1}{\sqrt{2\pi}}x \quad\forall x\ge0\,.
    \ee
    In particular, we obtain that
    \bb
        \frac12\left\|\mathcal{N}[m_1,V]-\mathcal{N}[m_2,V]\right\|_1&=2 \int_{0}^{\frac{|m_2-m_1|}{2\sqrt{V}}} \mathrm{d}y\, \mathcal{N}[0,1](y)\\
        &=2 \int_{0}^{t} \mathrm{d}y\, \mathcal{N}[0,1](y)\\
        &\le \sqrt{\frac{2}{\pi}}t\,.
    \ee
    Moreover, it holds that
    \bb
        \frac12\left\|\mathcal{N}[m_1,V]-\mathcal{N}[m_2,V]\right\|_1&=2 \int_{0}^{\frac{|m_2-m_1|}{2\sqrt{V}}} \mathrm{d}y\, \mathcal{N}[0,1](y)\\
        &=1-2\int_{\frac{|m_2-m_1|}{2\sqrt{V}}}^{\infty} \mathrm{d}y\, \mathcal{N}[0,1](y)\\
        &=1-2\int_{t}^{\infty} \mathrm{d}y\, \mathcal{N}[0,1](y)\\
        &\le 1-\left(\frac{1}{6}e^{-t^2}+\sqrt{\frac{2}{\pi}}\frac{e^{-t^2/2}}{t+1}\right)\,,
    \ee
    where in the last inequality we exploited the forthcoming Lemma~\ref{lemma_bounds_Q}.
\end{proof}

\begin{lemma}[(Tail bounds on the standard Gaussian probability distribution)]\label{lemma_bounds_Q}
    For any $t\ge0$ it holds that
    \bb
      \frac{1}{12}e^{-t^2}+\sqrt{\frac{1}{2\pi}}\frac{e^{-t^2/2}}{t+1}\le\int_{t}^{\infty} \mathrm{d}y\, \mathcal{N}[0,1](y)\le \frac{1}{2}e^{-t^2/2} \,,
    \ee
\end{lemma}
\begin{proof}
The first inequality is proved in~\cite{tight_bounds_gaus}, while the second one is proved in \cite{new_bound_qfunc}.
\end{proof}

\newpage

\section{Experimental Platforms }\label{sec:exp_platforms}

In this section, we detail different experimental platforms and explain how sensing a signal of interest $f(t)$ motivates our sensing model. Specifically, we show how a physical signal displaces a specific quadrature of a bosonic mode, as in the evolution generated by the Hamiltonian $H(t) = f(t) \, \hat{p}$.

\subsection{Gravitational wave sensing with LIGO}
We now examine how LIGO (Laser Interferometer Gravitational-Wave Observatory)~\cite{abbott2016observation,tse2019quantumenhanced} detects gravitational waves by encoding the signal into the phase of a laser. To build intuition, we first analyze the interferometer's response using classical electrodynamics. We then perform a full quantum optics analysis to show that this physical setup produces a quadrature displacement of the same form as that generated by the Hamiltonian $H(t)=f(t)\,\hat{p}$ on a bosonic mode, followed by a homodyne measurement.
\begin{figure}[h!]
    \centering
    \includegraphics[width=0.6\linewidth]{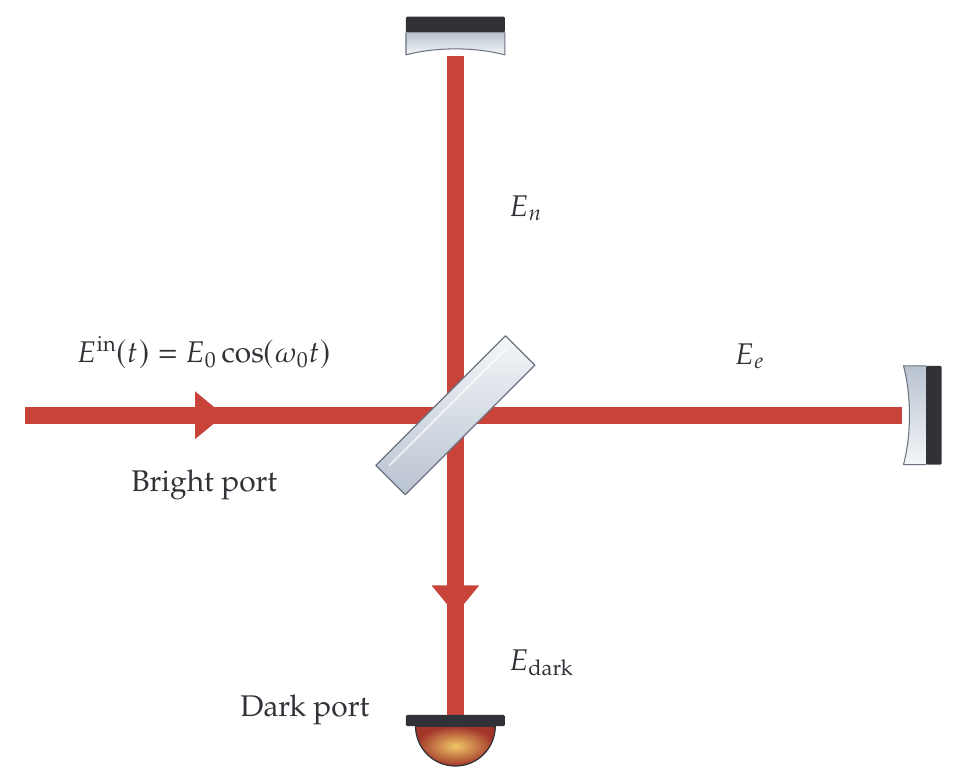}
    \caption{Simplified schematic of LIGO as a Michelson interferometer. The laser enters as $E^{\text{in}}$ a beamsplitter and is split into two orthogonal arms (north $E_n$ and east $E_e$). After traveling to the mirrors at the end and back, the light recombines via the same beamsplitter, and the signal is measured at the dark port. We first show how the gravitational wave signal is detectable in a photodetector in Section~\ref{sec:cl_ligo}, and proceed to homodyne detection in Section~\ref{sec:qu_ligo}.}
    \label{fig:LIGO}
\end{figure}

\subsubsection{Classical description of LIGO as an interferometer}
\label{sec:cl_ligo}
Figure \ref{fig:LIGO} illustrates a simplified Michelson interferometer, such as the one used in LIGO.  A Michelson interferometer splits a laser into two orthogonal arms (north and east) via a beamsplitter, reflects the light off mirrors, and recombines via the same beamsplitter. The signal is read out at the antisymmetric (dark) port via a photodetector. Following the analysis presented in \cite{danilishin2012quantum}, we model the input laser as a classical, monochromatic plane wave with amplitude $E_0$ and laser frequency $\omega_0$ \cite{danilishin2012quantum}, 
\bb
E^{\text{in}}(t) = E_0 \cos(\omega_0 t).
\ee

Propagation of an electric signal over a distance $x$ adds phase $\Delta \phi = k_0 x = \frac{\omega_0 }{c} x.$ After a first pass through the beamsplitter and a round trip of length $2L_{n}$ in the north arm and $2L_e$ in the east arm, the signals returning to the beamsplitter are 
\bb 
E_{n/e}(t) = \frac{E_0}{\sqrt{2}} \cos\left(\omega_0 t + \phi_0 - \frac{2 \omega_0 L_{n/e}}{c}\right).
\ee

When these signals ($E_{n}(t)$ and $E_e(t)$) recombine at the beam splitter, they interfere to form the output at the bright port (back towards the laser) and the dark port (photodiode, labelled in \ref{fig:LIGO}). 
The constructive (bright port) and destructive (dark port) interference follow from the beamsplitter phase relations. We will use the same convention as \cite{danilishin2012quantum, cahillane2022review} where a beam reflected off the back of the mirror suffers a $\pi$ phase flip, but a beam reflected off the front, or just transmitted, suffers no phase flip. 
For the bright port, $E_n$ is reflected off the front of the beamsplitter and $E_e$ is transmitted, accumulating no additional boundary phase.
However for the dark port, $E_n$ is transmitted while $E_e$ reflects off the back, resulting in a factor of $-1$. Thus, the electric signals are
\bb
E_{\text{bright}} = \frac{E_n(t) + E_e(t)}{\sqrt{2}}, \qquad E_{\text{dark}} = \frac{E_n(t) - E_e(t)}{\sqrt{2}}.
\ee
Specifically the electric signal at the dark port is 
\bb
E_{\text{dark}} = E_0 \sin\left( \frac{\omega_0 (L_e - L_n)}{c} \right)\sin\left(\omega_0 t - \frac{\omega_0(L_n+L_e)}{c}\right).
\ee
The photodetector measures intensity $(\mathcal{I})$, which is proportional to the time-average of the signal squared 
\bb
\mathcal{I} \propto \langle E^2 \rangle_t = \frac{E_0^2}{2} \sin^2\left( \frac{\omega_0 (L_e - L_n)}{c} \right) = \frac{E_0^2}{4} \left(1 - \cos \frac{2 \omega_0 (L_e - L_n)}{c}\right) \approx \frac{E_0^2}{2} \left( \frac{\omega_0}{c} \delta L \right)^2.
\ee
A passing gravitational wave induces a differential change in the arms length $\delta L = L_e - L_n$, which generates a nonzero signal that can be detected by the photodiode. In practice, to ensure the measurement readout is linear in $\delta L$ rather than quadratic, LIGO performs a homodyne measurement of a specific phase-space quadrature, as discussed in the next section.

\subsubsection{Quantizing the electric signal} \label{sec:qu_ligo}
To rigorously formalize homodyne detection, the injection of squeezed light, and the connection to the Hamiltonian evolution $H(t)=f(t)\,\hat{p}$, we must transition to a full quantum optics formulation.

We treat the incident laser not as a classical electromagnetic wave, but as a quantized electromagnetic signal. The beam splitter is now modelled as a unitary transformation with two input modes, one at the bright port $(\hat{a})$ and one at the dark port $(\hat{b})$ \cite{caves1980quantum, caves1981quantum} (see Figure~\ref{fig:LIGO_qu}). The incident laser enters at the bright port as a coherent state $\ket{\alpha}_a$, and if nothing is sent through the dark port, its state is a standard vacuum state $\ket{0}_b$. 

\begin{figure}[h!]
    \centering
    \includegraphics[width=0.6\linewidth]{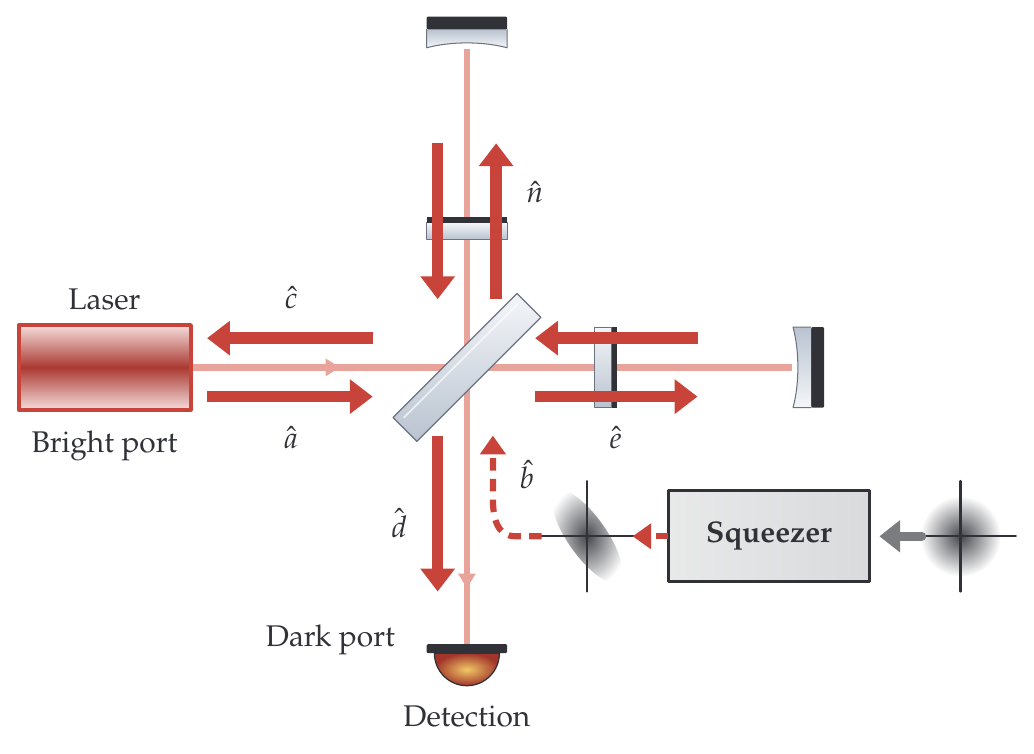}
    \caption{Simplified schematic of LIGO as a Michelson interferometer, with a squeezed light/vacuum state. Figure adapted from \cite{tse2022squeezed}.}
    \label{fig:LIGO_qu}
\end{figure}

There are two conventions for a 50/50 lossless beam splitter \cite{gerry2023introductory},
$
\frac{1}{\sqrt{2}} \begin{pmatrix}
    1 & i \\ i & 1
\end{pmatrix} \text{ or } \frac{1}{\sqrt{2}} \begin{pmatrix}
    -1 & 1 \\ 1 & 1
\end{pmatrix}.
$
To match the classical case, we will use the second.
\bb
\begin{pmatrix} \hat{n} \\ \hat{e} \end{pmatrix} = \frac{1}{\sqrt{2}}
\begin{pmatrix}
-1 & 1 \\
1 & 1
\end{pmatrix}
\begin{pmatrix} \hat{a} \\ \hat{b} \end{pmatrix}.
\ee
After propagation through the arms and reflecting off the mirrors, the operators accumulate their respective phases again and become 
\bb
\hat{n}' = e^{i \phi_n} \hat{n} = e^{i 2k L_n} \hat{n}, \quad \hat{e}' = e^{i \phi_e} \hat{e} = e^{i 2k L_e} \hat{e}.
\ee
Passing back through the beamsplitter again gives output operators at the bright $(\hat{c})$ and dark $(\hat{d})$ detection ports, which can again be calculated as
\bb
\begin{pmatrix} \hat{c} \\ \hat{d} \end{pmatrix} = \frac{1}{\sqrt{2}}
\begin{pmatrix}
-1 & 1\\
1 & 1
\end{pmatrix}
\begin{pmatrix} \hat{n}' \\ \hat{e}' \end{pmatrix}.
\ee
The operator at the dark detection port 
\bb
    \hat{d} &= \frac{1}{\sqrt{2}} \left( e^{i\phi_n} \hat{n} + e^{i\phi_e} \hat{e} \right)\\
    &= \frac{1}{2} \left( e^{i\phi_n} (\hat{b} - \hat{a}) + e^{i\phi_e} (\hat{a} + \hat{b}) \right) \\
    &= e^{ik(L_e + L_n)}\left(i \hat{a} \sin(k \delta L) + \hat{b} \cos(k \delta L)\right) \,.
\ee

Note that the interferometer is usually tuned close to a dark fringe $(L_n \approx L_e$), such that the $\delta L \ll \lambda$. This allows us to use $\cos(k \delta L) \approx 1$ and $\sin(k\delta L) \approx k\delta L$ \cite{danilishin2012quantum}.  
Furthermore, as originally proposed by Caves \cite{caves1981quantum}, we can decompose the strong laser signal operator into its large classical mean amplitude and its quantum fluctuations ($\hat{a} = \alpha + \delta \hat{a})$~\cite{danilishin2012quantum}:

\bb
\hat{d} &\approx e^{ik(L_e + L_n)}\left(i \hat{a} k \delta L + \hat{b}\right) \\
&\approx e^{ik(L_e + L_n)}\left(i (\alpha + \delta \hat{a}) k \delta L + \hat{b}\right) \\
&\approx e^{ik(L_e + L_n)}\left(i \alpha k \delta L + \hat{b}\right) \\
&\approx i \alpha k \delta L + \hat{b}\,,
\ee
where in the second to last line we neglect the signal-dependent fluctuation term $i\delta \hat{a} k \delta L$, since it is a product of the small signal-induced phase shift $k\delta L$ and the laser field fluctuation $\delta\hat a$, and is therefore a second-order term in the linearized approximation, negligible compared to the leading vacuum fluctuations from $\hat{b}$.

For the last line, the global phase $e^{ik(L_e + L_n)}$ gets absorbed in the homodyne detection, so we can ignore it.

The signal at the dark port is the vacuum state from the unused port $\hat{b}$, linearly displaced by $i \alpha k \delta L$. To see how this setup realizes our sensing model $H(t)=f(t)\,\hat{p}$ at the level of the induced quadrature displacement, we convert the annihilation operators $\hat{b}, \hat{d}$ to their position and momentum quadratures,
\bb
\hat{x}_{\text{out}} = \frac{\hat{d} + \hat{d}^\dagger}{\sqrt{2}}, \quad& \hat{p}_{\text{out}} = \frac{\hat{d} - \hat{d}^\dagger}{\sqrt{2} i} \\
\hat{x}_{\text{dark}} = \frac{\hat{b} + \hat{b}^\dagger}{\sqrt{2}}, \quad&\hat{p}_{\text{dark}} = \frac{\hat{b} - \hat{b}^\dagger}{\sqrt{2} i}.
\ee

Because $\alpha$ is the complex amplitude of the driven laser, we can tune its phase to be purely imaginary $\alpha = -i|\alpha|$ so that the displacement is purely real, $|\alpha|k\delta L$. Solving for $\hat{x}_{\text{out}}$ and $\hat{p}_{\text{out}}$,
\bb
\hat{x}_{\text{out}} &= \hat{x}_{\text{dark}} + \sqrt{2}|\alpha| k \delta L, \\
\hat{p}_{\text{out}} &= \hat{p}_{\text{dark}}.
\ee
The physical signal is a displacement in the $\hat{x}$ quadrature, which has the same form as that generated by the unitary time-evolution of the Hamiltonian $H(t)=f(t)\,\hat{p}$, with the integrated signal $\int f(t)\,\mathrm{d}t$ proportional to $\delta L$, which encodes the gravitational-wave signal.

In LIGO, this displaced phase quadrature is extracted via homodyne detection. Crucially, within the linearized model above, we showed that the leading quantum noise on the measurement of $\hat{x}_{\text{out}} = \hat{x}_{\text{dark}} + \sqrt{2} |\alpha| k \delta L$ does not come from the laser. Instead, it arises from $\hat{x}_{\text{dark}}$, the vacuum fluctuations entering from the unused dark port $\hat{b}$.

If we instead input a squeezed state in the dark port, aligning the squeezing with the desired quadrature, we can suppress the measurement noise on the gravitational-wave signal. LIGO has successfully demonstrated up to $2.15$ dB of high-frequency noise reduction, and over $3$ dB of broadband squeezing \cite{SGravi2, tse2019quantumenhanced}.

\subsection{Superconducting bosonic cavities}

Next, we go into how superconducting cavities can measure different signals, and provide a direct realization of our model $H(t) = f(t) \, \hat{p}$.

A superconducting cavity can act as a sensor for signals that induce phase-space displacements of its electromagnetic mode. We will go into how incoming microwave signals drive the stored electromagnetic mode.

For a single mode in the cavity, the undriven Hamiltonian is
\bb 
\hat H_0 = \hbar\omega_r\left(\hat a^\dagger\hat a+\frac12\right). 
\ee
A microwave signal is delivered through a transmission line and couples into the mode. The resulting drive Hamiltonian is
\bb \hat H_{\mathrm d}(t) = \hbar\left[ \varepsilon(t)\hat a^\dagger e^{-i\omega_{\mathrm d}t-i\phi_{\mathrm d}} + \varepsilon^*(t)\hat a e^{i\omega_{\mathrm d}t+i\phi_{\mathrm d}} \right], \ee
with
\bb \varepsilon(t)=i\sqrt{\kappa}\,A(t). \ee
Here $A(t)$ is the coherent input envelope, $\omega_{\mathrm d}$ is its carrier frequency, $\phi_{\mathrm d}$ is its phase, and $\kappa$ is the coupling. Note that $|A(t)|^2$ is the incident photon flux. We point the reader to \cite{blais2021circuit} for a full derivation.

Transforming to a frame rotating at the drive frequency $\omega_{\mathrm d}$ and dropping the constant zero-point energy gives
\bb 
\hat H_{\mathrm{rot}}(t) = \hbar(\omega_r-\omega_{\mathrm d})\hat a^\dagger\hat a + \hbar\left[ \varepsilon(t)e^{-i\phi_{\mathrm d}}\hat a^\dagger + \varepsilon^*(t)e^{i\phi_{\mathrm d}}\hat a \right]. 
\ee

For a real input signal $A(t)$, a resonant signal $(\omega_{\mathrm d}=\omega_r)$, and phase $\phi_{\mathrm d}=0$, the signal interaction becomes
\bb 
\hat H(t) = i\hbar\sqrt{\kappa}\,A(t) (\hat a^\dagger-\hat a). 
\ee
Substituting in $\hat p=\frac{\hat a-\hat a^\dagger}{i\sqrt2},$
we get
\bb
\hat H_{\mathrm{sig}}(t)=f(t)\hat p, \qquad f(t)=\hbar\sqrt{2\kappa}\,A(t). 
\ee

Thus, the microwave envelope controls the strength of a displacement along a fixed direction in cavity phase space.
The envelope $A(t)$, not the microwave carrier oscillation, is the waveform appearing in this rotating-frame model.

Experiments have directly realized this microwave signal sensing we've outlined \cite{deng, penasa}. In both experiments, an applied coherent microwave signal displaces a mode stored in a superconducting cavity, and the displacement amplitude is inferred from measurements of the resulting quantum state. Ref.~\cite{deng} use Fock-state probes and superconducting-qubit-assisted parity measurements, whereas Ref.~\cite{penasa} use an entangled Rydberg atom cavity probe and atomic state detection. These experiments therefore provide direct implementations of the microwave signal sensing explained here, using different quantum probes and readout procedures.

\subsection{Cavity optomechanics}
Related continuous-variable sensing arise in cavity optomechanics, although these need not use superconducting microwave cavities. An external force couples to a mechanical displacement through $\hat H_F(t)=-F(t)\hat q$. An optical cavity can read out the resulting motion through its displacement dependent resonance \cite{aspelmeyer2014cavity}. This is analogous, at the level of optical displacement readout, to the cavity-enhanced interferometric measurement used in LIGO.

Experimental demonstrations include squeezed light enhanced optomechanical magnetometry, in which a magnetic signal induces an effective mechanical force \cite{li}, measurements of the displacement fluctuations of a nanomechanical resonator \cite{amir2012}, and continuous force and displacement measurements below the standard quantum limit \cite{mason}. Optomechanical generation of squeezed light provides a related demonstration of quantum measurement correlations \cite{amir}. Together, these experiments illustrate the ability to sense with cavity optomechanics.

\newpage


\section{Sensing signals via Gaussian protocols with bounded squeezing}
\label{sec:gaussian_sensing_model}

We consider a single-mode continuous-variable system used to sense a
time-dependent signal. The signal is modelled as a coupling to the
momentum quadrature, with Hamiltonian
\bb\label{hamiltonian_classical_signal}
    \hat{H}(t)\coloneqq f(t)\,\hat{p},
\ee
where \(\hat{p}\) is the momentum operator and
\(f:\mathbb{R}^{+}\to\mathbb{R}\) describes the time profile of the signal.
Throughout the paper, we assume that the signal is composed of a sequence
of unit-duration pulses with a fixed shape. Specifically, we write
\bb\label{function_classical_signal}
    f(t)\coloneqq
    \sum_{i\in\mathbb{N}^{+}}
    m_i\,\varphi(t-i+1),
\ee
where each $m_i\in\mathbb{R}$ is a coefficient and
$\varphi:\mathbb{R}\to\mathbb{R}$ is a fixed and known pulse shape
satisfying
\begin{equation}\label{eq:pulse_assumptions}
    \varphi\in C(\mathbb{R}),
    \qquad
    \operatorname{supp}(\varphi)\subseteq[0,1],
    \qquad
    \varphi(t)>0\ \text{for }t\in(0,1),
    \qquad
    \int_0^1\varphi(t)\,\mathrm{d}t=1.
\end{equation}
Thus, on the $i$-th unit-time interval, the signal has the same temporal
profile $\varphi$, scaled by the coefficient $m_i$. A typical realization
is illustrated in Fig.~\ref{fig:pulsed_signal}. In particular, the
normalization of $\varphi$ implies
\begin{equation}\label{eq:pulse_bin_area}
    \int_{i-1}^{i}f(t)\,\mathrm{d}t=m_i.
\end{equation}
Hence $m_i$ represents the total integrated signal over the $i$-th
time interval. The precise shape of $\varphi$ will not play any role in
our results, provided that it satisfies the conditions above.

\begin{figure}[h!]
    \centering
    \includegraphics[width=\linewidth]{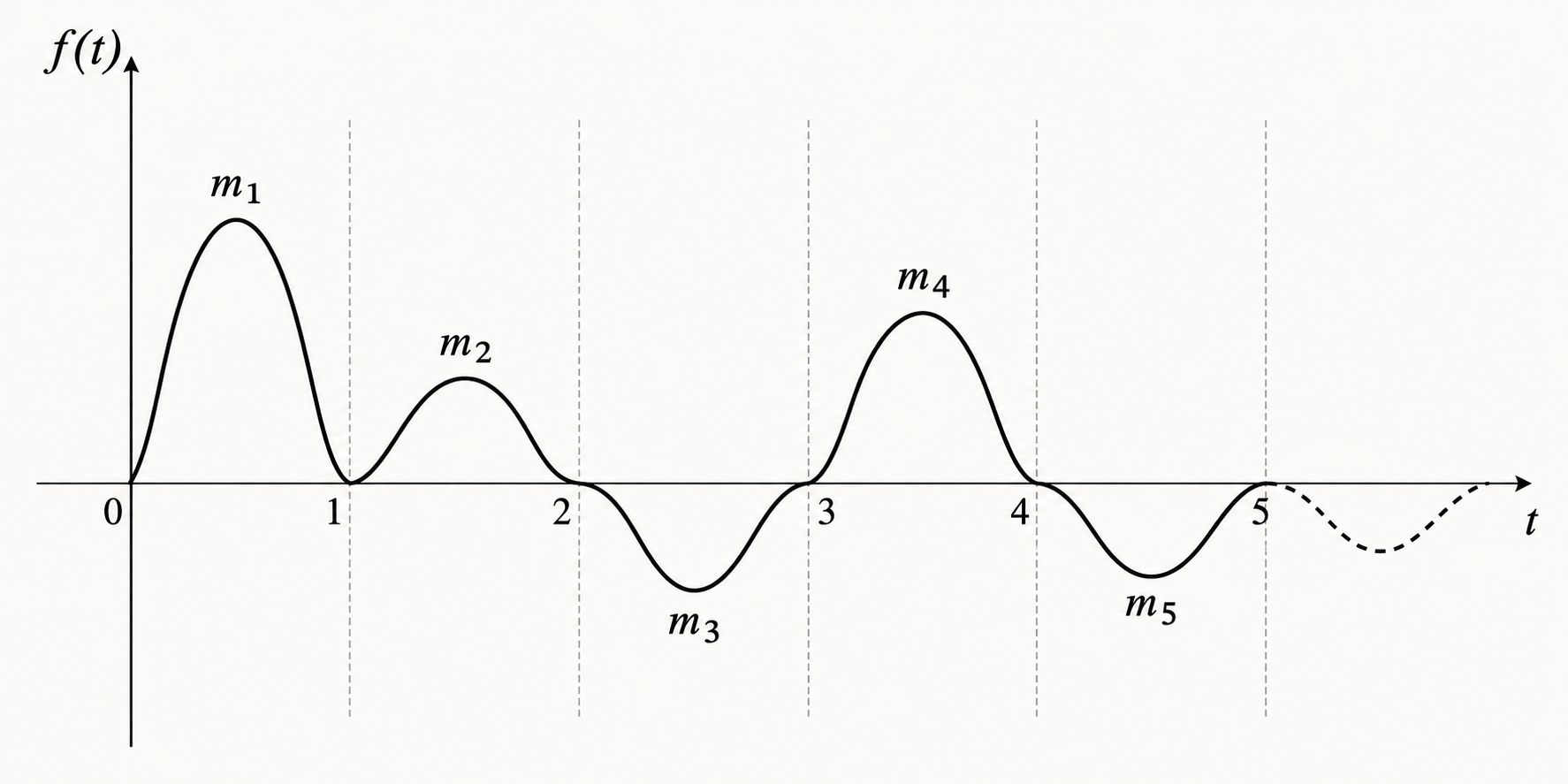}
    \caption{
        \textbf{A continuous pulsed signal.}
        A typical realization of the signal
        $f(t)=\sum_{i\ge1} m_i\,\varphi(t-i+1)$.
        Each coefficient $m_i$ scales the same unit-duration pulse shape
        $\varphi$, so that the signal is continuous and supported on
        consecutive unit-time intervals.
    }
    \label{fig:pulsed_signal}
\end{figure}

The key feature of this model is that the signal is imprinted
directly as a displacement of the position quadrature. By
Eq.~\eqref{eq:pulse_bin_area}, evolution over the complete $i$-th
unit-time interval, without intervening controls, induces the unitary
$e^{-im_i\hat p}$. In the Heisenberg picture, this gives
\[
    e^{im_i\hat p}\hat x e^{-im_i\hat p}=\hat x+m_i.
\]
More generally, in the absence of intervening controls, the Hamiltonians
at different times commute. Thus, over an interval
$[t_{\mathrm{in}},t_{\mathrm{out}}]$, the evolution generated by
\eqref{hamiltonian_classical_signal} is

\[
    \exp\!\left(
        -i
        \left(
            \int_{t_{\mathrm{in}}}^{t_{\mathrm{out}}}
            f(t)\,\mathrm{d}t
        \right)
        \hat{p}
    \right),
\]
and therefore it displaces the position quadrature as
\[
    \hat{x}
    \longmapsto
    \hat{x}
    +
    \int_{t_{\mathrm{in}}}^{t_{\mathrm{out}}}
    f(t)\,\mathrm{d}t .
\]
Thus, the information about the signal is encoded in a phase-space
displacement, which can then be accessed through measurements of the
single-mode quantum sensor.

In this setting, a quantum sensing task consists of inferring a relevant property of the random sequence $(m_i)_{i\in\mathbb{N}^+}$ from measurements performed on the evolved sensor. The problems considered in this work are defined in terms of a parameter $T\in\mathbb{N}^+$, called the \emph{pattern size}. Given $T$, the sequence $(m_i)_{i\in\mathbb{N}^+}$ is partitioned into consecutive blocks of length $T$:
\bb\label{eq:block_decomposition_general}
    (m_i)_{i\in\mathbb{N}^+}
    =
    \bigl(m^{(1)},m^{(2)},\ldots,m^{(n)},\ldots\bigr),
\ee
where each block vector
\[
    m^{(n)}
    =
    \bigl(m^{(n)}_1,\ldots,m^{(n)}_T\bigr)
    \in\mathbb{R}^T
\]
is drawn independently from a probability distribution belonging to a prescribed family. The learner knows this family of possible distributions, but not which member of the family is generating the blocks. A quantum sensing problem, in this framework, consists of identifying the unknown distribution, or a parameter specifying it, with high probability after accessing the Hamiltonian evolution for a finite amount of time.

We study two concrete problems within this setting. In Section~\ref{sec_cv_parity}, the block distribution is determined by an unknown parity constraint; this gives the \emph{CV Learning Parity problem}. In Section~\ref{sec_hyp_test}, the learner is instead promised that the block distribution is one of two known distributions; this gives the \emph{CV Hypothesis Testing problem}.
\subsection{Gaussian protocols with bounded squeezing}

We now specify the physically relevant class of quantum sensing protocols considered in this paper. We focus on Gaussian protocols in which the amount of squeezing is bounded by a parameter $z\ge 1$. Operationally, $z$ is the maximum squeezing that the learner is allowed to use when preparing the Gaussian state of the single-mode sensor. Equivalently, it determines the smallest quadrature variance that can be engineered at the input of each sensing experiment.

The case $z=1$ corresponds to unsqueezed Gaussian states, such as coherent states, and represents the natural baseline regime. Larger values of $z$ allow the learner to prepare squeezed states, and therefore to reduce the noise in the quadrature in which the displacement generated by the signal is read out. From an experimental viewpoint, coherent states are routinely available in optical platforms, whereas strongly squeezed states require nonlinear optical interactions or related advanced quantum technologies~\cite{dowling2002quantumtechnologysecondquantum,BUCCO}. Thus, in our model, the squeezing parameter $z$ plays the role of the quantum resource: small values of $z$ correspond to weakly quantum, essentially classical, Gaussian sensing protocols, while large values of $z$ correspond to protocols that use genuinely nonclassical states of the sensor. 

It is useful to keep in mind the limiting case $z=\infty$, even though
it is not physically attainable. In that idealized limit, the pulse
coefficients $m_i$ could be read exactly. For instance, one could prepare
an infinitely position-squeezed state centered at the origin, let it
evolve over one complete unit-time bin, and then perform an ideal
position homodyne measurement. By Eq.~\eqref{eq:pulse_bin_area}, the
integrated signal over that bin is $m_i$, so the measurement outcome would reveal this coefficient without noise. 

It is therefore natural to ask how the performance of quantum sensing depends on the available squeezing, and in particular at what scale of $z$ the sensing task becomes efficiently solvable. The relevant notions of efficiency will be defined below, but first we define a Gaussian protocol with squeezing bounded by a finite parameter $z$.


\begin{Def}[(Gaussian protocols with squeezing bounded by $z$)]
\label{def:gaussian_protocol_bounded_squeezing}
Let $z\ge1$. A sensing protocol for the single-mode system evolving under the Hamiltonian in \eqref{hamiltonian_classical_signal} over the time interval $[0,T_{\mathrm{sensing}}]$ is called a \emph{Gaussian protocol with squeezing bounded by $z$} if it is restricted to the following operations:
\begin{itemize}
    \item preparation of single-mode Gaussian input states whose squeezing is at most $z$;
    \item application of single-mode Gaussian unitaries during the evolution, provided that these unitaries introduce no additional squeezing, i.e.~only phase-space rotations and displacements are allowed;
    \item application of arbitrary Gaussian measurements, followed by arbitrary classical post-processing of the measurement outcomes.
\end{itemize}
\end{Def}

The protocol may be fully adaptive. In particular, the learner may choose the next input state, waiting time, Gaussian unitary, and Gaussian measurement as an arbitrary function of all previous measurement outcomes. We also emphasize that the model is prepare-and-measure: after a measurement, the system is discarded and a fresh Gaussian state may be prepared. Non-destructive measurements are not included.

Equivalently, the most general protocol in this class can be described as follows. At time $0$, the learner prepares a Gaussian state with squeezing at most $z$, chooses a waiting time $t_1$, applies arbitrary rotations and displacements during the evolution, and performs a Gaussian measurement at time $t_1$. After observing the outcome, the learner chooses a new Gaussian input state with squeezing at most $z$, a new waiting time $t_2$, new rotations and displacements, and a new Gaussian measurement at time $t_1+t_2$. This procedure is repeated adaptively until the final measurement is performed at time
\[
    t_1+\cdots+t_N=T_{\mathrm{sensing}}.
\]
The total number of measurements $N$ may itself depend on the outcomes observed during the protocol.

\subsection{Sensing time and computational time}

We distinguish between two resources. The first is the physical time during which the sensor has access to the signal. The second is the total time needed to complete the protocol, including classical post-processing.

\begin{Def}[(Sensing time and computational time)] 
\label{def:sens}
Fix a quantum sensing problem.
\begin{itemize}
    \item The \emph{sensing time} $T_{\mathrm{sensing}}$ of a protocol is the total duration for which the protocol has access to the Hamiltonian evolution. Equivalently, the protocol interacts with the signal only during the time window $[0,T_{\mathrm{sensing}}]$.
    \item The \emph{computational time} of a protocol is the total running time required to solve the sensing problem, including both the sensing phase and all subsequent classical post-processing.
\end{itemize}
\end{Def}

The computational time is always at least the sensing time. Nevertheless, the two notions are not equivalent. A protocol may collect enough physical data in polynomial sensing time, but require exponential classical post-processing to extract the answer.

\begin{Def}[(Sensing-time efficiency and computational efficiency)] 
\label{def:sens2}
Fix a quantum sensing problem with pattern size $T$. 
\begin{itemize}
    \item A protocol that solves the problem is called \emph{sensing-time efficient} if its sensing time $T_{\mathrm{sensing}}$ scales polynomially with $T$. Otherwise, it is called \emph{sensing-time inefficient}.
    \item  A protocol that solves the problem is called \emph{computationally efficient} if its total computational time scales polynomially with $T$. Otherwise, it is called \emph{computationally inefficient}.
\end{itemize}
\end{Def}
Thus computational efficiency is the stronger requirement. If a protocol is sensing-time inefficient, then it is automatically computationally inefficient. However, a sensing-time efficient protocol may still be computationally inefficient if its classical post-processing is too expensive.

\newpage

\section{CV Learning Parity problem}
\label{sec_cv_parity}

We now introduce the first sensing task studied in this work. The block vectors are bit strings, and the unknown distribution is specified by a hidden parity constraint. For every $s\in\{0,1\}^T$, define
\bb\label{eq:def_Cs}
    C_s
    \coloneqq
    \left\{
    x\in\{0,1\}^T:
    x\cdot s\equiv0 \pmod 2
    \right\},
\ee
where $x\cdot s\coloneqq \sum_{i=1}^T x_i s_i$.  In the CV Learning Parity problem, there is an unknown string $s\in\{0,1\}^T$. The pulse coefficients are generated as follows. The underlying bit sequence is decomposed into blocks
\bb\label{eq:parity_block_decomposition}
    (x_i)_{i\in\mathbb{N}^+}
    =
    \bigl(x^{(1)},x^{(2)},\ldots,x^{(n)},\ldots\bigr),
\ee
where each block $x^{(n)}\in\{0,1\}^T$ is sampled independently and uniformly from $C_s$. The corresponding signal is obtained by setting $m_i=x_i$ for all $i\in\mathbb{N}^+$ in \eqref{function_classical_signal}. The learner has access to the Hamiltonian evolution in \eqref{hamiltonian_classical_signal} generated by this signal and must recover the hidden string $s$.

\begin{problem}[(CV Learning Parity problem)] 
\label{def:PROB}
Let $T\in\mathbb{N}^+$ be the pattern size. Let $s\in\{0,1\}^T$ be an unknown bit string. The learner is given access to the time evolution generated by the Hamiltonian \eqref{hamiltonian_classical_signal}, where the signal is constructed from independent block vectors sampled uniformly from
\[
    C_s
    =
    \left\{
    x\in\{0,1\}^T:
    x\cdot s\equiv0 \pmod 2
    \right\}.
\]
The goal is to output a guess $\tilde{s}\in\{0,1\}^T$ such that $\tilde{s}=s$ with probability at least $2/3$.
\end{problem}
The success probability of $2/3$ is arbitrary and can be amplified straightforwardly. The central question is how the minimum sensing time depends on the squeezing parameter $z$. Our main result is the following.

\begin{thm}[(Complexity of the CV Learning Parity problem)]
\label{thm_main_compl}
Let $T\in\mathbb{N}^+$ be the pattern size and let $z\ge1$ be the squeezing parameter.
Every Gaussian protocol with squeezing bounded by $z$ that solves the CV Learning Parity problem  must use a sensing time of 
\bb\label{eq:main_parity_lower_bound}
    T_{\mathrm{sensing}}
    \ge
    T\left[
    \frac13
    \exp\!\left(\frac{e^{-3z^2}}{90}T\right)
    -1
    \right].
\ee
Conversely, there exists a Gaussian protocol with squeezing bounded by $z$ that solves the CV Learning Parity problem using sensing time
\bb\label{eq_upp_bound_protocol}
    T_{\mathrm{sensing}}
    =
    T\left\lceil
    8\log\!\left(3\cdot2^T\right)
    \exp\!\left(4e^{-z^2/4}T\right)
    \right\rceil .
\ee
Moreover, if $z\ge 2\sqrt{2\log(10T)}$, then there exists a Gaussian protocol with squeezing bounded by $z$ that solves the CV Learning Parity problem with computational time polynomial in $T$.
\end{thm}
The proof of the lower bound is given in Theorem~\ref{thm_lower_bound_learning_parity_gaussian}. The sensing-time upper bound is proved in Theorem~\ref{thm_upp1}, while the polynomial-time upper bound in the high-squeezing regime is proved in Theorem~\ref{thm_upp1_comp}.

The above Theorem~\ref{thm_main_compl} shows that \emph{all} Gaussian protocols with constant squeezing are inherently inefficient: they require a sensing time that grows exponentially with the pattern size $T$. More explicitly, if $z$ is constant, then $e^{-3z^2}$ is a positive constant, and the lower bound in \eqref{eq:main_parity_lower_bound} implies that the minimum sensing time grows exponentially in $T$. More generally, if $z=o(\sqrt{\log T})$, then the minimum sensing time is still exponential in $T$, with a lower bound of $\exp(T^{1-o(1)})$. On the other hand, once the squeezing reaches the scale $\sqrt{\log T}$, the above upper bounds show that the minimum sensing time scales polynomially in $T$, and the problem becomes efficiently solvable. Thus, up to constants, the transition occurs at
\[
    z=\Theta(\sqrt{\log T}) .
\]
We summarise this result in the following corollary:

\begin{cor}[(Exponential separations in CV Learning Parity problem)]\label{cor:cv_parity_advantage}
Let $z$ be the squeezing parameter and $T$ the pattern size. Then:
\begin{itemize}
    \item If $z=o(\sqrt{\log T})$, all Gaussian protocols with squeezing bounded by $z$ that solve the CV Learning Parity problem must be inefficient, with sensing time scaling exponentially in $T$. More precisely, the sensing time must scale at least as $\exp(T^{1-o(1)})$.
    \item If $z=\omega(\sqrt{\log T})$, there exists a Gaussian protocol with squeezing bounded by $z$ that efficiently solves the CV Learning Parity problem, with computational time scaling polynomially in $T$.
\end{itemize}
\end{cor}
This corollary reveals a sharp computational phase transition: when the squeezing parameter crosses the threshold $\Theta(\sqrt{\log T})$, the computational time drops from exponential to polynomial. Specifically, there is a Gaussian protocol with high squeezing ($z= \omega(\sqrt{\log T})$) that achieves an \emph{exponential speedup} over all low-squeezing protocols ($z = o(\sqrt{\log T})$). As discussed earlier, protocols with no squeezing ($z = 1$) are effectively \emph{classical}, while those with high squeezing are \emph{quantum}. In this sense, our result can also be seen as a form of \emph{exponential quantum advantage}: quantum protocols improve exponentially over any classical protocol.

\subsection{Reduction to a simpler problem}

We now introduce a variant of the CV Learning Parity problem, which we call the
\emph{CV Learning Parity with Partially Revealed Information} problem. This variant will be
used as an intermediate step in the lower-bound proof of the CV learning parity problem. Instead of asking the learner to recover
the whole hidden string $s$, we ask it to distinguish between two possibilities, $s=0$ and
$s=\bar{s}$, where the non-zero string $\bar{s}$ is revealed only after the sensing phase. Thus the
learner is helped at the post-processing stage, but not while interacting with the signal.

\begin{problem}[(CV Learning Parity with Partially Revealed Information)]
\label{def:PROB2}
Let $T\in\mathbb{N}^+$ be the pattern size. The problem consists of four phases:
\begin{itemize}
    \item \textbf{Referee's choice phase.} A referee samples a non-zero string
    \[
        \bar{s}\in\{0,1\}^T\setminus\{0\}
    \]
    uniformly at random. The referee then chooses the hidden string $s$ to be either
    $s=0$ or $s=\bar{s}$, with equal probability. At this stage, the learner knows neither
    $s$ nor $\bar{s}$.

    \item \textbf{Evolution phase.} The learner has access to the Hamiltonian evolution generated
    by the signal associated with the hidden string $s$. Equivalently, the signal is
    constructed from independent block vectors sampled uniformly from
    \[
        C_s
        \coloneqq
        \left\{
        x\in\{0,1\}^T:
        x\cdot s\equiv 0 \pmod 2
        \right\}.
    \]
    After this phase, the learner no longer has access to the Hamiltonian evolution.

    \item \textbf{Partial information reveal phase.} The referee reveals the string $\bar{s}$, but
    does not reveal whether the true hidden string was $0$ or $\bar{s}$.

    \item \textbf{Classical post-processing phase.} Using the data collected during the evolution
    phase, together with the revealed string $\bar{s}$, the learner must decide whether the
    true hidden string was $s=0$ or $s=\bar{s}$, with success probability at least $2/3$.
\end{itemize}
\end{problem}

The order of the phases is the important point. Since $\bar{s}$ is revealed only after the sensing
phase, the learner cannot choose its sensing strategy as a function of the parity direction that
will later be tested. During the evolution phase, the protocol must collect data without knowing
which non-zero parity direction will become relevant.

Problem~\ref{def:PROB2} is no harder than the original CV Learning Parity problem. Indeed, any
protocol that learns the hidden string $s$ also solves the partially revealed problem: the learner
runs the same sensing protocol, obtains an estimate of $s$, and, after $\bar{s}$ is revealed, decides
whether the estimate is $0$ or $\bar{s}$. Hence any lower bound for Problem~\ref{def:PROB2}
immediately implies the same lower bound for Problem~\ref{def:PROB}. This is why the partially
revealed formulation is useful: it reduces the lower-bound argument to a binary testing problem,
while still retaining the obstruction caused by not knowing the relevant parity direction during
the sensing phase.

The quantitative statement for this variant is the following.

\begin{thm}[(Complexity of CV Learning Parity with Partially Revealed Information)]
\label{thm:partial_revealed_complexity}
Let $T\in\mathbb{N}^+$ be the pattern size and let $z\ge1$ be the squeezing parameter.
Every Gaussian protocol with squeezing bounded by $z$ that solves
Problem~\ref{def:PROB2} must use sensing time
\bb
    T_{\mathrm{sensing}}
    \ge
    T\left[
    \frac13
    \exp\!\left(\frac{e^{-3z^2}}{90}T\right)
    -1
    \right].
\ee
Conversely, there exists a Gaussian protocol with squeezing bounded by $z$ that solves the same problem using sensing time
\bb
    T_{\mathrm{sensing}}
    =
    10\,T
    \exp\!\left(4e^{-z^2/4}T\right)
\ee
and computational time $O(T_{\mathrm{sensing}})$.
\end{thm}

The proof of the lower bound is given in Theorem~\ref{thm_lower_bound_learning_parity_gaussian}, while the upper bound is proved in Theorem~\ref{thm_correct2}.

Finally, let us stress that the timing of the reveal is crucial for obtaining an exponential lower bound in $T$ on the minimum sensing time. If the learner already knows $\bar{s}$ before the sensing phase, then the situation changes completely: in Section~\ref{sec_partial_info_bef} we show that, when $\bar{s}$ is known beforehand, the corresponding problem of distinguishing whether the true string is $s=\bar{s}$ or $s=0$ can be solved with sensing time linear in $T$, even with no squeezing, i.e.~with $z=1$. For example, this includes the task of distinguishing whether the block vector is sampled uniformly from $\{0,1\}^T$ or uniformly from the strings satisfying a known global parity constraint such as $\sum_{j=1}^T x_j\equiv0\pmod 2$. Once the parity direction is known in advance, such a binary test can be performed efficiently by an appropriate strategy with no squeezing. We refer to Section~\ref{sec_partial_info_bef} for the explicit protocol and its analysis.

\section{CV Hypothesis Testing problem}
\label{sec_hyp_test}

We now introduce the second quantum sensing task studied in this work: the CV Hypothesis Testing problem. In the previous section, the goal was to recover an unknown parity string. Here the setting is different: the learner is given two known distributions over the \emph{signs} of the signal, and must decide which one is generating the signal.

Specifically, the block vectors $m\in\mathbb{R}^T$ of the signal are i.i.d.~random variables constructed as follows. Let $p$ and $q$ be two probability distributions over $\{0,1\}^T$. In each block of length $T$, a bit string
\[
    x=(x_1,\ldots,x_T)\in\{0,1\}^T
\]
is sampled from one of these two distributions. Independently, amplitudes $b_1,\ldots,b_T$ are sampled uniformly from the interval $[\frac12,\frac32]$. The corresponding block vector of the signal is defined as
\bb\label{eq:hyp_testing_block}
    m
    =
    \bigl(
    (-1)^{x_1}b_1,\ldots,(-1)^{x_T}b_T
    \bigr)
    \in\mathbb{R}^T .
\ee
Thus, $b_i$ is the magnitude of the $i$-th pulse coefficient, while $x_i$ determines its sign. The signal is obtained by inserting these
coefficients into Eq.~\eqref{function_classical_signal}. Since
$\varphi>0$ inside each bin, the signal has the sign of its coefficient
throughout the interior of that bin. The promise is that, throughout the whole evolution, the sign string $x$ is sampled independently in each block either always from $p$ or always from $q$. The distributions $p$ and $q$ are known to the learner, and the task is to decide which of these two alternatives holds. The problem can be formalized as follows.

\begin{problem}[(CV Hypothesis Testing problem)]
\label{def:cv_hypothesis_testing}
Let $T\in\mathbb{N}^+$ be the pattern size, and let $p$ and $q$ be two known probability distributions over $\{0,1\}^T$. The learner has access to the Hamiltonian evolution in \eqref{hamiltonian_classical_signal} generated by a signal whose block vectors are of the form \eqref{eq:hyp_testing_block}. In each block, the bit string $x$ is sampled either always from $p$ or always from $q$, while the amplitudes $b_1,\ldots,b_T$ are sampled independently and uniformly from $[\frac12,\frac32]$. The goal is to decide which of the two hypotheses is true with success probability at least $2/3$.
\end{problem}
As before, the choice of the constant success probability $2/3$ is arbitrary and can be amplified by repetition. The interval $[\frac12,\frac32]$ is chosen for concreteness; a similar analysis applies to any fixed interval $[a,b]$ with $0<a<b$, with corresponding changes in the constants.

We now identify a broad class of pairs of distributions for which this sensing task exhibits an exponential separation. The relevant condition is that the two sign distributions agree on all sufficiently small subsets of coordinates. More precisely, for an integer $c\in[T]$, we say that two distributions $p$ and $q$ over $\{0,1\}^T$ have \emph{matching marginals up to size $T-c$} if, for every subset $S\subseteq[T]$ with $|S|\le T-c$, the marginal distributions of $p$ and $q$ on the coordinates in $S$ coincide. Equivalently, any test that observes at most $T-c$ bits of the sign pattern has exactly the same distribution under $p$ and under $q$.

A simple example is obtained by taking $p$ to be the uniform distribution over $\{0,1\}^T$, and $q$ to be the uniform distribution over
\[
    C_s
    =
    \left\{
    x\in\{0,1\}^T:
    x\cdot s\equiv0 \pmod 2
    \right\},
\]
where the parity string $s$ has Hamming weight $|s|\ge T-c+1$. In this case, the parity constraint cannot be detected unless one observes all coordinates in the support of $s$. Hence the marginals of $p$ and $q$ coincide on every subset of size at most $T-c$. At the same time, the two distributions are separated by a constant in total variation distance; indeed, for every non-zero $s$, $\frac12\|p-q\|_1=\frac12$.

The theorem below shows that, for any pair of distributions satisfying this matching-marginal condition, Gaussian protocols with low squeezing can distinguish the two hypotheses only very slowly. Conversely, once the squeezing is large enough, the signs of the signal can be read accurately enough to solve the corresponding hypothesis test efficiently.

\begin{thm}[(Complexity of the CV Hypothesis Testing problem)]
\label{thm:hypothesis_testing_complexity}
Let $T\in\mathbb{N}^+$ be the pattern size and let $z\ge1$ be the squeezing parameter. Let $c\in[T]$, and let $p,q$ be two distinct probability distributions over $\{0,1\}^T$ whose marginals on every subset of at most $T-c$ coordinates coincide. Every Gaussian protocol with squeezing bounded by $z$ that solves Problem~\ref{def:cv_hypothesis_testing} must use sensing time 
\bb\label{eq:hyp_testing_lower_bound}
    T_{\mathrm{sensing}}
    \ge
    T\left(
    \frac{
    \frac23
    \exp\!\left(
    \frac{e^{-25z^2}}{30}
    \bigl(T-(c-1)\bigr)
    \right)
    }{
    \|p-q\|_1
    \left(\sum_{r=0}^{c-1}\binom{T}{r}\right)
    }
    -1
    \right).
\ee
Conversely, there exists a Gaussian protocol with squeezing bounded by $z$ that solves the same problem with sensing time bounded by
\bb\label{eq:hyp_testing_achievability_clean}
    T_{\mathrm{sensing}}
    \le
    \frac{12\,T\exp\!\left(
        4T e^{-z^2/4}
    \right)}{\|p-q\|_1^2}
    .
\ee
\end{thm}
The proof of the lower bound is given in Theorem~\ref{thm_lower_bound_ht_gaussian}, while the upper bound is proved in Theorem~\ref{thm_upper_bound_hyp_testing}.

 Let us spell out the consequence of
Theorem~\ref{thm:hypothesis_testing_complexity}. Suppose that
$c=O(T^{0.99})$,\footnote{\label{fn:c-growth}The exponent $0.99$ is
arbitrary: the same conclusion holds for $c=O(T^{1-\delta})$, for any
fixed constant $\delta\in(0,1)$. For bounded squeezing, the
exponential lower bound holds more generally whenever $c=o(T)$.} so
that $\sum_{r=0}^{c-1}\binom{T}{r}$ is at most
$\exp(O(T^{0.99}\log T))$, and that we are in the low-squeezing regime
$z=o(\sqrt{\log T})$. Then, as a consequence of the lower bound in
\eqref{eq:hyp_testing_lower_bound}, the minimum sensing time required
to solve the CV Hypothesis Testing problem must be at least
$\exp(T^{1-o(1)})$, and is therefore exponential in the pattern
size $T$. In this regime, the problem is fundamentally inefficient.

On the other hand, in the high-squeezing regime
$z=\omega(\sqrt{\log T})$, the term $e^{-z^2/4}$ is $o(1/T)$.
Hence, the exponential factor in
\eqref{eq:hyp_testing_achievability_clean} remains bounded by a
constant. In this regime, the sensing time is polynomial in $T$
whenever $\|p-q\|_1$ is at least inverse-polynomially large. 

This is the natural regime in which to discuss efficiency. Indeed, in the
opposite regime, where \(\|p-q\|_1\) is exponentially small, the minimum
sensing time required to solve the CV Hypothesis Testing problem is already
exponential in \(T\) for any value of the squeezing, even in the idealized
limit of infinite squeezing. In that limit, the learner could read the sign
patterns exactly, and the problem would reduce to ordinary classical
hypothesis testing between \(p\) and \(q\). As shown in
Theorem~\ref{thm_upper_bound_hyp_testing}, any test with success probability
at least \(2/3\) requires a number of samples \(N\) satisfying
\[
    N
    \ge
    \frac{\log(9/8)}
    {
        2\log\!\left(
            \frac{1}{
                1-\frac12\|p-q\|_1
            }
        \right)
    }.
\]
In particular, as \(\|p-q\|_1\to0\), this lower bound scales as
\(N=\Omega(1/\|p-q\|_1)\), and is therefore exponential in \(T\) whenever
\(\|p-q\|_1\) is exponentially small. Thus, in the regime where
\(\|p-q\|_1\) is at least inverse-polynomially large, we conclude that, for
\(z=\omega(\sqrt{\log T})\), the CV Hypothesis Testing problem can be solved
efficiently.

In summary, under this natural assumption on \(\|p-q\|_1\), as in the CV
Learning Parity problem, the transition between efficient and inefficient
sensing occurs at the scale \(z=\Theta(\sqrt{\log T})\).
We summarize this conclusion in the following corollary.

\begin{cor}[(Exponential separation in the CV Hypothesis Testing problem)]\label{cor:cv_hypothesis_testing_advantage}
Let $z\ge 1$ be the squeezing parameter and let $T$ be the pattern size. Consider any pair of distributions $p,q$ over $\{0,1\}^T$ whose marginals agree on every subset of at most $T-c$ coordinates, where $c=O(T^{0.99})$.\footref{fn:c-growth} Then:
\begin{itemize}
    \item If $z=o(\sqrt{\log T})$, all Gaussian protocols with squeezing bounded by $z$ that solve the CV Hypothesis Testing problem must be inefficient, with sensing time scaling exponentially in $T$. More precisely, the sensing time must scale at least as $\exp(T^{1-o(1)})$.
    \item If $z=\omega(\sqrt{\log T})$, there exists a Gaussian protocol with squeezing bounded by $z$ that solves the CV Hypothesis Testing problem with sensing time polynomial in $T$, provided that $\|p-q\|_1$ is bounded below by the inverse of a polynomial in $T$.\footnote{This assumption excludes the case in which the underlying classical hypothesis test is already exponentially hard, even with exact access to the sign patterns (which corresponds to the case of infinite squeezing).}
\end{itemize}
\end{cor}
This gives an exponential separation in sensing time for a broad family of signals. The matching-marginal condition captures the case in which the difference between $p$ and $q$ is hidden in global correlations of the sign pattern: locally, the two hypotheses look identical. Low-squeezing Gaussian protocols, including the coherent-state case $z=1$, cannot access these correlations efficiently. By contrast, squeezing of order $\sqrt{\log T}$ makes the signs readable with sufficiently high accuracy; after that, the problem reduces to an ordinary hypothesis test between the two known sign distributions, which can be performed in a sensing-time efficient manner.


\newpage

\section{Lower bound on the minimum sensing time}
\label{sec:lower_bound_minimum_sensing_time}

In this section we prove lower bounds on the minimum sensing time required by Gaussian protocols with squeezing bounded by $z$. The results apply to both the CV Learning Parity problem and the CV Hypothesis Testing problem introduced above.

We recall the common structure of the two problems. In both cases,
the sensor is a single-mode continuous-variable system evolving under
$\hat H(t)=f(t)\,\hat p$, where $f$ is the continuous pulsed signal in
Eq.~\eqref{function_classical_signal}. The random coefficients
$m_i\in\mathbb{R}$ are grouped into consecutive blocks of length $T$,
where $T$ is the pattern size:
\bb\label{eq:block_decomposition_lower_bound}
    (m_i)_{i\in\mathbb{N}^+}
    =
    \bigl(m^{(1)},m^{(2)},\ldots,m^{(n)},\ldots\bigr),
\ee
with each block $m^{(n)}\in\mathbb{R}^T$ sampled independently from an unknown distribution belonging to a prescribed family. The task of the sensing protocol is to identify this underlying distribution, or the parameter specifying it, with high probability.

We first reduce protocols for the pulsed signal to protocols for a
piecewise-constant signal, preserving the squeezing bound and the number
of blocks reached. For this auxiliary model, the lower-bound argument
has two main steps. First, we reduce each Gaussian measurement step to a canonical Gaussian observation: the unknown signal appears only through two linear functionals in the mean, while the covariance is fixed by the squeezing bound. This reduction is sufficient for lower bounds, since the actual Gaussian measurement outcome can be simulated from the canonical observation by a classical post-processing independent of the unknown signal. Second, we account for the adaptivity of the protocol through a tree representation. Each root-to-leaf path records one possible transcript of the adaptive sensing procedure, and the tree structure lets us propagate local bounds on the distinguishability generated by a single experiment into a bound on the distinguishability of the full protocol.

\subsection{From continuous pulses to piecewise-constant signals}
\label{sec:pulse_to_piecewise}

To prove our lower bounds, we first reduce to the simpler case of a
signal that is constant on each unit-time interval. Any protocol for the pulsed signal can be simulated in this model by changing the times of its operations within its unit time interval. It therefore suffices to lower bound
the number of blocks of T required to sense in the piecewise-constant model.

For the same coefficient sequence $(m_i)_{i\ge1}$, define the piecewise function
\begin{equation}\label{eq:step_reference_signal}
    f_{\mathrm{step}}(t)\coloneqq
    \sum_{i\in\mathbb{N}^{+}}m_i\,\mathbf{1}_{[i-1,i)}(t),
\end{equation}
and the corresponding Hamiltonian
\begin{equation}\label{eq:step_reference_hamiltonian}
    \hat H_{\mathrm{step}}(t)\coloneqq f_{\mathrm{step}}(t)\hat p.
\end{equation}
This signal is illustrated in Fig.~\ref{fig:piecewise_constant_signal}.
The sensing problems are otherwise unchanged: the coefficients have
the same distributions, the learner has the same allowed operations,
and the goal and the order of information revelation remain the same.

\begin{figure}[h!]
    \centering
    \includegraphics[width=0.68\linewidth]{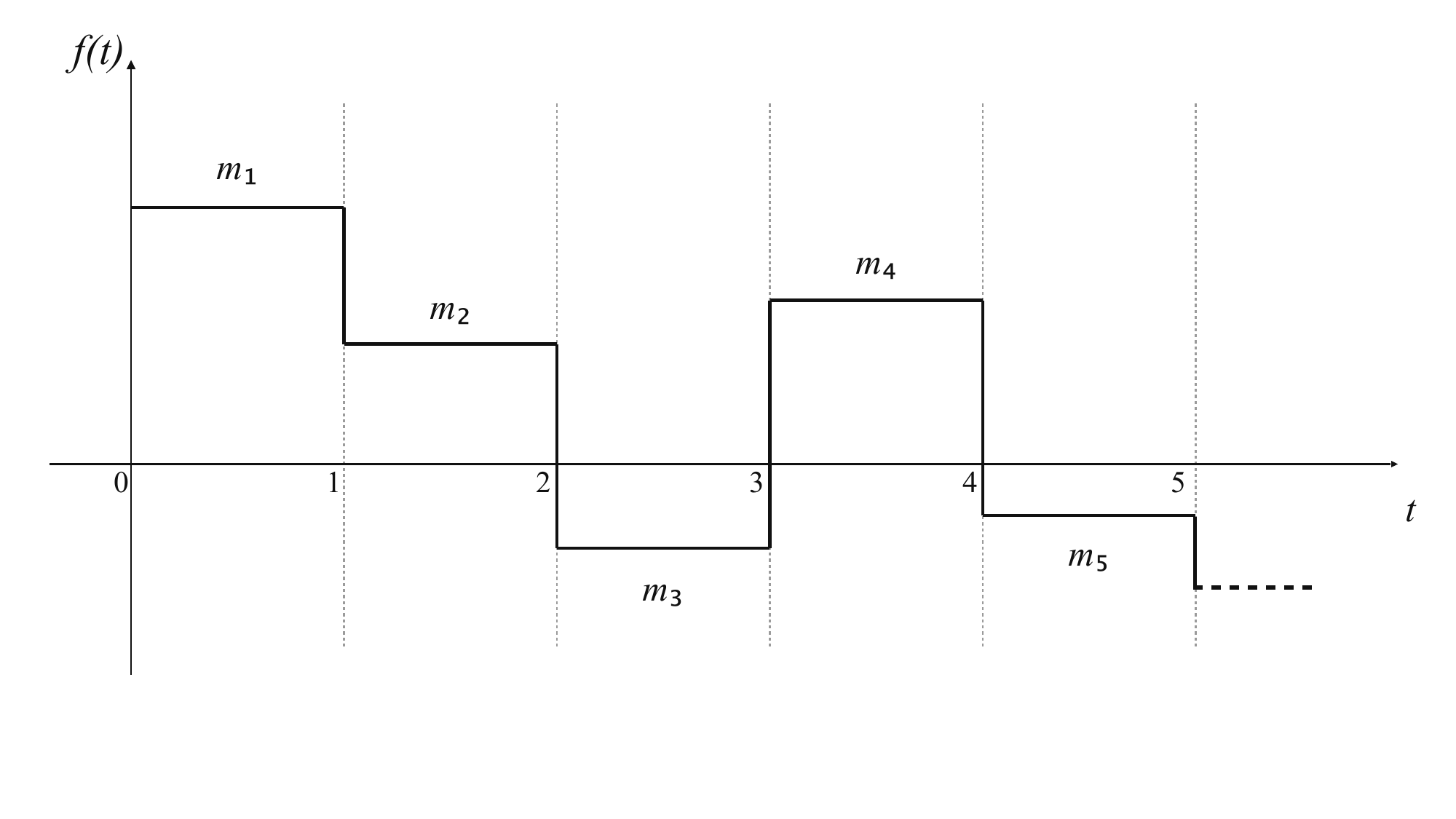}
    \caption{
        \textbf{Auxiliary piecewise-constant signal.}
        A realization of $f_{\mathrm{step}}$ in
        Eq.~\eqref{eq:step_reference_signal}, with value $m_i$ on
        $[i-1,i)$. It uses the same coefficients as the pulsed signal in
        Fig.~\ref{fig:pulsed_signal}; each rectangle has the same signed
        area as the corresponding pulse.
    }
    \label{fig:piecewise_constant_signal}
\end{figure}
\begin{lemma}[(Reduction to piecewise-constant signals)]
\label{lem:pulse_to_piecewise}
Fix $T\in\mathbb{N}^{+}$ and $z\ge1$, and let $\varphi$ satisfy
Eq.~\eqref{eq:pulse_assumptions}. Suppose that a Gaussian protocol with squeezing bounded by $z$ solves
one of our sensing problems for the pulsed signal
\[
    f(t)=\sum_{i\ge1}m_i\varphi(t-i+1)
\]
using sensing time $S$.

Then there exists a Gaussian protocol with the same squeezing bound
and the same success probability for the piecewise-constant signal
\[
    f_{\mathrm{step}}(t)
    =
    \sum_{i\ge1}m_i\,\mathbf{1}_{[i-1,i)}(t),
\]
whose sensing time $S_{\mathrm{step}}$ satisfies
\[
    \left\lceil\frac{S_{\mathrm{step}}}{T}\right\rceil
    =
    \left\lceil\frac{S}{T}\right\rceil .
\]
\end{lemma}
\begin{proof}
Let $S$ be the deterministic sensing-time budget of the original
protocol. We construct the piecewise-constant protocol by changing
only the times at which its operations are performed.

Define the cumulative pulse area
\[
    G_{\varphi}(v)\coloneqq\int_0^v\varphi(w)\,\mathrm{d}w,
    \qquad 0\le v\le1,
\]
and the change of clock
\begin{equation}\label{eq:pulse_clock}
    \Lambda_{\varphi}(t)\coloneqq
    \lfloor t\rfloor+G_{\varphi}(t-\lfloor t\rfloor),
    \qquad t\ge0.
\end{equation}
Within each unit interval, this replaces elapsed time by accumulated
pulse area. Since $\varphi$ is positive on $(0,1)$ and has unit integral,
$G_{\varphi}$ is continuous and strictly increasing, with
$G_{\varphi}(0)=0$ and $G_{\varphi}(1)=1$. Thus $\Lambda_{\varphi}$
is a continuous strictly increasing bijection of $[0,\infty)$ onto
itself, and
\[
    \Lambda_{\varphi}(k)=k
    \qquad\text{for every nonnegative integer }k.
\]

We first check that this change of clock reproduces the signal-induced
evolution. Fix any coefficient sequence $(m_i)_{i\ge1}$.
For $0\le a\le b\le1$, integration within the $i$-th bin gives
\[
\begin{aligned}
    \int_{i-1+a}^{i-1+b}f(t)\,\mathrm{d}t
    &=m_i\bigl(G_{\varphi}(b)-G_{\varphi}(a)\bigr) =\int_{\Lambda_{\varphi}(i-1+a)}^{\Lambda_{\varphi}(i-1+b)}
        f_{\mathrm{step}}(u)\,\mathrm{d}u.
\end{aligned}
\]
Splitting an arbitrary interval at integer times therefore gives
\begin{equation}\label{eq:pulse_integral_identity}
    \int_{t_{\mathrm{in}}}^{t_{\mathrm{out}}}f(t)\,\mathrm{d}t
    =
    \int_{\Lambda_{\varphi}(t_{\mathrm{in}})}
         ^{\Lambda_{\varphi}(t_{\mathrm{out}})}
        f_{\mathrm{step}}(u)\,\mathrm{d}u.
\end{equation}
This identity implies that their evolutions between corresponding times are identical
whenever no controls are applied in between.

The simulated protocol uses the original protocol's decision rules.
Whenever these rules prescribe a preparation, control, or measurement
at time $t$, it performs the same operation at time
$\Lambda_{\varphi}(t)$ instead. For instance, a waiting interval $[a,b]$ becomes
an interval of duration $\Lambda_{\varphi}(b)-\Lambda_{\varphi}(a)$.

It remains to compare the number of blocks reached. If $S=0$, then
$S_{\mathrm{step}}=0$ and the claim is immediate. Otherwise, set
$n\coloneqq\lceil S/T\rceil$, so that
\[
    (n-1)T<S\le nT.
\]
Since $T$ is an integer, both endpoints are fixed by
$\Lambda_{\varphi}$. Its strict monotonicity therefore gives
\[
    (n-1)T
    =\Lambda_{\varphi}((n-1)T)
    <\Lambda_{\varphi}(S)
    \le\Lambda_{\varphi}(nT)
    =nT.
\]
Consequently,
\[
    \left\lceil\frac{S_{\mathrm{step}}}{T}\right\rceil
    =n
    =\left\lceil\frac{S}{T}\right\rceil,
\]
as claimed.
\end{proof}

We can therefore carry out the remaining lower-bound analysis for
$f_{\mathrm{step}}$. In the single-experiment reductions and tree
analysis below, all times refer to this piecewise-constant model.
We return to the pulsed signal when proving the final sensing-time
bounds, using Lemma~\ref{lem:pulse_to_piecewise}.

\subsection{Simplified outcome probability distribution}
\label{sec_simpl_outcome_prob}

Let us start by analysing protocols obeying
Definition~\ref{def:gaussian_protocol_bounded_squeezing}, with the signal
Hamiltonian replaced by $\hat H_{\mathrm{step}}(t)$. We recall that such protocols consist of preparing single-mode Gaussian states with squeezing at most $z$, applying Gaussian unitaries that introduce no additional squeezing, namely phase-space rotations and displacements, and performing Gaussian measurements. All these operations may be carried out at arbitrary times during the Hamiltonian evolution, and may be chosen adaptively as a function of the measurement outcomes obtained so far.

Each Gaussian measurement (as in Eq.~\eqref{dist_gauss_meas}) produces an outcome $\mathbf{r}\in\mathbb{R}^2$. Conditioned on the underlying signal, both outcome distributions are Gaussian. The purpose of this subsection is to show that, for lower-bound purposes, every such outcome distribution can be reduced to a simple canonical form: a Gaussian distribution whose covariance is fixed by the squeezing bound, and whose mean contains the dependence on the unknown signal. We first consider measurements with a Gaussian seed and treat the homodyne case directly at the end of the proof.

First, we focus on a single measurement step. Let $t_{\mathrm{in}}$ be the time at which the learner prepares a Gaussian state with squeezing at most $z$, and let $t_{\mathrm{out}}$ be the time at which the corresponding measurement is performed. We begin with the case in which both $t_{\mathrm{in}}$ and $t_{\mathrm{out}}$ are integers. Let
\bb
    m \coloneqq (m_1,\dots,m_{t_{\mathrm{out}}-t_{\mathrm{in}}})
    \in \mathbb{R}^{t_{\mathrm{out}}-t_{\mathrm{in}}}
\ee
denote the values of the auxiliary signal in Eq.~\eqref{eq:step_reference_signal} during the interval $[t_{\mathrm{in}},t_{\mathrm{out}})$. These are precisely the random variables that determine the Hamiltonian evolution during this measurement step.

The following lemma shows that a sample from the actual outcome distribution can be generated, by a post-processing procedure independent of the unknown vector $m$, from a sample drawn from the Gaussian distribution
\bb\label{simp_form}
    \NN\!\left[
    \begin{pmatrix}
        a\cdot m\\
        b\cdot m
    \end{pmatrix},
    \frac{\mathbb{1}_2}{2z^2}
    \right],
\ee
where $a,b\in[-1,1]^{t_{\mathrm{out}}-t_{\mathrm{in}}}$ are vectors independent of $m$ and satisfy
\[
    a_j^2+b_j^2\le 1
    \qquad
    \forall\,j\in[t_{\mathrm{out}}-t_{\mathrm{in}}].
\]
Here, $\NN[\mathbf{u},V]$ denotes the two-dimensional Gaussian distribution with mean $\mathbf{u}\in\mathbb{R}^2$ and covariance matrix $V\in\mathbb{R}^{2\times2}$. We refer to the vectors $a$ and $b$ as the \emph{rotation coefficients}.

Thus, when proving lower bounds on the minimum sensing time, it is enough to work with the simplified distribution in \eqref{simp_form}. In this reduced description, the learner may choose the rotation coefficients $a$ and $b$, subject only to the constraint above, while the unknown signal enters the observation only through the two linear functionals $a\cdot m$ and $b\cdot m$.
\begin{lemma}[(Simulating the outcome distribution with a simplified distribution)]\label{sec_simplifying_prob}
Consider the above setting, where $z$ is the maximum allowed squeezing in the Gaussian protocol, $t_{\mathrm{in}}\in\mathbb{N}$ is the time at which the Gaussian state is prepared, and $t_{\mathrm{out}}\in\mathbb{N}$ is the time at which the Gaussian measurement is performed. Let $P_m$ denote the outcome distribution of the final measurement, which may be either a Gaussian measurement with a Gaussian seed or an ideal homodyne measurement, conditioned on the unknown vector $m\in\mathbb{R}^{t_{\mathrm{out}}-t_{\mathrm{in}}}$ describing the signal during the interval $[t_{\mathrm{in}},t_{\mathrm{out}})$. Then there exist vectors $a,b\in[-1,1]^{t_{\mathrm{out}}-t_{\mathrm{in}}}$, satisfying $a_j^2+b_j^2\le 1$ for all $j\in[t_{\mathrm{out}}-t_{\mathrm{in}}]$, such that the following holds: there exists an algorithm, independent of $m$, which takes as input a sample
\bb\label{eq_nn_prob}
    \mathbf{r}\sim\NN\!\left[
    \begin{pmatrix}
        a\cdot m\\
        b\cdot m
    \end{pmatrix},
    \frac{\mathbb{1}_2}{2z^2}
    \right]
\ee
and outputs a sample distributed according to $P_m$.

Equivalently, a sample from $P_m$ can be generated by a procedure independent of the unknown vector $m$, starting from a sample drawn from the Gaussian distribution in \eqref{eq_nn_prob}. In particular, the dependence on $m$ is entirely encoded in the mean, while the covariance is fixed by the squeezing parameter $z$.
\end{lemma}

\begin{proof}
We first consider the case in which the final measurement has a Gaussian seed $\sigma$, and then treat the case of an ideal homodyne measurement. The Hamiltonian evolution is Gaussian and may be interleaved with arbitrary displacements and phase-space rotations chosen by the learner. Since Gaussian states remain Gaussian under Gaussian unitaries, the state $\rho'$ immediately before the measurement is again Gaussian. Since the measurement itself is Gaussian, Eq.~\eqref{dist_gauss_meas} implies that $P_m$ is a Gaussian probability distribution of the form
\bb
    P_m=\NN\!\left[\mathbf{m}(\rho'),\frac{V(\rho')+V(\sigma)}{2}\right],
\ee
where the mean is given by the first moment $\mathbf{m}(\rho')$, and the covariance matrix is the arithmetic mean of the covariance matrices of the pre-measurement state $\rho'$ and of the Gaussian seed $\sigma$ associated with the measurement.

We now determine the most general form of the first moment $\mathbf{m}(\rho')$ and of the covariance matrix $V(\rho')$. To this end, observe that for any choice of the underlying vector
\bb
    m=(m_1,m_2,\ldots,m_{t_{\mathrm{out}}-t_{\mathrm{in}}})\in\mathbb{R}^{t_{\mathrm{out}}-t_{\mathrm{in}}},
\ee
the Hamiltonian evolution induces a displacement transformation. Indeed, for the evolution over a time $\tau\in(0,1)$ generated by the Hamiltonian $m_1\hat{p}$, with $m_1\in\mathbb{R}$, the corresponding unitary is
\bb
    \exp\left(-i\tau m_1 \hat{p}\right)
    = \exp\left(-i 
    \begin{pmatrix}
        \tau m_1,& \!\!0
    \end{pmatrix}
    \begin{pmatrix}
        0 & 1 \\
        -1 & 0
    \end{pmatrix}
    \begin{pmatrix}
        \hat{x} \\
        \hat{p}
    \end{pmatrix}
    \right)
    = D_{(\tau m_1,\,0)}\,.
\ee
This evolution may be interleaved with arbitrary displacements and phase-space rotations chosen by the learner. As recalled in the preliminary section, the first moment transforms under phase-space rotations and displacements according to
\bb
    \mathbf{m}\!\left(R_\theta\rho R_\theta^\dagger\right)&= O_\theta\,\mathbf{m}(\rho)\qquad\forall\,\theta\in[0,2\pi),\\
    \mathbf{m}\!\left(D_{\mathbf{r}}\rho D_{\mathbf{r}}^\dagger\right)&= \mathbf{m}(\rho)+\mathbf{r}\qquad\forall\,\mathbf{r}\in\mathbb{R}^2.
\ee

\begin{figure}[h!]
    \centering
    \includegraphics[width=0.8\linewidth]{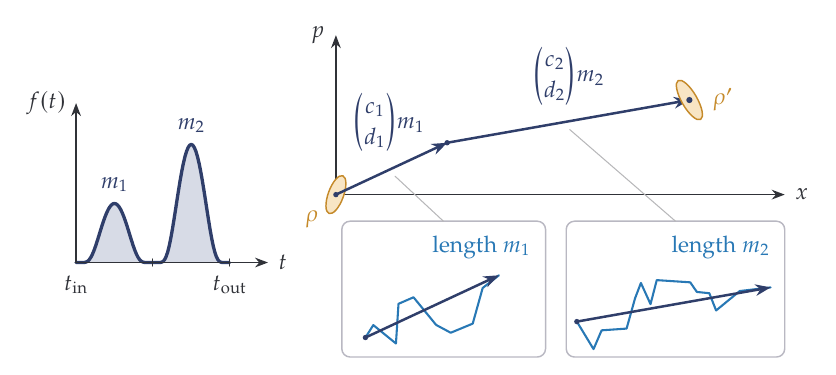}
    \caption{Phase space trajectory of a probe $\rho$, if a state was prepared at $t=0$ and evolved until $t=2$ under the field with $m_1$ and $m_2$ during unit time intervals 1 and 2. Over time, the probe's direction of travel is determined by the learner's phase space rotation it applies, as seen in the light blue paths. In each time unit $j$, the signal forces the probe to travel a length of $m_j$. The resulting dark blue vector $\begin{psmallmatrix} c_j \\ d_j \end{psmallmatrix} m_j$ is the overall displacement of the probe's first moment.}
    \label{fig:probe_displace_cd}
\end{figure}

As pictorially depicted in Fig.~\ref{fig:probe_displace_cd}, it follows that the final first moment $\mathbf{m}(\rho')$ can be written in the form
\bb
    \mathbf{m}(\rho')=
    \begin{pmatrix}
        c\cdot m\\
        d\cdot m
    \end{pmatrix}
    +\mathbf{r}_0,
\ee
for some vectors $c,d\in[-1,1]^{t_{\mathrm{out}}-t_{\mathrm{in}}}$ satisfying $c_j^2+d_j^2\le 1$ for all $j\in[t_{\mathrm{out}}-t_{\mathrm{in}}]$, and for some vector $\mathbf{r}_0\in\mathbb{R}^2$. All these quantities are independent of $m$. The vectors $c$ and $d$, as well as $\mathbf{r}_0$, depend on the phase-space rotations and displacements applied by the learner during the evolution, together with the first moment of the initial Gaussian state $\rho$. The bound $c_j^2+d_j^2\le 1$ follows because the contribution of the $j$th time bin is obtained from a displacement of Euclidean norm $|m_j|$ by composing rotations; hence the coefficient pair associated with each individual variable $m_j$ has Euclidean norm at most $1$. Intuition for these constraints is depicted in an example of the probe's movement in \cref{fig:probe_displace_cd}.

Similarly, the covariance matrix is invariant under displacements, namely
\bb
   V\!\left(D_{\mathbf{r}}\rho D_{\mathbf{r}}^\dagger\right)= V(\rho)\qquad\forall\,\mathbf{r}\in\mathbb{R}^2,
\ee
and transforms under rotations as
\bb
    V\!\left(R_\theta\rho R_\theta^\dagger\right)= O_\theta V(\rho) O_\theta^\intercal\qquad\forall\,\theta\in[0,2\pi).
\ee
Therefore the final covariance matrix has the form
\bb
    V(\rho')=OV(\rho)O^\intercal,
\ee
where $V(\rho)$ is the covariance matrix of the initial Gaussian state $\rho$, and $O$ is a suitable orthogonal rotation matrix. Importantly, $O$ depends on the phase-space rotations chosen by the learner, but is independent of $m$.

We conclude that the outcome distribution of the Gaussian measurement is of the form
\bb\label{eqxz}
    P_m=
    \NN\!\left[
    \begin{pmatrix}
        c\cdot m\\
        d\cdot m
    \end{pmatrix}
    +\mathbf{r}_0,\,
    \frac{OV(\rho)O^\intercal+V(\sigma)}{2}
    \right].
\ee
Thus, $P_m$ is Gaussian, with mean $\begin{pmatrix}
        c\cdot m\\
        d\cdot m
    \end{pmatrix}
    +\mathbf{r}_0$ and covariance matrix $\frac{OV(\rho)O^\intercal+V(\sigma)}{2}$. Moreover, all the quantities $c,d,\mathbf{r}_0,O,V(\rho)$, and $V(\sigma)$ are independent of the unknown vector $m$.

Since
\bb
    \frac{OV(\rho)O^\intercal+V(\sigma)}{2}\ge \frac{OV(\rho)O^\intercal}{2},
\ee
Lemma~\ref{reduction_via_convolution} implies that a sample from $P_m$ can be obtained by applying a post-processing procedure, independent of $m$, to a sample drawn from the Gaussian distribution
\bb
    \NN\!\left[
    \begin{pmatrix}
        c\cdot m\\
        d\cdot m
    \end{pmatrix}
    +\mathbf{r}_0,\,
    \frac{OV(\rho)O^\intercal}{2}
    \right].
\ee
Next, we apply Lemma~\ref{reduction_via_change} with the change of variable $\mathbf{r}\mapsto O^\intercal(\mathbf{r}-\mathbf{r}_0)$. This shows that the above distribution can be simulated from
\bb
    \NN\!\left[
    O^\intercal
    \begin{pmatrix}
        c\cdot m\\
        d\cdot m
    \end{pmatrix},
    \frac{V(\rho)}{2}
    \right].
\ee
We now use Eq.~\eqref{eq_covsss}, which gives the general parametrization of the covariance matrix of a single-mode Gaussian state. Namely,
\bb\label{most_general_v}
    V(\rho)
    =
    (2\nu+1)\,O_\phi
    \begin{pmatrix}
        \bar z^{-2} & 0\\
        0 & \bar z^2
    \end{pmatrix}
    O_\phi^\intercal,
\ee
where $\nu\ge 0$ is the mean photon number, $\phi\in[0,2\pi)$ is a rotation angle, and $\bar z\in[1,z]$ is the squeezing parameter chosen by the learner. The restriction $\bar z\in[1,z]$ follows from the assumption that the protocol has squeezing bounded by $z$.

Substituting \eqref{most_general_v} into the previous Gaussian distribution, we find that $P_m$ can be simulated from
\bb
    \NN\!\left[
    O^\intercal
    \begin{pmatrix}
        c\cdot m\\
        d\cdot m
    \end{pmatrix},
    \left(\nu+\frac12\right)
    O_\phi
    \begin{pmatrix}
        \bar z^{-2} & 0\\
        0 & \bar z^2
    \end{pmatrix}
    O_\phi^\intercal
    \right].
\ee
Applying Lemma~\ref{reduction_via_change} once again, now with the change of variable $\mathbf{r}\mapsto O_\phi^\intercal\mathbf{r}$, we obtain that this distribution can in turn be simulated from
\bb
    \NN\!\left[
    O_\phi^\intercal O^\intercal
    \begin{pmatrix}
        c\cdot m\\
        d\cdot m
    \end{pmatrix},
    \left(\nu+\frac12\right)
    \begin{pmatrix}
        \bar z^{-2} & 0\\
        0 & \bar z^2
    \end{pmatrix}
    \right].
\ee
Now observe that
\bb
    \left(\nu+\frac12\right)
    \begin{pmatrix}
        \bar z^{-2} & 0\\
        0 & \bar z^2
    \end{pmatrix}
    \ge
    \frac{\mathbb{1}_2}{2z^2},
\ee
because $\nu+\frac12\ge \frac12$, $\bar z\le z$, and $\bar z\ge 1$. Therefore, by another application of Lemma~\ref{reduction_via_convolution}, we conclude that the distribution $P_m$ can be simulated from
\bb
    \NN\!\left[
    O_\phi^\intercal O^\intercal
    \begin{pmatrix}
        c\cdot m\\
        d\cdot m
    \end{pmatrix},
    \frac{\mathbb{1}_2}{2z^2}
    \right].
\ee
The corresponding simulation procedure is independent of the unknown vector $m$.

Finally, since $c,d\in[-1,1]^{t_{\mathrm{out}}-t_{\mathrm{in}}}$ satisfy $c_j^2+d_j^2\le 1$ for all $j$, and since $O_\phi^\intercal O^\intercal$ is a rotation matrix, the mean vector can be rewritten as
\bb
    O_\phi^\intercal O^\intercal
    \begin{pmatrix}
        c\cdot m\\
        d\cdot m
    \end{pmatrix}
    =
    \begin{pmatrix}
        a\cdot m\\
        b\cdot m
    \end{pmatrix}
\ee
for suitable vectors $a,b\in[-1,1]^{t_{\mathrm{out}}-t_{\mathrm{in}}}$ satisfying $a_j^2+b_j^2\le 1$ for all $j\in[t_{\mathrm{out}}-t_{\mathrm{in}}]$. This proves the claim when the final measurement has a Gaussian seed.

It remains to consider an ideal homodyne measurement. Up to an invertible
rescaling of its outcome, independent of the unknown signal, such a
measurement measures a quadrature $u^\intercal\hat{\mathbf{R}}$ for some
$u\in\mathbb{R}^2$ with $\|u\|_2=1$. Using the expressions derived above for
the first moment and covariance matrix of $\rho'$, the corresponding
normalized homodyne outcome $Y_m$ is distributed as
\bb
    Y_m
    \sim
    \NN\!\left[
        a\cdot m+u^\intercal\mathbf{r}_0,\,
        \sigma_u^2
    \right],
\ee
where
$a_j\coloneqq u^\intercal(c_j,d_j)^\intercal$ and
$\sigma_u^2\coloneqq \frac12u^\intercal OV(\rho)O^\intercal u$.
By the Cauchy--Schwarz inequality,
$a_j^2\le c_j^2+d_j^2\le1$ for every $j$. Moreover,
Eq.~\eqref{most_general_v} implies $V(\rho)\ge z^{-2}\mathbb{1}_2$, and
therefore $\sigma_u^2\ge1/(2z^2)$. Set $b_j\coloneqq0$ for every $j$. The
first coordinate of a sample from the Gaussian distribution in
Eq.~\eqref{eq_nn_prob} is then distributed as
$\NN[a\cdot m,1/(2z^2)]$. By adding the known shift
$u^\intercal\mathbf{r}_0$ and an independent centered Gaussian random
variable of variance $\sigma_u^2-1/(2z^2)\ge0$, and by undoing the initial
rescaling if necessary, one obtains exactly the homodyne outcome
distribution $P_m$. This post-processing is independent of $m$. Hence the
homodyne outcome distribution can also be simulated from the Gaussian
distribution in Eq.~\eqref{eq_nn_prob}, which proves the claim.

\end{proof}

The same reasoning extends to the case in which $t_{\mathrm{in}}$ and
$t_{\mathrm{out}}$ are not necessarily integers, for both types of
measurements considered above. If both $t_{\mathrm{in}}$ and
$t_{\mathrm{out}}$ are integers, the result follows directly from
Lemma~\ref{sec_simplifying_prob}. We therefore only need to consider the
case in which at least one of them is not an integer.

The only modification is that the rotation coefficients $(a_1,b_1)$ and
$(a_k,b_k)$ corresponding to the first and last time bins, respectively,
must reflect the fact that these bins may be only partially traversed during
the evolution. Indeed, the proof is identical, except that the first and last
displacement contributions are weighted by the corresponding traversal times
rather than necessarily by a full unit time.

Specifically, let
$k\coloneqq \lceil t_{\mathrm{out}}\rceil-\lfloor t_{\mathrm{in}}\rfloor$
and let $m\coloneqq(m_1,\ldots,m_k)$ be the values of the signal
over the interval $[t_{\mathrm{in}},t_{\mathrm{out}})$.

If $k\ge2$, so that the evolution intersects at least two distinct time bins,
then a sample from the outcome distribution $P_m$ can be simulated from the
Gaussian distribution
\bb\label{eq_nonint1}
    \NN\!\left[
    \begin{pmatrix}
        a\cdot m\\
        b\cdot m
    \end{pmatrix},
    \frac{\mathbb{1}_2}{2z^2}
    \right],
\ee
where $a,b\in[-1,1]^k$ are independent of $m$ and satisfy
$\sqrt{a_1^2+b_1^2}\le
\lfloor t_{\mathrm{in}}\rfloor+1-t_{\mathrm{in}}$,
$a_j^2+b_j^2\le1$ for all $j=2,\ldots,k-1$, and
$\sqrt{a_k^2+b_k^2}\le
t_{\mathrm{out}}-\lceil t_{\mathrm{out}}\rceil+1$.

If instead $k=1$, then the evolution is entirely contained in a single time
bin. In that case, the signal is determined by a single value
$m_1$, and a sample from $P_m$ can be simulated from the Gaussian distribution
\bb\label{eq_nonint2}
    \NN\!\left[
    \begin{pmatrix}
        a_1m_1\\
        b_1m_1
    \end{pmatrix},
    \frac{\mathbb{1}_2}{2z^2}
    \right],
\ee
where $a_1,b_1\in[-1,1]$ satisfy
$\sqrt{a_1^2+b_1^2}\le t_{\mathrm{out}}-t_{\mathrm{in}}$.

\subsection{Simplified access model}

In the previous subsection, we showed that the outcome distribution associated with a single measurement step---namely, the preparation of a Gaussian input state with squeezing bounded by $z$ at time $t=t_{\mathrm{in}}$, followed by Hamiltonian evolution interleaved with phase-space rotations and displacements, and ending with a Gaussian measurement at time $t=t_{\mathrm{out}}$---can be reduced to a Gaussian distribution with covariance matrix $\frac{\mathbb{1}_2}{2z^2}$ and mean encoding the relevant information about the unknown signal.

We now show that, at the cost of increasing the squeezing parameter by a factor of $\sqrt{5}$, it is enough to restrict the analysis to input times $t_{\mathrm{in}}$ and measurement times $t_{\mathrm{out}}$ satisfying one of the following conditions:
\begin{itemize}
    \item \emph{Condition 1}: there exists $n\in\mathbb{N}$ such that
    \bb\label{cond_1}
        n \le t_{\mathrm{in}} \le t_{\mathrm{out}} \le n+1.
    \ee
    In other words, the evolution is entirely contained within a single unit time interval. In this case, the rotation coefficients satisfy $a,b\in\mathbb{R}$ and
    \bb
        \sqrt{a^2+b^2}\le t_{\mathrm{out}}-t_{\mathrm{in}}.
    \ee
    
    \item \emph{Condition 2}: $t_{\mathrm{in}},t_{\mathrm{out}}\in\mathbb{N}$ and there exists $n\in\mathbb{N}$ such that
    \bb\label{cond_2}
        nT \le t_{\mathrm{in}} \le t_{\mathrm{out}} \le (n+1)T.
    \ee
    In other words, the Gaussian state is prepared at an integer time $t_{\mathrm{in}}$, the Gaussian measurement is performed at an integer time $t_{\mathrm{out}}$, and both times belong to the same time block of length $T$. In this case, the rotation coefficients satisfy $a,b\in\mathbb{R}^{\,t_{\mathrm{out}}-t_{\mathrm{in}}}$ and
    \bb
        \sqrt{a_j^2+b_j^2}\le 1
        \qquad\forall\, j\in[t_{\mathrm{out}}-t_{\mathrm{in}}].
    \ee
    
    \item \emph{Condition 3}: there exist $n,k\in\mathbb{N}$ with $k>n$ such that
    \bb\label{cond_3}
        t_{\mathrm{in}}=nT
        \qquad\text{and}\qquad
        t_{\mathrm{out}}=kT.
    \ee
    In other words, the Gaussian state is prepared at the beginning of a time block of length $T$, and the Gaussian measurement is performed at the end of a later time block of length $T$. In this case, the rotation coefficients satisfy $a,b\in\mathbb{R}^{\,t_{\mathrm{out}}-t_{\mathrm{in}}}$ and
    \bb
        \sqrt{a_j^2+b_j^2}\le 1
        \qquad\forall\, j\in[t_{\mathrm{out}}-t_{\mathrm{in}}].
    \ee
\end{itemize}

See \cref{fig:3_cont_sqrt5}(a) for a visual of these three conditions. Next, we show that we may work with the following access model.

\begin{figure}
    \centering
    \includegraphics[width=\linewidth]{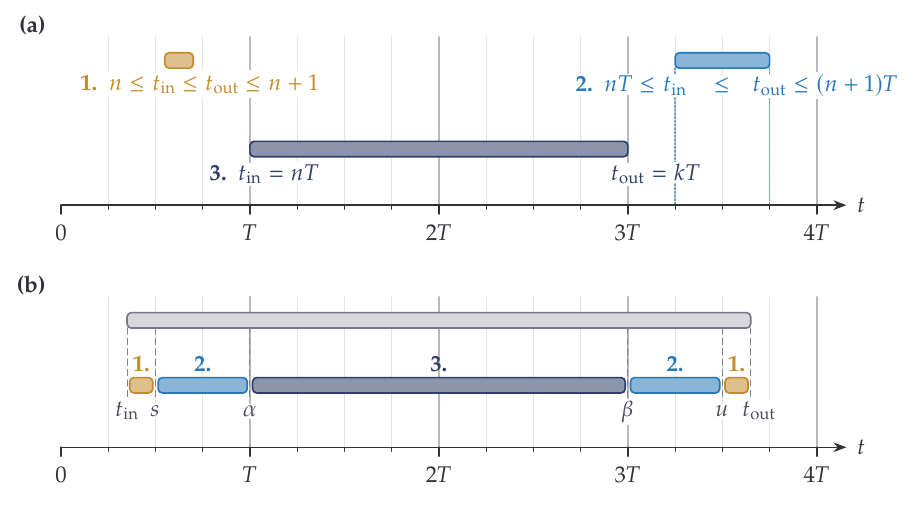}
    \caption{(a) The three conditions of Eqs.~\eqref{cond_1}, \eqref{cond_2} and \eqref{cond_3}. In Condition 1, the evolution lies within a single unit time bin. In Condition 2, the start and end times are integers in the same block of $T$. For Condition 3, the segment starts and ends at block boundaries. (b) In the proof of Lemma~\ref{lemma_access_model}, an arbitrary step $[t_{\mathrm{in}},t_{\mathrm{out}})$ is split into at most $5$ segments: $[t_{\mathrm{in}},s)\cup[s,\alpha)\cup[\alpha, \beta)\cup[\beta,u)\cup[u,t_{\mathrm{out}})$, each satisfying one of the three conditions. To simulate the final measurement of $[t_{\mathrm{in}},t_{\mathrm{out}})$ with covariance $\frac{\mathbb{1}_2}{2z^2}$, each of these 5 segments must be measured with covariance $\frac{\mathbb{1}_2}{10z^2}$. The cost of using this access model is a factor of $\sqrt{5}$ in squeezing.}
    \label{fig:3_cont_sqrt5}
\end{figure}

\begin{lemma}[(Simplified access model)]\label{lemma_access_model}
Any Gaussian protocol with squeezing bounded by $z$, performing sensing of the auxiliary Hamiltonian $\hat H_{\mathrm{step}}(t)$ in
Eq.~\eqref{eq:step_reference_hamiltonian} over a total sensing time $T_{\mathrm{sensing}}$, can be simulated by an access model of the following form:
\begin{itemize}
    \item At time $t=0$, the learner chooses an evolution time $t_1$ such that the pair $(0,t_1)$ satisfies one of Conditions 1, 2, or 3, and chooses the corresponding rotation coefficients $a^{(1)},b^{(1)}$. The learner then receives an outcome sampled from the Gaussian distribution
    \bb
        \NN\!\left[
        \begin{pmatrix}
            a^{(1)}\cdot m^{(1)}\\
            b^{(1)}\cdot m^{(1)}
        \end{pmatrix},
        \frac{\mathbb{1}_2}{10z^2}
        \right],
    \ee
    where $m^{(1)}$ denotes the unknown vector describing the signal during the interval $[0,t_1)$.
    
    \item At time $t=t_1$, based on the outcome obtained so far, the learner adaptively chooses a new evolution time $t_2$ such that the pair $(t_1,t_1+t_2)$ satisfies one of Conditions 1, 2, or 3, and chooses the corresponding rotation coefficients $a^{(2)},b^{(2)}$. The learner then receives an outcome sampled from the Gaussian distribution
    \bb
        \NN\!\left[
        \begin{pmatrix}
            a^{(2)}\cdot m^{(2)}\\
            b^{(2)}\cdot m^{(2)}
        \end{pmatrix},
        \frac{\mathbb{1}_2}{10z^2}
        \right],
    \ee
    where $m^{(2)}$ denotes the unknown vector describing the signal during the interval $[t_1,t_1+t_2)$.
    
    \item This process is then repeated adaptively.
    
    \item Finally, at time $t=t_1+\cdots+t_{N-1}$, the learner chooses the last evolution time $t_N$ so that the pair
    \bb
        \bigl(t_1+\cdots+t_{N-1},\,t_1+\cdots+t_N\bigr)
    \ee
    satisfies one of Conditions 1, 2, or 3, and chooses the corresponding rotation coefficients $a^{(N)},b^{(N)}$. The learner then receives an outcome sampled from the Gaussian distribution
    \bb
        \NN\!\left[
        \begin{pmatrix}
            a^{(N)}\cdot m^{(N)}\\
            b^{(N)}\cdot m^{(N)}
        \end{pmatrix},
        \frac{\mathbb{1}_2}{10z^2}
        \right],
    \ee
    where $m^{(N)}$ denotes the unknown vector describing the signal during the interval $\bigl[t_1+\cdots+t_{N-1},\,t_1+\cdots+t_N\bigr)$.
\end{itemize}
\end{lemma}

\begin{proof}
Consider a single measurement step, with preparation time $t_{\mathrm{in}}$ and measurement time $t_{\mathrm{out}}$. By Lemma~\ref{sec_simplifying_prob}, together with the extension to non-integer endpoints given in \eqref{eq_nonint1} and \eqref{eq_nonint2}, its outcome distribution can be simulated by
\bb\label{eq_single_step_marg_start}
    \NN\!\left[
    \begin{pmatrix}
        a\cdot m\\
        b\cdot m
    \end{pmatrix},
    \frac{\mathbb{1}_2}{2z^2}
    \right],
\ee
where $a,b\in\mathbb{R}^L$ are the corresponding rotation coefficients,
\bb
    L\coloneqq \lceil t_{\mathrm{out}}\rceil-\lfloor t_{\mathrm{in}}\rfloor,
\ee
and $m\in\mathbb{R}^L$ is the unknown vector describing the signal during the interval $[t_{\mathrm{in}},t_{\mathrm{out}})$.

We claim that the interval $[t_{\mathrm{in}},t_{\mathrm{out}})$ can always be partitioned into at most five subintervals, each satisfying one of Conditions 1, 2, or 3. Once such a partition is fixed, let
\bb
    [t_{\mathrm{in}},t_{\mathrm{out}})=I_1\cup\cdots\cup I_r
\ee
be the corresponding decomposition into consecutive subintervals, with $r\le 5$. For each $\ell\in[r]$, let $m^{(\ell)}$ denote the vector of values of the signal on $I_\ell$, and let $a^{(\ell)}$ and $b^{(\ell)}$ be the corresponding restrictions of the rotation coefficients. By construction, these restricted coefficients satisfy the admissibility constraints prescribed by the corresponding condition. Moreover,
\bb
    a\cdot m=\sum_{\ell=1}^r a^{(\ell)}\cdot m^{(\ell)},
    \qquad
    b\cdot m=\sum_{\ell=1}^r b^{(\ell)}\cdot m^{(\ell)}.
\ee
Therefore, by Lemma~\ref{reduction_via_marg}, the Gaussian distribution in \eqref{eq_single_step_marg_start} can be simulated by the product distribution
\bb\label{eq_tensor_cov_r}
    \bigotimes_{\ell=1}^r
    \NN\!\left[
    \begin{pmatrix}
        a^{(\ell)}\cdot m^{(\ell)}\\
        b^{(\ell)}\cdot m^{(\ell)}
    \end{pmatrix},
    \frac{\mathbb{1}_2}{2rz^2}
    \right].
\ee
Since $r\le 5$, we have $\frac{\mathbb{1}_2}{2rz^2}\ge \frac{\mathbb{1}_2}{10z^2}$. Hence, by Lemma~\ref{reduction_via_convolution}, each factor in \eqref{eq_tensor_cov_r} can in turn be simulated from the corresponding Gaussian distribution with covariance matrix $\frac{\mathbb{1}_2}{10z^2}$. Thus the original measurement step can be simulated by at most five measurement steps of the form described in the statement. In the simulated protocol, one first collects these at most five sub-outcomes, then classically post-processes them---via a procedure independent of the unknown signal---to reconstruct a sample from the original outcome distribution, and only after that proceeds to the next adaptive choice. Repeating this replacement independently at each step of the protocol therefore preserves the original adaptive behavior and yields the claimed access model.

It remains to prove the existence of the above partition. Set
\bb
    s\coloneqq \lceil t_{\mathrm{in}}\rceil,
    \qquad
    u\coloneqq \lfloor t_{\mathrm{out}}\rfloor.
\ee
We distinguish the following cases.

\smallskip
\noindent
\emph{Case 1:} $L\le 3$. In this case, the interval $[t_{\mathrm{in}},t_{\mathrm{out}})$ intersects at most three unit time bins. We split it at the intermediate integer times between $t_{\mathrm{in}}$ and $t_{\mathrm{out}}$. This yields at most three consecutive subintervals, each entirely contained in a single unit time interval. Hence every subinterval satisfies Condition 1.

\smallskip
\noindent
\emph{Case 2:} $L\ge 4$ and $s$ and $u$ belong to the same block of length $T$. Equivalently, there exists $n\in\mathbb{N}$ such that
\bb
    nT\le s\le u\le (n+1)T.
\ee
In this case we decompose
\bb
    [t_{\mathrm{in}},t_{\mathrm{out}})
    =
    [t_{\mathrm{in}},s)\cup[s,u)\cup[u,t_{\mathrm{out}}),
\ee
omitting any empty interval. The first and last intervals satisfy Condition 1, while the middle interval satisfies Condition 2.

\smallskip
\noindent
\emph{Case 3:} $L\ge 4$, the integers $s$ and $u$ do not belong to the same block of length $T$, and there is no full block strictly contained between them. Define
\bb
    \alpha \coloneqq T\Bigl\lceil \frac{s}{T}\Bigr\rceil,
    \qquad
    \beta \coloneqq T\Bigl\lfloor \frac{u}{T}\Bigr\rfloor.
\ee
Since $s$ and $u$ lie in different blocks and there is no full block strictly between them, we have $\alpha=\beta$. We then decompose
\bb
    [t_{\mathrm{in}},t_{\mathrm{out}})
    =
    [t_{\mathrm{in}},s)\cup[s,\alpha)\cup[\alpha,u)\cup[u,t_{\mathrm{out}}),
\ee
again omitting any empty interval. The first and last intervals satisfy Condition 1, while the two middle intervals have integer endpoints and lie inside single blocks of length $T$, so they satisfy Condition 2.

\smallskip
\noindent
\emph{Case 4:} $L\ge 4$ and there is at least one full block strictly contained between $s$ and $u$. With the same definitions
\bb
    \alpha \coloneqq T\Bigl\lceil \frac{s}{T}\Bigr\rceil,
    \qquad
    \beta \coloneqq T\Bigl\lfloor \frac{u}{T}\Bigr\rfloor,
\ee
this case is characterized by $\alpha<\beta$. We decompose
\bb
    [t_{\mathrm{in}},t_{\mathrm{out}})
    =
    [t_{\mathrm{in}},s)
    \cup
    [s,\alpha)
    \cup
    [\alpha,\beta)
    \cup
    [\beta,u)
    \cup
    [u,t_{\mathrm{out}}),
\ee
omitting any empty interval. The first and last intervals satisfy Condition 1. The second and fourth intervals have integer endpoints and lie within single blocks of length $T$, so they satisfy Condition 2. Finally, the middle interval has endpoints that are multiples of $T$, hence it satisfies Condition 3. See \cref{fig:3_cont_sqrt5}(b) for an example decomposition.

In all cases we obtain a decomposition into at most five subintervals satisfying Conditions 1, 2, or 3, as claimed. This concludes the proof.

\end{proof}

\subsection{Hierarchical tree representation of the simplified access model}

The simplified access model of Lemma~\ref{lemma_access_model} is inherently adaptive. After each measurement, the learner may choose the next waiting time and the corresponding rotation coefficients as a function of the full transcript observed so far. It is therefore natural to represent such a protocol by means of a rooted tree: each node records the transcript accumulated up to that point together with the next experiment prescribed by the learner, while each child corresponds to one possible outcome of that experiment. Accordingly, every root-to-leaf path represents one possible execution of the adaptive protocol. Similar representation trees have recently been used in quantum learning theory to model adaptive quantum algorithms~\cite{chen2021exponential,huang2022quantum,chen2022complexitynisq,Chen_2024,chen2025efficientpaulichannelestimation}.

Lemma~\ref{lemma_access_model} naturally leads to a nested tree representation because it reduces every measurement step to one of three canonical types, namely Conditions~1, 2, and~3 of the previous subsection. These three conditions correspond to three temporal resolutions, and the three trees introduced below simply encode the protocol at those three scales. We call this nested tree representation the \emph{hierarchical tree representation}.

At the coarsest level, Condition~3 in \eqref{cond_3} describes experiments whose input and output times are multiples of $T$, namely intervals of the form $[nT,kT)$. This is the natural scale at which to represent the protocol over the whole sensing time $T_{\mathrm{sensing}}$, and it gives rise to the \emph{full-time tree}. If such an experiment spans several blocks of length $T$, then Lemma~\ref{lemma_access_model} already describes its outcome directly by a Gaussian distribution. If instead it spans exactly one block, then the protocol inside that block may still be adaptive and must be resolved at a finer level. This is the role of Condition~2 in \eqref{cond_2}, which describes experiments with integer input and output times lying inside the same block of length $T$; this gives rise to the \emph{$T$-time tree}. Finally, even inside a $T$-time tree, an experiment of duration exactly one time unit may still contain adaptive structure. This is precisely the situation covered by Condition~1 in \eqref{cond_1}, which describes experiments entirely contained in a single unit-time interval, and it gives rise to the \emph{unit-time tree}. Thus, the full-time tree, the $T$-time tree, and the unit-time tree arise directly from the three classes of intervals identified by Lemma~\ref{lemma_access_model}.

In summary, we introduce three nested trees:
\begin{itemize}
    \item the \emph{full-time tree}, which describes the protocol over the entire sensing time $\left \lceil \frac{T_{\mathrm{sensing}}}{T} \right\rceil T$;
    \item the \emph{$T$-time tree}, which describes the protocol inside a single block of length $T$;
    \item the \emph{unit-time tree}, which describes the protocol inside a single unit-time interval.
\end{itemize}
The full-time tree is the top-level object. Whenever a node of the full-time tree corresponds to an experiment of duration exactly one block of length $T$, its internal adaptive structure is refined by a $T$-time tree. Likewise, whenever a node of a $T$-time tree corresponds to an experiment of duration exactly one time unit, its internal adaptive structure is refined by a unit-time tree.

Figure~\ref{fig:tree} summarizes this organization. The same adaptive protocol is represented at three temporal scales: the full-time tree at the scale of block-to-block evolution, the $T$-time tree at the scale of integer times within a single block of length $T$, and the unit-time tree at the scale of continuous-time evolution within a single unit interval. These three levels correspond exactly to Conditions~3, 2, and~1 of Lemma~\ref{lemma_access_model}.

In all three cases, a node represents the transcript accumulated so far together with the next experiment prescribed by the learner, while a child of that node corresponds to one possible outcome of such an experiment.

\begin{figure}[t]
    \includegraphics[width=\linewidth]{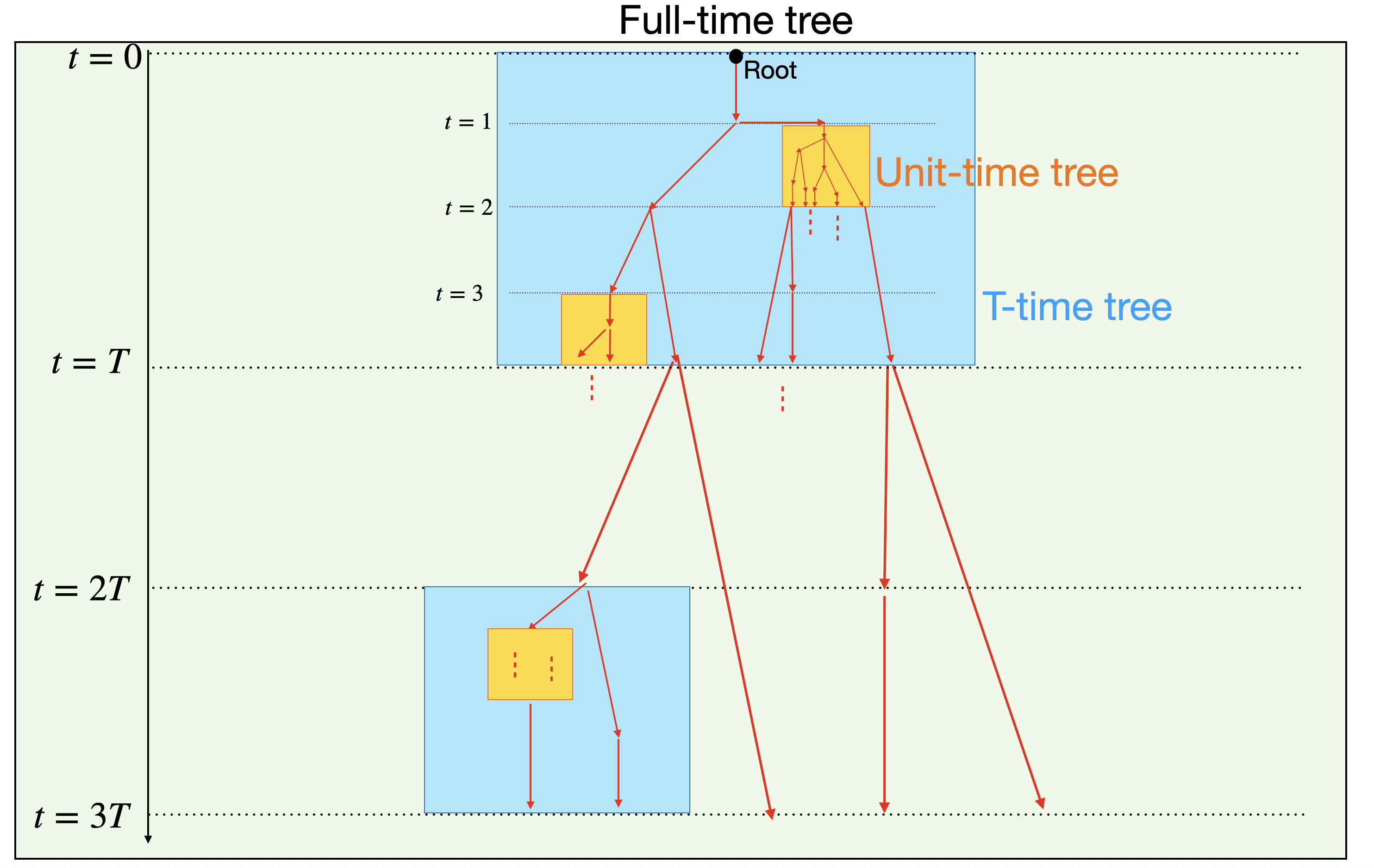}
    \caption{\textbf{Hierarchical tree representation induced by Lemma~\ref{lemma_access_model}.} Any adaptive protocol compatible with Lemma~\ref{lemma_access_model} can be represented by a nested tree structure. At the coarsest level, the \emph{full-time tree} describes the protocol at the scale of blocks of length $T$, namely experiments whose input and output times are multiples of $T$; this is precisely the setting of Condition~3. Whenever such an experiment is confined to a single block of length $T$, its internal adaptive structure is resolved at a finer scale by a \emph{$T$-time tree}, corresponding to Condition~2, where the input and output times are integers within the same block. Finally, whenever an experiment inside a $T$-time tree is confined to a single unit-time interval, its remaining continuous-time adaptive structure is resolved by a \emph{unit-time tree}, corresponding to Condition~1. In this way, the full-time tree, the $T$-time tree, and the unit-time tree provide a hierarchical representation of the same adaptive protocol.}
    \label{fig:tree}
\end{figure}

We now introduce notation that will be used uniformly for all three trees. Let $r$ denote the root of the tree under consideration. Every node $v\neq r$ has a unique parent, which we denote by $u$. Since the tree is rooted, for every node $v$ there is a unique path from the root to $v$. We write this path as
\[
r=u_0^{(v)} \to u_1^{(v)} \to \cdots \to u_{N(v)-1}^{(v)} \to u_{N(v)}^{(v)}=v.
\]
The integer $N(v)$ is called the \emph{tree depth} of $v$. Equivalently, $N(v)$ is the number of experiments that have been performed along the branch from the root to $v$. In particular, $N(r)\coloneqq 0$.

For each $i\in[N(v)]$, the edge $u_{i-1}^{(v)} \to u_i^{(v)}$ corresponds to the $i$-th experiment performed along this branch, together with its observed outcome. We denote that outcome by $\mathbf{y}_i^{(v)}$. Hence,
\[
\mathbf{y}_1^{(v)},\ldots,\mathbf{y}_{N(v)}^{(v)}
\]
is exactly the sequence of outcomes observed along the branch from the root to $v$. This is the transcript available to the learner when the protocol reaches the node $v$.

Besides the tree depth, each node also carries a \emph{time depth}, denoted by $t(v)$. This quantity records how much sensing time has elapsed along the branch from the root of the corresponding tree to the node $v$. Its precise meaning depends on which tree is being considered:
\begin{itemize}
    \item in the full-time tree, $t(v)$ is measured in blocks of length $T$, so that the elapsed reference time is $t(v)T$;
    \item in the $T$-time tree, $t(v)$ is measured in integer time units inside a fixed block of length $T$;
    \item in the unit-time tree, $t(v)$ is measured in ordinary continuous time inside a fixed unit interval.
\end{itemize}
In all three cases, $t(r)\coloneqq 0$, and if $v$ is a child of $u$, then $t(v)$ is defined as the sum of $t(u)$ and the duration of the experiment prescribed at $u$. This notion is local to the tree under consideration. In particular, when a node of the full-time tree is refined by a $T$-time tree, the time depth in the attached $T$-time tree is reset and measured from the beginning of the corresponding block; likewise, when a node of a $T$-time tree is refined by a unit-time tree, the time depth in the attached unit-time tree is reset and measured from the beginning of the corresponding unit interval.

Thus, for any node $v$, the quantity $N(v)$ tells us how many experiments have been performed before reaching $v$, while $t(v)$ tells us how much sensing time has elapsed. The ordered lists
\[
u_0^{(v)},u_1^{(v)},\ldots,u_{N(v)}^{(v)}
\qquad\text{and}\qquad
\mathbf{y}_1^{(v)},\ldots,\mathbf{y}_{N(v)}^{(v)}
\]
record, respectively, the sequence of nodes and the sequence of observed outcomes along the unique branch from the root to $v$.

\subsubsection{Full-time tree}

We begin with the coarsest time scale. A \emph{full-time tree} represents the protocol over the entire sensing time $nT$, where $n\coloneqq \lceil \frac{T_{\mathrm{sensing}}}{T}\rceil $. Its time depth is measured in units of $T$: if a node $v$ has time depth $t(v)\in\{0,1,\ldots,n\}$, then the protocol has already consumed reference time $t(v)T$ by the time it reaches $v$.

This is the natural representation associated with Condition~3 of Lemma~\ref{lemma_access_model}, which describes experiments whose input and output times are both multiples of $T$. Accordingly, at this level each experiment has duration $kT$ for some $k\in\mathbb{N}^+$. When $k\ge 2$, the experiment is described directly by the Gaussian law of Condition~3. When $k=1$, the experiment is confined to a single block of length $T$; in that case, its internal adaptive structure is more naturally described by a $T$-time tree, corresponding to the finer resolution given by Condition~2 inside that block.

We consider a rooted tree, possibly with infinitely many children, since measurement outcomes may take values in continuous spaces. The root $r$ represents the beginning of the protocol and has time depth $t(r)\coloneqq 0$. It specifies the first experiment to be performed. Every non-root node $v$ is determined by the following data:
\begin{itemize}
    \item its parent $u$;
    \item the outcome of the experiment prescribed at $u$, denoted by $\mathbf{y}_{N(v)}^{(v)}$; this is the outcome that determines the child $v$ among the children of $u$;
    \item the next waiting time $\tau(v)\in\mathbb{N}^+$, meaning that the experiment prescribed at $v$ lasts $\tau(v)T$; since the protocol cannot exceed the total sensing time, this quantity must satisfy
    \[
        \tau(v)\le n-t(v);
    \]
    \item the time depth $t(v)$, defined recursively by $t(v)\coloneqq t(u)+\tau(u)$;
    \item the details of the experiment prescribed at $v$.
\end{itemize}

These details depend on the value of $\tau(v)$. If $\tau(v)\ge 2$, then the experiment spans $\tau(v)$ consecutive blocks of length $T$, and its description is given by blockwise rotation coefficients $a^{(v,1)},b^{(v,1)},\ldots,a^{(v,\tau(v))},b^{(v,\tau(v))}\in[-1,1]^T$ satisfying $\sqrt{(a^{(v,i)}_j)^2+(b^{(v,i)}_j)^2}\le 1$ for every $i\in[\tau(v)]$ and every $j\in[T]$. If instead $\tau(v)=1$, then the experiment is confined to a single block of length $T$, and the node stores the description of a \emph{$T$-time tree}, which represents the adaptive protocol executed inside that block. The root of such a $T$-time tree is, by definition, attached to the node $v$ of the full-time tree.

The recursive construction is therefore as follows. The root of the full-time tree specifies the first experiment. Once the outcome of that experiment is observed, one obtains a child of the root. More generally, if the protocol is at a node $u$, then performing the experiment prescribed at $u$ and observing one of its possible outcomes determines a child $v$ of $u$. Repeating this procedure generates a branch of the tree, and the construction continues until the time depth reaches $n$.

The probability law of the outcome associated with an edge $u\to v$ depends on the waiting time of the parent node $u$:
\begin{itemize}
\item If $\tau(u)\ge 2$, let $m^{(u,1)},\ldots,m^{(u,\tau(u))}\in\mathbb{R}^T$ denote the unknown vectors describing, block by block, the signal on the interval
\[
[t(u)T,(t(u)+\tau(u))T)
=
[t(u)T,(t(u)+1)T)\cup\cdots\cup[(t(u)+\tau(u)-1)T,(t(u)+\tau(u))T).
\]
Then, by Condition~3 of Lemma~\ref{lemma_access_model}, the outcome attached to the edge $u\to v$ is a two-dimensional Gaussian random variable distributed as
\begin{equation}\label{tau3}
\mathbf{y}_{N(v)}^{(v)} \sim
\NN\!\left[
\begin{pmatrix}
\sum_{i=1}^{\tau(u)} a^{(u,i)}\!\cdot m^{(u,i)}\\[1mm]
\sum_{i=1}^{\tau(u)} b^{(u,i)}\!\cdot m^{(u,i)}
\end{pmatrix},
\frac{\mathbb{1}_2}{10z^2}
\right].
\end{equation}

\item If instead $\tau(u)=1$, then the node $u$ stores a $T$-time tree. In this case, the outcome associated with the edge $u\to v$ is not a vector in $\mathbb{R}^2$, but rather a \emph{leaf} of the $T$-time tree encoded in $u$, namely the complete transcript produced by the adaptive protocol executed inside the corresponding block of length $T$.
\end{itemize}

Finally, a node $l$ of the full-time tree is called a \emph{leaf} if its time depth is equal to $n$, that is, if $t(l)=n$. Thus, a leaf represents a complete execution of the protocol over the full sensing time $nT$.

\subsubsection{$T$-time tree}

We now move to the intermediate time scale, namely a single block of length $T$. A \emph{$T$-time tree} represents the protocol inside one such block. Its time depth is measured in integer time units: if a node $v$ has time depth $t(v)\in\{0,1,\ldots,T\}$, then the protocol has already consumed time $t(v)$ within the current block by the time it reaches $v$.

This is the natural representation associated with Condition~2 of Lemma~\ref{lemma_access_model}, which describes experiments whose input and output times are integers lying inside the same block of length $T$. Accordingly, each experiment at this level has an integer duration $\tau(v)\in\mathbb{N}^+$. When $\tau(v)\ge 2$, the experiment is described directly by the Gaussian law of Condition~2. When $\tau(v)=1$, the experiment is confined to a single unit-time interval; in that case, its internal continuous-time adaptive structure is described by a unit-time tree, corresponding to the finer resolution given by Condition~1.

We again consider a rooted tree, possibly with infinitely many children. The root $r$ represents the beginning of the protocol inside the fixed block and has time depth $t(r)\coloneqq 0$. It specifies the first experiment to be performed. Every non-root node $v$ is determined by the following data:
\begin{itemize}
    \item its parent $u$;
    \item the outcome of the experiment prescribed at $u$, denoted by $\mathbf{y}_{N(v)}^{(v)}$; this is the outcome that determines the child $v$ among the children of $u$;
    \item the next waiting time $\tau(v)\in\mathbb{N}^+$, meaning that the experiment prescribed at $v$ lasts $\tau(v)$ time units; since the protocol must remain inside the fixed block of length $T$, this quantity must satisfy $\tau(v)\le T-t(v)$;
    \item the time depth $t(v)$, defined recursively by $t(v)\coloneqq t(u)+\tau(u)$;
    \item the details of the experiment prescribed at $v$.
\end{itemize}

These details depend on the value of $\tau(v)$. If $\tau(v)\ge 2$, then the experiment lasts $\tau(v)$ integer time steps and is described by rotation coefficients $a^{(v)},b^{(v)}\in[-1,1]^{\tau(v)}$ satisfying $\sqrt{(a^{(v)}_i)^2+(b^{(v)}_i)^2}\le 1$ for every $i\in[\tau(v)]$. If instead $\tau(v)=1$, then the experiment is confined to a single unit-time interval, and the node stores the description of a \emph{unit-time tree}, which represents the adaptive protocol executed inside that interval. The root of such a unit-time tree is, by definition, attached to the node $v$ of the $T$-time tree.

The recursive construction is the same as before. The root of the $T$-time tree specifies the first experiment. Once the outcome of that experiment is observed, one obtains a child of the root. More generally, if the protocol is at a node $u$, then performing the experiment prescribed at $u$ and observing one of its possible outcomes determines a child $v$ of $u$. Repeating this procedure generates a branch of the tree, and the construction continues until the time depth reaches $T$.

The probability law of the outcome associated with an edge $u\to v$ depends on the waiting time of the parent node $u$:
\begin{itemize}
    \item If $\tau(u)\ge 2$, let $m_u\in\mathbb{R}^{\tau(u)}$ denote the vector of values of the signal on the interval $[t(u),t(u)+\tau(u))$, after reindexing this interval locally. Then, by Condition~2 of Lemma~\ref{lemma_access_model}, the outcome attached to the edge $u\to v$ is a two-dimensional Gaussian random variable distributed as
\begin{equation}\label{tau2}
\mathbf{y}_{N(v)}^{(v)} \sim
\NN\!\left[
\begin{pmatrix}
a^{(u)}\!\cdot m_u\\
b^{(u)}\!\cdot m_u
\end{pmatrix},
\frac{\mathbb{1}_2}{10z^2}
\right].
\end{equation}

\item If instead $\tau(u)=1$, then the node $u$ stores a unit-time tree. In this case, the outcome associated with the edge $u\to v$ is not a vector in $\mathbb{R}^2$, but rather a \emph{leaf} of the unit-time tree encoded in $u$, namely the complete transcript produced by the adaptive protocol executed inside the corresponding unit-time interval.
\end{itemize}

Finally, a node $l$ of the $T$-time tree is called a \emph{leaf} if its time depth is equal to $T$, that is, if $t(l)=T$. Thus, a leaf represents a complete execution of the protocol inside the given block of length $T$.

\subsubsection{Unit-time tree}

We now arrive at the finest time scale, namely a single unit-time interval. A \emph{unit-time tree} represents the protocol inside such an interval. If a node $v$ has time depth $t(v)\in[0,1]$, then the protocol has already consumed time $t(v)$ within the current unit interval by the time it reaches $v$.

This is the natural representation associated with Condition~1 of Lemma~\ref{lemma_access_model}, which describes experiments entirely contained in a single unit-time interval, with arbitrary input and output times inside that interval. At this scale, the learner may still behave adaptively in continuous time.

We again consider a rooted tree, possibly with infinitely many children. The root $r$ represents the beginning of the protocol inside the unit interval and has time depth $t(r)\coloneqq 0$. It specifies the first experiment to be performed. Every non-root node $v$ is determined by the following data:
\begin{itemize}
    \item its parent $u$;
    \item the outcome of the experiment prescribed at $u$, denoted by $\mathbf{y}_{N(v)}^{(v)}$; this is the outcome that determines the child $v$ among the children of $u$;
    \item the next waiting time $\tau(v)\in(0,1]$; since the protocol must remain inside the current unit-time interval, this quantity must satisfy
    \[
        \tau(v)\le 1-t(v);
    \]
    \item the time depth $t(v)$, defined recursively by $t(v)\coloneqq t(u)+\tau(u)$;
    \item the rotation coefficients $a^{(v)},b^{(v)}\in[-1,1]$ of the experiment prescribed at $v$.
\end{itemize}

By Condition~1 of Lemma~\ref{lemma_access_model}, these coefficients satisfy the constraint $\sqrt{(a^{(v)})^2+(b^{(v)})^2}\le \tau(v)$. Thus, unlike the previous two trees, at this scale the description of a node is fully specified by its waiting time and its two rotation coefficients.

The recursive construction is exactly as above. The root of the unit-time tree specifies the first experiment. Once the outcome of that experiment is observed, one obtains a child of the root. More generally, if the protocol is at a node $u$, then performing the experiment prescribed at $u$ and observing one of its possible outcomes determines a child $v$ of $u$. Repeating this procedure generates a branch of the tree, and the construction continues until the time depth reaches $1$.

Let $m\in\mathbb{R}$ be the value of the signal on the unit interval under consideration. Then, by Condition~1 of Lemma~\ref{lemma_access_model}, the outcome associated with an edge $u\to v$ is a two-dimensional Gaussian random variable distributed as
\begin{equation}\label{eq_6z}
\mathbf{y}_{N(v)}^{(v)} \sim
\NN\!\left[
m\begin{pmatrix}
a^{(u)}\\
b^{(u)}
\end{pmatrix},
\frac{\mathbb{1}_2}{10z^2}
\right].
\end{equation}
Finally, a node $l$ of the unit-time tree is called a \emph{leaf} if its time depth is equal to $1$, that is, if $t(l)=1$. Thus, a leaf represents a complete execution of the protocol inside the given unit-time interval.

\subsection{Lower bound in the simplified access model}

We now apply the tree framework introduced above to derive lower bounds on the minimum sensing time required by Gaussian protocols with bounded squeezing to solve the CV Learning Parity problem and the CV Hypothesis Testing problem.

As explained at the beginning of this section, both problems can be viewed as distribution-identification tasks. At each block of length $T$, the signal is described by a random vector $m\in\mathbb{R}^T$, sampled independently from an underlying distribution. We denote this distribution by $q_s$, where $s\in\{0,1\}^T$ is a fixed but unknown string. The family of distributions $\{q_s\}_{s\in\{0,1\}^T}$ is known to the learner, but the actual string $s$ is not.

This formulation includes the CV Learning Parity problem by taking $q_s$ to be the uniform distribution over all strings $x\in\{0,1\}^T$ satisfying $x\cdot s\equiv 0 \pmod 2$. It also includes the CV Hypothesis Testing problem by choosing two distributions, say $q_0$ and $q_1$, and setting $q_s=q_1$ for every non-zero string $s\in\{0,1\}^T\setminus\{0\}$. In the parity-learning case, the goal is to identify the unknown string $s$. In the hypothesis-testing case, the goal is instead to decide whether the underlying distribution is $q_0$ or $q_1$, which is captured in this notation by distinguishing the case $s=0$ from the case $s\ne 0$.

By Lemma~\ref{lem:pulse_to_piecewise}, followed by
Lemma~\ref{lemma_access_model}, any Gaussian protocol with bounded
squeezing can be simulated by the simplified access model introduced
in the latter lemma. Therefore, for the purpose of proving lower bounds,
it is enough to study algorithms in the simplified access model.
More precisely, if a Gaussian protocol solves the problem with sensing
time $T_{\mathrm{sensing}}$, then there is a protocol in the simplified
access model that solves the same problem with sensing time $nT$, where
\[
    n\coloneqq
    \left\lceil\frac{T_{\mathrm{sensing}}}{T}\right\rceil.
\]

We first formulate the corresponding learning task in the simplified access model.

\begin{problem}[(Learning $s$ in the simplified access model)] 
\label{def:PROB_simplified}
    Let $T\in\mathbb{N}^+$ be the pattern size. Let $s\in\{0,1\}^T$ be an unknown string. The learner has access to the simplified access model of Lemma~\ref{lemma_access_model} for a total sensing time of $nT$, where, at each block of length $T$, the vector $m\in\mathbb{R}^T$ is sampled independently according to $q_s$. The goal is to output a guess $\tilde{s}\in\{0,1\}^T$ such that
    \bb
        \Pr\!\left(\tilde{s}=s\right)\ge \frac23\,.
    \ee
\end{problem}

To prove a lower bound on the minimum sensing time, we shall use the following simpler task. This task gives the learner additional classical information after the sensing phase, and is therefore no harder than Problem~\ref{def:PROB_simplified}.

\begin{problem}[(Learning $s$ with partial revealed information in the simplified access model)]
\label{def:PROB_revealed}
    Let $T\in\mathbb{N}^+$ be the pattern size. The problem consists of four phases:
    \begin{itemize}
        \item \textbf{Referee's choice phase:} A referee selects a non-zero string $\bar{s}\in\{0,1\}^T\setminus\{0\}$ uniformly at random. Then the referee chooses the hidden string $s$ to be either $s=0$ or $s=\bar{s}$, with equal probability. Both $s$ and $\bar{s}$ are initially unknown to the learner.

        \item \textbf{Evolution phase:} The learner has access to the simplified access model of Lemma~\ref{lemma_access_model} for a total sensing time of $nT$, where, at each block of length $T$, the vector $m\in\mathbb{R}^T$ is sampled independently according to $q_s$. After this phase, the learner no longer has access to the access model.

        \item \textbf{Partial information reveal phase:} The referee reveals the string $\bar{s}$, but does not reveal whether the true string was $s=0$ or $s=\bar{s}$.

        \item \textbf{Classical post-processing phase:} Using the transcript gathered during the evolution phase, together with the revealed string $\bar{s}$, the learner must decide whether the true string was $0$ or $\bar{s}$, with success probability at least \(2/3\).
    \end{itemize}
\end{problem}

Any algorithm that solves Problem~\ref{def:PROB_simplified} also solves Problem~\ref{def:PROB_revealed} with the same sensing time: it first runs the learning algorithm during the evolution phase, and after $\bar{s}$ is revealed, it simply checks whether the estimated string is $0$ or $\bar{s}$. Thus, a lower bound for Problem~\ref{def:PROB_revealed} immediately implies the same lower bound for Problem~\ref{def:PROB_simplified}.

Let us now fix the pattern size $T$, and consider an algorithm $\mathcal{A}$ that solves Problem~\ref{def:PROB_revealed} using total sensing time $nT$. Our goal is to lower bound the number of blocks $n$. By the previous subsection, the adaptive algorithm $\mathcal{A}$ can be represented by a full-time tree. For each string $s\in\{0,1\}^T$, let $P^{(\mathcal{A})}_{s}$ denote the probability distribution over the leaves of the full-time tree, conditioned on the fact that, in each block of length $T$, the vector $m$ is sampled independently according to $q_s$.

Now fix a non-zero string $s\in\{0,1\}^T\setminus\{0\}$. Conditioned on the event that the referee revealed $\bar{s}=s$, the learner's task in Problem~\ref{def:PROB_revealed} is exactly a symmetric binary hypothesis testing problem between the two possible leaf distributions
\[
    P^{(\mathcal{A})}_{0}
    \qquad\text{and}\qquad
    P^{(\mathcal{A})}_{s},
\]
with equal prior probabilities. By the operational interpretation of the total variation distance between two probability distributions~\cite{KHATRI}, the optimal success probability in this symmetric binary testing problem is
\bb
    P^{(s)}_{\mathrm{best}}
    =
    \frac12\left(
        1+\frac12\left\|P^{(\mathcal{A})}_{0}
        -
        P^{(\mathcal{A})}_{s}\right\|_1
    \right),
\ee
where
\bb
    \left\|P^{(\mathcal{A})}_{0}-P^{(\mathcal{A})}_{s}\right\|_1
    =
    \sum_{l\in\mathrm{leaves}(r)}
    \left|P^{(\mathcal{A})}_{0}(l)-P^{(\mathcal{A})}_{s}(l)\right| .
\ee
Here $\mathrm{leaves}(r)$ denotes the set of leaves of the tree rooted at $r$; when the leaf space has continuous components, the sum is understood as the corresponding integral.

Since $\mathcal{A}$ solves Problem~\ref{def:PROB_revealed} with success probability at least $2/3$, and the revealed non-zero string $s$ is chosen uniformly at random, we must have
\[
    \frac23
    \le
    \mathop{\mathbb{E}}_{\substack{s\in\{0,1\}^T\\ s\ne 0}}
    P^{(s)}_{\mathrm{best}} .
\]
Using the expression for $P^{(s)}_{\mathrm{best}}$ above, this gives the key inequality
\bb\label{eq_crucial}
    \frac13
    \le
    \mathop{\mathbb{E}}_{\substack{s\in\{0,1\}^T\\ s\ne 0}}
    \frac12\left\|P^{(\mathcal{A})}_{0}
    -
    P^{(\mathcal{A})}_{s}\right\|_1 .
\ee
This inequality is the starting point for the lower-bound argument: to prove that the sensing time must be large, it is enough to upper bound the average total variation distance on the right-hand side as a function of the number of sensed blocks $n$.

We now express the leaf distribution $P^{(\mathcal{A})}_{s}$ in terms of the local transition probabilities along the full-time tree. We use the notation introduced in the previous subsection. Let $r$ be the root of the full-time tree, and let
\[
    r=u_0^{(l)}\to u_1^{(l)}\to\cdots\to u_{N(l)-1}^{(l)}\to u_{N(l)}^{(l)}=l
\]
be the unique path from the root to a leaf $l$. For each $i\in[N(l)]$, the edge $u_{i-1}^{(l)}\to u_i^{(l)}$ corresponds to the experiment prescribed at the node $u_{i-1}^{(l)}$, and its observed outcome is denoted by $\mathbf{y}^{(l)}_i$.

Conditioned on the hidden string $s$, the probability of reaching the leaf $l$ is
\bb\label{eq_leaf_distribution_product}
    P^{(\mathcal{A})}_s(l)
    =
    \prod_{i=1}^{N(l)}
    p_s^{(u_{i-1}^{(l)})}\!\left(\mathbf{y}^{(l)}_i\right),
\ee
where $p_s^{(u)}$ denotes the conditional outcome distribution of the experiment prescribed at the node $u$, conditioned on the hidden string being $s$. Here the dependence on the previous outcomes is already encoded in the node $u$, since the node represents the full transcript accumulated up to that point.

It remains to specify the local distributions $p_s^{(u)}$. Let $u$ be a node of the full-time tree.

\begin{enumerate}
    \item If $\tau(u)\ge 2$, then the experiment prescribed at $u$ spans $\tau(u)$ consecutive blocks of length $T$. In this case, the outcome takes values in $\mathbb{R}^2$, and its distribution is
    \bb\label{eq_full_time_local_distribution}
        p_s^{(u)}(\mathbf{y})
        \coloneqq
        \mathop{\mathbb{E}}_{m_1,\ldots,m_{\tau(u)}\sim q_s}
        \NN\!\left[
        \begin{pmatrix}
            \sum_{i=1}^{\tau(u)} a^{(u,i)}\!\cdot m_i\\[1mm]
            \sum_{i=1}^{\tau(u)} b^{(u,i)}\!\cdot m_i
        \end{pmatrix},
        \frac{\mathbb{1}_2}{10z^2}
        \right](\mathbf{y}),
        \qquad \mathbf{y}\in\mathbb{R}^2.
    \ee
    Here $m_1,\ldots,m_{\tau(u)}$ are sampled independently from $q_s$, and, for each $i\in[\tau(u)]$, the vectors $a^{(u,i)},b^{(u,i)}\in[-1,1]^T$ satisfy
    \[
        \sqrt{(a^{(u,i)}_j)^2+(b^{(u,i)}_j)^2}\le 1
        \qquad
        \forall\,j\in[T].
    \]

    \item If $\tau(u)=1$, then the experiment prescribed at $u$ is refined by the $T$-time tree encoded in $u$. In this case, the outcome $\mathbf{y}$ is a leaf of that attached $T$-time tree, and $p_s^{(u)}(\mathbf{y})$ denotes the corresponding leaf probability distribution, conditioned on the block variable being sampled according to $q_s$. This leaf probability is computed recursively from the transition probabilities of the attached $T$-time tree.
\end{enumerate}

We now turn to the estimate of the right-hand side of \eqref{eq_crucial}. The argument proceeds in two steps. First, we show that any experiment spanning several blocks can be reduced, in total variation distance, to the sum of its one-block contributions. Second, we use this local estimate in Lemma~\ref{lemma_full_time_tree_global_bound} to bound the right-hand side of \eqref{eq_crucial} in terms of the total number of sensed blocks $n$.

\begin{lemma}\label{lemma_full_time_reduction_to_one_block}
Let $T\in\mathbb{N}^+$ and $z\ge 1$. Let $\{q_s\}_{s\in\{0,1\}^T}$ be a family of probability distributions over $\mathbb{R}^T$. Define
\bb\label{eq:def_C_one_block}
    D
    \coloneqq
    \sup_{\substack{a,b\in[-1,1]^T\\
    \sqrt{a_j^2+b_j^2}\le 1\,\,\forall j\in[T]}}
    \mathop{\mathbb{E}}_{\substack{s\in\{0,1\}^{T}\\ s\ne 0}}
    \frac12
    \Bigg\|
    \left(
    \mathop{\mathbb{E}}_{m\sim q_s}
    -
    \mathop{\mathbb{E}}_{m\sim q_0}
    \right)
    \NN\!\left[
    \begin{pmatrix}
        a\cdot m\\
        b\cdot m
    \end{pmatrix},
    \frac{\mathbb{1}_2}{10z^2}
    \right]
    \Bigg\|_1 .
\ee
Then, for every $\tau\in\mathbb{N}^+$ and every collection of vectors
$a^{(1)},b^{(1)},\ldots,a^{(\tau)},b^{(\tau)}\in[-1,1]^T$
satisfying
\[
    \sqrt{(a^{(i)}_j)^2+(b^{(i)}_j)^2}\le 1
    \qquad
    \forall\, i\in[\tau],\ \forall\, j\in[T],
\]
it holds that
\bb\label{eq:full_time_reduction_to_one_block}
    &\mathop{\mathbb{E}}_{\substack{s\in\{0,1\}^{T}\\ s\ne 0}}
    \frac12
    \Bigg\|
    \left(
    \mathop{\mathbb{E}}_{m_1,\ldots,m_{\tau}\sim q_s}
    -
    \mathop{\mathbb{E}}_{m_1,\ldots,m_{\tau}\sim q_0}
    \right)
    \NN\!\left[
    \begin{pmatrix}
        \sum_{i=1}^{\tau} a^{(i)}\!\cdot m_i\\[1mm]
        \sum_{i=1}^{\tau} b^{(i)}\!\cdot m_i
    \end{pmatrix},
    \frac{\mathbb{1}_2}{10z^2}
    \right]
    \Bigg\|_1
    \le
    \tau D .
\ee
Here all expectations are over independent samples.
\end{lemma}

\begin{proof}
Fix admissible coefficients
$a^{(1)},b^{(1)},\ldots,a^{(\tau)},b^{(\tau)}$.
For each non-zero string $s\in\{0,1\}^T$ and each $\ell\in\{0,\ldots,\tau\}$, define
\bb
    H_\ell^{(s)}
    \coloneqq
    \mathop{\mathbb{E}}_{\substack{
        m_1,\ldots,m_\ell\sim q_0\\
        m_{\ell+1},\ldots,m_\tau\sim q_s}}
    \NN\!\left[
    \begin{pmatrix}
        \sum_{i=1}^{\tau} a^{(i)}\!\cdot m_i\\[1mm]
        \sum_{i=1}^{\tau} b^{(i)}\!\cdot m_i
    \end{pmatrix},
    \frac{\mathbb{1}_2}{10z^2}
    \right].
\ee
When $\ell=0$, all variables are sampled from $q_s$, while when $\ell=\tau$, all variables are sampled from $q_0$. Thus $H_0^{(s)}$ and $H_\tau^{(s)}$ are the two distributions appearing on the left-hand side of \eqref{eq:full_time_reduction_to_one_block}. By the telescoping identity and the triangle inequality,
\bb
    \frac12\|H_0^{(s)}-H_\tau^{(s)}\|_1
    &=
    \frac12\left\|
        \sum_{\ell=0}^{\tau-1}
        \bigl(H_\ell^{(s)}-H_{\ell+1}^{(s)}\bigr)
    \right\|_1 \le
    \sum_{\ell=0}^{\tau-1}
    \frac12\|H_\ell^{(s)}-H_{\ell+1}^{(s)}\|_1 .
\ee
The distributions $H_\ell^{(s)}$ and $H_{\ell+1}^{(s)}$ differ only in the law of $m_{\ell+1}$. All remaining variables enter the Gaussian mean through the same additive shift. Therefore, by convexity of total variation under mixtures and by translation invariance of total variation,
\bb
    \frac12\|H_\ell^{(s)}-H_{\ell+1}^{(s)}\|_1
    \le
    \frac12
    \Bigg\|
    \left(
    \mathop{\mathbb{E}}_{m\sim q_s}
    -
    \mathop{\mathbb{E}}_{m\sim q_0}
    \right)
    \NN\!\left[
    \begin{pmatrix}
        a^{(\ell+1)}\!\cdot m\\[1mm]
        b^{(\ell+1)}\!\cdot m
    \end{pmatrix},
    \frac{\mathbb{1}_2}{10z^2}
    \right]
    \Bigg\|_1 .
\ee
Averaging over the uniformly chosen non-zero string $s$ and summing over $\ell$ gives
\begin{equation*}
    \mathop{\mathbb{E}}_{\substack{s\in\{0,1\}^{T}\\ s\ne 0}}
    \frac12\|H_0^{(s)}-H_\tau^{(s)}\|_1
    \le
    \sum_{\ell=0}^{\tau-1}
    \mathop{\mathbb{E}}_{\substack{s\in\{0,1\}^{T}\\ s\ne 0}}
    \frac12
    \Bigg\|
    \left(
    \mathop{\mathbb{E}}_{m\sim q_s}
    -
    \mathop{\mathbb{E}}_{m\sim q_0}
    \right)
    \NN\!\left[
    \begin{pmatrix}
        a^{(\ell+1)}\!\cdot m\\[1mm]
        b^{(\ell+1)}\!\cdot m
    \end{pmatrix},
    \frac{\mathbb{1}_2}{10z^2}
    \right]
    \Bigg\|_1 \le
    \sum_{\ell=0}^{\tau-1} D
    =
    \tau D,
\end{equation*}
where each pair $a^{(\ell+1)},b^{(\ell+1)}$ is admissible in the supremum defining $D$. This proves \eqref{eq:full_time_reduction_to_one_block}.
\end{proof}

\begin{lemma}\label{lemma_full_time_tree_global_bound}
Consider the simplified access model, with $T\in\mathbb{N}^+$ and $z\ge 1$. Assume that there exists a constant $C$ such that, for every $T$-time tree $\mathcal{T}$,
\bb\label{eq:one_block_attached_tree_bound}
    \mathop{\mathbb{E}}_{\substack{s\in\{0,1\}^T\\ s\ne 0}}
    \frac12
    \left\|P_s^{(\mathcal{T})}-P_0^{(\mathcal{T})}\right\|_1
    \le
    C,
\ee
where $P_s^{(\mathcal{T})}$ denotes the leaf distribution of $\mathcal{T}$ when the block vector $m\in\mathbb{R}^T$ is sampled according to $q_s$.

Let $\mathcal{A}$ be a full-time tree with total time depth $n$. Then
\bb\label{eq:global_tree_bound}
    \mathop{\mathbb{E}}_{\substack{s\in\{0,1\}^T\\ s\ne 0}}
    \frac12
    \left\|P_s^{(\mathcal{A})}-P_0^{(\mathcal{A})}\right\|_1
    \le
    nC .
\ee
\end{lemma}

\begin{proof}
We use the notation introduced above for the full-time tree. Sums over children and leaves are understood as integrals whenever the corresponding outcome space is continuous.

We first establish a local bound for every internal node $u$ of the full-time tree:
\bb\label{eq:local_full_time_bound}
    \mathop{\mathbb{E}}_{\substack{s\in\{0,1\}^T\\ s\ne 0}}
    \frac12
    \left\|p_s^{(u)}-p_0^{(u)}\right\|_1
    \le
    \tau(u)C .
\ee
If $\tau(u)=1$, then the experiment prescribed at $u$ is confined to a single block of length $T$. By construction of the full-time tree, this experiment is represented by the $T$-time tree attached to $u$. Hence $p_s^{(u)}$ is exactly the leaf distribution of that attached $T$-time tree, conditioned on the block variable being sampled according to $q_s$. Therefore \eqref{eq:local_full_time_bound} follows directly from \eqref{eq:one_block_attached_tree_bound}.

Suppose instead that $\tau(u)\ge 2$. Then the experiment prescribed at $u$ spans $\tau(u)$ consecutive blocks of length $T$, and its outcome distribution is the Gaussian mixture described in \eqref{eq_full_time_local_distribution}. By Lemma~\ref{lemma_full_time_reduction_to_one_block},
\[
    \mathop{\mathbb{E}}_{\substack{s\in\{0,1\}^T\\ s\ne 0}}
    \frac12
    \left\|p_s^{(u)}-p_0^{(u)}\right\|_1
    \le
    \tau(u)D.
\]
Moreover, $D\le C$. Indeed, every one-block Gaussian experiment appearing in the definition of $D$ is a particular $T$-time tree: the learner chooses the coefficients $a,b$, waits for the whole block of length $T$, and performs the corresponding Gaussian measurement at the end. Since \eqref{eq:one_block_attached_tree_bound} holds for every $T$-time tree, it also holds for this subclass. Hence $D\le C$, and \eqref{eq:local_full_time_bound} follows also when $\tau(u)\ge 2$.

For every node $v$ of the full-time tree, define the path weight
\[
    W_s(v)
    \coloneqq
    \prod_{i=1}^{N(v)}
    p_s^{(u_{i-1}^{(v)})}\!\left(\mathbf{y}^{(v)}_i\right),
\]
with the convention $W_s(r)=1$ at the root. Thus, if $l$ is a leaf of the full-time tree, then $W_s(l)=P_s^{(\mathcal{A})}(l)$.

We now propagate the local estimate through the full-time tree. We prove, by backward induction on the remaining time depth $n-t(v)$, that for every node $v$,
\bb\label{eq:tree_induction_claim}
    \sum_{l\in\mathrm{leaves}(v)}
    \mathop{\mathbb{E}}_{\substack{s\in\{0,1\}^T\\ s\ne 0}}
    \left|W_s(l)-W_0(l)\right|
    \le
    \mathop{\mathbb{E}}_{\substack{s\in\{0,1\}^T\\ s\ne 0}}
    \left|W_s(v)-W_0(v)\right|
    +
    2C\bigl(n-t(v)\bigr)W_0(v).
\ee
Here $\mathrm{leaves}(v)$ denotes the set of leaves of the subtree rooted at $v$.

If $v$ is a leaf, then $t(v)=n$ and $\mathrm{leaves}(v)=\{v\}$. Thus \eqref{eq:tree_induction_claim} holds with equality.

Now let $u$ be an internal node, and assume that \eqref{eq:tree_induction_claim} holds for every child $v$ of $u$. Since the experiment prescribed at $u$ has duration $\tau(u)$, each child $v$ of $u$ has time depth $t(v)=t(u)+\tau(u)$. Applying the induction hypothesis to the subtrees rooted at the children of $u$, we obtain
\bb\label{eq:after_induction_children}
    &\sum_{l\in\mathrm{leaves}(u)}
    \mathop{\mathbb{E}}_{\substack{s\in\{0,1\}^T\\ s\ne 0}}
    \left|W_s(l)-W_0(l)\right|\\
    &\le
    \sum_{v\in\mathrm{children}(u)}
    \mathop{\mathbb{E}}_{\substack{s\in\{0,1\}^T\\ s\ne 0}}
    \left|W_s(v)-W_0(v)\right|
    +
    2C\bigl(n-t(u)-\tau(u)\bigr)
    \sum_{v\in\mathrm{children}(u)} W_0(v)\\
    &=
    \sum_{v\in\mathrm{children}(u)}
    \mathop{\mathbb{E}}_{\substack{s\in\{0,1\}^T\\ s\ne 0}}
    \left|W_s(v)-W_0(v)\right|
    +
    2C\bigl(n-t(u)-\tau(u)\bigr)W_0(u),
\ee
where we used
\[
    W_0(v)=W_0(u)p_0^{(u)}(\mathbf{y}_v),
    \qquad
    \sum_{v\in\mathrm{children}(u)}p_0^{(u)}(\mathbf{y}_v)=1.
\]

It remains to bound the first term on the right-hand side. For every child $v$ of $u$,
\[
    W_s(v)=W_s(u)\,p_s^{(u)}(\mathbf{y}_v),
    \qquad
    W_0(v)=W_0(u)\,p_0^{(u)}(\mathbf{y}_v).
\]
Therefore,
\bb\label{eq:first_child_sum_bound}
    &\sum_{v\in\mathrm{children}(u)}
    \mathop{\mathbb{E}}_{\substack{s\in\{0,1\}^T\\ s\ne 0}}
    \left|W_s(v)-W_0(v)\right|\\
    &=
    \sum_{v\in\mathrm{children}(u)}
    \mathop{\mathbb{E}}_{\substack{s\in\{0,1\}^T\\ s\ne 0}}
    \left|
    W_s(u)p_s^{(u)}(\mathbf{y}_v)
    -
    W_0(u)p_0^{(u)}(\mathbf{y}_v)
    \right|\\
    &=
    \sum_{v\in\mathrm{children}(u)}
    \mathop{\mathbb{E}}_{\substack{s\in\{0,1\}^T\\ s\ne 0}}
    \left|
    \bigl(W_s(u)-W_0(u)\bigr)p_s^{(u)}(\mathbf{y}_v)
    +
    W_0(u)\bigl(p_s^{(u)}(\mathbf{y}_v)-p_0^{(u)}(\mathbf{y}_v)\bigr)
    \right|\\
    &\le
    \sum_{v\in\mathrm{children}(u)}
    \mathop{\mathbb{E}}_{\substack{s\in\{0,1\}^T\\ s\ne 0}}
    \left|
    W_s(u)-W_0(u)
    \right|
    p_s^{(u)}(\mathbf{y}_v)\\
    &\qquad+
    W_0(u)
    \sum_{v\in\mathrm{children}(u)}
    \mathop{\mathbb{E}}_{\substack{s\in\{0,1\}^T\\ s\ne 0}}
    \left|
    p_s^{(u)}(\mathbf{y}_v)-p_0^{(u)}(\mathbf{y}_v)
    \right|\\
    &=
    \mathop{\mathbb{E}}_{\substack{s\in\{0,1\}^T\\ s\ne 0}}
    \left|W_s(u)-W_0(u)\right|
    +
    W_0(u)
    \mathop{\mathbb{E}}_{\substack{s\in\{0,1\}^T\\ s\ne 0}}
    \left\|p_s^{(u)}-p_0^{(u)}\right\|_1\\
    &\le
    \mathop{\mathbb{E}}_{\substack{s\in\{0,1\}^T\\ s\ne 0}}
    \left|W_s(u)-W_0(u)\right|
    +
    2C\,\tau(u)\,W_0(u),
\ee
where the last step follows from the local estimate \eqref{eq:local_full_time_bound}. Combining \eqref{eq:after_induction_children} and \eqref{eq:first_child_sum_bound}, we obtain
\[
    \sum_{l\in\mathrm{leaves}(u)}
    \mathop{\mathbb{E}}_{\substack{s\in\{0,1\}^T\\ s\ne 0}}
    \left|W_s(l)-W_0(l)\right|
    \le
    \mathop{\mathbb{E}}_{\substack{s\in\{0,1\}^T\\ s\ne 0}}
    \left|W_s(u)-W_0(u)\right|
    +
    2C\bigl(n-t(u)\bigr)W_0(u).
\]
This proves the induction step, and hence the induction claim.

Finally, apply \eqref{eq:tree_induction_claim} to the root $r$. Since $t(r)=0$ and $W_s(r)=W_0(r)=1$, we obtain
\[
    \sum_{l\in\mathrm{leaves}(r)}
    \mathop{\mathbb{E}}_{\substack{s\in\{0,1\}^T\\ s\ne 0}}
    \left|P_s^{(\mathcal{A})}(l)-P_0^{(\mathcal{A})}(l)\right|
    \le
    2nC,
\]
or equivalently,
\[
    \mathop{\mathbb{E}}_{\substack{s\in\{0,1\}^T\\ s\ne 0}}
    \frac12
    \left\|P_s^{(\mathcal{A})}-P_0^{(\mathcal{A})}\right\|_1
    \le
    nC.
\]
This proves the claim.
\end{proof}

Combining \eqref{eq_crucial} and Lemma~\ref{lemma_full_time_tree_global_bound}, we obtain the following.

\begin{lemma}\label{lemma_sensing_time_from_T_tree_bound}
Consider the simplified access model with pattern size $T\in\mathbb{N}^+$ and $z\ge 1$. Assume that there exists a constant $C>0$ such that, for every $T$-time tree $\mathcal{T}$,
\bb\label{eq:one_block_attached_tree_bound_final}
    \mathop{\mathbb{E}}_{\substack{s\in\{0,1\}^T\\ s\ne 0}}
    \frac12
    \left\|P_s^{(\mathcal{T})}-P_0^{(\mathcal{T})}\right\|_1
    \le
    C,
\ee
where $P_s^{(\mathcal{T})}$ denotes the leaf distribution of $\mathcal{T}$ when the block vector $m\in\mathbb{R}^T$ is sampled according to $q_s$.

If Problem~\ref{def:PROB_revealed} can be solved with success probability at least \(2/3\) using total sensing time $nT$, then
\bb\label{eq:sensing_time_lower_bound_from_C}
    n
    \ge
    \frac{1}{3C}.
\ee
\end{lemma}

\begin{proof}
Let $\mathcal{A}$ be a full-time tree representing an algorithm that solves Problem~\ref{def:PROB_revealed} with total sensing time $nT$ and success probability at least \(2/3\). By \eqref{eq_crucial},
\[
    \frac13
    \le
    \mathop{\mathbb{E}}_{\substack{s\in\{0,1\}^T\\ s\ne 0}}
    \frac12
    \left\|P^{(\mathcal{A})}_{0}
    -
    P^{(\mathcal{A})}_{s}\right\|_1 .
\]
On the other hand, Lemma~\ref{lemma_full_time_tree_global_bound} gives
\[
    \mathop{\mathbb{E}}_{\substack{s\in\{0,1\}^T\\ s\ne 0}}
    \frac12
    \left\|P^{(\mathcal{A})}_{0}
    -
    P^{(\mathcal{A})}_{s}\right\|_1
    \le
    nC .
\]
Combining the two inequalities yields $1/3\le nC$, and therefore $n\ge 1/(3C)$.
\end{proof}

\subsection{Lower bound for the CV Learning Parity problem}

We now specialize Lemma~\ref{lemma_sensing_time_from_T_tree_bound} to the CV Learning Parity problem. By that lemma, it is enough to upper bound, uniformly over all $T$-time trees $\mathcal{T}$, the average total variation distance
\[
    \mathop{\mathbb{E}}_{\substack{s\in\{0,1\}^T\\ s\ne 0}}
    \frac12
    \left\|P_s^{(\mathcal{T})}-P_0^{(\mathcal{T})}\right\|_1 .
\]
Here $P_s^{(\mathcal{T})}$ denotes the leaf distribution of the $T$-time tree $\mathcal{T}$ when the block string is sampled from the parity distribution associated with $s$.

Recall that, in the CV Learning Parity problem, the block variable is a bit string $x\in\{0,1\}^T$, sampled uniformly from
\[
    C_s
    =
    \{x\in\{0,1\}^T:x\cdot s\equiv 0 \pmod 2\}.
\]
In particular, for $s=0$ we have $C_0=\{0,1\}^T$, so $P_0^{(\mathcal{T})}$ corresponds to the case in which the block string is sampled uniformly from the whole hypercube.

Let $\mathcal{T}$ be a fixed $T$-time tree. To simplify notation, throughout this subsection we write
\[
    P_s(l)\equiv P_s^{(\mathcal{T})}(l)
\]
for the probability of reaching the leaf $l$ of $\mathcal{T}$, conditioned on the hidden parity string being $s$.

Let $r$ be the root of $\mathcal{T}$, and let $l\in\mathrm{leaves}(\mathcal{T})$. We denote the unique path from $r$ to $l$ by
\[
    r=u_0^{(l)}
    \to u_1^{(l)}
    \to \cdots
    \to u_{N(l)-1}^{(l)}
    \to u_{N(l)}^{(l)}=l.
\]
For each $i\in[N(l)]$, the edge
\[
    u_{i-1}^{(l)}\to u_i^{(l)}
\]
corresponds to the experiment prescribed at the node $u_{i-1}^{(l)}$, and its observed outcome is denoted by $\mathbf{y}_i^{(l)}$.

Along this fixed branch, any block string $x\in\{0,1\}^T$ can be decomposed according to the successive time intervals explored by the experiments:
\[
    x=
    \left(
    x_1^{(l)},x_2^{(l)},\ldots,x_{N(l)}^{(l)}
    \right),
\]
where $x_i^{(l)}
    \in
    \{0,1\}^{\tau(u_{i-1}^{(l)})}$. Thus $x_i^{(l)}$ is the substring of $x$ corresponding to the experiment prescribed at the node $u_{i-1}^{(l)}$. Since $l$ is a leaf of the $T$-time tree, the durations along the branch add up to the whole block:
\[
    \sum_{i=1}^{N(l)}
    \tau\!\left(u_{i-1}^{(l)}\right)
    =
    T.
\]
We also introduce the following local notation. Let $u$ be a node of the $T$-time tree, and let $\xi\in\{0,1\}^{\tau(u)}$ be the substring of the block variable associated with the experiment prescribed at $u$. We denote by $p_{\xi}^{(u)}(\mathbf{y})$ the conditional distribution of the outcome of that experiment, conditioned on this local substring being equal to $\xi$.

More explicitly, if $\tau(u)\ge 2$, then the experiment at $u$ is an ordinary $T$-time-tree experiment. In this case, the outcome belongs to $\mathbb{R}^2$, and
\bb\label{eq_gausscond2}
    p_{\xi}^{(u)}(\mathbf{y})
    =
    \NN\!\left[
    \begin{pmatrix}
        a^{(u)}\!\cdot \xi\\
        b^{(u)}\!\cdot \xi
    \end{pmatrix},
    \frac{\mathbb{1}_2}{10z^2}
    \right](\mathbf{y}),
    \qquad
    \mathbf{y}\in\mathbb{R}^2.
\ee
If instead $\tau(u)=1$, then the experiment at $u$ is refined by the unit-time tree attached to $u$. In this case, $\xi\in\{0,1\}$, the outcome $\mathbf{y}$ is a leaf, equivalently the full transcript, of that attached unit-time tree, and $p_{\xi}^{(u)}(\mathbf{y})$ denotes its leaf probability distribution conditioned on the bit value $\xi$.

With this notation, the probability of reaching a leaf $l$ under the hidden string $s$ is
\bb\label{eq:T_time_leaf_distribution_parity}
    P_s(l)
    =
    \mathop{\mathbb{E}}_{\substack{x\in\{0,1\}^T\\ x\cdot s\equiv 0}}
    \prod_{i=1}^{N(l)}
    p_{x_i^{(l)}}^{(u_{i-1}^{(l)})}
    \!\left(\mathbf{y}_i^{(l)}\right),
\ee
where the expectation is uniform over the strings satisfying the parity constraint. In particular,
\[
    P_0(l)
    =
    \mathop{\mathbb{E}}_{x\in\{0,1\}^T}
    \prod_{i=1}^{N(l)}
    p_{x_i^{(l)}}^{(u_{i-1}^{(l)})}
    \!\left(\mathbf{y}_i^{(l)}\right).
\]
The next lemma shows that, after conditioning on a leaf, the dependence on the hidden parity string factorizes along the branch.

\begin{lemma}\label{lemma_adapt_prob}
Let $\mathcal{T}$ be a $T$-time tree. For every node $u$ of $\mathcal{T}$, define
\bb\label{eq:def_f_u_parity}
    f^{(u)}(\mathbf{y})
    \coloneqq
    \frac{1}{2^{\tau(u)}}
    \mathop{\mathbb{E}}_{\xi\in\{0,1\}^{\tau(u)}}
    p_{\xi}^{(u)}(\mathbf{y})
    +
    \left(1-\frac{1}{2^{\tau(u)}}\right)
    \mathop{\mathbb{E}}_{\substack{\sigma\in\{0,1\}^{\tau(u)}\\ \sigma\ne 0}}
    \frac12
    \left|
    \mathop{\mathbb{E}}_{\substack{\xi\in\{0,1\}^{\tau(u)}\\ \xi\cdot\sigma\equiv 1}}
    p_{\xi}^{(u)}(\mathbf{y})
    -
    \mathop{\mathbb{E}}_{\substack{\xi\in\{0,1\}^{\tau(u)}\\ \xi\cdot\sigma\equiv 0}}
    p_{\xi}^{(u)}(\mathbf{y})
    \right|.
\ee
Then, for every leaf $l\in\mathrm{leaves}(\mathcal{T})$,
\bb\label{eq:adapt_prob_product_bound}
    \mathop{\mathbb{E}}_{\substack{s\in\{0,1\}^{T}\\ s\ne 0}}
    \left|P_s(l)-P_0(l)\right|
    \le
    \frac{1}{1-2^{-T}}
    \prod_{i=1}^{N(l)}
    f^{(u_{i-1}^{(l)})}
    \!\left(\mathbf{y}_i^{(l)}\right).
\ee
\end{lemma}

\begin{proof}
Fix a leaf $l\in\mathrm{leaves}(\mathcal{T})$. To simplify notation inside the proof, write
\[
    u_i\coloneqq u_{i-1}^{(l)},
    \qquad
    \tau_i\coloneqq \tau(u_i),
    \qquad
    \mathbf{y}_i\coloneqq \mathbf{y}_i^{(l)}
\]
for $i\in[N(l)]$. Along the branch leading to $l$, any string $x\in\{0,1\}^T$ decomposes as
\[
    x=(x_1,\ldots,x_{N(l)}),
    \qquad
    x_i\in\{0,1\}^{\tau_i}.
\]
We decompose $s\in\{0,1\}^T$ in the same way:
\[
    s=(s_1,\ldots,s_{N(l)}),
    \qquad
    s_i\in\{0,1\}^{\tau_i}.
\]
For each node $u_i$ and each local string $\sigma\in\{0,1\}^{\tau_i}$, define
\[
    A_{\sigma}^{(u_i)}(\mathbf{y}_i)
    \coloneqq
    \mathop{\mathbb{E}}_{\xi\in\{0,1\}^{\tau_i}}
    (-1)^{\xi\cdot\sigma}
    p_{\xi}^{(u_i)}(\mathbf{y}_i).
\]
If $\sigma=0$, then
\[
    A_{0}^{(u_i)}(\mathbf{y}_i)
    =
    \mathop{\mathbb{E}}_{\xi\in\{0,1\}^{\tau_i}}
    p_{\xi}^{(u_i)}(\mathbf{y}_i).
\]
If $\sigma\ne0$, then exactly half of the strings $\xi\in\{0,1\}^{\tau_i}$ satisfy $\xi\cdot\sigma\equiv0$ and half satisfy $\xi\cdot\sigma\equiv1$. Therefore,
\[
    A_{\sigma}^{(u_i)}(\mathbf{y}_i)
    =
    \frac12
    \left(
    \mathop{\mathbb{E}}_{\substack{\xi\in\{0,1\}^{\tau_i}\\ \xi\cdot\sigma\equiv 0}}
    p_{\xi}^{(u_i)}(\mathbf{y}_i)
    -
    \mathop{\mathbb{E}}_{\substack{\xi\in\{0,1\}^{\tau_i}\\ \xi\cdot\sigma\equiv 1}}
    p_{\xi}^{(u_i)}(\mathbf{y}_i)
    \right).
\]
Consequently,
\bb\label{eq:f_as_average_fourier}
    f^{(u_i)}(\mathbf{y}_i)
    =
    \mathop{\mathbb{E}}_{\sigma\in\{0,1\}^{\tau_i}}
    \left|A_{\sigma}^{(u_i)}(\mathbf{y}_i)\right|.
\ee
Now define
\[
    G_l(x)
    \coloneqq
    \prod_{i=1}^{N(l)}
    p_{x_i}^{(u_i)}(\mathbf{y}_i).
\]
For every $s\ne0$, using that
\[
    \mathop{\mathbb{E}}_{\substack{x\in\{0,1\}^T\\ x\cdot s\equiv0}}
    G_l(x)
    =
    \mathop{\mathbb{E}}_{x\in\{0,1\}^T}
    G_l(x)\bigl(1+(-1)^{x\cdot s}\bigr),
\]
we obtain
\bb
    P_s(l)-P_0(l)
    =
    \mathop{\mathbb{E}}_{x\in\{0,1\}^T}
    G_l(x)(-1)^{x\cdot s}
    =
    \mathop{\mathbb{E}}_{x_1,\ldots,x_{N(l)}}
    \prod_{i=1}^{N(l)}
    p_{x_i}^{(u_i)}(\mathbf{y}_i)(-1)^{x_i\cdot s_i}
    =
    \prod_{i=1}^{N(l)}
    A_{s_i}^{(u_i)}(\mathbf{y}_i).
\ee
Taking absolute values gives
\bb\label{eq:leaf_diff_product_fourier}
    \left|P_s(l)-P_0(l)\right|
    =
    \prod_{i=1}^{N(l)}
    \left|A_{s_i}^{(u_i)}(\mathbf{y}_i)\right|
    \qquad
    \forall\,s\ne0.
\ee
We now average over the uniformly chosen non-zero string $s$. Since the decomposition $s=(s_1,\ldots,s_{N(l)})$ is compatible with the decomposition of the block, averaging uniformly over $s\in\{0,1\}^T$ is the same as averaging independently over each local string $s_i\in\{0,1\}^{\tau_i}$. Hence, using \eqref{eq:leaf_diff_product_fourier},
\bb
    \mathop{\mathbb{E}}_{\substack{s\in\{0,1\}^T\\ s\ne0}}
    \left|P_s(l)-P_0(l)\right|
    &=
    \frac{1}{1-2^{-T}}
    \left[
    \mathop{\mathbb{E}}_{s\in\{0,1\}^T}
    \prod_{i=1}^{N(l)}
    \left|A_{s_i}^{(u_i)}(\mathbf{y}_i)\right|
    -
    2^{-T}
    \prod_{i=1}^{N(l)}
    \left|A_{0}^{(u_i)}(\mathbf{y}_i)\right|
    \right]\\
    &\le
    \frac{1}{1-2^{-T}}
    \prod_{i=1}^{N(l)}
    \mathop{\mathbb{E}}_{\sigma\in\{0,1\}^{\tau_i}}
    \left|A_{\sigma}^{(u_i)}(\mathbf{y}_i)\right|\\
    &=
    \frac{1}{1-2^{-T}}
    \prod_{i=1}^{N(l)}
    f^{(u_i)}(\mathbf{y}_i),
\ee
where in the last step we used \eqref{eq:f_as_average_fourier}. This concludes the proof.
\end{proof}

We now need to control the local factors appearing in Lemma~\ref{lemma_adapt_prob}. The bound in \eqref{eq:adapt_prob_product_bound} reduces the leafwise distinguishability to a product of terms $f^{(u)}$, one for each node visited by the branch. Thus, the next step is to bound $\|f^{(u)}\|_1$ uniformly over the nodes of the $T$-time tree. There are two cases: if $\tau(u)=1$, the node is resolved by an attached unit-time tree; if $\tau(u)\ge2$, the transition is directly given by the Gaussian law in \eqref{eq_gausscond2}. We begin with the first case.

\subsubsection{Unit-time tree analysis}

Let $u$ be a node of the $T$-time tree with $\tau(u)=1$. Then the local substring is a single bit $\xi\in\{0,1\}$, and $p_\xi^{(u)}$ is the leaf distribution of the unit-time tree attached to $u$, conditioned on the signal value $\xi$. The next lemma shows that no amount of adaptivity inside this unit interval can make the two cases $\xi=0$ and $\xi=1$ more distinguishable than a single Gaussian observation with total signal strength equal to one. This gives the required estimate on the local factor $f^{(u)}$.

\begin{lemma}\label{lemma_unit_time_tree_contribution}
Let $u$ be a node of a $T$-time tree with $\tau(u)=1$, and let $\mathcal{U}_u$ be the unit-time tree attached to $u$. Then
\bb\label{eq:unit_time_fu_bound}
    \left\|f^{(u)}\right\|_1
    \le
    \exp\!\left(-\frac{1}{12}e^{-5z^2/2}\right).
\ee
\end{lemma}

\begin{proof}
Let $P_0^{(\mathcal{U}_u)}$ and $P_1^{(\mathcal{U}_u)}$ denote the leaf distributions of the attached unit-time tree when the value of the signal inside the unit interval is respectively $0$ and $1$. We first show that
\bb\label{eq:unit_time_tree_tv_bound}
    \left\|P_1^{(\mathcal{U}_u)}-P_0^{(\mathcal{U}_u)}\right\|_1
    \le
    \left\|
    \NN\!\left[1,\frac{1}{10z^2}\right]
    -
    \NN\!\left[0,\frac{1}{10z^2}\right]
    \right\|_1 .
\ee
Let $r$ be the root of $\mathcal{U}_u$. For every node $v$ of $\mathcal{U}_u$ and every $x\in\{0,1\}$, define the path weight
\[
    W_x(v)
    \coloneqq
    \prod_{i=1}^{N(v)}
    p_x^{(u_{i-1}^{(v)})}\!\left(\mathbf{y}_i^{(v)}\right),
\]
with the convention $W_x(r)=1$. Thus, if $l$ is a leaf of $\mathcal{U}_u$, then $W_x(l)=P_x^{(\mathcal{U}_u)}(l)$.

We prove, by backward induction on the unit-time tree, that for every node $v$,
\bb\label{eq:unit_time_induction_claim}
    \sum_{l\in\mathrm{leaves}(v)}
    \left|W_1(l)-W_0(l)\right|
    \le
    \int_{\mathbb{R}}
    \left|
    W_1(v)
    \NN\!\left[1-t(v),\frac{1}{10z^2}\right](r')
    -
    W_0(v)
    \NN\!\left[0,\frac{1}{10z^2}\right](r')
    \right|
    \mathrm{d}r' .
\ee
If $v$ is a leaf, then $t(v)=1$ and $\mathrm{leaves}(v)=\{v\}$. Hence the right-hand side of \eqref{eq:unit_time_induction_claim} is
\[
    \int_{\mathbb{R}}
    \left|
    W_1(v)
    -
    W_0(v)
    \right|
    \NN\!\left[0,\frac{1}{10z^2}\right](r')\mathrm{d}r'
    =
    |W_1(v)-W_0(v)|,
\]
so the claim holds with equality.

Now let $w$ be an internal node, and assume that \eqref{eq:unit_time_induction_claim} holds for every child $v$ of $w$. Since each child satisfies $t(v)=t(w)+\tau(w)$, the induction hypothesis gives
\bb\label{eq:aaaa}
    &\sum_{l\in\mathrm{leaves}(w)}
    \left|W_1(l)-W_0(l)\right|\\
    &\le
    \sum_{v\in\mathrm{children}(w)}
    \int_{\mathbb{R}}
    \left|
    W_1(v)\NN\!\left[1-t(w)-\tau(w),\frac{1}{10z^2}\right](r')
    -
    W_0(v)\NN\!\left[0,\frac{1}{10z^2}\right](r')
    \right|
    \mathrm{d}r' .
\ee
For every child $v$ of $w$, we have $W_x(v)=W_x(w)p_x^{(w)}(\mathbf{y}_v)$. Moreover, by the unit-time transition law \eqref{eq_6z},
\[
    p_1^{(w)}(\mathbf{y})
    =
    \NN\!\left[
    \begin{pmatrix}
        a^{(w)}\\
        b^{(w)}
    \end{pmatrix},
    \frac{\mathbb{1}_2}{10z^2}
    \right](\mathbf{y}),
    \qquad
    p_0^{(w)}(\mathbf{y})
    =
    \NN\!\left[
    \begin{pmatrix}
        0\\
        0
    \end{pmatrix},
    \frac{\mathbb{1}_2}{10z^2}
    \right](\mathbf{y}).
\]
Thus the right-hand side of \eqref{eq:aaaa} becomes
\bb\label{eq:unit_time_step_integral}
    &\int_{\mathbb{R}^2}\mathrm{d}^2\mathbf{y}
    \int_{\mathbb{R}}\mathrm{d}r'\,
    \Bigg|
    W_1(w)
    \NN\!\left[
    \begin{pmatrix}
        a^{(w)}\\
        b^{(w)}
    \end{pmatrix},
    \frac{\mathbb{1}_2}{10z^2}
    \right](\mathbf{y})
    \NN\!\left[1-t(w)-\tau(w),\frac{1}{10z^2}\right](r')\\
    &\hspace{38mm}
    -
    W_0(w)
    \NN\!\left[
    \begin{pmatrix}
        0\\
        0
    \end{pmatrix},
    \frac{\mathbb{1}_2}{10z^2}
    \right](\mathbf{y})
    \NN\!\left[0,\frac{1}{10z^2}\right](r')
    \Bigg|.
\ee
Equivalently, this is the $L^1$ distance between two Gaussian densities in $\mathbb{R}^3$, with the same covariance matrix $\mathbb{1}_3/(10z^2)$ and means
\[
    \begin{pmatrix}
        a^{(w)}\\
        b^{(w)}\\
        1-t(w)-\tau(w)
    \end{pmatrix}
    \qquad\text{and}\qquad
    \begin{pmatrix}
        0\\
        0\\
        0
    \end{pmatrix},
\]
weighted respectively by $W_1(w)$ and $W_0(w)$. Since the covariance is proportional to the identity, we may rotate the variables without changing the integral. After rotating the first mean onto a single coordinate and integrating out the two orthogonal coordinates, \eqref{eq:unit_time_step_integral} is equal to
\bb\label{eq:unit_time_delta_bound}
    \int_{\mathbb{R}}
    \left|
    W_1(w)
    \NN\!\left[\Delta_w,\frac{1}{10z^2}\right](r')
    -
    W_0(w)
    \NN\!\left[0,\frac{1}{10z^2}\right](r')
    \right|
    \mathrm{d}r',
\ee
where
\[
    \Delta_w
    \coloneqq
    \sqrt{
        (a^{(w)})^2
        +
        (b^{(w)})^2
        +
        \bigl(1-t(w)-\tau(w)\bigr)^2
    }.
\]
If $\Delta_w=0$, then \eqref{eq:unit_time_delta_bound} reduces to $|W_1(w)-W_0(w)|$, and the desired induction step follows immediately. We may therefore assume that $\Delta_w>0$. By Condition~1 of Lemma~\ref{lemma_access_model}, the rotation coefficients satisfy $\sqrt{(a^{(w)})^2+(b^{(w)})^2}\le \tau(w)$. Hence
\[
    \Delta_w
    \le
    \sqrt{(a^{(w)})^2+(b^{(w)})^2}+1-t(w)-\tau(w)
    \le
    1-t(w).
\]
Thus the quantity $\gamma_w
    \coloneqq
    \frac{1-t(w)}{\Delta_w}$ is well defined and satisfies $\gamma_w\ge 1$.
By the change of variables $r'\mapsto \gamma_w r'$, the integral in \eqref{eq:unit_time_delta_bound} is equal to
\[
    \int_{\mathbb{R}}
    \left|
    W_1(w)
    \NN\!\left[1-t(w),\frac{\gamma_w^2}{10z^2}\right](r')
    -
    W_0(w)
    \NN\!\left[0,\frac{\gamma_w^2}{10z^2}\right](r')
    \right|
    \mathrm{d}r' .
\]
Since $\gamma_w\ge 1$, Lemma~\ref{conv_gauss} gives
\[
    \NN\!\left[1-t(w),\frac{\gamma_w^2}{10z^2}\right]
    =
    \NN\!\left[0,\frac{\gamma_w^2-1}{10z^2}\right]
    *
    \NN\!\left[1-t(w),\frac{1}{10z^2}\right],
\]
and
\[
    \NN\!\left[0,\frac{\gamma_w^2}{10z^2}\right]
    =
    \NN\!\left[0,\frac{\gamma_w^2-1}{10z^2}\right]
    *
    \NN\!\left[0,\frac{1}{10z^2}\right].
\]
Convolution with a fixed probability density cannot increase the $L^1$ norm: indeed, for every probability density $K$ and every integrable signed density $h$,
\[
    \|K*h\|_1
    \le
    \int\!\!\int K(r-r')|h(r')|\,\mathrm{d}r'\mathrm{d}r
    =
    \|h\|_1 .
\]
Applying this inequality with $K=\NN\!\left[0,\frac{\gamma_w^2-1}{10z^2}\right]$ and
\[
    h
    =
    W_1(w)\NN\!\left[1-t(w),\frac{1}{10z^2}\right]
    -
    W_0(w)\NN\!\left[0,\frac{1}{10z^2}\right],
\]
we obtain
\[
    \sum_{l\in\mathrm{leaves}(w)}
    |W_1(l)-W_0(l)|
    \le
    \int_{\mathbb{R}}
    \left|
    W_1(w)
    \NN\!\left[1-t(w),\frac{1}{10z^2}\right](r')
    -
    W_0(w)
    \NN\!\left[0,\frac{1}{10z^2}\right](r')
    \right|
    \mathrm{d}r'.
\]
This proves the induction step, and therefore proves \eqref{eq:unit_time_induction_claim}.

Applying \eqref{eq:unit_time_induction_claim} to the root $r$ of $\mathcal{U}_u$, and using $t(r)=0$ and $W_1(r)=W_0(r)=1$, we get
\[
    \left\|P_1^{(\mathcal{U}_u)}-P_0^{(\mathcal{U}_u)}\right\|_1
    =
    \sum_{l\in\mathrm{leaves}(\mathcal{U}_u)}
    |W_1(l)-W_0(l)|
    \le
    \left\|
    \NN\!\left[1,\frac{1}{10z^2}\right]
    -
    \NN\!\left[0,\frac{1}{10z^2}\right]
    \right\|_1,
\]
which proves \eqref{eq:unit_time_tree_tv_bound}.

It remains to translate this into a bound on the local factor $f^{(u)}$. Since $\tau(u)=1$, the definition \eqref{eq:def_f_u_parity} gives
\[
    f^{(u)}(\mathbf{y})
    =
    \frac14\left(p_0^{(u)}(\mathbf{y})+p_1^{(u)}(\mathbf{y})\right)
    +
    \frac14\left|p_1^{(u)}(\mathbf{y})-p_0^{(u)}(\mathbf{y})\right|.
\]
Here $p_0^{(u)}=P_0^{(\mathcal{U}_u)}$ and $p_1^{(u)}=P_1^{(\mathcal{U}_u)}$ are probability densities over the leaves of the attached unit-time tree. Hence
\bb
    \left\|f^{(u)}\right\|_1
    &=
    \frac14\left\|p_0^{(u)}+p_1^{(u)}\right\|_1
    +
    \frac14\left\|p_1^{(u)}-p_0^{(u)}\right\|_1\\
    &=
    \frac12
    +
    \frac14
    \left\|P_1^{(\mathcal{U}_u)}-P_0^{(\mathcal{U}_u)}\right\|_1\\
    &\le
    \frac12
    +
    \frac14
    \left\|
    \NN\!\left[1,\frac{1}{10z^2}\right]
    -
    \NN\!\left[0,\frac{1}{10z^2}\right]
    \right\|_1 .
\ee
We now bound the last expression using Lemma~\ref{tvd_one_g}. Apply that lemma with $m_1=1$, $m_2=0$, and $V=1/(10z^2)$. Then
\[
    t=\frac{|m_1-m_2|}{2\sqrt{V}}=\frac{\sqrt{10}\,z}{2},
    \qquad
    t^2=\frac{5z^2}{2}.
\]
By \eqref{eq_2_l},
\bb
    \frac12
    \left\|
    \NN\!\left[1,\frac{1}{10z^2}\right]
    -
    \NN\!\left[0,\frac{1}{10z^2}\right]
    \right\|_1
    &\le
    1-
    \left(
    \frac16 e^{-5z^2/2}
    +
    \sqrt{\frac{2}{\pi}}\,
    \frac{e^{-5z^2/4}}{\frac{\sqrt{10}}{2}z+1}
    \right) \le
    1-\frac16 e^{-5z^2/2}.
\ee
Therefore,
\bb
    \frac12
    +
    \frac14
    \left\|
    \NN\!\left[1,\frac{1}{10z^2}\right]
    -
    \NN\!\left[0,\frac{1}{10z^2}\right]
    \right\|_1
    &=
    \frac12
    +
    \frac12\cdot
    \frac12
    \left\|
    \NN\!\left[1,\frac{1}{10z^2}\right]
    -
    \NN\!\left[0,\frac{1}{10z^2}\right]
    \right\|_1\\
    &\le
    1-\frac{1}{12}e^{-5z^2/2}\\
    &\le
    \exp\!\left(-\frac{1}{12}e^{-5z^2/2}\right),
\ee
where the last inequality uses $1-x\le e^{-x}$ for $x\ge0$. Combining this estimate with the previous bound on $\|f^{(u)}\|_1$ proves \eqref{eq:unit_time_fu_bound}.
\end{proof}

\subsubsection{Analysis for $\tau(u)\ge2$}

We now consider the second case, namely a node $u$ of the $T$-time tree with $\tau(u)\ge2$. In this case the transition is directly described by the Gaussian law in \eqref{eq_gausscond2}. Set $t\coloneqq \tau(u)$ and, for every local string $\xi\in\{0,1\}^t$, write
\[
    p_{\xi}^{(u)}(\mathbf{y})
    =
    \NN\!\left[
    \begin{pmatrix}
        a^{(u)}\!\cdot \xi\\
        b^{(u)}\!\cdot \xi
    \end{pmatrix},
    \frac{\mathbb{1}_2}{10z^2}
    \right](\mathbf{y}),
    \qquad \mathbf{y}\in\mathbb{R}^2 .
\]
We use the same Boolean-Fourier notation introduced in the proof of Lemma~\ref{lemma_adapt_prob}: for every $\sigma\in\{0,1\}^t$, define
\[
    A_{\sigma}^{(u)}(\mathbf{y})
    \coloneqq
    \mathop{\mathbb{E}}_{\xi\in\{0,1\}^t}
    (-1)^{\xi\cdot\sigma}
    p_{\xi}^{(u)}(\mathbf{y}).
\]
By \eqref{eq:f_as_average_fourier}, the local factor appearing in Lemma~\ref{lemma_adapt_prob} satisfies
\[
    f^{(u)}(\mathbf{y})
    =
    \mathop{\mathbb{E}}_{\sigma\in\{0,1\}^t}
    \left|A_{\sigma}^{(u)}(\mathbf{y})\right|.
\]
Consequently,
\bb\label{eq_to_use_1}
    \left\|f^{(u)}\right\|_1
    =
    \frac{1}{2^t}
    +
    \left(1-\frac{1}{2^t}\right)
    \mathop{\mathbb{E}}_{\substack{\sigma\in\{0,1\}^t\\ \sigma\ne0}}
    \left\|A_{\sigma}^{(u)}\right\|_1,
\ee
because $A_0^{(u)}=\mathop{\mathbb{E}}_{\xi}p_\xi^{(u)}$ is a probability density. Hence controlling $\|f^{(u)}\|_1$ reduces to controlling the average $L^1$ norm of the non-zero coefficients $A_\sigma^{(u)}$. We begin with an $L^2$ estimate.

\begin{lemma}\label{lemma_l2_A_tau_ge_2}
Let $u$ be a node of the $T$-time tree with $t=\tau(u)\ge2$. Then
\[
    \mathop{\mathbb{E}}_{\substack{\sigma\in\{0,1\}^t\\ \sigma\ne0}}
    \left\|A_{\sigma}^{(u)}\right\|_2
    \le
    \sqrt{\frac{5}{\pi}}\,
    \frac{z}{2^{t/2}} .
\]
\end{lemma}

\begin{proof}
We use the Fourier transform convention
\[
    \widehat h(\boldsymbol{\omega})
    \coloneqq
    \frac{1}{2\pi}
    \int_{\mathbb{R}^2}
    h(\mathbf{y})e^{-i\boldsymbol{\omega}\cdot\mathbf{y}}\,
    \mathrm{d}^2\mathbf{y},
    \qquad
    \boldsymbol{\omega}=(\omega_1,\omega_2),
\]
so that Plancherel's identity reads $\|h\|_2=\|\widehat h\|_2$. Set $v_z\coloneqq 1/(10z^2)$. For every $\xi\in\{0,1\}^t$, the density $p_\xi^{(u)}$ is the two-dimensional Gaussian
\[
    p_\xi^{(u)}(\mathbf{y})
    =
    \NN\!\left[
    \begin{pmatrix}
        a^{(u)}\!\cdot \xi\\
        b^{(u)}\!\cdot \xi
    \end{pmatrix},
    v_z\mathbb{1}_2
    \right](\mathbf{y}),
\]
and therefore
\[
    \widehat{p_\xi^{(u)}}(\boldsymbol{\omega})
    =
    \frac{1}{2\pi}
    \exp\!\left(
        -i\omega_1 a^{(u)}\!\cdot\xi
        -i\omega_2 b^{(u)}\!\cdot\xi
    \right)
    \exp\!\left(
        -\frac{v_z}{2}\|\boldsymbol{\omega}\|_2^2
    \right).
\]
By Jensen's inequality and Plancherel's identity,
\begin{equation}\label{eq_proof_cr}
\begin{aligned}
    \left(
    \mathop{\mathbb{E}}_{\substack{\sigma\in\{0,1\}^t\\ \sigma\ne0}}
    \left\|A_\sigma^{(u)}\right\|_2
    \right)^2
    &\le
    \mathop{\mathbb{E}}_{\substack{\sigma\in\{0,1\}^t\\ \sigma\ne0}}
    \left\|A_\sigma^{(u)}\right\|_2^2 =
    \int_{\mathbb{R}^2}
    \mathop{\mathbb{E}}_{\substack{\sigma\in\{0,1\}^t\\ \sigma\ne0}}
    \left|
    \widehat{A_\sigma^{(u)}}(\boldsymbol{\omega})
    \right|^2
    \mathrm{d}^2\boldsymbol{\omega}.
\end{aligned}
\end{equation}
By linearity of the Fourier transform,
\[
    \widehat{A_\sigma^{(u)}}(\boldsymbol{\omega})
    =
    \mathop{\mathbb{E}}_{\xi\in\{0,1\}^t}
    (-1)^{\xi\cdot\sigma}
    \widehat{p_\xi^{(u)}}(\boldsymbol{\omega}).
\]
Thus, for fixed $\boldsymbol{\omega}$, adding the missing term $\sigma=0$ gives
\[
\begin{aligned}
    \mathop{\mathbb{E}}_{\substack{\sigma\in\{0,1\}^t\\ \sigma\ne0}}
    \left|
    \widehat{A_\sigma^{(u)}}(\boldsymbol{\omega})
    \right|^2
    &\le
    \frac{1}{2^t-1}
    \sum_{\sigma\in\{0,1\}^t}
    \left|
    \mathop{\mathbb{E}}_{\xi\in\{0,1\}^t}
    (-1)^{\xi\cdot\sigma}
    \widehat{p_\xi^{(u)}}(\boldsymbol{\omega})
    \right|^2 .
\end{aligned}
\]
Expanding the square,
\[
\begin{aligned}
    &\sum_{\sigma\in\{0,1\}^t}
    \left|
    \mathop{\mathbb{E}}_{\xi\in\{0,1\}^t}
    (-1)^{\xi\cdot\sigma}
    \widehat{p_\xi^{(u)}}(\boldsymbol{\omega})
    \right|^2 =
    \frac{1}{2^{2t}}
    \sum_{\xi,\xi'\in\{0,1\}^t}
    \widehat{p_\xi^{(u)}}(\boldsymbol{\omega})
    \overline{\widehat{p_{\xi'}^{(u)}}(\boldsymbol{\omega})}
    \sum_{\sigma\in\{0,1\}^t}
    (-1)^{(\xi+\xi')\cdot\sigma}.
\end{aligned}
\] 
The inner sum is $2^t$ if $\xi=\xi'$ and $0$ otherwise. Hence only the diagonal terms survive, and
\[
\begin{aligned}
    \sum_{\sigma\in\{0,1\}^t}
    \left|
    \mathop{\mathbb{E}}_{\xi\in\{0,1\}^t}
    (-1)^{\xi\cdot\sigma}
    \widehat{p_\xi^{(u)}}(\boldsymbol{\omega})
    \right|^2
    &=
    \frac{1}{2^t}
    \sum_{\xi\in\{0,1\}^t}
    \left|
    \widehat{p_\xi^{(u)}}(\boldsymbol{\omega})
    \right|^2  =
    \mathop{\mathbb{E}}_{\xi\in\{0,1\}^t}
    \left|
    \widehat{p_\xi^{(u)}}(\boldsymbol{\omega})
    \right|^2 .
\end{aligned}
\]
Using the explicit Fourier transform of $p_\xi^{(u)}$, whose modulus is independent of $\xi$, we obtain
\[
\begin{aligned}
    \mathop{\mathbb{E}}_{\substack{\sigma\in\{0,1\}^t\\ \sigma\ne0}}
    \left|
    \widehat{A_\sigma^{(u)}}(\boldsymbol{\omega})
    \right|^2
    &\le
    \frac{1}{2^t-1}
    \mathop{\mathbb{E}}_{\xi\in\{0,1\}^t}
    \left|
    \widehat{p_\xi^{(u)}}(\boldsymbol{\omega})
    \right|^2  =
    \frac{1}{2^t-1}\,
    \frac{1}{4\pi^2}
    e^{-v_z\|\boldsymbol{\omega}\|_2^2}.
\end{aligned}
\]
Substituting this into \eqref{eq_proof_cr} yields
\[
\begin{aligned}
    \left(
    \mathop{\mathbb{E}}_{\substack{\sigma\in\{0,1\}^t\\ \sigma\ne0}}
    \left\|A_\sigma^{(u)}\right\|_2
    \right)^2
    &\le
    \frac{1}{2^t-1}
    \frac{1}{4\pi^2}
    \int_{\mathbb{R}^2}
    e^{-v_z\|\boldsymbol{\omega}\|_2^2}
    \mathrm{d}^2\boldsymbol{\omega}  =
    \frac{1}{2^t-1}
    \frac{1}{4\pi^2}
    \frac{\pi}{v_z}
    =
    \frac{5z^2}{2\pi(2^t-1)}.
\end{aligned}
\]
Since $2^t-1\ge 2^{t-1}$, this gives
\[
    \left(
    \mathop{\mathbb{E}}_{\substack{\sigma\in\{0,1\}^t\\ \sigma\ne0}}
    \left\|A_\sigma^{(u)}\right\|_2
    \right)^2
    \le
    \frac{5z^2}{\pi\,2^t}.
\]
Taking the square root concludes the proof.
\end{proof}

We now convert the $L^2$ estimate into an $L^1$ estimate. The only additional ingredient is a Gaussian tail bound, which allows us to restrict the integral to a square of controlled area.

\begin{lemma}\label{lemma_l1_A_tau_ge_2}
Let $u$ be a node of the $T$-time tree with $t=\tau(u)\ge2$. Then
\[
    \mathop{\mathbb{E}}_{\substack{\sigma\in\{0,1\}^t\\ \sigma\ne0}}
    \left\|A_{\sigma}^{(u)}\right\|_1
    \le
    \frac{4\sqrt{5}\,tz}{\sqrt{\pi}\,2^{t/2}}
    +
    2e^{-t^2/2}.
\]
\end{lemma}

\begin{proof}
Let
\[
    \sigma_z\coloneqq \frac{1}{\sqrt{10}z},
    \qquad
    R\coloneqq t+\sigma_z t,
    \qquad
    B_R\coloneqq\{\mathbf{y}\in\mathbb{R}^2:\|\mathbf{y}\|_\infty\le R\}.
\]
For every $\xi\in\{0,1\}^t$, the density $p_\xi^{(u)}$ has covariance matrix $\sigma_z^2\mathbb{1}_2$ and mean $\mu_\xi
    =
    \begin{pmatrix}
        a^{(u)}\!\cdot\xi\\
        b^{(u)}\!\cdot\xi
    \end{pmatrix}$. Since $a^{(u)},b^{(u)}\in[-1,1]^t$ and $\xi\in\{0,1\}^t$, both coordinates of $\mu_\xi$ have absolute value at most $t$. Thus, if $\mathbf{Y}=(Y_1,Y_2)\sim p_\xi^{(u)}$, then each coordinate can be written as $Y_j=\mu_j+\sigma_z Z_j$, with $Z_j\sim\NN[0,1]$ and $|\mu_j|\le t$.

The choice $R=t+\sigma_z t$ ensures that $[\mu_j-\sigma_z t,\mu_j+\sigma_z t]\subseteq[-R,R]$ for every $|\mu_j|\le t$. Hence the event $\{|Y_j|>R\}$ implies $|Y_j-\mu_j|>\sigma_z t$, and therefore
\[
    \Pr(|Y_j|>R)
    \le
    \Pr(|Y_j-\mu_j|>\sigma_z t)
    =
    2\int_t^\infty \NN[0,1](r)\,\mathrm{d}r .
\]
By the union bound and the Gaussian tail bound in Lemma~\ref{lemma_bounds_Q},
\begin{equation}\label{eq_useful}
\begin{aligned}
    \Pr(\mathbf{Y}\notin B_R)
    &\le
    \Pr(|Y_1|>R)+\Pr(|Y_2|>R)
    \le
    4\int_t^\infty \NN[0,1](r)\,\mathrm{d}r
    \le
    2e^{-t^2/2}.
\end{aligned}
\end{equation}

Now fix $\sigma\ne0$. Since $A_\sigma^{(u)}=\mathop{\mathbb{E}}_{\xi}(-1)^{\xi\cdot\sigma}p_\xi^{(u)}$, we have
\[
\begin{aligned}
    \int_{B_R^c}
    \left|A_\sigma^{(u)}(\mathbf{y})\right|
    \mathrm{d}^2\mathbf{y}
    &\le
    \mathop{\mathbb{E}}_{\xi\in\{0,1\}^t}
    \int_{B_R^c}
    p_\xi^{(u)}(\mathbf{y})
    \mathrm{d}^2\mathbf{y}
    \le
    2e^{-t^2/2},
\end{aligned}
\]
where the last inequality uses \eqref{eq_useful}. On the square $B_R$, Cauchy--Schwarz gives
\[
    \int_{B_R}
    \left|A_\sigma^{(u)}(\mathbf{y})\right|
    \mathrm{d}^2\mathbf{y}
    \le
    \sqrt{\operatorname{Vol}(B_R)}
    \left\|A_\sigma^{(u)}\right\|_2
    =
    2R\left\|A_\sigma^{(u)}\right\|_2.
\]
Combining the two estimates and averaging over $\sigma\ne0$ yields
\[
\begin{aligned}
    \mathop{\mathbb{E}}_{\substack{\sigma\in\{0,1\}^t\\ \sigma\ne0}}
    \left\|A_\sigma^{(u)}\right\|_1
    &\le
    2R
    \mathop{\mathbb{E}}_{\substack{\sigma\in\{0,1\}^t\\ \sigma\ne0}}
    \left\|A_\sigma^{(u)}\right\|_2
    +
    2e^{-t^2/2}  =
    2t\left(1+\frac{1}{\sqrt{10}z}\right)
    \mathop{\mathbb{E}}_{\substack{\sigma\in\{0,1\}^t\\ \sigma\ne0}}
    \left\|A_\sigma^{(u)}\right\|_2
    +
    2e^{-t^2/2}.
\end{aligned}
\]
Using Lemma~\ref{lemma_l2_A_tau_ge_2} and $1+\frac{1}{\sqrt{10}z}\le2$, which follows from $z\ge1$, we obtain
\[
\begin{aligned}
    \mathop{\mathbb{E}}_{\substack{\sigma\in\{0,1\}^t\\ \sigma\ne0}}
    \left\|A_\sigma^{(u)}\right\|_1
    &\le
    2t\left(1+\frac{1}{\sqrt{10}z}\right)
    \sqrt{\frac{5}{\pi}}\,
    \frac{z}{2^{t/2}}
    +
    2e^{-t^2/2} \le
    \frac{4\sqrt{5}\,tz}{\sqrt{\pi}\,2^{t/2}}
    +
    2e^{-t^2/2}.
\end{aligned}
\]
This concludes the proof.
\end{proof}

The bound obtained above is useful when $t$ is moderately large. When $t$ is small, however, the term $\frac{4\sqrt{5}\,tz}{\sqrt{\pi}\,2^{t/2}}$ may exceed one, making the estimate essentially trivial. In this regime we use a more direct argument, given in the following lemma.

\begin{lemma}\label{lemma_small_t_A_tau_ge_2}
Let $u$ be a node of the $T$-time tree with $t=\tau(u)\ge2$. Then
\[
    \mathop{\mathbb{E}}_{\substack{\sigma\in\{0,1\}^t\\ \sigma\ne0}}
    \left\|A_\sigma^{(u)}\right\|_1
    \le
    \exp\!\left(-\frac13 e^{-5z^2/2}\right).
\]
\end{lemma}

\begin{proof}
Fix a non-zero string $\sigma\in\{0,1\}^t$. Since $\sigma\ne0$, there is at least one coordinate at which $\sigma$ is equal to $1$. Let
\[
    j(\sigma)\coloneqq \min\{j\in[t]:\sigma_j=1\},
\]
and write simply $j=j(\sigma)$. Let $e_j\in\{0,1\}^t$ be the $j$th standard basis vector. We split the sum defining $A_\sigma^{(u)}$ according to the value of the $j$th bit. Every string with $j$th bit equal to $1$ can be written uniquely as $\eta+e_j$, where $\eta_j=0$ and the addition is modulo $2$. Hence
\[
\begin{aligned}
    A_\sigma^{(u)}
    &=
    \frac{1}{2^t}
    \sum_{\xi\in\{0,1\}^t}
    (-1)^{\xi\cdot\sigma}p_\xi^{(u)}  =
    \frac{1}{2^t}
    \sum_{\substack{\eta\in\{0,1\}^t\\ \eta_j=0}}
    \left(
    (-1)^{\eta\cdot\sigma}p_\eta^{(u)}
    +
    (-1)^{(\eta+e_j)\cdot\sigma}p_{\eta+e_j}^{(u)}
    \right).
\end{aligned}
\]
Since $\sigma_j=1$, we obtain
\[
    A_\sigma^{(u)}
    =
    \frac{1}{2^t}
    \sum_{\substack{\eta\in\{0,1\}^t\\ \eta_j=0}}
    (-1)^{\eta\cdot\sigma}
    \left(
        p_\eta^{(u)}-p_{\eta+e_j}^{(u)}
    \right).
\]
Taking the $L^1$ norm and using the triangle inequality gives
\bb\label{eq_12345}
    \left\|A_\sigma^{(u)}\right\|_1
    \le
    \frac{1}{2^t}
    \sum_{\substack{\eta\in\{0,1\}^t\\ \eta_j=0}}
    \left\|p_\eta^{(u)}-p_{\eta+e_j}^{(u)}\right\|_1.
\ee
For each such $\eta$, the two densities $p_\eta^{(u)}$ and $p_{\eta+e_j}^{(u)}$ are Gaussian densities with the same covariance matrix $\mathbb{1}_2/(10z^2)$. Their means differ by $\begin{pmatrix}
        a_j^{(u)}\\
        b_j^{(u)}
    \end{pmatrix}$. Thus, by translation invariance of the $L^1$ distance,
\[
    \left\|p_\eta^{(u)}-p_{\eta+e_j}^{(u)}\right\|_1
    =
    \left\|
    \NN\!\left[
    \begin{pmatrix}
        a_j^{(u)}\\
        b_j^{(u)}
    \end{pmatrix},
    \frac{\mathbb{1}_2}{10z^2}
    \right]
    -
    \NN\!\left[
    \begin{pmatrix}
        0\\
        0
    \end{pmatrix},
    \frac{\mathbb{1}_2}{10z^2}
    \right]
    \right\|_1 .
\]
Since the covariance matrix is proportional to the identity, we may rotate the plane and integrate out the orthogonal coordinate. Setting $d_j\coloneqq \sqrt{(a_j^{(u)})^2+(b_j^{(u)})^2}$, we obtain
\[
    \left\|p_\eta^{(u)}-p_{\eta+e_j}^{(u)}\right\|_1
    =
    \left\|
    \NN\!\left[d_j,\frac{1}{10z^2}\right]
    -
    \NN\!\left[0,\frac{1}{10z^2}\right]
    \right\|_1.
\]
By the admissibility condition on the rotation coefficients, $d_j\le1$. Moreover, Lemma~\ref{tvd_one_g} shows that the total variation distance between one-dimensional Gaussians with the same variance is increasing in the distance between their means. Hence
\[
    \left\|p_\eta^{(u)}-p_{\eta+e_j}^{(u)}\right\|_1
    \le
    \left\|
    \NN\!\left[1,\frac{1}{10z^2}\right]
    -
    \NN\!\left[0,\frac{1}{10z^2}\right]
    \right\|_1.
\]
There are exactly $2^{t-1}$ strings $\eta$ with $\eta_j=0$, so \eqref{eq_12345} implies
\[
    \left\|A_\sigma^{(u)}\right\|_1
    \le
    \frac12
    \left\|
    \NN\!\left[1,\frac{1}{10z^2}\right]
    -
    \NN\!\left[0,\frac{1}{10z^2}\right]
    \right\|_1.
\]
We now apply \eqref{eq_2_l} of Lemma~\ref{tvd_one_g} with $m_1=1$, $m_2=0$, and $V=1/(10z^2)$. In this case
\[
    \alpha
    \coloneqq
    \frac{|m_1-m_2|}{2\sqrt{V}}
    =
    \frac{\sqrt{10}\,z}{2},
    \qquad
    \alpha^2=\frac{5z^2}{2}.
\]
Therefore
\[
\begin{aligned}
    \left\|A_\sigma^{(u)}\right\|_1
    &\le
    1-
    \left(
    \frac16 e^{-5z^2/2}
    +
    \sqrt{\frac{2}{\pi}}\,
    \frac{e^{-5z^2/4}}{\frac{\sqrt{10}}{2}z+1}
    \right) \le
    1-\frac13 e^{-5z^2/2}
    \le
    \exp\!\left(-\frac13 e^{-5z^2/2}\right).
\end{aligned}
\]
In the second inequality we used that, for every $z\ge1$,
\[
    \sqrt{\frac{2}{\pi}}\,
    \frac{e^{-5z^2/4}}{\frac{\sqrt{10}}{2}z+1}
    \ge
    \frac16 e^{-5z^2/2},
\]
and in the last inequality we used $1-x\le e^{-x}$. The bound is uniform over all non-zero strings $\sigma$, and averaging over $\sigma\ne0$ gives the claim.
\end{proof}

Combining Lemma~\ref{lemma_l1_A_tau_ge_2} and Lemma~\ref{lemma_small_t_A_tau_ge_2}, we obtain a non-trivial bound on the local factor $\|f^{(u)}\|_1$ for all values of the waiting time \(\tau(u)\ge2\).

\begin{lemma}\label{lemma_tau_ge_2_fu_bound}
Let \(u\) be a node of the \(T\)-time tree with \(\tau(u)\ge2\). Then
\[
    \left\|f^{(u)}\right\|_1
    \le
    \exp\!\left(-\frac{e^{-3z^2}}{90}\,\tau(u)\right).
\]
\end{lemma}

\begin{proof}
Set \(t\coloneqq\tau(u)\). From \eqref{eq_to_use_1}, we have
\[
    \left\|f^{(u)}\right\|_1
    =
    \frac{1}{2^t}
    +
    \left(1-\frac{1}{2^t}\right)
    \mathop{\mathbb{E}}_{\substack{\sigma\in\{0,1\}^t\\ \sigma\ne0}}
    \left\|A_\sigma^{(u)}\right\|_1 .
\]
Combining Lemmas~\ref{lemma_l1_A_tau_ge_2} and~\ref{lemma_small_t_A_tau_ge_2}, we obtain
\[
\begin{aligned}
    \left\|f^{(u)}\right\|_1
    &\le
    \frac{1}{2^t}
    +
    \left(1-\frac{1}{2^t}\right)
    \min\left\{
    \frac{4\sqrt{5}\,tz}{\sqrt{\pi}\,2^{t/2}}
    +
    2e^{-t^2/2},
    \exp\!\left(-\frac13 e^{-5z^2/2}\right)
    \right\}.
\end{aligned}
\]
We now show that the right-hand side is at most $\exp\!\left(-\frac{e^{-3z^2}}{90}t\right)$.

First assume that \(t\ge 12\log z+28\). Using the first term in the minimum, together with
\[
    2e^{-t^2/2}\le 2^{-t},
    \qquad
    \frac{4\sqrt{5}\,t}{\sqrt{\pi}\,2^{t/2}}\le e^{-t/6}
    \qquad\forall\,t\ge28,
\]
we get
\[
\begin{aligned}
    \left\|f^{(u)}\right\|_1
    &\le
    \frac{1}{2^t}
    +
    \frac{4\sqrt{5}\,tz}{\sqrt{\pi}\,2^{t/2}}
    +
    2e^{-t^2/2}
    \le
    \frac{2}{2^t}+ze^{-t/6}.
\end{aligned}
\]
Since \(t\ge12\log z\), we have \(ze^{-t/6}\le e^{-t/12}\). Hence
\[
    \left\|f^{(u)}\right\|_1
    \le
    \frac{2}{2^t}+e^{-t/12}
    \le
    e^{-t/90}.
\]
Since \(e^{-3z^2}\le1\), it follows that
\[
    \left\|f^{(u)}\right\|_1
    \le
    \exp\!\left(-\frac{e^{-3z^2}}{90}t\right).
\]
This proves the claim in the first regime.

Now assume that \(t\le 12\log z+28\). Using the second term in the minimum and the fact that \(t\ge2\), so that \(2^{-t}\le1/4\), we obtain
\[
\begin{aligned}
    \left\|f^{(u)}\right\|_1
    &\le
    \frac{1}{2^t}
    +
    \left(1-\frac{1}{2^t}\right)
    \exp\!\left(-\frac13 e^{-5z^2/2}\right) \le
    \frac14
    +
    \frac34
    \exp\!\left(-\frac13 e^{-5z^2/2}\right).
\end{aligned}
\]
Set \(x\coloneqq e^{-5z^2/2}\). Since \(0\le x\le1\), the inequality $\frac14+\frac34e^{-x/3}\le e^{-x/5}$ implies
\[
    \left\|f^{(u)}\right\|_1
    \le
    \exp\!\left(-\frac15 e^{-5z^2/2}\right).
\]
It remains to compare the exponent with the desired one. We claim that, for every \(z\ge1\),
\[
    \frac15 e^{-5z^2/2}
    \ge
    \frac{e^{-3z^2}}{90}(12\log z+28).
\]
Indeed, this is equivalent to $(12\log z+28)e^{-z^2/2}\le18$. The function \(g(z)\coloneqq(12\log z+28)e^{-z^2/2}\) is decreasing on \([1,\infty)\), since
\[
    g'(z)
    =
    e^{-z^2/2}
    \left(
        \frac{12}{z}
        -
        z(12\log z+28)
    \right)
    \le0
    \qquad \forall\,z\ge1.
\]
Therefore \(g(z)\le g(1)=28e^{-1/2}<18\), proving the claim. Since \(t\le12\log z+28\), we conclude that
\[
    \frac15 e^{-5z^2/2}
    \ge
    \frac{e^{-3z^2}}{90}(12\log z+28)
    \ge
    \frac{e^{-3z^2}}{90}t.
\]
Hence
\[
    \left\|f^{(u)}\right\|_1
    \le
    \exp\!\left(-\frac{e^{-3z^2}}{90}t\right).
\]
This proves the second regime and concludes the proof.
\end{proof}

\subsubsection{Analysis of the $T$-time tree}

We are now ready to propagate the local estimates through the whole $T$-time tree. The key point is that both local regimes, $\tau(u)=1$ and $\tau(u)\ge2$, can be written in a unified form.

\begin{lemma}\label{lemma_uniform_fu_bound_T_tree}
Let $u$ be an internal node of the $T$-time tree. Then
\[
    \left\|f^{(u)}\right\|_1
    \le
    \exp\!\left(-\frac{e^{-3z^2}}{90}\,\tau(u)\right).
\]
\end{lemma}

\begin{proof}
If $\tau(u)\ge2$, the claim is exactly Lemma~\ref{lemma_tau_ge_2_fu_bound}. It remains to consider the case $\tau(u)=1$. By Lemma~\ref{lemma_unit_time_tree_contribution},
\[
    \left\|f^{(u)}\right\|_1
    \le
    \exp\!\left(-\frac{1}{12}e^{-5z^2/2}\right).
\]
Since $\frac{1}{12}e^{-5z^2/2}
    \ge
    \frac{1}{90}e^{-3z^2}$ for all $z\ge1$, we obtain
\[
    \left\|f^{(u)}\right\|_1
    \le
    \exp\!\left(-\frac{e^{-3z^2}}{90}\right)
    =
    \exp\!\left(-\frac{e^{-3z^2}}{90}\tau(u)\right),
\]
because here $\tau(u)=1$. This proves the claim.
\end{proof}

We can now bound the average total variation distance between the two leaf distributions of an arbitrary $T$-time tree.

\begin{lemma}\label{lemma_T_time_tree_parity_bound}
Let $\mathcal{T}$ be any $T$-time tree. Then
\[
    \mathop{\mathbb{E}}_{\substack{s\in\{0,1\}^{T}\\ s\ne 0}}
    \frac12
    \left\|P_s^{(\mathcal{T})}-P_0^{(\mathcal{T})}\right\|_1
    \le
    \exp\!\left(-\frac{e^{-3z^2}}{90}T\right).
\]
\end{lemma}

\begin{proof}
Let $r$ be the root of $\mathcal{T}$. For a node $v$, define the prefix weight
\[
    \Phi(v)
    \coloneqq
    \prod_{i=1}^{N(v)}
    f^{(u_{i-1}^{(v)})}\!\left(\mathbf{y}_i^{(v)}\right),
\]
with the convention $\Phi(r)=1$. Thus, for a leaf $l$, the quantity $\Phi(l)$ is exactly the product of the local factors along the branch from the root to $l$.

By Lemma~\ref{lemma_adapt_prob}, for every leaf $l$,
\[
    \mathop{\mathbb{E}}_{\substack{s\in\{0,1\}^{T}\\ s\ne 0}}
    \left|P_s(l)-P_0(l)\right|
    \le
    \frac{1}{1-2^{-T}}\Phi(l).
\]
Therefore,
\[
\begin{aligned}
    \mathop{\mathbb{E}}_{\substack{s\in\{0,1\}^{T}\\ s\ne 0}}
    \frac12
    \left\|P_s^{(\mathcal{T})}-P_0^{(\mathcal{T})}\right\|_1
    &=
    \frac12
    \sum_{l\in\mathrm{leaves}(r)}
    \mathop{\mathbb{E}}_{\substack{s\in\{0,1\}^{T}\\ s\ne 0}}
    \left|P_s(l)-P_0(l)\right| \le
    \frac{1}{2(1-2^{-T})}
    \sum_{l\in\mathrm{leaves}(r)}\Phi(l).
\end{aligned}
\]
Since $T\ge1$, we have $\frac{1}{2(1-2^{-T})}\le1$. Hence it remains to prove that
\[
    \sum_{l\in\mathrm{leaves}(r)}\Phi(l)
    \le
    \exp\!\left(-\frac{e^{-3z^2}}{90}T\right).
\]
We prove the following stronger statement: for every node $v$ of the $T$-time tree,
\[
    \sum_{l\in\mathrm{leaves}(v)}\Phi(l)
    \le
    \Phi(v)\,
    \exp\!\left(-\frac{e^{-3z^2}}{90}(T-t(v))\right).
\]
The proof is by backward induction on the remaining time $T-t(v)$. If $v$ is a leaf, then $t(v)=T$ and $\mathrm{leaves}(v)=\{v\}$, so the claim holds with equality.

Now let $v$ be an internal node, and assume the claim holds for all children $w$ of $v$. Since every such child satisfies $t(w)=t(v)+\tau(v)$ and
\[
    \Phi(w)=\Phi(v)\,f^{(v)}(\mathbf{y}_w),
\]
we obtain
\[
\begin{aligned}
    \sum_{l\in\mathrm{leaves}(v)}\Phi(l)
    &=
    \sum_{w\in\mathrm{children}(v)}
    \sum_{l\in\mathrm{leaves}(w)}\Phi(l) \\
    &\le
    \sum_{w\in\mathrm{children}(v)}
    \Phi(w)
    \exp\!\left(-\frac{e^{-3z^2}}{90}(T-t(w))\right) \\
    &=
    \Phi(v)
    \exp\!\left(-\frac{e^{-3z^2}}{90}(T-t(v)-\tau(v))\right)
    \sum_{w\in\mathrm{children}(v)}
    f^{(v)}(\mathbf{y}_w) \\
    &=
    \Phi(v)
    \exp\!\left(-\frac{e^{-3z^2}}{90}(T-t(v)-\tau(v))\right)
    \left\|f^{(v)}\right\|_1 \\
    &\le
    \Phi(v)
    \exp\!\left(-\frac{e^{-3z^2}}{90}(T-t(v)-\tau(v))\right)
    \exp\!\left(-\frac{e^{-3z^2}}{90}\tau(v)\right) \\
    &=
    \Phi(v)
    \exp\!\left(-\frac{e^{-3z^2}}{90}(T-t(v))\right),
\end{aligned}
\]
where the last inequality uses Lemma~\ref{lemma_uniform_fu_bound_T_tree}. This proves the induction claim.

Applying it to the root $r$, and using $t(r)=0$ and $\Phi(r)=1$, gives
\[
    \sum_{l\in\mathrm{leaves}(r)}\Phi(l)
    \le
    \exp\!\left(-\frac{e^{-3z^2}}{90}T\right).
\]
Substituting this into the previous total-variation estimate concludes the proof.
\end{proof}

Combining Lemma~\ref{lemma_T_time_tree_parity_bound} with Lemma~\ref{lemma_sensing_time_from_T_tree_bound} gives the desired lower bound in the simplified access model.

\begin{lemma}\label{lemma_lower_bound_learning_parity_simplified}
Let $T\in\mathbb{N}^+$ and let $z\ge1$. Suppose that a protocol in the simplified access model solves Problem~\ref{def:PROB_revealed} with success probability at least \(2/3\) using total sensing time $nT$. Then
\[
    n
    \ge
    \frac13
    \exp\!\left(\frac{e^{-3z^2}}{90}T\right).
\]
\end{lemma}

\begin{proof}
By Lemma~\ref{lemma_T_time_tree_parity_bound}, for every $T$-time tree $\mathcal{T}$ we have
\[
    \mathop{\mathbb{E}}_{\substack{s\in\{0,1\}^T\\ s\ne0}}
    \frac12
    \left\|P_s^{(\mathcal{T})}-P_0^{(\mathcal{T})}\right\|_1
    \le
    \exp\!\left(-\frac{e^{-3z^2}}{90}T\right).
\]
Thus the hypothesis of Lemma~\ref{lemma_sensing_time_from_T_tree_bound} holds with $C = \exp\!\left(-\frac{e^{-3z^2}}{90}T\right)$. Applying Lemma~\ref{lemma_sensing_time_from_T_tree_bound}, we obtain
\[
    n
    \ge
    \frac{1}{3C}
    =
    \frac13
    \exp\!\left(\frac{e^{-3z^2}}{90}T\right),
\]
as claimed.
\end{proof}

We now return to the pulsed-signal problems. Combining
Lemmas~\ref{lem:pulse_to_piecewise} and~\ref{lemma_access_model}, a
Gaussian protocol with squeezing bounded by $z$ and physical sensing-time
budget $T_{\mathrm{sensing}}$ can be simulated in the simplified access
model over $n=\lceil T_{\mathrm{sensing}}/T\rceil$ blocks. Since learning
the full parity string also solves the partially revealed problem,
the preceding bound on $n$ gives the following theorem.

\begin{thm}\label{thm_lower_bound_learning_parity_gaussian}
Let $T\in\mathbb{N}^+$ and let $z\ge1$. Suppose that a Gaussian protocol with squeezing bounded by $z$ solves either the CV Learning Parity problem, Problem~\ref{def:PROB}, or the CV Learning Parity with Partially Revealed Information problem, Problem~\ref{def:PROB2}, with success probability at least \(2/3\) using total sensing time $T_{\mathrm{sensing}}$. Then
\[
    T_{\mathrm{sensing}}
    \ge
    T\left[
    \frac13
    \exp\!\left(\frac{e^{-3z^2}}{90}T\right)
    -1
    \right].
\]
\end{thm}

\begin{proof}
Set $n\coloneqq\lceil T_{\mathrm{sensing}}/T\rceil$.
A protocol for Problem~\ref{def:PROB} also solves
Problem~\ref{def:PROB2}, so it is enough to treat the latter.
Lemma~\ref{lem:pulse_to_piecewise} gives a protocol for its
piecewise-constant version with the same squeezing bound and budget
$nT$. Applying Lemma~\ref{lemma_access_model} gives a protocol for
Problem~\ref{def:PROB_revealed} in the simplified access model over
those $n$ blocks. Hence Lemma~\ref{lemma_lower_bound_learning_parity_simplified} gives
\[
    n
    \ge
    \frac13
    \exp\!\left(\frac{e^{-3z^2}}{90}T\right).
\]
Since \(n=\lceil T_{\mathrm{sensing}}/T\rceil\le T_{\mathrm{sensing}}/T+1\), the claim follows.
\end{proof}

\newpage
\subsection{Lower bound for the CV Hypothesis Testing problem}

We now specialize Lemma~\ref{lemma_sensing_time_from_T_tree_bound} to the CV Hypothesis Testing problem. 
As discussed above, the same abstract distribution-identification framework also covers binary hypothesis testing: one may introduce an auxiliary label \(s\in\{0,1\}^T\), set \(q_0\) to be the one-block distribution induced by \(p\), and set \(q_s\) to be the one-block distribution induced by \(q\) for every non-zero \(s\). With this convention, distinguishing the two hypotheses \(p\) and \(q\) is equivalent to distinguishing the case \(s=0\) from the case \(s\ne0\). The revealed string in Problem~\ref{def:PROB_revealed} then carries no additional information, since all non-zero labels correspond to the same distribution, and Lemma~\ref{lemma_sensing_time_from_T_tree_bound} applies directly.

Therefore, it is enough to upper bound, uniformly over all \(T\)-time trees \(\mathcal{T}\), the total variation distance
\[
    \frac12
    \left\|P_p^{(\mathcal{T})}-P_q^{(\mathcal{T})}\right\|_1 .
\]
Here \(P_p^{(\mathcal{T})}\) and \(P_q^{(\mathcal{T})}\) denote the leaf distributions of the \(T\)-time tree \(\mathcal{T}\) when the block of the signal is generated respectively from \(p\) and from \(q\).

Recall the definition of the signal in the CV Hypothesis Testing problem. Specifically, the block vectors \(m\in\mathbb{R}^T\) of the signal are i.i.d.~random variables constructed as follows. Let \(p\) and \(q\) be two probability distributions over \(\{0,1\}^T\). In each block of length \(T\), a bit string $x=(x_1,\ldots,x_T)\in\{0,1\}^T$ is sampled from one of these two distributions. Independently, amplitudes \(b_1,\ldots,b_T\) are sampled uniformly from the interval \([1/2,3/2]\). The corresponding block vector of the signal is defined as
\bb\label{eq:hyp_testing_block_recall}
    m
    =
    \bigl(
    (-1)^{x_1}b_1,\ldots,(-1)^{x_T}b_T
    \bigr)
    \in\mathbb{R}^T .
\ee
The goal is to decide whether the hidden bit string \(x\) was sampled
from \(p\) or from \(q\). Fix \(c\in[T]\), and throughout this
subsection assume that \(p\) and \(q\) have matching marginals on
every subset of at most \(T-c\) coordinates.

Let \(\mathcal{T}\) be a fixed \(T\)-time tree. To simplify notation, throughout this subsection we write
\[
    P_p(l)\equiv P_p^{(\mathcal{T})}(l),
    \qquad
    P_q(l)\equiv P_q^{(\mathcal{T})}(l)
\]
for the probability of reaching the leaf \(l\) of \(\mathcal{T}\) under the two hypotheses \(p\) and \(q\), respectively.

Let \(r\) be the root of \(\mathcal{T}\), and let \(l\in\mathrm{leaves}(\mathcal{T})\). We denote the unique path from \(r\) to \(l\) by
\[
    r=u_0^{(l)}
    \to u_1^{(l)}
    \to \cdots
    \to u_{N(l)-1}^{(l)}
    \to u_{N(l)}^{(l)}=l.
\]
For each \(i\in[N(l)]\), the edge $u_{i-1}^{(l)}\to u_i^{(l)}$ corresponds to the experiment prescribed at the node \(u_{i-1}^{(l)}\), and its observed outcome is denoted by \(\mathbf{y}_i^{(l)}\).

Along this fixed branch, any block string \(x\in\{0,1\}^T\) can be decomposed according to the successive time intervals explored by the experiments:
\[
    x=
    \left(
    x_1^{(l)},x_2^{(l)},\ldots,x_{N(l)}^{(l)}
    \right),
\]
where $x_i^{(l)}
    \in
    \{0,1\}^{\tau(u_{i-1}^{(l)})}$. Thus \(x_i^{(l)}\) is the substring of \(x\) corresponding to the experiment prescribed at the node \(u_{i-1}^{(l)}\). Since \(l\) is a leaf of the \(T\)-time tree, the durations along the branch add up to the whole block:
\[
    \sum_{i=1}^{N(l)}
    \tau\!\left(u_{i-1}^{(l)}\right)
    =
    T.
\]

We also introduce the following local notation. Let \(u\) be a node of the \(T\)-time tree, and let \(\xi\in\{0,1\}^{\tau(u)}\) be the substring of the hidden bit string associated with the experiment prescribed at \(u\). We denote by \(p_{\xi}^{(u)}(\mathbf{y})\) the conditional distribution of the outcome of that experiment, conditioned on this local substring being equal to \(\xi\), after averaging over the random amplitudes of the signal.

More explicitly, if \(\tau(u)\ge2\), then the experiment at \(u\) is an ordinary \(T\)-time-tree experiment. In this case, the outcome belongs to \(\mathbb{R}^2\), and
\bb\label{eq_gausscond2_ht}
    p_{\xi}^{(u)}(\mathbf{y})
    =
    \mathop{\mathbb{E}}_{b_1,\ldots,b_{\tau(u)}}
    \NN\!\left[
    \begin{pmatrix}
        \sum_{j=1}^{\tau(u)}(-1)^{\xi_j}b_j a_j^{(u)}\\[1mm]
        \sum_{j=1}^{\tau(u)}(-1)^{\xi_j}b_j d_j^{(u)}
    \end{pmatrix},
    \frac{\mathbb{1}_2}{10z^2}
    \right](\mathbf{y}),
    \qquad
    \mathbf{y}\in\mathbb{R}^2,
\ee
where \(b_1,\ldots,b_{\tau(u)}\) are independent and uniformly distributed on \([1/2,3/2]\), and
\[
    \sqrt{(a_j^{(u)})^2+(d_j^{(u)})^2}\le1
    \qquad
    \forall\,j\in[\tau(u)].
\]
If instead \(\tau(u)=1\), then the experiment at \(u\) is refined by the unit-time tree attached to \(u\). In this case, \(\xi\in\{0,1\}\), the outcome \(\mathbf{y}\) is a leaf, equivalently the full transcript, of that attached unit-time tree, and \(p_{\xi}^{(u)}(\mathbf{y})\) denotes its leaf probability distribution conditioned on the bit value \(\xi\), again after averaging over the random amplitude \(b\sim\mathrm{Unif}([1/2,3/2])\).

With this notation, the probability of reaching a leaf \(l\) under the hypothesis \(p\) is
\bb\label{eq:T_time_leaf_distribution_ht_p}
    P_p(l)
    =
    \mathop{\mathbb{E}}_{x\sim p}
    \prod_{i=1}^{N(l)}
    p_{x_i^{(l)}}^{(u_{i-1}^{(l)})}
    \!\left(\mathbf{y}_i^{(l)}\right),
\ee
and similarly
\bb\label{eq:T_time_leaf_distribution_ht_q}
    P_q(l)
    =
    \mathop{\mathbb{E}}_{x\sim q}
    \prod_{i=1}^{N(l)}
    p_{x_i^{(l)}}^{(u_{i-1}^{(l)})}
    \!\left(\mathbf{y}_i^{(l)}\right).
\ee
The next lemma is the analogue of Lemma~\ref{lemma_adapt_prob} for the CV Hypothesis Testing problem. Here the role of the hidden parity string is replaced by the Walsh--Fourier expansion of the signed measure \(p-q\). The key point is that the assumption of matching marginals removes all low-weight Fourier coefficients, while the remaining high-weight coefficients factorize along each branch of the tree.

\begin{lemma}\label{lemma_adapt_prob_hyp_testing}
Let \(\mathcal{T}\) be a \(T\)-time tree. For every node \(u\) of \(\mathcal{T}\), and for every local string
\(\sigma\in\{0,1\}^{\tau(u)}\), define
\bb\label{eq:def_A_u_hyp_testing}
    A_{\sigma}^{(u)}(\mathbf{y})
    \coloneqq
    \mathop{\mathbb{E}}_{\xi\in\{0,1\}^{\tau(u)}}
    (-1)^{\xi\cdot\sigma}
    p_{\xi}^{(u)}(\mathbf{y}).
\ee
Then, for every leaf \(l\in\mathrm{leaves}(\mathcal{T})\),
\bb\label{eq:adapt_prob_product_bound_hyp_testing}
    \left|P_p(l)-P_q(l)\right|
    \le
    \|p-q\|_1
    \sum_{s:\,|s|>T-c}
    \prod_{i=1}^{N(l)}
    \left|A_{s_i^{(l)}}^{(u_{i-1}^{(l)})}
    \!\left(\mathbf{y}_i^{(l)}\right)\right|,
\ee
where, along the branch leading to \(l\), the string \(s\in\{0,1\}^T\) is decomposed as
\[
    s=
    \left(
    s_1^{(l)},s_2^{(l)},\ldots,s_{N(l)}^{(l)}
    \right),
    \qquad
    s_i^{(l)}
    \in
    \{0,1\}^{\tau(u_{i-1}^{(l)})}.
\]
\end{lemma}

\begin{proof}
Fix a leaf \(l\in\mathrm{leaves}(\mathcal{T})\). To simplify notation inside the proof, write
\[
    u_i\coloneqq u_{i-1}^{(l)},
    \qquad
    \tau_i\coloneqq \tau(u_i),
    \qquad
    \mathbf{y}_i\coloneqq \mathbf{y}_i^{(l)}
\]
for \(i\in[N(l)]\). Along the branch leading to \(l\), any string \(x\in\{0,1\}^T\) decomposes as
\[
    x=(x_1,\ldots,x_{N(l)}),
    \qquad
    x_i\in\{0,1\}^{\tau_i}.
\]
We decompose \(s\in\{0,1\}^T\) in the same way:
\[
    s=(s_1,\ldots,s_{N(l)}),
    \qquad
    s_i\in\{0,1\}^{\tau_i}.
\]
For each node \(u_i\) and each local string \(\sigma\in\{0,1\}^{\tau_i}\), define
\[
    A_{\sigma}^{(u_i)}(\mathbf{y}_i)
    \coloneqq
    \mathop{\mathbb{E}}_{\xi\in\{0,1\}^{\tau_i}}
    (-1)^{\xi\cdot\sigma}
    p_{\xi}^{(u_i)}(\mathbf{y}_i).
\]
By \eqref{eq:T_time_leaf_distribution_ht_p} and \eqref{eq:T_time_leaf_distribution_ht_q}, we have
\[
    P_p(l)-P_q(l)
    =
    \sum_{x\in\{0,1\}^T}
    \bigl(p(x)-q(x)\bigr)\prod_{i=1}^{N(l)}
    p_{x_i}^{(u_i)}(\mathbf{y}_i).
\]
We expand the signed measure \(p-q\) in the Walsh--Fourier basis:
\[
    p(x)-q(x)
    =
    \sum_{s\in\{0,1\}^T}
    \frac{1}{2^T}
    \left[
        \sum_{x'\in\{0,1\}^T}
        (-1)^{s\cdot x'}
        \bigl(p(x')-q(x')\bigr)
    \right]
    (-1)^{x\cdot s}.
\]
Substituting this expansion gives
\[
\begin{aligned}
    P_p(l)-P_q(l)
    &=
    \sum_{s\in\{0,1\}^T}
    \left[
        \sum_{x'\in\{0,1\}^T}
        (-1)^{s\cdot x'}
        \bigl(p(x')-q(x')\bigr)
    \right]
    \mathop{\mathbb{E}}_{x\in\{0,1\}^T}
    \prod_{i=1}^{N(l)}
    p_{x_i}^{(u_i)}(\mathbf{y}_i)(-1)^{x\cdot s}.
\end{aligned}
\]
Since both \(x\) and \(s\) decompose along the branch, the last expectation factorizes:
\[
\begin{aligned}
    \mathop{\mathbb{E}}_{x\in\{0,1\}^T}
    \prod_{i=1}^{N(l)}
    p_{x_i}^{(u_i)}(\mathbf{y}_i)(-1)^{x\cdot s}
    &=
    \mathop{\mathbb{E}}_{x_1,\ldots,x_{N(l)}}
    \prod_{i=1}^{N(l)}
    p_{x_i}^{(u_i)}(\mathbf{y}_i)(-1)^{x_i\cdot s_i}  =
    \prod_{i=1}^{N(l)}
    A_{s_i}^{(u_i)}(\mathbf{y}_i).
\end{aligned}
\]
Hence
\bb\label{eq:leaf_difference_fourier_hyp_testing}
    P_p(l)-P_q(l)
    =
    \sum_{s\in\{0,1\}^T}
    \left[
        \sum_{x'\in\{0,1\}^T}
        (-1)^{s\cdot x'}
        \bigl(p(x')-q(x')\bigr)
    \right]
    \prod_{i=1}^{N(l)}
    A_{s_i}^{(u_i)}(\mathbf{y}_i).
\ee
We now use the matching-marginals assumption. If \(|s|\le T-c\), then the function $x'\mapsto (-1)^{s\cdot x'}$ depends only on the coordinates in the support of \(s\), whose size is at most \(T-c\). Since \(p\) and \(q\) have identical marginals on every subset of at most \(T-c\) coordinates, it follows that
\[
    \sum_{x'\in\{0,1\}^T}
    (-1)^{s\cdot x'}
    \bigl(p(x')-q(x')\bigr)
    =
    0
    \qquad
    \text{whenever } |s|\le T-c.
\]
Therefore only the Fourier modes with \(|s|>T-c\) remain in \eqref{eq:leaf_difference_fourier_hyp_testing}. Taking absolute values and using the triangle inequality, we obtain
\[
\begin{aligned}
    \left|P_p(l)-P_q(l)\right|
    &\le
    \sum_{s:\,|s|>T-c}
    \left|
        \sum_{x'\in\{0,1\}^T}
        (-1)^{s\cdot x'}
        \bigl(p(x')-q(x')\bigr)
    \right|
    \prod_{i=1}^{N(l)}
    \left|
        A_{s_i}^{(u_i)}(\mathbf{y}_i)
    \right|  \\
    &\le
    \|p-q\|_1
    \sum_{s:\,|s|>T-c}
    \prod_{i=1}^{N(l)}
    \left|A_{s_i}^{(u_i)}(\mathbf{y}_i)\right|,
\end{aligned}
\]
where in the last step we used
\[
    \left|
        \sum_{x'\in\{0,1\}^T}
        (-1)^{s\cdot x'}
        \bigl(p(x')-q(x')\bigr)
    \right|
    \le
    \sum_{x'\in\{0,1\}^T}|p(x')-q(x')|
    =
    \|p-q\|_1 .
\]
This proves \eqref{eq:adapt_prob_product_bound_hyp_testing}.
\end{proof}

We now need to control the local factors appearing in Lemma~\ref{lemma_adapt_prob_hyp_testing}. The bound in \eqref{eq:adapt_prob_product_bound_hyp_testing} reduces the leafwise distinguishability to a sum over high-weight Fourier modes, and, for each such mode, to a product of local terms \(A_{\sigma}^{(u)}\) along the branch. Thus, the next step is to bound \(\|A_{\sigma}^{(u)}\|_1\) uniformly over the nodes of the \(T\)-time tree, with a decay proportional to the local Hamming weight \(|\sigma|\). There are two cases: if \(\tau(u)=1\), the node is resolved by an attached unit-time tree; if \(\tau(u)\ge2\), the transition is directly given by the Gaussian mixture in \eqref{eq_gausscond2_ht}. We begin with the first case.

\subsubsection{Unit-time tree analysis}

Let $u$ be a node of the $T$-time tree with $\tau(u)=1$. Then the local substring is a single bit $\xi\in\{0,1\}$. In the present hypothesis-testing problem, this bit does not determine the signal value by itself: rather, conditioned on $\xi$, the signal inside the unit interval is equal to $(-1)^\xi b$, where the amplitude $b$ is sampled uniformly from $[\frac12,\frac32]$. Thus $p_\xi^{(u)}$ is the leaf distribution of the unit-time tree attached to $u$, averaged over this random amplitude.

The next lemma gives the corresponding unit-time estimate for the local Walsh coefficients. The point is that, after conditioning on the amplitude $b$, the two cases $+b$ and $-b$ are statistically equivalent, by a hypothesis-independent shift of the transcript, to the two signals $2b$ and $0$. Since $2b\le3$, we can reduce the analysis to the unit-time tree estimate already proved for the CV Learning Parity problem, with effective signal strength at most $3$.

\begin{lemma}\label{lemma_unit_time_tree_contribution_ht}
Let $u$ be a node of a $T$-time tree with $\tau(u)=1$, and let $\mathcal{U}_u$ be the unit-time tree attached to $u$. For $\sigma\in\{0,1\}$, define
\[
    A_{\sigma}^{(u)}(\mathbf{y})
    \coloneqq
    \mathop{\mathbb{E}}_{\xi\in\{0,1\}}
    (-1)^{\xi\sigma}
    p_{\xi}^{(u)}(\mathbf{y}).
\]
Then, for every $z\ge1$,
\bb\label{eq:unit_time_A_sigma_ht_bound}
    \left\|A_{\sigma}^{(u)}\right\|_1
    \le
    \exp\!\left(
        -\frac16 e^{-45z^2/2}\,|\sigma|
    \right).
\ee 
\end{lemma}

\begin{proof}
The case $\sigma=0$ is immediate. Indeed,
\[
    A_0^{(u)}
    =
    \frac12\left(p_0^{(u)}+p_1^{(u)}\right)
\]
is a probability density, and therefore $\|A_0^{(u)}\|_1=1$.

It remains to consider $\sigma=1$. Let $P_{\xi,b}^{(\mathcal{U}_u)}$ denote the leaf distribution of the attached unit-time tree when the local bit is $\xi$ and the amplitude is fixed to be $b$. Then $p_\xi^{(u)}
    =
    \mathop{\mathbb{E}}_{b}
    P_{\xi,b}^{(\mathcal{U}_u)}$. Hence, by convexity of the $L^1$ norm,
\bb\label{eq:condition_on_b_unit_ht}
    \left\|A_1^{(u)}\right\|_1
    =
    \frac12
    \left\|
    p_0^{(u)}-p_1^{(u)}
    \right\|_1
    \le
    \mathop{\mathbb{E}}_b
    \frac12
    \left\|
    P_{0,b}^{(\mathcal{U}_u)}
    -
    P_{1,b}^{(\mathcal{U}_u)}
    \right\|_1 .
\ee
Thus it is enough to bound the distinguishability at fixed $b$.

Fix $b\in[\frac12,\frac32]$. Under the two hypotheses $\xi=0$ and $\xi=1$, the signal in the unit interval is respectively $+b$ and $-b$. Consider any node $v$ of the unit-time tree. If the experiment prescribed at $v$ has rotation coefficients $(a^{(v)},d^{(v)})$, then the corresponding Gaussian outcome has mean
\[
    (-1)^\xi b
    \begin{pmatrix}
        a^{(v)}\\
        d^{(v)}
    \end{pmatrix}.
\]
Since $b$ is now fixed and known, we may apply to each observed outcome the deterministic shift
\[
    \mathbf{y}
    \longmapsto
    \mathbf{y}
    +
    b
    \begin{pmatrix}
        a^{(v)}\\
        d^{(v)}
    \end{pmatrix}.
\]
This transformation is invertible and independent of the hypothesis $\xi$. Therefore it preserves the total variation distance between the two induced transcript distributions. After this shift, the case $\xi=1$ has zero mean, while the case $\xi=0$ has mean
\[
    2b
    \begin{pmatrix}
        a^{(v)}\\
        d^{(v)}
    \end{pmatrix}.
\]
Thus, conditioned on $b$, the problem of distinguishing the two signals $+b$ and $-b$ is equivalent to distinguishing the two signals $2b$ and $0$.

We can now invoke the unit-time tree bound proved in the CV Learning Parity analysis, applied with signal strength $2b$. That bound, stated in \eqref{eq:unit_time_tree_tv_bound}, gives
\bb\label{eq:fixed_b_unit_ht}
    \frac12
    \left\|
    P_{0,b}^{(\mathcal{U}_u)}
    -
    P_{1,b}^{(\mathcal{U}_u)}
    \right\|_1
    \le
    \frac12
    \left\|
    \NN\!\left[2b,\frac{1}{10z^2}\right]
    -
    \NN\!\left[0,\frac{1}{10z^2}\right]
    \right\|_1 .
\ee
Since $2b\le3$ and the total variation distance between one-dimensional Gaussians with the same variance is monotone in the distance between their means, we obtain
\bb\label{eq:fixed_b_unit_ht_max}
    \frac12
    \left\|
    P_{0,b}^{(\mathcal{U}_u)}
    -
    P_{1,b}^{(\mathcal{U}_u)}
    \right\|_1
    \le
    \frac12
    \left\|
    \NN\!\left[3,\frac{1}{10z^2}\right]
    -
    \NN\!\left[0,\frac{1}{10z^2}\right]
    \right\|_1 .
\ee
Combining \eqref{eq:condition_on_b_unit_ht} and \eqref{eq:fixed_b_unit_ht_max}, we get
\bb\label{eq:A_one_gaussian_bound_ht}
    \left\|A_1^{(u)}\right\|_1
    \le
    \frac12
    \left\|
    \NN\!\left[3,\frac{1}{10z^2}\right]
    -
    \NN\!\left[0,\frac{1}{10z^2}\right]
    \right\|_1 .
\ee

It remains only to estimate the last Gaussian total variation distance. We apply Lemma~\ref{tvd_one_g} with $m_1=3$, $m_2=0$, and $V=1/(10z^2)$. In this case
\[
    \alpha
    \coloneqq
    \frac{|m_1-m_2|}{2\sqrt V}
    =
    \frac{3\sqrt{10}\,z}{2},
    \qquad
    \alpha^2
    =
    \frac{45z^2}{2}.
\]
Therefore,
\[
    \frac12
    \left\|
    \NN\!\left[3,\frac{1}{10z^2}\right]
    -
    \NN\!\left[0,\frac{1}{10z^2}\right]
    \right\|_1
    \le
    1-\frac16 e^{-45z^2/2}.
\]
Using $1-x\le e^{-x}$ for $x\ge0$, we conclude that
\[
    \left\|A_1^{(u)}\right\|_1
    \le
    \exp\!\left(
        -\frac16 e^{-45z^2/2}
    \right).
\]
Together with the trivial identity $\|A_0^{(u)}\|_1=1$, this proves \eqref{eq:unit_time_A_sigma_ht_bound}.
\end{proof}

\subsubsection{Analysis for $\tau(u)\ge2$}

We now consider the second case, namely a node $u$ of the $T$-time tree with $\tau(u)\ge2$. Set \(t\coloneqq\tau(u)\). In the CV hypothesis-testing problem, the local string \(\xi\in\{0,1\}^t\) determines only the signs of the signal, while the amplitudes are sampled independently. Thus the transition at \(u\) is described by the Gaussian mixture
\bb\label{eq_gausscond_ht_tau_ge_2}
    p_{\xi}^{(u)}(\mathbf{y})
    =
    \mathop{\mathbb{E}}_{b_1,\ldots,b_t}
    \NN\!\left[
    \begin{pmatrix}
        \sum_{j=1}^{t}(-1)^{\xi_j}b_j a_j^{(u)}\\[0.5ex]
        \sum_{j=1}^{t}(-1)^{\xi_j}b_j d_j^{(u)}
    \end{pmatrix},
    \frac{\mathbb{1}_2}{10z^2}
    \right](\mathbf{y}),
    \qquad
    \mathbf{y}\in\mathbb{R}^2,
\ee
where \(b_1,\ldots,b_t\) are independent and uniformly distributed in \([\frac12,\frac32]\), and the coefficients satisfy
\[
    \bigl(a_j^{(u)}\bigr)^2+\bigl(d_j^{(u)}\bigr)^2\le1
    \qquad
    \forall\,j\in[t].
\]
For every \(\sigma\in\{0,1\}^t\), define the local Walsh coefficient
\[
    A_{\sigma}^{(u)}(\mathbf{y})
    \coloneqq
    \mathop{\mathbb{E}}_{\xi\in\{0,1\}^t}
    (-1)^{\xi\cdot\sigma}
    p_{\xi}^{(u)}(\mathbf{y}).
\]
Unlike in the CV Learning Parity problem, here we do not average
these coefficients over all local strings \(\sigma\). Instead, we
bound each coefficient \(A_\sigma^{(u)}\) separately in terms of
\(|\sigma|\).

We first consider the all-one Walsh coefficient. For \(z\ge1\),
an integer \(r\ge1\), and vectors
\(\mathbf{v}_1,\ldots,\mathbf{v}_r\in\mathbb{R}^2\) satisfying
\(\|\mathbf{v}_j\|_2\le1\), define
\begin{equation}\label{eq:all_one_coefficient_ht}
    B_r(\mathbf{y})
    \coloneqq
    \mathop{\mathbb{E}}_{\eta,b}
    (-1)^{|\eta|}
    \NN\!\left[
        \sum_{j=1}^r(-1)^{\eta_j}b_j\mathbf{v}_j,
        \frac{\mathbb{1}_2}{10z^2}
    \right](\mathbf{y}),
\end{equation}
where \(\eta\) is uniform on \(\{0,1\}^r\),
\(b_1,\ldots,b_r\) are independent and uniform on
\([\frac12,\frac32]\), and \(\eta\) is independent of \(b\).

\begin{lemma}\label{lemma_large_t_A_tau_ge_2_ht}
Let \(B_r\) be defined by Eq.~\eqref{eq:all_one_coefficient_ht},
with \(r\ge2\). Then
\(\|B_r\|_1\le3zr(0.91)^r+e^{-r^2/2}\).
\end{lemma}

\begin{proof}
Set \(\sigma_z\coloneqq1/(\sqrt{10}z)\),
\(R\coloneqq(\frac32+\sigma_z)r\), and
\(B_R\coloneqq\{\mathbf{y}\in\mathbb{R}^2:
\|\mathbf{y}\|_2\le R\}\).
Every Gaussian component in Eq.~\eqref{eq:all_one_coefficient_ht}
has mean
\(\boldsymbol{\mu}_{\eta,b}
=\sum_{j=1}^r(-1)^{\eta_j}b_j\mathbf{v}_j\),
with \(\|\boldsymbol{\mu}_{\eta,b}\|_2\le3r/2\).
Thus, if \(\mathbf{Y}\) is sampled from such a component and
\(\mathbf{G}\sim\NN[0,\mathbb{1}_2]\), then
\[
    \Pr(\mathbf{Y}\notin B_R)
    \le
    \Pr(\|\mathbf{G}\|_2>r)
    =
    \int_r^\infty s e^{-s^2/2}\,\mathrm{d}s
    =
    e^{-r^2/2}.
\]
Averaging this estimate gives
\(\int_{B_R^c}|B_r(\mathbf{y})|\,\mathrm{d}^2\mathbf{y}
\le e^{-r^2/2}\).
On \(B_R\), Cauchy--Schwarz therefore yields
\[
    \|B_r\|_1
    \le
    \sqrt{\pi}R\|B_r\|_2+e^{-r^2/2}.
\]

It remains to bound \(\|B_r\|_2\). We use the Fourier transform
convention
\(\widehat h(\boldsymbol{\omega})
=(2\pi)^{-1}\int_{\mathbb{R}^2}
h(\mathbf{y})e^{-i\boldsymbol{\omega}\cdot\mathbf{y}}\,
\mathrm{d}^2\mathbf{y}\),
for which Plancherel's identity reads
\(\|h\|_2=\|\widehat h\|_2\).
Independence of the signs and amplitudes gives
\[
    \widehat B_r(\boldsymbol{\omega})
    =
    \frac{(-i)^r}{2\pi}
    e^{-\sigma_z^2\|\boldsymbol{\omega}\|_2^2/2}
    \prod_{j=1}^r
    \mathop{\mathbb{E}}_{b_j}
    \sin\!\left(b_j\boldsymbol{\omega}\cdot\mathbf{v}_j\right).
\]
For \(b\) uniform on \([\frac12,\frac32]\), write
\[
    \phi(x)
    \coloneqq
    \mathop{\mathbb{E}}_b\sin(bx)
    =
    \sin(x)\frac{\sin(x/2)}{x/2},
\]
with the value at \(x=0\) understood by continuity. It is elementary to check that
\(\sup_{x\in\mathbb{R}}|\phi(x)|<0.91\).

Consequently, Plancherel's identity gives
\[
    \|B_r\|_2
    \le
    \frac{(0.91)^r}{2\pi}
    \left(
        \int_{\mathbb{R}^2}
        e^{-\sigma_z^2\|\boldsymbol{\omega}\|_2^2}
        \,\mathrm{d}^2\boldsymbol{\omega}
    \right)^{1/2}
    =
    \frac{(0.91)^r}{2\sqrt{\pi}\,\sigma_z}.
\]
Substituting into the \(L^1\) estimate yields
\[
    \|B_r\|_1
    \le
    \frac{\frac32+\sigma_z}{2\sigma_z}\,
    r(0.91)^r+e^{-r^2/2}.
\]
Finally,
\((\frac32+\sigma_z)/(2\sigma_z)
=3\sqrt{10}z/4+1/2\le3z\),
since \(z\ge1\). This proves the claim.
\end{proof}

The preceding estimate is useful when \(r\) is large. For small
values of \(r\), we use a direct comparison between the two
parity classes.

\begin{lemma}\label{lemma_small_t_A_tau_ge_2_ht}
Let \(B_r\) be defined by Eq.~\eqref{eq:all_one_coefficient_ht},
with \(r\ge1\). Then
\(\|B_r\|_1\le\exp(-\frac16e^{-45z^2/2})\).
\end{lemma}

\begin{proof}
Condition on the amplitudes \(b_1,\ldots,b_r\) and the bits
\(\eta_2,\ldots,\eta_r\), and average over \(\eta_1\).
The two values of \(\eta_1\) have opposite parity, and the
corresponding Gaussian means differ only in the first contribution.
By convexity of the \(L^1\) norm and translation invariance,
\[
    \|B_r\|_1
    \le
    \frac12\mathop{\mathbb{E}}_{b_1}
    \left\|
        \NN\!\left[
            b_1\mathbf{v}_1,\frac{\mathbb{1}_2}{10z^2}
        \right]
        -
        \NN\!\left[
            -b_1\mathbf{v}_1,\frac{\mathbb{1}_2}{10z^2}
        \right]
    \right\|_1.
\]
Rotating the first axis along \(\mathbf{v}_1\) and integrating out
the other coordinate reduces this distance to the one-dimensional
case. If \(\mathbf{v}_1=0\), the distance is zero.
Since \(2b_1\|\mathbf{v}_1\|_2\le3\), monotonicity of the Gaussian
total variation distance in the separation of the means gives
\[
    \|B_r\|_1
    \le
    \frac12
    \left\|
        \NN\!\left[3,\frac{1}{10z^2}\right]
        -
        \NN\!\left[0,\frac{1}{10z^2}\right]
    \right\|_1.
\]
Applying Lemma~\ref{tvd_one_g}, with
\(\alpha=3\sqrt{10}z/2\), we obtain
\[
    \|B_r\|_1
    \le
    1-\frac16e^{-45z^2/2}
    \le
    \exp\!\left(-\frac16e^{-45z^2/2}\right).
\]
This proves the claim.
\end{proof}

We now combine the two bounds.

\begin{lemma}\label{lemma_min_bound_ht}
Let \(z\ge1\), and let \(r\ge2\) be an integer. Then
\[
    \min\left\{
        \exp\!\left(-\frac16 e^{-45z^2/2}\right),
        \;
        3zr(0.91)^r+e^{-r^2/2}
    \right\}
    \le
    \exp\!\left(
        -\frac{e^{-25z^2}}{30}\,r
    \right).
\]
\end{lemma}
\begin{proof}
Set \(\gamma\coloneqq e^{-25z^2}/30\) and
\(S\coloneqq5e^{5z^2/2}\). Since
\(\gamma S=\frac16e^{-45z^2/2}\), the first term in the minimum is
\(e^{-\gamma S}\).

If \(r\le S\), then the minimum is at most
\(e^{-\gamma S}\le e^{-\gamma r}\). We may therefore assume that
\(r\ge S\).

Since \(z\ge1\), we have
\(\log z\le z-1\le\frac52(z^2-1)\). Hence
\[
    z
    \le
    e^{\frac52(z^2-1)}
    =
    \frac{S}{5e^{5/2}}
    \le
    \frac{r}{5e^{5/2}}.
\]
Moreover,
\(0.094+\gamma
\le0.094+e^{-25}/30
<-\log(0.91)\), and therefore
\((0.91)^r\le e^{-0.094r}e^{-\gamma r}\). It follows that
\[
\begin{aligned}
    3zr(0.91)^r
    &\le
    \frac{3r^2}{5e^{5/2}}\,
    e^{-0.094r}e^{-\gamma r} \le
    15e^{5/2}e^{-0.47e^{5/2}}
    e^{-\gamma r}
    <
    \frac23e^{-\gamma r}.
\end{aligned}
\]
Here, in the second inequality, we used the fact that
\(x\mapsto x^2e^{-0.094x}\) is decreasing for
\(x\ge5e^{5/2}\).

Finally, since \(r\ge S>5\) and \(\gamma\le1/30\), we have
\[
    \frac{r^2}{2}-\gamma r
    =
    r\left(\frac r2-\gamma\right)
    \ge
    5\left(\frac52-\frac1{30}\right)
    >
    \log 3.
\]
Thus \(e^{-r^2/2}\le\frac13e^{-\gamma r}\). Combining the last two
estimates gives
\[
    3zr(0.91)^r+e^{-r^2/2}
    \le
    e^{-\gamma r},
\]
which concludes the proof.
\end{proof}

We can now state the local estimate in the form needed for the hypothesis-testing lower bound.

 \begin{lemma}\label{lemma_tau_ge_2_A_bound_ht}
Let \(u\) be a node of the \(T\)-time tree with
\(t=\tau(u)\ge2\). Then, for every
\(\sigma\in\{0,1\}^t\),
\[
    \left\|A_{\sigma}^{(u)}\right\|_1
    \le
    \exp\!\left(-\frac{e^{-25z^2}}{30}\,|\sigma|\right).
\]
\end{lemma}

\begin{proof}
If \(\sigma=0\), then
\(A_0^{(u)}
=
\mathop{\mathbb{E}}_{\xi\in\{0,1\}^t}
p_\xi^{(u)}\)
is a probability density. Hence
\(\|A_0^{(u)}\|_1=1\), which proves the claim in this case.

Assume now that \(\sigma\ne0\), and set
\(r\coloneqq|\sigma|\). Write
\(S\coloneqq\operatorname{supp}(\sigma)
=\{j_1,\ldots,j_r\}\), and set
\(\mathbf{v}_j^{(u)}
\coloneqq
(a_j^{(u)},d_j^{(u)})^\intercal\)
for every \(j\in[t]\). Define
\[
\begin{aligned}
    \widetilde A^{(u)}(\mathbf{y})
    \coloneqq
    \mathop{\mathbb{E}}_{\eta\in\{0,1\}^r}
    \mathop{\mathbb{E}}_{b_1,\ldots,b_r}
    (-1)^{|\eta|}
    \NN\!\left[
        \sum_{\ell=1}^r
        (-1)^{\eta_\ell}b_\ell
        \mathbf{v}_{j_\ell}^{(u)},
        \frac{\mathbb{1}_2}{10z^2}
    \right](\mathbf{y}),
\end{aligned}
\]
where \(b_1,\ldots,b_r\) are independent and uniformly distributed
in \([\frac12,\frac32]\).

Let
\[
    \mathbf{Z}
    \coloneqq
    \sum_{j\notin S}
    (-1)^{\zeta_j}\bar b_j
    \mathbf{v}_j^{(u)},
\]
where the bits \(\zeta_j\) are independent and uniformly distributed
in \(\{0,1\}\), the amplitudes \(\bar b_j\) are independent and
uniformly distributed in \([\frac12,\frac32]\), and all these
variables are mutually independent and independent of the variables
entering the definition of \(\widetilde A^{(u)}\). If \(S=[t]\), we
use the convention \(\mathbf{Z}=0\).

For brevity, set
\(\mathbf{M}_{\eta,b}
\coloneqq
\sum_{\ell=1}^r
(-1)^{\eta_\ell}b_\ell
\mathbf{v}_{j_\ell}^{(u)}\).
Since
\(\xi\cdot\sigma
\equiv
\sum_{\ell=1}^r\xi_{j_\ell}\pmod 2\),
separating the coordinates inside and outside \(S\) gives
\[
\begin{aligned}
    A_{\sigma}^{(u)}(\mathbf{y})
    &=
    \mathop{\mathbb{E}}_{(\zeta_j,\bar b_j)_{j\notin S}}
    \mathop{\mathbb{E}}_{\eta\in\{0,1\}^r}
    \mathop{\mathbb{E}}_{b_1,\ldots,b_r}
    (-1)^{|\eta|}
    \NN\!\left[
        \mathbf{M}_{\eta,b}+\mathbf{Z},
        \frac{\mathbb{1}_2}{10z^2}
    \right](\mathbf{y})\\
    &=
    \mathop{\mathbb{E}}_{(\zeta_j,\bar b_j)_{j\notin S}}
    \mathop{\mathbb{E}}_{\eta\in\{0,1\}^r}
    \mathop{\mathbb{E}}_{b_1,\ldots,b_r}
    (-1)^{|\eta|}
    \NN\!\left[
        \mathbf{M}_{\eta,b},
        \frac{\mathbb{1}_2}{10z^2}
    \right](\mathbf{y}-\mathbf{Z})\\
    &=
    \mathop{\mathbb{E}}_{(\zeta_j,\bar b_j)_{j\notin S}}
    \widetilde A^{(u)}(\mathbf{y}-\mathbf{Z})\\
    &=
    \mathop{\mathbb{E}}_{\mathbf{Z}}
    \widetilde A^{(u)}(\mathbf{y}-\mathbf{Z}).
\end{aligned}
\]
Here, in the second equality, we used
\(\NN[\mathbf{m}+\mathbf{z},\Sigma](\mathbf{y})
=
\NN[\mathbf{m},\Sigma](\mathbf{y}-\mathbf{z})\), while the last
equality is simply expectation with respect to the distribution of
the random vector \(\mathbf{Z}\).
Consequently, by convexity of the absolute value and translation
invariance of the \(L^1\) norm,
\[
\begin{aligned}
    \left\|A_{\sigma}^{(u)}\right\|_1
    &=
    \int_{\mathbb{R}^2}
    \left|
        \mathop{\mathbb{E}}_{\mathbf{Z}}
        \widetilde A^{(u)}(\mathbf{y}-\mathbf{Z})
    \right|
    \mathrm{d}^2\mathbf{y}\le
    \mathop{\mathbb{E}}_{\mathbf{Z}}
    \int_{\mathbb{R}^2}
    \left|
        \widetilde A^{(u)}(\mathbf{y}-\mathbf{Z})
    \right|
    \mathrm{d}^2\mathbf{y}
    =
    \left\|\widetilde A^{(u)}\right\|_1.
\end{aligned}
\]

After relabelling the coordinates in \(S\) as \([r]\),
\(\widetilde A^{(u)}\) is exactly the function \(B_r\) in
Eq.~\eqref{eq:all_one_coefficient_ht}, with coefficient vectors
\(\mathbf{v}_{j_1}^{(u)},\ldots,\mathbf{v}_{j_r}^{(u)}\).
These vectors satisfy \(\|\mathbf{v}_{j_\ell}^{(u)}\|_2\le1\)
for every \(\ell\in[r]\).

If \(r=1\), Lemma~\ref{lemma_small_t_A_tau_ge_2_ht} gives
\(\|\widetilde A^{(u)}\|_1
\le\exp(-\frac16e^{-45z^2/2})\).
Since
\(\frac16e^{-45z^2/2}\ge\frac1{30}e^{-25z^2}\),
this proves the desired bound in this case.

If \(r\ge2\), Lemmas~\ref{lemma_large_t_A_tau_ge_2_ht}
and~\ref{lemma_small_t_A_tau_ge_2_ht} give
\[
    \left\|\widetilde A^{(u)}\right\|_1
    \le
    \min\left\{
        \exp\!\left(-\frac16e^{-45z^2/2}\right),
        \;
        3zr(0.91)^r+e^{-r^2/2}
    \right\}.
\]
Applying Lemma~\ref{lemma_min_bound_ht}, we obtain
\(\|\widetilde A^{(u)}\|_1
\le\exp(-\frac{e^{-25z^2}}{30}r)\).
Recalling that \(r=|\sigma|\) and combining this estimate with
\(\|A_{\sigma}^{(u)}\|_1
\le\|\widetilde A^{(u)}\|_1\)
concludes the proof.
\end{proof}

\subsubsection{Analysis of the $T$-time tree}

We are now ready to propagate the local estimates through the whole \(T\)-time tree. The key point is that Lemma~\ref{lemma_adapt_prob_hyp_testing} has already reduced the leafwise distinguishability to a sum over high-weight Fourier modes, and, for each such mode, to a product of local factors along the branch. Thus, it remains only to sum these products over the leaves of the tree.

Throughout this paragraph, set
\(\gamma_z\coloneqq e^{-25z^2}/30\).
\begin{lemma}\label{lemma_uniform_A_bound_T_tree_ht}
Let \(u\) be an internal node of the \(T\)-time tree, and let
\(\sigma\in\{0,1\}^{\tau(u)}\). Then
\[
    \left\|A_{\sigma}^{(u)}\right\|_1
    \le \exp\!\left(-\gamma_z|\sigma|\right).
\]
\end{lemma}

\begin{proof}
If \(\sigma=0\), then
\(A_0^{(u)}=\mathop{\mathbb{E}}_{\xi\in\{0,1\}^{\tau(u)}}p_\xi^{(u)}\)
is a probability density. Hence
\(\|A_0^{(u)}\|_1=1=\exp(-\gamma_z|0|)\), as desired.

We may therefore assume that \(\sigma\ne0\). If \(\tau(u)\ge2\), the
claim follows directly from Lemma~\ref{lemma_tau_ge_2_A_bound_ht}.
It remains to consider the case \(\tau(u)=1\). Then necessarily
\(\sigma=1\), and Lemma~\ref{lemma_unit_time_tree_contribution_ht}
gives \(\|A_1^{(u)}\|_1\le\exp(-\frac16e^{-45z^2/2})\).
Since \(z\ge1\), we have
\(\frac16e^{-45z^2/2}\ge\frac1{30}e^{-25z^2}=\gamma_z\).
Therefore \(\|A_1^{(u)}\|_1\le\exp(-\gamma_z)
=\exp(-\gamma_z|\sigma|)\), which concludes the proof.
\end{proof}

We can now bound the total variation distance between the two leaf distributions induced by \(p\) and \(q\) on an arbitrary \(T\)-time tree.

\begin{lemma}\label{lemma_T_time_tree_ht_bound}
Let \(c\in[T]\), and let \(p\) and \(q\) be two probability
distributions over \(\{0,1\}^T\) whose marginals on every subset of
at most \(T-c\) coordinates coincide. Let \(\mathcal{T}\) be any \(T\)-time tree. Then
\[
    \frac12
    \left\|P_p^{(\mathcal{T})}-P_q^{(\mathcal{T})}\right\|_1
    \le
    \frac{\|p-q\|_1}{2}
    \left(\sum_{r=0}^{c-1}\binom{T}{r}\right)
    \exp\!\left(
        -\gamma_z\,(T-c+1)
    \right).
\]
\end{lemma}

\begin{proof}
Let \(r\) be the root of \(\mathcal{T}\). Fix a string \(s\in\{0,1\}^T\) with \(|s|>T-c\). For a node \(v\) of the \(T\)-time tree, define the prefix weight
\[
    \Phi_s(v)
    \coloneqq
    \prod_{i=1}^{N(v)}
    \left|A_{s_i^{(v)}}^{(u_{i-1}^{(v)})}
    \!\left(\mathbf{y}_i^{(v)}\right)\right|,
\]
with the convention \(\Phi_s(r)=1\). Here, along the path from the root to \(v\), the string \(s\) is decomposed as
\[
    s=
    \left(
    s_1^{(v)},s_2^{(v)},\ldots,s_{N(v)}^{(v)},s_{\mathrm{rem}}^{(v)}
    \right),
\]
where \(s_i^{(v)}\) is the restriction of \(s\) to the interval explored by the \(i\)-th experiment along the path, and \(s_{\mathrm{rem}}^{(v)}\) is the restriction of \(s\) to the portion of the block that has not yet been explored after reaching \(v\).

We first prove that, for every node \(v\),
\bb\label{eq:induction_claim_ht_tree}
    \sum_{l\in\mathrm{leaves}(v)}
    \Phi_s(l)
    \le
    \Phi_s(v)
    \exp\!\left(
        -\gamma_z |s_{\mathrm{rem}}^{(v)}|
    \right).
\ee
The proof is by backward induction on the remaining time \(T-t(v)\). If \(v\) is a leaf, then \(s_{\mathrm{rem}}^{(v)}=0\) and \(\mathrm{leaves}(v)=\{v\}\), so \eqref{eq:induction_claim_ht_tree} holds with equality.

Now let \(v\) be an internal node, and assume that the claim holds for every child \(w\) of \(v\). Let \(\sigma_v\) be the restriction of \(s\) to the interval explored by the experiment prescribed at \(v\). Then, for every child \(w\) of \(v\),
\[
    |s_{\mathrm{rem}}^{(v)}|
    =
    |\sigma_v|
    +
    |s_{\mathrm{rem}}^{(w)}|.
\]
Using the induction hypothesis at the children of \(v\), we obtain
\[
\begin{aligned}
    \sum_{l\in\mathrm{leaves}(v)}
    \Phi_s(l)
    &=
    \sum_{w\in\mathrm{children}(v)}
    \sum_{l\in\mathrm{leaves}(w)}
    \Phi_s(l)\\
    &\le
    \sum_{w\in\mathrm{children}(v)}
    \Phi_s(w)
    \exp\!\left(
        -\gamma_z |s_{\mathrm{rem}}^{(w)}|
    \right)\\
    &=
    \Phi_s(v)
    \exp\!\left(
        -\gamma_z\bigl(|s_{\mathrm{rem}}^{(v)}|-|\sigma_v|\bigr)
    \right)
    \sum_{w\in\mathrm{children}(v)}
    \left|A_{\sigma_v}^{(v)}(\mathbf{y}_w)\right|\\
    &=
    \Phi_s(v)
    \exp\!\left(
        -\gamma_z\bigl(|s_{\mathrm{rem}}^{(v)}|-|\sigma_v|\bigr)
    \right)
    \left\|A_{\sigma_v}^{(v)}\right\|_1 .
\end{aligned}
\]
By Lemma~\ref{lemma_uniform_A_bound_T_tree_ht},
\(\|A_{\sigma_v}^{(v)}\|_1\le\exp(-\gamma_z|\sigma_v|)\). Substituting this bound into the previous estimate yields
\[
    \sum_{l\in\mathrm{leaves}(v)}
    \Phi_s(l)
    \le
    \Phi_s(v)
    \exp\!\left(
        -\gamma_z |s_{\mathrm{rem}}^{(v)}|
    \right),
\]
which proves the induction step.

Applying \eqref{eq:induction_claim_ht_tree} to the root \(r\), and using \(\Phi_s(r)=1\) and \(s_{\mathrm{rem}}^{(r)}=s\), gives
\bb\label{eq:fixed_s_tree_product_bound_ht}
    \sum_{l\in\mathrm{leaves}(\mathcal{T})}
    \prod_{i=1}^{N(l)}
    \left|A_{s_i^{(l)}}^{(u_{i-1}^{(l)})}
    \!\left(\mathbf{y}_i^{(l)}\right)\right|
    \le
    \exp(-\gamma_z|s|).
\ee

We now use Lemma~\ref{lemma_adapt_prob_hyp_testing}. Summing \eqref{eq:adapt_prob_product_bound_hyp_testing} over all leaves and multiplying by \(1/2\), we obtain
\[
\begin{aligned}
    \frac12
    \left\|P_p^{(\mathcal{T})}-P_q^{(\mathcal{T})}\right\|_1
    &=
    \frac12
    \sum_{l\in\mathrm{leaves}(\mathcal{T})}
    \left|P_p(l)-P_q(l)\right|\\
    &\le
    \frac{\|p-q\|_1}{2}
    \sum_{s:\,|s|>T-c}
    \sum_{l\in\mathrm{leaves}(\mathcal{T})}
    \prod_{i=1}^{N(l)}
    \left|A_{s_i^{(l)}}^{(u_{i-1}^{(l)})}
    \!\left(\mathbf{y}_i^{(l)}\right)\right|\\
    &\le
    \frac{\|p-q\|_1}{2}
    \sum_{s:\,|s|>T-c}
    \exp(-\gamma_z|s|),
\end{aligned}
\]
where the last step uses \eqref{eq:fixed_s_tree_product_bound_ht}. Since every string in the sum satisfies \(|s|\ge T-c+1\), we get
\[
    \sum_{s:\,|s|>T-c}
    \exp(-\gamma_z|s|)
    \le
    \left|\{s\in\{0,1\}^T:|s|>T-c\}\right|
    \exp\!\left(-\gamma_z(T-c+1)\right).
\]
Finally,
\[
    \left|\{s\in\{0,1\}^T:|s|>T-c\}\right|
    =
    \sum_{r=0}^{c-1}\binom{T}{r},
\]
because choosing a string of weight larger than \(T-c\) is equivalent to choosing at most \(c-1\) zero coordinates. Therefore
\[
    \frac12
    \left\|P_p^{(\mathcal{T})}-P_q^{(\mathcal{T})}\right\|_1
    \le
    \frac{\|p-q\|_1}{2}
    \left(\sum_{r=0}^{c-1}\binom{T}{r}\right)
    \exp\!\left(
        -\gamma_z(T-c+1)
    \right),
\]
as claimed.
\end{proof}

Combining Lemma~\ref{lemma_T_time_tree_ht_bound} with the full-time-tree reduction in Lemma~\ref{lemma_sensing_time_from_T_tree_bound} gives the desired lower bound in the simplified access model.

\begin{lemma}\label{lemma_lower_bound_ht_simplified}
Let \(T\in\mathbb{N}^+\), let \(z\ge1\), and let \(c\in[T]\).
Let \(p\) and \(q\) be two distinct probability distributions over
\(\{0,1\}^T\) whose marginals on every subset of at most \(T-c\)
coordinates coincide. Suppose that a protocol in the simplified
access model solves the CV Hypothesis Testing problem using total
sensing time \(nT\). Then
\[
    n
    \ge
    \frac{2}{3\|p-q\|_1}
    \left(\sum_{r=0}^{c-1}\binom{T}{r}\right)^{-1}
    \exp\!\left(\frac{e^{-25z^2}}{30}(T-c+1)\right).
\]
\end{lemma}

\begin{proof}
By Lemma~\ref{lemma_T_time_tree_ht_bound}, every \(T\)-time tree
\(\mathcal{T}\) satisfies
\(\frac12\|P_p^{(\mathcal{T})}-P_q^{(\mathcal{T})}\|_1\le C\), where
\[
    C
    \coloneqq
    \frac{\|p-q\|_1}{2}
    \left(\sum_{r=0}^{c-1}\binom{T}{r}\right)
    \exp\!\left(-\frac{e^{-25z^2}}{30}(T-c+1)\right).
\]
Applying Lemma~\ref{lemma_sensing_time_from_T_tree_bound} gives
\(n\ge1/(3C)\). Substituting the value of \(C\) proves the statement.
\end{proof}

We now return to the pulsed-signal hypothesis-testing problem.
Combining Lemmas~\ref{lem:pulse_to_piecewise} and~\ref{lemma_access_model},
a Gaussian protocol with physical sensing-time budget
$T_{\mathrm{sensing}}$ can be simulated in the simplified access model
over $n=\lceil T_{\mathrm{sensing}}/T\rceil$ blocks. The preceding bound
on $n$ therefore gives the following theorem.

\begin{thm}\label{thm_lower_bound_ht_gaussian}
Let \(T\in\mathbb{N}^+\), let \(z\ge1\), and let \(c\in[T]\).
Let \(p\) and \(q\) be two distinct probability distributions over
\(\{0,1\}^T\) whose marginals on every subset of at most \(T-c\)
coordinates coincide. Suppose that a Gaussian protocol with
squeezing bounded by \(z\) solves the CV Hypothesis Testing problem
(Problem~\ref{def:cv_hypothesis_testing}) using total sensing time
\(T_{\mathrm{sensing}}\). Then
\[
    T_{\mathrm{sensing}}
    \ge
    T\left[
    \frac{2}{3\|p-q\|_1}
    \left(\sum_{r=0}^{c-1}\binom{T}{r}\right)^{-1}
    \exp\!\left(\frac{e^{-25z^2}}{30}(T-c+1)\right)
    -1
    \right].
\]
\end{thm}

\begin{proof}
Set \(n\coloneqq\lceil T_{\mathrm{sensing}}/T\rceil\).
Lemma~\ref{lem:pulse_to_piecewise} gives a protocol for the
piecewise-constant version with the same squeezing bound and budget
$nT$. Lemma~\ref{lemma_access_model} simulates it in the simplified
access model over the same $n$ blocks. Hence
Lemma~\ref{lemma_lower_bound_ht_simplified} gives
\[
    n
    \ge
    \frac{2}{3\|p-q\|_1}
    \left(\sum_{r=0}^{c-1}\binom{T}{r}\right)^{-1}
    \exp\!\left(\frac{e^{-25z^2}}{30}(T-c+1)\right).
\]
Since \(n\le T_{\mathrm{sensing}}/T+1\), the claim follows.
\end{proof}

\newpage

\section{Protocols to solve the CV Learning Parity problem and its variants \label{sec:upper_bound}}

In this section, we give explicit Gaussian sensing protocols for the
CV Learning Parity problem and its variants, establishing the upper
bounds stated above. All these protocols work directly with the pulsed
signal and only apply controls or measurements at integer times.

By Eq.~\eqref{eq:pulse_bin_area}, the uncontrolled evolution over the
$i$-th unit-time interval is exactly $e^{-im_i\hat p}$. In particular,
a zero-mean Gaussian input with covariance matrix
$\operatorname{diag}(z^{-2},z^2)$, followed by position homodyne
detection at the end of the interval, gives
\begin{equation}\label{eq:pulse_bin_readout}
    R_i\mid m_i
    \sim
    \mathcal{N}\!\left(m_i,\frac{1}{2z^2}\right).
\end{equation}
Thus, the pulse shape $\varphi(t)$ does not affect the readout distributions or the
integer-time control sequences used in the protocols below.
\subsection{Protocol to solve the CV Learning Parity problem}

In Table~\ref{table_algo_learnings} we introduce a Gaussian protocol with bounded squeezing that solves the CV Learning Parity problem (Problem~\ref{def:PROB}) with success probability at least $2/3$, and we show its correctness in Theorem~\ref{thm_upp1}. Notably, this protocol is sensing-time efficient if the squeezing satisfies, for instance, \(z=\omega(\sqrt{\log T})\), with the sensing time scaling polynomially with the pattern size \(T\). However, the protocol is computationally inefficient, as computing the quantities \(Y_r\) for all non-zero \(r\in\{0,1\}^T\) requires time that grows exponentially with \(T\).

\begin{figure}[h]
    \centering
    \includegraphics[width=\linewidth]{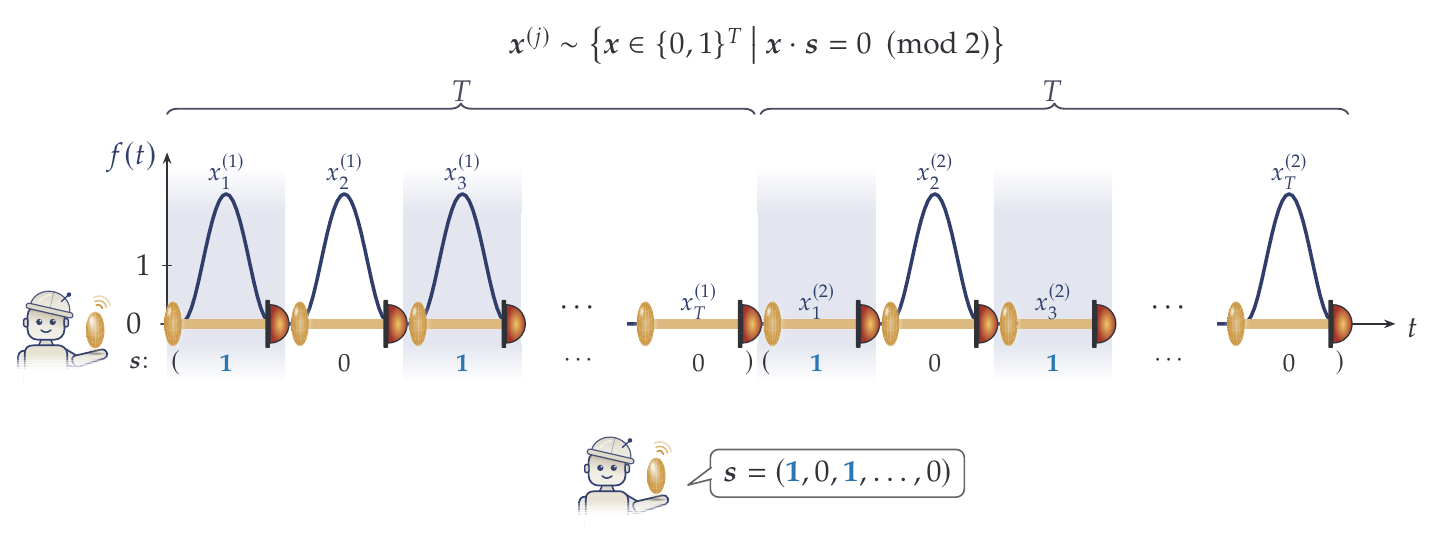}
    \caption{The protocol to solve the CV Learning Parity problem (Problem~\ref{def:PROB}), which involves preparing a new probe at each unit time, and then measuring at the end of this unit time. In the end, some classical post processing can be done to learn the hidden string $\bm{s}$. A similar unit time readout protocol solves the CV Hypothesis Testing problem (Problem~\ref{def:cv_hypothesis_testing}), but with different classical post processing.}
    \label{fig:unit_time_protocol}
\end{figure}

\begin{table}[!htbp]
  \caption{Gaussian protocol with bounded squeezing that solves the CV Learning Parity problem (Problem~\ref{def:PROB})}
  \label{table_algo_learnings}
  \begin{mdframed}[linewidth=2pt, roundcorner=10pt, backgroundcolor=white!10, innerbottommargin=10pt, innertopmargin=10pt]
\textbf{Input:}
\begin{itemize}[topsep=4pt,itemsep=2pt,parsep=0pt,partopsep=0pt]
        \item Squeezing \(z\ge1\);
        \item Pattern size \(T\);
        \item Setting of the CV Learning Parity problem (Problem~\ref{def:PROB}). Specifically, the learner has access to a single-mode system subjected to the Hamiltonian evolution \(\hat{H}(t)\) in Eq.~\eqref{hamiltonian_classical_signal} governed by an unknown parity pattern \(s\in\{0,1\}^T\) over the time interval \(t\in[0,T_{\mathrm{sensing}}]\), where 
        \bb
            T_{\mathrm{sensing}}\coloneqq T\left\lceil 8\exp\left(4e^{-z^2/4}T\right)\log\left(3\cdot 2^T\right) \right\rceil\,.
        \ee
    \end{itemize}
    \textbf{Output:} A bit string \(\tilde{s}\in\{0,1\}^T\) such that \(\tilde{s}=s\) with probability at least \(2/3\).
    \begin{algorithmic}[1]
    \State The time is set to be \(t=0\), and the learner starts to have access to the Hamiltonian evolution.
    \State Define \(n\coloneqq \left\lceil 8\exp\left(4e^{-z^2/4}T\right)\log\left(3\cdot 2^T\right) \right\rceil\).
      \For{\(j \leftarrow 1\) \textbf{to} \(n\)}
        \For{\(i \leftarrow 1\) \textbf{to} \(T\)}
         \State Prepare the system in the Gaussian state with zero first moment and covariance matrix \(\diag(z^{-2},z^2)\).
        \State Wait a unit time, so that the system evolves for a unit time under the Hamiltonian evolution.
        \State Measure the position observable (homodyne detection with respect to the \(x\)-axis).
          \If{the outcome is larger than \(\frac12\)}
              \State Set \(\tilde{x}^{(j)}_i\coloneqq 1\).
          \Else
              \State Set \(\tilde{x}^{(j)}_i\coloneqq 0\).
          \EndIf
      \EndFor
      \State Set \(\tilde{x}^{(j)}\coloneqq \left(\tilde{x}^{(j)}_1,\tilde{x}^{(j)}_2,\ldots, \tilde{x}^{(j)}_T\right)\).
       \EndFor
     \State For each non-zero \(r\in\{0,1\}^T\), compute \(Y_r\coloneqq \sum_{j=1}^n\frac{1+(-1)^{\tilde{x}^{(j)}\cdot r}}{2}\) (i.e.~the total number of \(j\)'s such that \(\tilde{x}^{(j)}\cdot r\equiv 0\)).
    \State Find \(r'\) such that \(Y_{r'}=\max\left\{Y_r: r\in\{0,1\}^T,\ r\ne0\right\}\) (i.e.~the non-zero string \(r'\) that maximizes the total number of \(j\)'s such that \(\tilde{x}^{(j)}\cdot r\equiv 0\)).
    \If{\(Y_{r'}\ge \frac n2+\frac{n}{4}\exp\left(-2e^{-z^2/4}T\right)\)} 
    \State\Return \(\tilde{s}\coloneqq r'\).
    \Else 
    \State\Return  \(\tilde{s}\coloneqq 0\).
    \EndIf
    \end{algorithmic}
  \end{mdframed}
\end{table}

\begin{thm}[(Protocol to solve the CV Learning Parity problem)]\label{thm_upp1}
Let \(T\in\mathbb{N}^+\) denote the pattern size and let \(z\ge 1\) be the squeezing parameter. The algorithm in Table~\ref{table_algo_learnings} is a Gaussian protocol with squeezing bounded by \(z\) that solves the CV Learning Parity problem (Problem~\ref{def:PROB}) with success probability at least \(2/3\), using sensing time
\bb
    T_{\mathrm{sensing}}
    =
    T\left\lceil
    8\exp\left(4e^{-z^2/4}T\right)\log\left(3\cdot 2^T\right)
    \right\rceil .
\ee
Its computational time scales exponentially with \(T\). In particular, if \(z=\omega(\sqrt{\log T})\), then the sensing time \(T_{\mathrm{sensing}}\) scales polynomially with the pattern size \(T\).
\end{thm}

\begin{proof}
Let us introduce some notation and some preliminaries. Let
\bb
    n&\coloneqq \left\lceil 8\exp\left(4e^{-z^2/4}T\right)\log\left(3\cdot 2^T\right) \right\rceil,\\
    A&\coloneqq \frac{n}{2}\exp\left(-2e^{-z^2/4}T\right) -\sqrt{\frac{n}{2}\log\left(3\cdot 2^T\right)},\\
    B&\coloneqq \sqrt{\frac{n}{2}\log\left(3\cdot 2^T\right)}.
\ee
The definition of \(n\) implies that \(A\ge B\). Moreover, the threshold used by the algorithm is exactly
\[
    \frac n2+\frac{A+B}{2}
    =
    \frac n2+\frac{n}{4}\exp\left(-2e^{-z^2/4}T\right).
\]
Thus the output \(\tilde{s}\) of the algorithm satisfies \(\tilde{s}=r'\) if
\[
    Y_{r'}\ge \frac n2+\frac{A+B}{2},
\]
and \(\tilde{s}=0\) otherwise. We now prove correctness by considering separately the cases \(s\ne0\) and \(s=0\).

First, suppose that \(s\ne0\). It suffices to prove that
\bb\label{eq_0_union}
    \Pr\left[ Y_s\le \frac n2+ \frac{A+B}{2}\right]\le \frac{1}{3\cdot 2^T},
\ee
and that
\bb\label{eq_1_union}
    \Pr\left[ Y_r\ge \frac n2+ \frac{A+B}{2}\right]\le \frac{1}{3\cdot 2^T}
    \qquad
    \forall r\in\{0,1\}^T\setminus \{0,s\}.
\ee 
Indeed, by the union bound, these two estimates imply that, with probability at least \(2/3\), the string \(s\) both passes the threshold and strictly maximizes \(Y_r\) among all non-zero \(r\). Hence the algorithm returns \(\tilde{s}=s\).

Second, suppose that \(s=0\). It suffices to prove that
\bb\label{eq_2_union}
    \Pr\left[ Y_r\ge \frac n2+ \frac{A+B}{2}\right]\le \frac{1}{3\cdot 2^T}
    \qquad
    \forall r\in\{0,1\}^T\setminus \{0\}.
\ee 
In fact, by the union bound, this implies that no non-zero string passes the threshold with probability at least \(2/3\), and therefore the algorithm returns \(\tilde{s}=0=s\).

Let \(s\in\{0,1\}^T\) be the unknown parity pattern. Recall that the bit strings
\[
    x^{(j)}\coloneqq (x_i^{(j)})_{i\in[T]},
    \qquad j\in[n],
\]
are i.i.d.~samples from the uniform distribution over
\[
    C_s\coloneqq \left\{x\in\{0,1\}^T:\, x\cdot s\equiv 0\right\}.
\]
In the case \(s=0\), sampling uniformly from \(C_s\) is simply sampling uniformly from the full hypercube \(\{0,1\}^T\). In the case \(s\ne0\), if \(k\) is any index such that \(s_k=1\), then sampling \(x\) uniformly from \(C_s\) is equivalent to sampling the bits \((x_i)_{i\in[T]\setminus\{k\}}\) uniformly at random, and then setting
\[
    x_k\equiv \sum_{i\in[T]\setminus\{k\}}x_is_i \pmod 2.
\]

For each \(j\in[n]\) and \(i\in[T]\), the outcome probability distribution of the homodyne measurement at Step~6 of Table~\ref{table_algo_learnings}, conditioned on the underlying bit \(x_i^{(j)}\), is
\[
    \mathcal{N}\!\left[x_i^{(j)},\frac{1}{2z^2}\right].
\]
The protocol sets \(\tilde{x}^{(j)}_i=1\) if the outcome is larger than \(1/2\), and \(\tilde{x}^{(j)}_i=0\) otherwise. Hence, for each fixed \(x_i^{(j)}\in\{0,1\}\), the probability of reading the bit incorrectly is
\bb\label{eq_p}
    p
    \coloneqq
    \Pr\left[\tilde{x}^{(j)}_i\ne x^{(j)}_i\right]
    =
    \int_{1/2}^\infty \mathrm{d}y\,
    \mathcal{N}\left[0,\frac{1}{2z^2}\right](y)
    =
    \int_{\frac{z}{\sqrt{2}}}^\infty \mathrm{d}y\,
    \mathcal{N}\left[0,1\right](y).
\ee
Equivalently, we can write
\bb\label{eq_ee}
    \tilde{x}_i^{(j)}
    =
    x_i^{(j)}+e_i^{(j)}
    \pmod 2,
\ee
where the error variables \(e_i^{(j)}\) are independent Bernoulli random variables satisfying
\bb\label{eq_eeA}
    \Pr[e_i^{(j)}=1]&=p,\qquad\Pr[e_i^{(j)}=0]&=1-p.
\ee

We now prove Eq.~\eqref{eq_0_union}. Assume \(s\ne0\). We have
\bb\label{eq_ys1}
    Y_s=\sum_{j=1}^n y^{(j)},
\ee
where
\[
    y^{(j)}
    \coloneqq
    \begin{cases}
        1, & \text{if } \tilde{x}^{(j)}\cdot s\equiv0,\\
        0, & \text{otherwise.}
    \end{cases}
\]
The random variables \(y^{(1)},\ldots,y^{(n)}\) are i.i.d.~Bernoulli random variables. Let
\[
    q\coloneqq \Pr\!\left[\tilde{x}^{(1)}\cdot s\equiv0\right].
\]
Since \(x^{(1)}\cdot s\equiv0\), we get
\bb\label{eq_ys2}
    q
    &=
    \Pr\left[\tilde{x}^{(1)}\cdot s\equiv 0\right] =
    \Pr\left[\sum_{i=1}^T e_i^{(1)}s_i\equiv 0\right] =
    \sum_{\substack{0\le l\le |s|\\ l\equiv0}}
    \binom{|s|}{l}p^l(1-p)^{|s|-l} 
    =
    \frac{1+\left(1-2p\right)^{|s|}}{2},
\ee
where \(|s|\coloneqq\sum_{i=1}^Ts_i\). In the last equality we used the identity
\[
    \sum_{\substack{0\le l\le m\\l\equiv 0}}
    \binom{m}{l}a^lb^{m-l}
    =
    \frac{(a+b)^m+(b-a)^m}{2},
\]
which is valid for all \(a,b\in\mathbb{R}\) and \(m\in\mathbb{N}\). Thus \(Y_s\sim\mathrm{Bin}(n,q)\). By Hoeffding's inequality,
\[
    \Pr\left[ Y_s\le nq-t\right]\le e^{-2t^2/n}
    \qquad
    \forall\,t>0.
\]
Setting
\[
    t\coloneqq \sqrt{\frac{n}{2}\log\left(3\cdot 2^T\right)}
\]
gives
\bb\label{y_s}
    \Pr\left[
    Y_s
    \le
    nq-\sqrt{\frac{n}{2}\log\left(3\cdot 2^T\right)}
    \right]
    \le
    \frac{1}{3\cdot 2^T}.
\ee

We now lower bound \(q\). From Eq.~\eqref{eq_p} and the Gaussian tail bound in Lemma~\ref{lemma_bounds_Q},
\[
    2p
    =
    2\int_{\frac{z}{\sqrt{2}}}^\infty
    \mathcal{N}[0,1](y)\,\mathrm{d}y
    \le
    e^{-z^2/4}.
\]
Since \(p\in(0,1/2)\) and \(|s|\le T\), Eq.~\eqref{eq_ys2} gives
\bb\label{boundd_q}
    q
    &=
    \frac{1+\left(1-2p\right)^{|s|}}{2} \ge
    \frac{1+\left(1-e^{-z^2/4}\right)^T}{2} \ge
    \frac{1+\exp\left(-2e^{-z^2/4}T\right)}{2}.
\ee
In the last step we used \(1-y\ge e^{-2y}\), valid for \(y\in[0,e^{-1/4}]\), together with \(z\ge1\).

Combining \eqref{y_s} and \eqref{boundd_q}, we obtain
\bb\label{eq_0thm1}
    \Pr\left[
    Y_s
    \le
    \frac{n}{2}
    +
    \frac{n}{2}\exp\left(-2e^{-z^2/4}T\right)
    -
    \sqrt{\frac{n}{2}\log\left(3\cdot 2^T\right)}
    \right]
    \le
    \frac{1}{3\cdot 2^T}.
\ee
By the definition of \(A\), this says
\[
    \Pr\left[Y_s\le \frac n2+A\right]\le \frac{1}{3\cdot 2^T}.
\]
Since \(A\ge B\), we have \(A\ge (A+B)/2\), and therefore
\[
    \Pr\left[Y_s\le \frac n2+\frac{A+B}{2}\right]
    \le
    \Pr\left[Y_s\le \frac n2+A\right]
    \le
    \frac{1}{3\cdot 2^T}.
\]
This proves Eq.~\eqref{eq_0_union}.

We now prove Eq.~\eqref{eq_1_union}. Assume again that \(s\ne0\), and fix
\[
    r\in\{0,1\}^T\setminus\{0,s\}.
\]
We first recall the following fact:
\bb\label{eq_12_PROB}
    \Pr_{x\sim C_s}\left[x\cdot r\equiv 0\right]=\frac12.
\ee
Indeed, \(C_s=s^\perp\) is a hyperplane over \(\mathbb{F}_2^T\). Since \(r\notin\{0,s\}\), the linear functional \(x\mapsto x\cdot r\) is non-zero when restricted to \(C_s\), and therefore it is balanced on \(C_s\).

Define
\bb\label{eq_yr1}
    Y_r=\sum_{j=1}^n y_r^{(j)},
\ee
where
\[
    y_r^{(j)}
    \coloneqq
    \begin{cases}
        1, & \text{if } \tilde{x}^{(j)}\cdot r\equiv0,\\
        0, & \text{otherwise.}
    \end{cases}
\]
The random variables \(y_r^{(1)},\ldots,y_r^{(n)}\) are i.i.d.~Bernoulli random variables. Let
\[
    q_r\coloneqq \Pr\left[\tilde{x}^{(1)}\cdot r\equiv0\right].
\]
Using \(\tilde{x}^{(1)}=x^{(1)}+e^{(1)}\pmod2\), independence of the measurement noise from \(x^{(1)}\), and the balance property \eqref{eq_12_PROB}, we obtain
\bb\label{eq_yr2}
    q_r
    &=
    \Pr\left[x^{(1)}\cdot r+e^{(1)}\cdot r\equiv0\right] =
    \frac12\Pr\left[e^{(1)}\cdot r\equiv0\right]
    +
    \frac12\Pr\left[e^{(1)}\cdot r\equiv1\right] =
    \frac12.
\ee
Hence \(Y_r\sim\mathrm{Bin}(n,1/2)\). By Hoeffding's inequality,
\[
    \Pr\left[ Y_r\ge \frac{n}{2}+t\right]\le e^{-2t^2/n}
    \qquad
    \forall\,t>0.
\]
Setting
\[
    t\coloneqq \sqrt{\frac{n}{2}\log\left(3\cdot 2^T\right)}
\]
gives
\bb\label{y_r}
    \Pr\left[
    Y_r
    \ge
    \frac{n}{2}
    +
    \sqrt{\frac{n}{2}\log\left(3\cdot 2^T\right)}
    \right]
    \le
    \frac{1}{3\cdot 2^T}.
\ee
Thus
\[
    \Pr\left[Y_r\ge \frac n2+B\right]\le \frac{1}{3\cdot 2^T}.
\]
Since \(A\ge B\), we have \((A+B)/2\ge B\), and therefore
\[
    \Pr\left[
    Y_r
    \ge
    \frac n2+\frac{A+B}{2}
    \right]
    \le
    \Pr\left[Y_r\ge \frac n2+B\right]
    \le
    \frac{1}{3\cdot 2^T}.
\]
This proves Eq.~\eqref{eq_1_union}.

It remains to prove Eq.~\eqref{eq_2_union}. Suppose that \(s=0\). Then \(x^{(1)},\ldots,x^{(n)}\) are i.i.d.~uniform samples from \(\{0,1\}^T\). Hence, for every non-zero \(r\in\{0,1\}^T\), the random variable \(x^{(j)}\cdot r\) is uniform on \(\{0,1\}\). Arguing exactly as above, we obtain
\[
    Y_r\sim\mathrm{Bin}\left(n,\frac12\right)
    \qquad
    \forall r\in\{0,1\}^T\setminus\{0\}.
\]
Therefore Hoeffding's inequality gives
\[
    \Pr\left[ Y_r\ge \frac{n}{2}+B\right]
    \le
    \frac{1}{3\cdot 2^T}
    \qquad
    \forall r\in\{0,1\}^{T}\setminus\{0\}.
\]
Since \((A+B)/2\ge B\), this implies
\[
    \Pr\left[
    Y_r
    \ge
    \frac{n}{2}+\frac{A+B}{2}
    \right]
    \le
    \frac{1}{3\cdot 2^T}
    \qquad
    \forall r\in\{0,1\}^{T}\setminus\{0\}.
\]
This proves Eq.~\eqref{eq_2_union}, and hence the theorem.
\end{proof}

\subsubsection{Protocol to solve the CV Learning Parity problem with improved computational time}

In the previous subsection, we introduced a sensing-time efficient but computationally inefficient protocol for solving the CV Learning Parity problem. The computational inefficiency comes from the classical post-processing step, where the learner must compute $Y_r$ for every non-zero string $r\in\{0,1\}^T$. We now present a second protocol which becomes computationally efficient once the squeezing is at least $z\ge2\sqrt{2\log(10T)}$. Its classical post-processing consists only of computing a rank and solving a linear system over \(\mathbb{F}_2\), and therefore runs in time polynomial in \(T\).

\begin{table}[!htbp]
  \caption{Gaussian protocol with bounded squeezing that solves the CV Learning Parity problem (Problem~\ref{def:PROB}) with an improved computational time}
  \label{table_algo_learnings:comp}
  \begin{mdframed}[linewidth=2pt, roundcorner=10pt, backgroundcolor=white!10, innerbottommargin=10pt, innertopmargin=10pt]
\textbf{Input:}
\begin{itemize}[topsep=4pt,itemsep=2pt,parsep=0pt,partopsep=0pt]
        \item Pattern size \(T\);
        \item Squeezing \(z\ge 2\sqrt{2\log(10T)}\);
        \item Setting of the CV Learning Parity problem (Problem~\ref{def:PROB}). Specifically, the learner has access to a single-mode system subjected to the Hamiltonian evolution \(\hat{H}(t)\) in Eq.~\eqref{hamiltonian_classical_signal} governed by an unknown parity pattern \(s\in\{0,1\}^T\) over the time interval \(t\in[0,T_{\mathrm{sensing}}]\), where $T_{\mathrm{sensing}}\coloneqq 25\,T^2$.
    \end{itemize}
    \textbf{Output:} A bit string \(\tilde{s}\in\{0,1\}^T\) such that \(\tilde{s}=s\) with probability at least \(2/3\).
    \begin{algorithmic}[1]
    \State The time is set to be \(t=0\), and the learner starts to have access to the Hamiltonian evolution.
    \State Define \(n\coloneqq 25\,T\).
      \For{\(j \leftarrow 1\) \textbf{to} \(n\)}
        \For{\(i \leftarrow 1\) \textbf{to} \(T\)}
         \State Prepare the system in the Gaussian state with zero first moment and covariance matrix \(\diag(z^{-2},z^2)\).
        \State Wait a unit time.
        \State Measure the position observable.
          \If{the outcome is larger than \(\frac12\)}
              \State Set \(\tilde{x}^{(j)}_i\coloneqq 1\).
          \Else
              \State Set \(\tilde{x}^{(j)}_i\coloneqq 0\).
          \EndIf
      \EndFor
      \State Set \(\tilde{x}^{(j)}\coloneqq \left(\tilde{x}^{(j)}_1,\tilde{x}^{(j)}_2,\ldots,\tilde{x}^{(j)}_T\right)\).
       \EndFor
       \If{the rank over \(\mathbb{F}_2\) of the list \((\tilde{x}^{(1)},\tilde{x}^{(2)},\ldots,\tilde{x}^{(n)})\) is \(T-1\)}
       \State Set \(\tilde{s}\in\{0,1\}^T\) to be the unique non-zero solution of the linear system \(\tilde{x}^{(j)}\cdot \tilde{s}\equiv 0\) for all \(j\in[n]\).
       \Else
       \State Set \(\tilde{s}\coloneqq 0\).
       \EndIf
    \State\Return \(\tilde{s}\).
    \end{algorithmic}
  \end{mdframed}
\end{table}

\begin{thm}[(Protocol to solve the CV Learning Parity problem with improved computational time)]\label{thm_upp1_comp}
Let \(T\in\mathbb{N}^+\) denote the pattern size, and let \(z\) be a squeezing parameter satisfying $z\ge 2\sqrt{2\log(10T)}$. Then the algorithm in Table~\ref{table_algo_learnings:comp} is a Gaussian protocol with squeezing bounded by \(z\) that solves the CV Learning Parity problem (Problem~\ref{def:PROB}) with success probability at least \(2/3\), using sensing time $T_{\mathrm{sensing}}=25\,T^2$ and computational time polynomial in \(T\).
\end{thm}

\begin{proof}
We use the notation introduced in the proof of Theorem~\ref{thm_upp1}, together with the notation from Table~\ref{table_algo_learnings:comp}. Let \(n=25T\). Recall that the bit strings \(x^{(1)},\ldots,x^{(n)}\), which define the Hamiltonian evolution, are sampled independently and uniformly from
\[
    C_s
    \coloneqq
    \left\{
    x\in\{0,1\}^T:
    x\cdot s\equiv 0
    \right\}.
\]

First, consider the case \(s=0\). Then \(C_s=\{0,1\}^T\), so the strings \(x^{(1)},\ldots,x^{(n)}\) are sampled uniformly from the whole vector space \(\mathbb{F}_2^T\). By Lemma~\ref{lemma_geom}, with probability at least \(\sqrt{3/4}\), the list contains \(T\) linearly independent strings.

Now consider the case \(s\ne0\). Then \(C_s=s^\perp\) is a hyperplane of dimension \(T-1\). Thus all sampled strings belong to this hyperplane, and therefore their rank can never exceed \(T-1\). Moreover, after choosing any coordinate \(k\) such that \(s_k=1\), sampling uniformly from \(C_s\) is equivalent to sampling the \(T-1\) coordinates \((x_i)_{i\in[T]\setminus\{k\}}\) uniformly at random, and then setting
\[
    x_k\equiv \sum_{i\in[T]\setminus\{k\}}x_is_i \pmod 2.
\]
Equivalently, \(C_s\) is linearly isomorphic to \(\mathbb{F}_2^{T-1}\). Hence, by applying Lemma~\ref{lemma_geom} in dimension \(T-1\) (with the case \(T=1\) being trivial), the list \(x^{(1)},\ldots,x^{(n)}\) contains \(T-1\) linearly independent strings with probability at least \(\sqrt{3/4}\). In summary, with probability at least
\bb\label{eq_p_proof}
    q_1\coloneqq \sqrt{\frac34},
\ee
the true sampled strings have rank \(T\) in the case \(s=0\), and rank \(T-1\) in the case \(s\ne0\).

We now control the probability that the noisy readout recovers all sampled bits correctly. As in Eq.~\eqref{eq_p}, the single-bit error probability is
\[
    p
    =
    \int_{\frac{z}{\sqrt{2}}}^\infty
    \mathcal{N}[0,1](y)\,\mathrm{d}y.
\]
Therefore, using Lemma~\ref{lemma_bounds_Q}, $p\le \frac12 e^{-z^2/4}$. Since there are \(nT=25T^2\) measured bits, the probability \(q_2\) that all measured bits are correct satisfies
\bb\label{eq_q_proof}
    q_2
    &\coloneqq
    \Pr\left[\tilde{x}^{(j)}=x^{(j)}\quad\forall j\in[n]\right]\\
    &=
    (1-p)^{nT}\\
    &\ge
    \left(1-\frac12e^{-z^2/4}\right)^{nT}\\
    &\ge
    1-\frac{nT}{2}e^{-z^2/4}\\
    &=
    1-\frac{25}{2}T^2 e^{-z^2/4}\\
    &\ge
    \sqrt{\frac34}.
\ee
In the last line we used \(z\ge2\sqrt{2\log(10T)}\), which implies \(e^{-z^2/4}\le(100T^2)^{-1}\).

The rank event in \eqref{eq_p_proof} depends only on the sampled strings \(x^{(1)},\ldots,x^{(n)}\), while the event in \eqref{eq_q_proof} depends only on the independent measurement noises in the representation \(\tilde{x}^{(j)}=x^{(j)}+e^{(j)}\pmod2\). Hence the two events are independent. Therefore, with probability at least
\[
    q_1q_2\ge \frac34,
\]
the noisy strings \(\tilde{x}^{(1)},\ldots,\tilde{x}^{(n)}\) coincide with the true strings \(x^{(1)},\ldots,x^{(n)}\), and these strings have rank \(T\) if \(s=0\), and rank \(T-1\) if \(s\ne0\).

On this event, the correctness of the algorithm is immediate. If \(s=0\), then the sampled strings span all of \(\mathbb{F}_2^T\), so the only solution of the system
\[
    \tilde{x}^{(j)}\cdot\tilde{s}\equiv0
    \qquad
    \forall j\in[n]
\]
is the zero vector. The algorithm therefore returns \(\tilde{s}=0=s\). If \(s\ne0\), then the sampled strings span the hyperplane \(C_s=s^\perp\). Hence the orthogonal complement of their span is one-dimensional and is generated by \(s\). Therefore the system has a unique non-zero solution, namely \(s\), and the algorithm returns \(\tilde{s}=s\).

We have proved that $\Pr[\tilde{s}=s]\ge \frac34\ge \frac23$. This establishes correctness. Finally, the computational time is polynomial in \(T\), since the only nontrivial classical post-processing consists of computing a rank and solving a linear system over \(\mathbb{F}_2\).
\end{proof}

\begin{lemma}\label{lemma_geom}
Let \(T\in\mathbb{N}^+\). The probability that a set of \(25T\) bit strings, each sampled uniformly at random from \(\{0,1\}^T\), contains \(T\) linearly independent bit strings is at least \(\sqrt{3/4}\).
\end{lemma}

\begin{proof}
Fix \(i\in[T]\), and suppose that \(i-1\) linearly independent bit strings in \(\{0,1\}^T\) have already been obtained. The probability \(p_i\) that a further uniformly sampled bit string is linearly independent of the previously chosen \(i-1\) strings is
\[
    p_i
    =
    \frac{2^T-2^{i-1}}{2^T}.
\]
Let \(X_i\) be the number of additional samples needed to obtain such an independent string. Then \(X_i\) is a geometric random variable with success probability \(p_i\), and therefore $\mathbb{E}[X_i]=\frac{1}{p_i}$. The random variable $X\coloneqq \sum_{i=1}^T X_i$ is the total number of samples needed to obtain \(T\) linearly independent strings. Its expectation satisfies
\[
    \mathbb{E}[X]
    =
    \sum_{i=1}^T\frac{2^T}{2^T-2^{i-1}}
    =
    \sum_{\ell=1}^T\frac{1}{1-2^{-\ell}}
    \le
    T+2.
\]
Hence, by Markov's inequality,
\[
    \Pr[X\ge25T]
    \le
    \frac{\mathbb{E}[X]}{25T}
    \le
    \frac{T+2}{25T}
    \le
    1-\sqrt{\frac34},
\]
where the last inequality holds for all \(T\ge1\). Therefore, with probability at least \(\sqrt{3/4}\), \(25T\) uniform samples contain \(T\) linearly independent strings.
\end{proof}
\subsection{Why unit-time readout is not optimal}
\label{sec:coherent_parity_readout}

The protocols above use a particularly simple sensing strategy: they prepare a fresh state and measure the signal separately after each unit-time interval, and only afterwards process the resulting data classically. We refer to this strategy as \emph{unit-time readout}. While unit-time readout is already sufficient to establish our polynomial upper bounds and the exponential separation with squeezing, it is in fact, not optimal.

In fact, for $z=1$, one can do strictly better by exploiting the freedom to control the sensor during its evolution. Rather than measuring after every interval, the same mode can interact coherently with an entire block, with only phase-space inversions applied in between to control how the different signal amplitudes accumulate. Only at the end of the $T$ block will there be a single homodyne measurement. As we show below, this gives an improvement over unit-time readout already for $z=1$. Thus, the continuous-time control allowed in our sensing model is genuinely useful, rather than merely a technical complication that must be accounted for in the lower-bound analysis.

\subsubsection{An improved protocol for low squeezing}

Fix $z\ge1$. For each block $x\in C_s$, the learner proceeds as follows:
\begin{enumerate}
    \item \textbf{Choose a random trial string.}
    Sample $\bm{\tilde x} \in\{0,1\}^T$ uniformly at random and prepare a zero-mean Gaussian state with covariance matrix $\diag(z^{-2},z^2)$. The string $\bm{\tilde x}$ is a guess of the current signal realization $x$, not of the hidden parity string $\bm{s}$.

    \item \textbf{Let the same mode sense the entire block and only apply gates.} For each unit time interval $i\in[T]$, if $\tilde x_i=1$, apply $R_\pi$ immediately before that interval and $R_\pi^\dagger$ immediately afterwards. No measurement is performed in between. (See \Cref{sec:gaussian_unitaries} for an explicit definition of $R_\pi$.)

    \item \textbf{Measure once at the end.}
    Perform a position homodyne measurement after the $T$ intervals and denote the outcome by $R$. By Eq.~\eqref{eq:pulse_bin_area}, each interval contributes $e^{-ix_i\hat p}$, while $R_\pi^\dagger e^{-ix_i\hat p}R_\pi=e^{ix_i\hat p}$. Hence
    \[
        R=\sum_{i=1}^T(-1)^{\tilde x_i}x_i+N,
        \qquad
        N\sim\NN\!\left[0,\frac{1}{2z^2}\right].
    \]
    Adding the known number $|\bm{\tilde x}|$ gives
    \begin{equation}\label{eq:coherent_parity_readout}
        y\coloneqq R+|\bm{\tilde x}|
        =
        |\bm{x}\oplus \bm{\tilde x}|+N,
    \end{equation}
    where $\oplus$ denotes bitwise addition $\mod 2$ and $|\cdot|$ denotes Hamming weight. Thus, a single homodyne measurement gives a noisy estimate of the number of positions in which $\bm{\tilde x}$ differs from $\bm{x}$.

    \item \textbf{Keep only sufficiently convincing guesses.}
    Retain $\bm{\tilde x}$ only if
    \begin{equation}\label{eq:coherent_parity_accept}
        y\le-\frac{\log(4T)}{2z^2}.
    \end{equation}
    As shown below, conditioned on this event, the retained trial string
    equals the actual signal realization with probability at least $3/4$.
\end{enumerate}

Repeat this procedure for a fixed number $n$ of independent blocks, using a fresh state and a fresh random trial string $\bm{\tilde{x}}$ each time, as depicted in \cref{fig:coh_protocol}.
Every block counts toward the sensing time, including those for which the trial string is rejected, so that $T_{\mathrm{sensing}}=nT$.

Relabel the $m$ retained strings as
$\bm{\tilde x^{(1)}},\ldots, \bm{\tilde x^{(m)}}$ and, as in
Table~\ref{table_algo_learnings}, compute
\[
    Y_{\bm{r}}
    =
    \sum_{j=1}^m
    \frac{1+(-1)^{\bm{\tilde x^{(j)}} \cdot \bm{r}}}{2}
\] for each $\bm{r}\in\{0,1\}^T\setminus\{0\}$.
Thus, $Y_{\bm{r}}$ is simply the number of retained trial strings whose parity with respect to $\bm{r}$ is even, i.e.~the number of $j$'s such that
$\bm{\tilde x^{(j)}} \cdot \bm{r} \equiv0\pmod 2$.

If $m=0$, output zero. Otherwise, output the unique $\bm{r}$ such that
$Y_{\bm{r}}\ge5m/8$, or output zero if no such unique string exists.

\begin{figure}
    \centering
    \includegraphics[width=\linewidth]{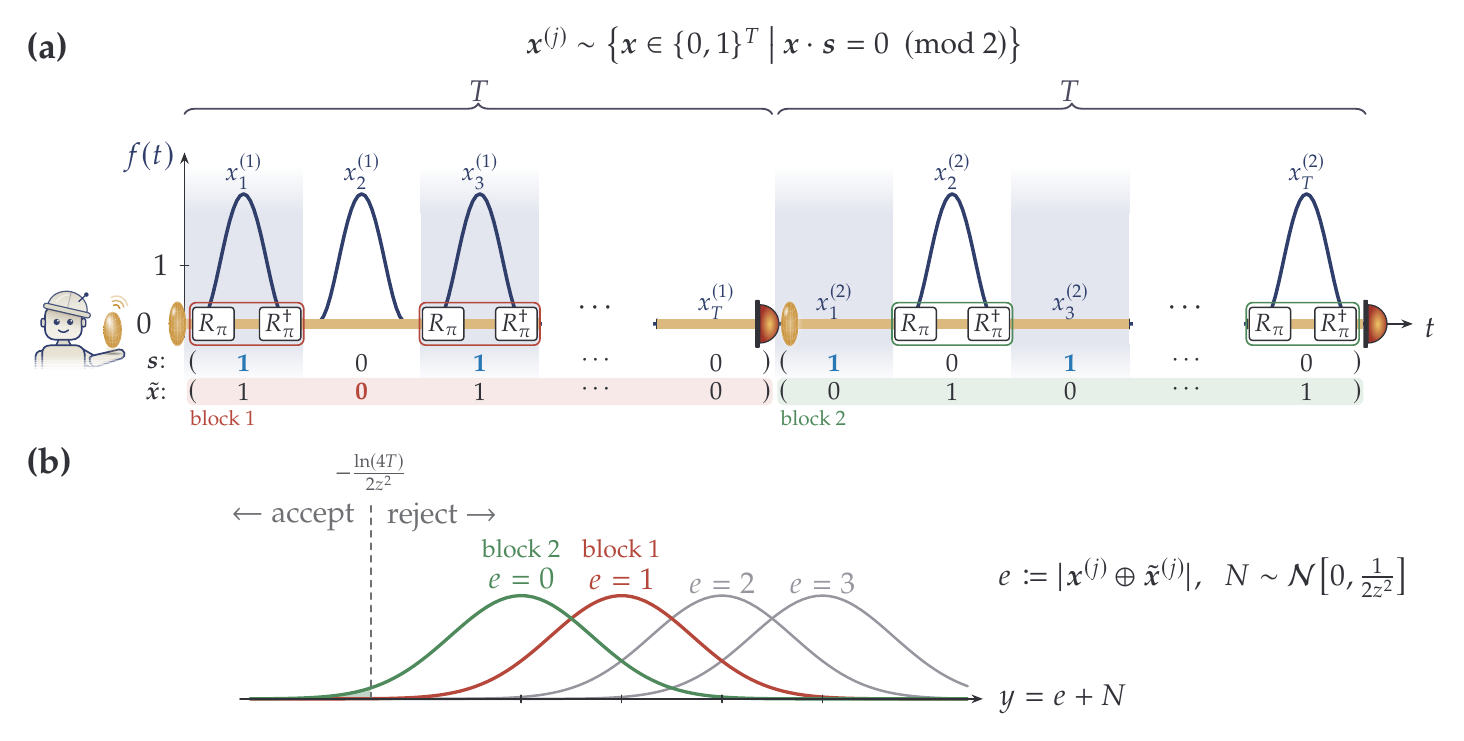}
    \caption{(a) The improved sensing protocol which solves the CV Learning Parity Problem (Problem~\ref{def:PROB}). This protocol outperforms any unit time readout protocols, for the case when $z=1$. Instead of measuring after every unit time, the probe now senses an entire block of length $T$. For each block, the learner draws a random guess of that block's signal $\bm{\tilde{x}} \in \{0,1\}^T$. For each $\tilde{x}_i = 1$, the learner applies $R_{\pi}$ and $R_{\pi}^\dagger$ immediately before and after that unit time interval. A homodyne measurement at the end of the block gives $y=|\bm x^{(j)}\oplus\tilde{\bm x}|+N$ where $N\sim\mathcal N[0,\frac{1}{2z^2}]$ is the noise from the measurement. In block 1, the guess differs by 1 ($e=1)$, and in block 2 it is exactly correct $(e=0)$. (b) A guess is only kept if $y\le-\frac{\ln(4T)}{2z^2}$, which happens essentially only when the guess is $e=0$. The kept strings are then classically processed, and the learner returns it's guess for $\bm{s}$.}
    \label{fig:coh_protocol}
\end{figure}

\subsubsection{Analysis of the improved protocol}
\begin{thm}
\label{thm:coherent_parity_readout}
For every $T\in\mathbb N^+$ and squeezing $z = 1$, the protocol above solves Problem~\ref{def:PROB} with success probability at least $2/3$ using sensing time at most
\begin{equation}\label{eq:coherent_parity_time}
    T_{\mathrm{sensing}}
    \le 2^T\exp\!\left(O((\log T)^2)\right).
\end{equation}
For $z=1$, this protocol has a better exponential factor than any unit-time readout protocol, as proved in Theorem~\ref{thm:unit_time_readout_lower_bound}.
\end{thm}

\begin{proof}
To start, we will compute the sensing time required for any $z \geq 1$. 
Let us first understand what happens when a trial string is retained. Let
\[
    \bm{e}\coloneqq \bm{x}\oplus\bm{\tilde x}
\] denote the indices in which $\bm{x}$ and $\bm{\tilde{x}}$ differ. Since $\bm{\tilde x}$ is chosen uniformly at random, the error string $\bm{e}$ is uniform on $\{0,1\}^T$ and independent of the actual signal string $\bm{x}$.
By Eq.~\eqref{eq:coherent_parity_readout},
\[
    y=|\bm{e}|+N,
    \qquad
    N\sim\NN\!\left[0,\frac{1}{2z^2}\right].
\]
Thus, conditioned on $\bm{e}$, the outcome $y$ is Gaussian with mean $|\bm{e}|$. In particular, a very negative value of $y$ is much more likely when $\bm{e}=0$ than when the trial string contains one or more errors.

Suppose that the trial is accepted, so
\[
    y\le-\frac{\log(4T)}{2z^2}.
\]
Under this condition, for any error string $\bm{e}$, its likelihood relative to the perfect guess
$\bm{e} = 0$ satisfies
\[
    \frac{\NN[|\bm{e}|,1/(2z^2)](y)}
         {\NN[0,1/(2z^2)](y)}
    =
    e^{2z^2|e|y-z^2|e|^2}
    \le
    (4T)^{-|e|}.
\]
Since all error strings have the same prior probability $2^{-T}$,
Bayes' rule gives
\[
\begin{aligned}
    \Pr(\bm{e}=0\mid y)
    &= \frac{\Pr(y| \bm{e}=0) \Pr(e=0)}{\sum_{\bm{e'}\in\{0,1\}^T} \Pr(y|\bm{e'}) \Pr(\bm{e'})} \\
    &= \frac{\NN[0,1/(2z^2)](y)}
    {\sum_{\bm{e'}\in\{0,1\}^T}
     \NN[|\bm{e'}|,1/(2z^2)](y)}\\
    &\ge
    \left[
        \sum_{j=0}^T
        \binom{T}{j}(4T)^{-j}
    \right]^{-1}\\
    &=
    \left(1+\frac1{4T}\right)^{-T}
    \ge
    e^{-1/4}
    >
    \frac34.
\end{aligned}
\]

Note that this bound is pointwise in $y$ and holds only for accepted
values of $y$. Since acceptance is determined by $y$ alone, averaging
over the accepted values of $y$ preserves the bound. Moreover,
$\bm{e}=\bm{0}$ and $\bm{\tilde x}=\bm{x}$ are the same event, and hence
\begin{equation}\label{eq:retained_guess_correct}
    \Pr(\bm{\tilde x}=\bm{x}\mid\text{accepted})\ge\frac34.
\end{equation}

There is one more useful consequence of the construction.
The acceptance event depends only on $e$ and on the Gaussian noise $N$, and not on $\bm{x}$. Since $\bm{x}$ and $\bm{e}$ are independent before conditioning, they remain independent after conditioning on acceptance. Moreover,
$x$ remains uniformly distributed on $C_s$.

We can now see why the parity decoder works. Suppose first that $\bm{s}\ne0$.
Since every $\bm{x}\in C_s$ satisfies $\bm{x}\cdot \bm{s}\equiv0$, we have
\[
    \bm{\tilde x}\cdot \bm{s}
    =
    (\bm{x}\oplus \bm{e})\cdot \bm{s}
    \equiv
    \bm{e}\cdot \bm{s}
    \pmod 2.
\]
In particular, whenever the trial string is exactly correct, i.e.~$\bm{\tilde{x}} = \bm{x}$ and $\bm{e}=0$,
we certainly have $\bm{\tilde x}\cdot \bm{s}\equiv0$. Hence
Eq.~\eqref{eq:retained_guess_correct} implies
\[
    \Pr(\bm{\tilde x}\cdot \bm{s}\equiv0\mid\text{accepted})
    \ge
    \frac34.
\]

Now fix any wrong candidate $\bm{r}$, so $\bm{r}\notin\{0,\bm{s}\}$. Since $\bm{x}$ is uniform on $C_s$,
Eq.~\eqref{eq_12_PROB} gives
\[
    \Pr(\bm{x}\cdot \bm{r}\equiv0)=\Pr(\bm{x}\cdot \bm{r}\equiv1)=\frac12.
\]
Moreover, $x$ is independent of the retained error string $e$.
Therefore, regardless of the value of $e\cdot r$, the bit
\[
    \tilde x\cdot r
    \equiv
    x\cdot r+e\cdot r
    \pmod 2
\]
remains uniform. Thus, 
\[
    \Pr(\tilde x\cdot r\equiv0\mid\text{accepted})
    =
    \frac12
    \qquad
    \forall\,r\notin\{0,s\}.
\] Or in other words, a wrong candidate $\bm{r}$ will look random on the real field strings $\bm{\tilde{x}}$.
If $s=0$, then $x$ is uniform on the whole hypercube, and the same
argument gives probability $1/2$ for every non-zero $r$.

Therefore, conditioned on retaining $m$ trial strings, the true parity $\bm{r} = \bm{s} \neq 0$ has expected count $Y_{\bm{s}}$ at least $3m/4$, whereas any non-zero wrong candidate $\bm{r} \notin \{0, \bm{s}\}$ has expected count $Y_{\bm{r}}$ exactly $m/2$. This explains the threshold $5m/8$, which lies halfway between these two values.
By Hoeffding's inequality,
\[
    \Pr\!\left(
        Y_s<\frac{5m}{8}
        \,\middle|\,
        m
    \right)
    \le
    e^{-m/32}
    \qquad
    (s\ne0),
\]
while, for every $r\notin\{0,s\}$,
\[
    \Pr\!\left(
        Y_r\ge\frac{5m}{8}
        \,\middle|\,
        m
    \right)
    \le
    e^{-m/32}.
\]
When $s=0$, the second bound holds for every non-zero $r$. A decoding error happens when $Y_s$ for the true $s$ falls below this $5m/8$ threshold, or a $r$ reaches it. So, a union bound over all candidates therefore gives
\begin{equation}\label{eq:conditional_coherent_error}
    \Pr(\text{decoding error}\mid m)
    \le
    2^T e^{-m/32}.
\end{equation}
Thus, with high probability, $s$ is the unique string above threshold when $s\ne0$, whereas no non-zero string is above threshold when $s=0$.

It remains to determine how many blocks are needed to obtain enough retained guesses. Let $p$ denote the probability that one trial is accepted. A simple lower bound is obtained by considering only the event that the random trial string $\bm{\tilde{x}}$ is exactly correct $\bm{\tilde{x}} = \bm{x}$. Since $\bm{\tilde x}$ is uniform on $\{0,1\}^T$, this event has probability $2^{-T}$.
Conditioned on $\bm{\tilde x}=\bm{x}$, Eq.~\eqref{eq:coherent_parity_readout}
reduces to $y=N$. Hence
\[
    p
    \ge
    2^{-T}
    \Pr\!\left(
        N\le-\frac{\log(4T)}{2z^2}
    \right).
\]
Using Lemma~\ref{lemma_bounds_Q}, we obtain
\begin{equation}
\label{eq:coherent_parity_acceptance_probability}
    p
    \ge
    \frac{
        2^{-T}
        \exp[-\log^2(4T)/(4z^2)]
    }{
        \sqrt{2\pi}
        \left(
            1+\frac{\log(4T)}{\sqrt2\,z}
        \right)
    }.
\end{equation}

After $n$ independent blocks, the number $m$ of retained guesses is
distributed as $\mathrm{Bin}(n,p)$. Averaging
Eq.~\eqref{eq:conditional_coherent_error} over $m$ gives
\[
\begin{aligned}
    \Pr(\text{decoding error})
    &\le
    2^T\mathbb E[e^{-m/32}]\\
    &=
    2^T(1-p+pe^{-1/32})^n\\
    &\le
    2^T e^{-np/64},
\end{aligned}
\]
where we used $1-e^{-1/32}\ge1/64$ and $1-u\le e^{-u}$.
Therefore it is enough to choose
\[
    n
    \ge
    \frac{64\log(3\cdot2^T)}{p}.
\]
Substituting the lower bound on $p$ from Eq.~\eqref{eq:coherent_parity_acceptance_probability}, rounding $n$ up, and multiplying by the block duration $T$ gives
\begin{equation}\label{eq:coherent_parity_time_explicit}
T_{\mathrm{sensing}}
    \le T+64\sqrt{2\pi}\,T\,2^T\log(3\cdot2^T)
    \left(1+\frac{\log(4T)}{\sqrt2\,z}\right)
    \exp\!\left(\frac{\log^2(4T)}{4z^2}\right).
\end{equation}
The additional $T$ accounts for the rounding.
For $z=1$, we have $\log(3\cdot2^T)=O(T)$,
$1+\log(4T)/\sqrt2=O(\log T)$, and
$\log^2(4T)=(\log T)^2+O(\log T)$. Therefore,
\[
\begin{aligned}
    T_{\mathrm{sensing}}
    &\le 2^T\exp\!\left(\frac14(\log T)^2+O(\log T)\right)\\
    &=2^T\exp\!\left(O((\log T)^2)\right),
\end{aligned}
\]

\end{proof}
We now compare the improved coherent protocol above with the simpler sensing strategy used in Protocol~\ref{table_algo_learnings}, where a fresh state is prepared and measured after every unit-time interval. We call this
\emph{unit-time readout}. To make the comparison as strong as possible, we do not restrict how the resulting data are processed: the learner keeps the full real-valued homodyne outcomes and may apply an arbitrary classical post processing. We show that the sensing time for a unit time readout protocol, with \textit{any} classical post processing, is still exponential in $T$ (for no squeezing $z=1$). Both the improved coherent protocol and any unit time readout protocol have exponential sensing time, but the coherent protocol has a better exponent.


\begin{thm}[(Lower bound for unit-time readout)]
\label{thm:unit_time_readout_lower_bound}
Consider the CV Learning Parity problem with no squeezing, i.e.~$z=1$.
Suppose that the learner performs unit-time readout: on every complete
unit-time interval, it prepares a fresh vacuum state, lets it interact
with the signal for that interval, and performs a position homodyne
measurement. The learner may retain the full real-valued outcomes and
apply any classical post-processing to them.

Any such protocol that solves Problem~\ref{def:PROB} with success
probability at least $2/3$ must use sensing time
\begin{equation}\label{eq:coherent_parity_unit_time_lower}
    T_{\mathrm{sensing}}
    \ge
    \frac{T}{3}\left(\frac52\right)^T .
\end{equation}
\end{thm}
\begin{proof}
It is enough to consider the two hypotheses $H_0:s=0$ and
$H_1:s=(1,\ldots,1)$. Any protocol that learns $s$ with success
probability at least $2/3$ must, in particular, distinguish these two
cases.

Suppose first that the learner observes $n$ complete blocks, and let
$P_0$ and $P_1$ denote the joint densities of all the homodyne outcomes
under $H_0$ and $H_1$, respectively. Successful discrimination requires
$\operatorname{TV}(P_0,P_1)\ge1/3$.
Let $L\coloneqq P_1/P_0$ denote the likelihood ratio of the complete
dataset. Then
\[
    \operatorname{TV}(P_0,P_1)
    =\frac12\int P_0\,|L-1|
    =\frac12\mathbb E_0[|L-1|]
    \le\frac12\sqrt{\mathbb E_0[(L-1)^2]},
\]
where $\mathbb E_0$ denotes expectation under $H_0$, and the last
inequality follows from Cauchy--Schwarz. Thus, to upper bound the total
variation distance, it is enough to upper bound
$\mathbb E_0[(L-1)^2]$.

We first determine the likelihood ratio produced by one block.
For a single unit-time interval, conditioned on the underlying bit being
$b\in\{0,1\}$, Eq.~\eqref{eq:pulse_bin_readout} gives the homodyne
density $p_b(y)\coloneqq\NN[b,1/2](y)$. Define
\[
    h(y)\coloneqq\frac{p_0(y)-p_1(y)}{p_0(y)+p_1(y)}.
\]
Under $H_0$, every bit is independently uniform, so each homodyne
outcome has density $q(y)\coloneqq(p_0(y)+p_1(y))/2$. Hence
$\mathbb E_0[h(R)]=0$. We will also need its second moment. A direct
calculation gives
\[
    \mathbb E_0[h(R)^2]
    =
    1-2\int_{\mathbb R}\frac{p_0p_1}{p_0+p_1}\,\mathrm{d}y.
\]
By Cauchy--Schwarz,
\[
    2\int_{\mathbb R}\frac{p_0p_1}{p_0+p_1}\,\mathrm{d}y
    \ge
    \left(\int_{\mathbb R}\sqrt{p_0p_1}\,\mathrm{d}y\right)^2.
\]
For $p_0=\NN[0,1/2]$ and $p_1=\NN[1,1/2]$, we have
$\int_{\mathbb R}\sqrt{p_0p_1}\,\mathrm{d}y=e^{-1/4}$, and therefore
\begin{equation}\label{eq:unit_time_h_second_moment}
    \mathbb E_0[h(R)^2]\le1-e^{-1/2}<\frac25.
\end{equation}

Now consider the $T$ outcomes $R=(R_1,\ldots,R_T)$ collected from one
complete block, and let $Q_0$ and $Q_1$ denote their joint densities.
Under $H_0$, the string $x$ is uniform on $\{0,1\}^T$, whereas under
$H_1$ it is uniform over the strings of even Hamming weight. Hence
\[
    Q_0(R)
    =
    2^{-T}\sum_{x\in\{0,1\}^T}\prod_{i=1}^T p_{x_i}(R_i)
    =
    2^{-T}\prod_{i=1}^T\bigl(p_0(R_i)+p_1(R_i)\bigr).
\]
For $Q_1$, using
$\mathbf{1}_{\{|x|\ {\rm even}\}}=(1+(-1)^{|x|})/2$, we obtain
\[
\begin{aligned}
    Q_1(R)
    &=
    2^{-(T-1)}
    \sum_{\substack{x\in\{0,1\}^T\\ |x|\ {\rm even}}}
    \prod_{i=1}^T p_{x_i}(R_i)\\
    &=
    2^{-T}
    \left[
        \prod_{i=1}^T\bigl(p_0(R_i)+p_1(R_i)\bigr)
        +
        \prod_{i=1}^T\bigl(p_0(R_i)-p_1(R_i)\bigr)
    \right].
\end{aligned}
\]
Indeed, the first product comes from summing over all strings $x$, while
the factor $(-1)^{|x|}=\prod_i(-1)^{x_i}$ changes
$p_0(R_i)+p_1(R_i)$ into $p_0(R_i)-p_1(R_i)$ in the second term.
Dividing by $Q_0(R)$ gives
\begin{equation}\label{eq:unit_time_block_likelihood}
    \frac{Q_1(R)}{Q_0(R)}
    =
    1+
    \prod_{i=1}^T
    \frac{p_0(R_i)-p_1(R_i)}
         {p_0(R_i)+p_1(R_i)}
    =
    1+\prod_{i=1}^T h(R_i).
\end{equation}
This is the crucial point: with unit-time readout, all the information
distinguishing the two hypotheses appears only through a correlation
involving all $T$ outcomes.

For the $j$-th block, let
$R^{(j)}=(R_1^{(j)},\ldots,R_T^{(j)})$ and set
$Z_j\coloneqq\prod_{i=1}^T h(R_i^{(j)})$.
Equation~\eqref{eq:unit_time_block_likelihood} says that the likelihood
ratio contributed by this block is simply
$Q_1(R^{(j)})/Q_0(R^{(j)})=1+Z_j$.

Since the $n$ blocks are independent under both hypotheses, the joint
densities of the complete dataset factorize as
\[
    P_0(R^{(1)},\ldots,R^{(n)})
    =
    \prod_{j=1}^n Q_0(R^{(j)}),
    \qquad
    P_1(R^{(1)},\ldots,R^{(n)})
    =
    \prod_{j=1}^n Q_1(R^{(j)}).
\]
Therefore their ratio is
\[
    L
    =
    \frac{P_1(R^{(1)},\ldots,R^{(n)})}
         {P_0(R^{(1)},\ldots,R^{(n)})}
    =
    \prod_{j=1}^n
    \frac{Q_1(R^{(j)})}{Q_0(R^{(j)})}
    =
    \prod_{j=1}^n(1+Z_j).
\]

Under $H_0$, the $T$ outcomes within each block are independent and each
has density $q$. Since $\mathbb E_0[h(R)]=0$,
\[
    \mathbb E_0[Z_j]
    =
    \prod_{i=1}^T\mathbb E_0[h(R_i^{(j)})]
    =
    0.
\]
Moreover, Eq.~\eqref{eq:unit_time_h_second_moment} gives
\[
    \mathbb E_0[Z_j^2]
    =
    \left(\mathbb E_0[h(R)^2]\right)^T
    \le
    \left(\frac25\right)^T.
\]
The variables $Z_1,\ldots,Z_n$ are also independent under $H_0$, and
therefore
\[
    \mathbb E_0[L^2]
    =
    \prod_{j=1}^n\mathbb E_0[(1+Z_j)^2]
    =
    \prod_{j=1}^n\left(1+\mathbb E_0[Z_j^2]\right)
    \le
    \left(1+\left(\frac25\right)^T\right)^n.
\]
Since $\mathbb E_0[L]=1$, it follows that
\begin{equation}\label{eq:unit_time_chi_squared}
    \mathbb E_0[(L-1)^2]
    =
    \mathbb E_0[L^2]-1
    \le
    \left(1+\left(\frac25\right)^T\right)^n-1.
\end{equation}

Substituting this into our original bound on the total variation distance
gives
\[
    \operatorname{TV}(P_0,P_1)
    \le
    \frac12
    \sqrt{
        \left(1+\left(\frac25\right)^T\right)^n-1
    }.
\]
But successful discrimination requires
$\operatorname{TV}(P_0,P_1)\ge1/3$. Hence
\[
    \frac13
    \le
    \frac12
    \sqrt{
        \left(1+\left(\frac25\right)^T\right)^n-1
    },
\]
which implies
\[
    n
    \ge
    \frac{\log(13/9)}
         {\log\!\left(1+(2/5)^T\right)}
    \ge
    \frac13\left(\frac52\right)^T,
\]
where we used $\log(1+u)\le u$ and $\log(13/9)>1/3$.

Thus at least $\frac13(5/2)^T$ complete blocks are necessary. Since
each block has duration $T$,
\[
    T_{\mathrm{sensing}}
    \ge
    \frac{T}{3}\left(\frac52\right)^T.
\]

Finally, observing part of one additional block cannot help. Under
$H_0$ and $H_1$, every proper subset of the $T$ underlying bits has
exactly the same uniform distribution. Hence any homodyne outcomes
obtained before completing that block have the same distribution under
the two hypotheses and provide no additional information. This concludes
the proof.
\end{proof}
Taken together, when there is no squeezing ($z=1$), Theorems~\ref{thm:coherent_parity_readout}
and~\ref{thm:unit_time_readout_lower_bound} show that unit-time readout
is not optimal. By letting the same mode evolve coherently across a full block, interleaving the evolution only with phase-space inversions, and measuring once at the end, one can achieve
\[
    T_{\mathrm{sensing}}=2^T \exp\left(O(\log T)^2)\right) = 2^{T + o(T)}.
\]
In contrast, measuring and re-preparing the mode after every unit-time
interval requires exponential time in T, as
\[
    T_{\mathrm{sensing}}
    \ge
    \frac{T}{3}\left(\frac52\right)^T,
\]
even with arbitrary classical post-processing of the full analog
homodyne outcomes. Since $5/2>2$, the improved coherent protocol is still exponential as well, but has a better exponent than any unit-time strategy.
The advantage comes from how the sensor is controlled during the
evolution: keeping the mode coherent allows the signal-induced
displacements to accumulate before the final measurement.

\subsection{Problem of CV Learning Parity with Partially Revealed Information}

In Table~\ref{table_algo_learnings2} we introduce a Gaussian protocol with bounded squeezing that solves the CV Learning Parity with Partially Revealed Information problem with success probability at least $2/3$, and we show its correctness in Theorem~\ref{thm_correct2}. Notably, in the regime where the squeezing satisfies $z=\omega(\sqrt{\log T})$, this protocol is efficient, with both its sensing time and computational time scaling polynomially with the pattern size $T$.

\begin{table}[!htbp]
  \caption{Protocol to solve the CV Learning Parity with Partially Revealed Information problem (Problem~\ref{def:PROB2})}
  \label{table_algo_learnings2}
  \begin{mdframed}[linewidth=2pt, roundcorner=10pt, backgroundcolor=white!10, innerbottommargin=10pt, innertopmargin=10pt]
\textbf{Input:}
\begin{itemize}[topsep=4pt,itemsep=2pt,parsep=0pt,partopsep=0pt]
        \item Squeezing $z\ge1$;
        \item Pattern size $T$;
        \item Setting of the CV Learning Parity with Partially Revealed Information problem (Problem~\ref{def:PROB2}). Specifically, the learner has access to a single-mode system subjected to the Hamiltonian evolution $\hat{H}(t)$ governed by the unknown parity pattern $s$, as defined in the statement of Problem~\ref{def:PROB2}, over the time interval $t\in[0,T_{\mathrm{sensing}}]$, where $T_{\mathrm{sensing}}
            \coloneqq
            10\,T
            \exp\!\left(4e^{-z^2/4}T\right)$.
    \end{itemize}
    \textbf{Output:} A bit string $\tilde{s}\in\{0,1\}^T$ such that $\tilde{s}=s$ with probability at least $2/3$.
    \begin{algorithmic}[1]
    \State The referee selects a non-zero bit string $\bar{s}\in\{0,1\}^T\setminus\{0\}$ uniformly at random. 
    \State The referee defines the hidden string $s$ to be either $s\coloneqq0$ or $s\coloneqq \bar{s}$, with equal probability.
    \State The time is set to be $t=0$, and the learner starts to have access to the Hamiltonian evolution.
    \State Define $n\coloneqq
        \left\lceil
        8\exp\left(4e^{-z^2/4}T\right)\log(3)
        \right\rceil$.
      \For{$j \leftarrow 1$ \textbf{to} $n$}
        \For{$i \leftarrow 1$ \textbf{to} $T$}
         \State Prepare the system in the Gaussian state with zero first moment and covariance matrix $\diag(z^{-2},z^2)$.
        \State Wait a unit time, so that the system evolves for a unit time under the Hamiltonian evolution.
        \State Measure the position observable (homodyne detection with respect to the $x$ axis).
          \If{the outcome is larger than $\frac12$}
              \State Set $\tilde{x}^{(j)}_i\coloneqq 1$.
          \Else
              \State Set $\tilde{x}^{(j)}_i\coloneqq 0$.
          \EndIf
      \EndFor
      \State Set $\tilde{x}^{(j)}\coloneqq \left(\tilde{x}^{(j)}_1,\tilde{x}^{(j)}_2,\ldots, \tilde{x}^{(j)}_T\right)$.
       \EndFor
    \State The referee reveals $\bar{s}$.
    \State The learner computes $Y_{\bar{s}}
        \coloneqq
        \sum_{j=1}^n
        \frac{1+(-1)^{\tilde{x}^{(j)}\cdot \bar{s}}}{2}$, i.e.~the total number of $j$'s such that $\tilde{x}^{(j)}\cdot \bar{s}\equiv 0$.
    \If{$Y_{\bar{s}}\ge \frac n2+\frac{n}{4}\exp\left(-2e^{-z^2/4}T\right)$} 
    \State\Return $\tilde{s}\coloneqq \bar{s}$.
    \Else 
    \State\Return  $\tilde{s}\coloneqq 0$.
    \EndIf
    \end{algorithmic}
  \end{mdframed}
\end{table}

\begin{thm}[(Protocol to solve the CV Learning Parity with Partially Revealed Information problem)]\label{thm_correct2}
Let $T\in\mathbb{N}^+$ denote the pattern size and let $z\ge 1$ be the squeezing parameter. The algorithm in Table~\ref{table_algo_learnings2} is a Gaussian protocol with squeezing bounded by $z$ that solves the CV Learning Parity with Partially Revealed Information problem (Problem~\ref{def:PROB2}) with success probability at least $2/3$, using sensing time at most
\bb
    T_{\mathrm{sensing}}
    =
    10\,T
    \exp\!\left(4e^{-z^2/4}T\right)
\ee
and computational time $O(T_{\mathrm{sensing}})$.
\end{thm}

\begin{proof}
The proof is completely analogous to that of Theorem~\ref{thm_upp1}, except that here the union bound over all non-zero strings is not required. Since
\[
T\left\lceil 8\exp\left(4e^{-z^2/4}T\right)\log(3) \right\rceil
\le
10\,T\exp\!\left(4e^{-z^2/4}T\right)
=
T_{\mathrm{sensing}},
\]
the protocol performs all its measurements within the announced sensing time. We therefore analyse the
\[
        n\coloneqq \left\lceil 8\exp\left(4e^{-z^2/4}T\right)\log(3) \right\rceil
\]
block repetitions actually used by the algorithm.

As in the proof of Theorem~\ref{thm_upp1}, define
\[
        A\coloneqq \frac{n}{2}\exp\left(-2e^{-z^2/4}T\right) -\sqrt{\frac{n}{2}\log(3)},
        \qquad
        B\coloneqq \sqrt{\frac{n}{2}\log(3)} .
\]
The definition of $n$ implies that $A\ge B$. Moreover, the threshold used by the algorithm is exactly
\[
    \frac n2+\frac{A+B}{2}
    =
    \frac n2+\frac n4\exp\left(-2e^{-z^2/4}T\right).
\]
Thus the algorithm outputs $\tilde{s}=\bar{s}$ if
\[
Y_{\bar{s}}\ge \frac n2+\frac{A+B}{2},
\]
and outputs $\tilde{s}=0$ otherwise. Hence, to prove the theorem, we consider two cases: $s=\bar{s}$ and $s=0$.

On the one hand, in the case $s=\bar{s}$, it suffices to prove that
\[
        \Pr\left[ Y_{\bar{s}}\le \frac n2+ \frac{A+B}{2}\right]\le \frac13 .
\]
Indeed, this implies that $\tilde{s}=s$ with probability at least $2/3$. This bound follows from the same argument used to prove Eq.~\eqref{eq_0_union}, with $\log(3\cdot 2^T)$ replaced by $\log(3)$.

On the other hand, in the case $s=0$, it suffices to prove that
\[
        \Pr\left[ Y_{\bar{s}}\ge \frac n2+ \frac{A+B}{2}\right]\le \frac13 .
\]
Indeed, this implies that $\tilde{s}=0=s$ with probability at least $2/3$. This bound follows from the same argument used to prove Eq.~\eqref{eq_2_union}, again with $\log(3\cdot 2^T)$ replaced by $\log(3)$. This concludes the proof.
\end{proof}

\subsection{Problem of CV Learning Parity with partial information revealed beforehand}\label{sec_partial_info_bef}

To obtain an exponential separation for the CV Learning Parity with Partially Revealed Information problem, it is crucial that the referee reveals the partial information \emph{after} the sensing phase. Indeed, if the referee reveals the partial information \emph{before} the sensing phase, then the corresponding problem can be solved efficiently even by a Gaussian protocol with no squeezing, i.e.~with $z=1$. To see this, let us consider the following variant of Problem~\ref{def:PROB2}.

\begin{problem}[(CV Learning Parity with partial information revealed beforehand)] 
\label{def:PROB3}
    Let $T\in\mathbb{N}^+$ denote the pattern size. The problem is divided into three phases:
    \begin{itemize}
        \item \textbf{Referee's choice phase:} A referee selects a non-zero bit string $\bar{s}\in\{0,1\}^T\setminus\{0\}$ and reveals it. Then, the referee defines the hidden string $s$ to be either $s\coloneqq0$ or $s\coloneqq \bar{s}$, with equal probability. Thus only $s$ remains unknown to the learner, while $\bar{s}$ is known.

        \item \textbf{Evolution phase:} The learner has access to the time evolution generated by the Hamiltonian $\hat{H}(t)$ from time $t=0$, where $\hat{H}(t)$ is constructed as described above using a sequence of samples from 
        \bb
            C_s
            \coloneqq
            \left\{
            x\in\{0,1\}^T:
            x\cdot s\equiv 0\pmod 2
            \right\}.
        \ee

        \item \textbf{Classical post-processing phase:} Using the data gathered during the evolution phase, the learner must determine whether the true string $s$ was $0$ or $\bar{s}$, with success probability at least $2/3$.
    \end{itemize} 
\end{problem}

\begin{table}[!htbp]
  \caption{Protocol to solve Problem~\ref{def:PROB3}}
  \label{table_algo_learnings3}
  \begin{mdframed}[linewidth=2pt, roundcorner=10pt, backgroundcolor=white!10, innerbottommargin=10pt, innertopmargin=10pt]
\textbf{Input:}
\begin{itemize}[topsep=4pt,itemsep=2pt,parsep=0pt,partopsep=0pt]
        \item Pattern size $T$;
        \item No squeezing, i.e.~$z=1$;
        \item Setting of Problem~\ref{def:PROB3}. Specifically, the learner has access to a single-mode system subjected to the Hamiltonian evolution $\hat{H}(t)$ governed by the unknown parity pattern $s$, as defined in the statement of Problem~\ref{def:PROB3}, over the time interval $t\in[0,T_{\mathrm{sensing}}]$, where $T_{\mathrm{sensing}}\coloneqq \left\lceil2000\log(3)\right\rceil T$.
    \end{itemize}
    \textbf{Output:} A bit string $\tilde{s}\in\{0,1\}^T$ such that $\tilde{s}=s$ with probability at least $2/3$.
    \begin{algorithmic}[1]
    \State The referee selects a non-zero bit string $\bar{s}\in\{0,1\}^T\setminus\{0\}$ and reveals it. 
    \State The referee defines the hidden string $s$ to be either $s\coloneqq0$ or $s\coloneqq \bar{s}$, with equal probability.
    \State The time is set to be $t=0$, and the learner starts to have access to the Hamiltonian evolution.
    \State Define $n\coloneqq \left\lceil2000\log(3)\right\rceil$.
      \For{$j \leftarrow 1$ \textbf{to} $n$}
        \State Prepare the system in the vacuum state.
        \For{$i \leftarrow 1$ \textbf{to} $T$}
        \If{$\bar{s}_i=0$}
        \State Apply a phase-space rotation of angle $\pi/2$.
        \EndIf
        \State Wait a unit time.
        \If{$\bar{s}_i=0$}
        \State Apply a phase-space rotation of angle $-\pi/2$.
        \EndIf        
         \EndFor
    \State Measure the position observable and denote the outcome by $\tilde{x}$.
    \If{$\tilde{x}\in\bigcup_{m\in\mathbb{Z}}[2m-\frac12,2m+\frac12]$}
              \State Set $Y_j\coloneqq 1$.
          \Else
              \State Set $Y_j\coloneqq 0$.
          \EndIf
    \EndFor
    \If{$\frac{1}{n}\sum_{j=1}^nY_j\ge \frac{\frac{1}{2}+0.55}{2}$} 
    \State\Return $\tilde{s}\coloneqq \bar{s}$.
    \Else 
    \State\Return  $\tilde{s}\coloneqq 0$.
    \EndIf
    \end{algorithmic}
  \end{mdframed}
\end{table}

\begin{thm}[(Protocol to solve Problem~\ref{def:PROB3})]\label{thm_correct3}
Let $T\in\mathbb{N}^+$ denote the pattern size. The algorithm in Table~\ref{table_algo_learnings3} is a Gaussian protocol with no squeezing, i.e.~with $z=1$, that solves Problem~\ref{def:PROB3} with success probability at least $2/3$, using sensing time
\bb
    T_{\mathrm{sensing}}
    =
    \left\lceil2000\log(3)\right\rceil T
\ee
and computational time $O(T_{\mathrm{sensing}})$. In particular, the computational time scales linearly with the pattern size $T$.
\end{thm}

\begin{proof}
We first identify the effective measurement performed by the protocol.
By Eq.~\eqref{eq:pulse_bin_area}, each complete bin contributes a
displacement of strength $x_i$, before accounting for the rotations. Whenever $\bar{s}_i=1$, the learner lets the signal displace the measured quadrature. Whenever $\bar{s}_i=0$, the learner rotates the state before the unit-time evolution and rotates it back afterwards; the corresponding displacement is then placed in the conjugate quadrature and does not affect the final position measurement. Therefore, in one block, the final position outcome has mean $x\cdot\bar{s}$ and variance $1/2$.

Hence, conditioned on the hidden string being $s=\bar{s}$, the probability density of the homodyne outcome $\tilde{x}$ in one repetition is
\bb
    p_{\bar{s}}(\tilde{x})
    \coloneqq
    \mathop{\mathbb{E}}_{\substack{x\in\{0,1\}^T\\ x\cdot \bar{s}\equiv 0}}
    \mathcal{N}\left[x\cdot \bar{s},\frac{1}{2}\right](\tilde{x})\,,
\ee
where the expectation is uniform over the strings satisfying the parity constraint. Moreover, conditioned on the hidden string being $s=0$, the probability density is
\bb
    p_0(\tilde{x})
    \coloneqq
    \mathop{\mathbb{E}}_{x\in\{0,1\}^T}
    \mathcal{N}\left[x\cdot \bar{s},\frac{1}{2}\right](\tilde{x})\,,
\ee
where the expectation is uniform over the whole hypercube.

Let
\[
    E
    \coloneqq
    \bigcup_{m\in\mathbb{Z}}
    \left[2m-\frac12,2m+\frac12\right].
\]
In the case $s=0$, the parity of $x\cdot\bar{s}$ is uniform, because $\bar{s}\ne0$ and $x$ is uniformly distributed over $\{0,1\}^T$. Moreover, shifting a Gaussian by an odd integer exchanges the set $E$ with its complement, up to boundaries of measure zero. Hence, for every $j\in[n]$,
\bb
        \Pr_{\tilde{x}\sim p_0}[Y_j=1]
        =
        \Pr_{\tilde{x}\sim p_0}[\tilde{x}\in E]
        =
        \frac{1}{2}.
\ee

In the case $s=\bar{s}$, the quantity $x\cdot\bar{s}$ is always an even integer. Therefore, by translation invariance modulo $2$,
\bb\label{pr_succ_odd}
        \Pr_{\tilde{x}\sim p_{\bar{s}}}[Y_j=1]
        &=
        \Pr_{\tilde{x}\sim p_{\bar{s}}}[\tilde{x}\in E]\\
        &=
        \sum_{m\in\mathbb{Z}}
        \int_{2m-\frac12}^{2m+\frac12}
        \mathrm{d}\tilde{x}\,
        \mathcal{N}\left[0,\frac{1}{2}\right](\tilde{x})\\
        &=
        \int_{-\frac12}^{\frac12}
        \mathrm{d}\tilde{x}\,
        \left(
        \sum_{m\in\mathbb{Z}}
        \mathcal{N}\left[0,\frac{1}{2}\right](2m+\tilde{x})
        \right).
\ee
By the Poisson summation formula,
\bb\label{Poisson}
        \sum_{m\in\mathbb{Z}}
        \mathcal{N}\left[0,\frac{1}{2}\right](2m+\tilde{x})
        =
        \sum_{\ell\in\mathbb{Z}}S(\ell),
\ee
where
\bb
        S(\ell)
        &\coloneqq
        \int_{\mathbb{R}}
        \mathrm{d}t\,
        \mathcal{N}\left[0,\frac{1}{2}\right](2t+\tilde{x})
        e^{- i\,2\pi t \ell}
        =
        \frac12
        e^{i\pi \ell \tilde{x}}
        e^{-\frac{\pi^2}{4}\ell^2}.
\ee
Consequently,
\bb
        \sum_{m\in\mathbb{Z}}
        \mathcal{N}\left[0,\frac{1}{2}\right](2m+\tilde{x})
        =
        \frac12
        +
        \sum_{\ell=1}^\infty
        e^{-\frac{\pi^2}{4}\ell^2}
        \cos(\pi \ell\tilde{x}) .
\ee
Together with Eq.~\eqref{pr_succ_odd}, this implies
\bb
        \Pr_{\tilde{x}\sim p_{\bar{s}}}[Y_j=1]
        =
        \frac12
        +
        \frac{2}{\pi}
        \sum_{k=0}^\infty
        \frac{(-1)^k}{2k+1}
        e^{-\frac{\pi^2}{4}(2k+1)^2}.
\ee
The right-hand side satisfies
\bb
        \frac12
        +
        \frac{2}{\pi}
        \sum_{k=0}^\infty
        \frac{(-1)^k}{2k+1}
        e^{-\frac{\pi^2}{4}(2k+1)^2}
        >
        0.55.
\ee
In summary, each $Y_j$ is a Bernoulli random variable such that
\[
    \Pr[Y_j=1]=\frac12
    \quad\text{if }s=0,
    \qquad
    \Pr[Y_j=1]>0.55
    \quad\text{if }s=\bar{s}.
\]

We now apply Chernoff's bound. If $Y_1,\ldots,Y_n$ are i.i.d.~Bernoulli random variables with $\Pr[Y_j=1]=p$, then
\bb
        \Pr\left[\frac{1}{n}\sum_{j=1}^nY_j\le (1-\eta)p\right]
        &\le
        \exp\left(-\frac{\eta^2 p}{2}n\right),\\
        \Pr\left[\frac{1}{n}\sum_{j=1}^nY_j\ge (1+\eta)p\right]
        &\le
        \exp\left(-\frac{\eta^2 p}{2+\eta}n\right),
\ee
for all $\eta\in(0,1)$.

In the case $s=\bar{s}$, the algorithm fails only if
\[
    \frac{1}{n}\sum_{j=1}^nY_j
    \le
    \frac{\frac12+0.55}{2}.
\]
Using the lower bound $\Pr[Y_j=1]>0.55$, we obtain
\bb
        \Pr_{\tilde{x}_1,\ldots,\tilde{x}_n\sim p_{\bar{s}}}
        \left[
        \frac{1}{n}\sum_{j=1}^nY_j
        \le
        \frac{\frac{1}{2}+0.55}{2}
        \right]
        &\le
        \exp\left(
        -\frac{1}{2}
        \left(\frac{0.55-\frac12}{2\cdot0.55}\right)^2
        0.55\,n
        \right)
        \le
        \exp\left(-\frac{n}{2000}\right).
\ee
Analogously, in the case $s=0$, the algorithm fails only if
\[
    \frac{1}{n}\sum_{j=1}^nY_j
    \ge
    \frac{\frac12+0.55}{2}.
\]
Since here $\Pr[Y_j=1]=1/2$, Chernoff's bound gives
\bb
        \Pr_{\tilde{x}_1,\ldots,\tilde{x}_n\sim p_0}
        \left[
        \frac{1}{n}\sum_{j=1}^nY_j
        \ge
        \frac{\frac{1}{2}+0.55}{2}
        \right]
        &\le
        \exp\left(
        -\frac{0.05^2}{2+0.05}\,0.5\,n
        \right)
        \le
        \exp\left(-\frac{n}{2000}\right).
\ee
Thus, choosing $n\coloneqq\lceil 2000\log(3)\rceil$ guarantees that the failure probability is at most $1/3$ in both cases. Since the sensing time of the protocol is equal to $nT$, this concludes the proof.
\end{proof}
\newpage
\section{Protocol to solve the CV Hypothesis Testing problem}\label{sec_cv_prob_hyph}

We now give a Gaussian protocol with bounded squeezing for the CV Hypothesis Testing problem. The protocol uses the same sensing step as in the previous section: it reads each coordinate of the signal by a homodyne measurement, thereby producing a noisy classical sample from the underlying sign distribution.

Specifically, the block vectors \(m\in\mathbb{R}^T\) of the signal are i.i.d.~random variables constructed as follows. Let \(p\) and \(q\) be two probability distributions over \(\{0,1\}^T\). In each block of length \(T\), a bit string
\[
    x=(x_1,\ldots,x_T)\in\{0,1\}^T
\]
is sampled either from \(p\) or from \(q\). Independently, amplitudes \(b_1,\ldots,b_T\) are sampled uniformly from \([\frac12,\frac32]\). The corresponding block vector of the signal is
\bb\label{eq:hyp_testing_block_protocol}
    m
    =
    \bigl(
    (-1)^{x_1}b_1,\ldots,(-1)^{x_T}b_T
    \bigr).
\ee
The task is to decide whether the hidden distribution of \(x\) is \(p\) or \(q\).

\begin{table}[!htbp]
  \caption{Gaussian protocol with bounded squeezing that solves the CV Hypothesis Testing problem}
  \label{table_algo_hyp_testing}
  \begin{mdframed}[linewidth=2pt, roundcorner=10pt, backgroundcolor=white!10, innerbottommargin=10pt, innertopmargin=10pt]
\textbf{Input:}
\begin{itemize}[topsep=4pt,itemsep=2pt,parsep=0pt,partopsep=0pt]
        \item Squeezing \(z\ge1\);
        \item Pattern size \(T\);
        \item Two probability distributions \(p,q\) over \(\{0,1\}^T\), with \(p\ne q\);
        \item Setting of the CV Hypothesis Testing problem. Specifically, the learner has access to a single-mode system subjected to the Hamiltonian evolution \(\hat H(t)\) in Eq.~\eqref{hamiltonian_classical_signal}, where the block vectors of the signal are distributed either always according to \(p\) or always according to \(q\), as in Eq.~\eqref{eq:hyp_testing_block_protocol}, over the time interval \(t\in[0,T_{\mathrm{sensing}}]\), where
\[
    T_{\mathrm{sensing}}
    \coloneqq
    T\left\lceil
        \frac{9}
        {
            \left(1-e^{-z^2/4}\right)^{2T}
            \|p-q\|_1^2
        }
    \right\rceil .
\]
    \end{itemize}

    \textbf{Output:} A hypothesis \(\widetilde{h}\in\{p,q\}\) which is correct with probability at least \(2/3\).

    \begin{algorithmic}[1]
    \State Set
    \[
        n
        \coloneqq
        \left\lceil
        \frac{9}
        {
        \left(1-e^{-z^2/4}\right)^{2T}
        \|p-q\|_1^2
        }
        \right\rceil .
    \]
    \State The time is set to be \(t=0\), and the learner starts to have access to the Hamiltonian evolution.
      \For{\(j \leftarrow 1\) \textbf{to} \(n\)}
        \For{\(i \leftarrow 1\) \textbf{to} \(T\)}
            \State Prepare the system in the Gaussian state with zero first moment and covariance matrix \(\diag(z^{-2},z^2)\).
            \State Wait a unit time.
            \State Measure the position observable.
            \If{the outcome is non-negative}
                \State Set \(\widetilde{x}^{(j)}_i\coloneqq 0\).
            \Else
                \State Set \(\widetilde{x}^{(j)}_i\coloneqq 1\).
            \EndIf
        \EndFor
        \State Set \(\widetilde{x}^{(j)}
        \coloneqq
        (\widetilde{x}^{(j)}_1,\ldots,\widetilde{x}^{(j)}_T)\).
      \EndFor
      \State Let \(Q_p\) and \(Q_q\) be the one-block output distributions induced by the above noisy readout when the hidden string is sampled from \(p\) and \(q\), respectively.
      \If{$\prod_{j=1}^n Q_p(\widetilde{x}^{(j)})
          \ge
          \prod_{j=1}^n Q_q(\widetilde{x}^{(j)})$}
          \State \Return \(\widetilde{h}\coloneqq p\).
      \Else
          \State \Return \(\widetilde{h}\coloneqq q\).
      \EndIf
    \end{algorithmic}
  \end{mdframed}
\end{table}

\begin{thm}[(Protocol to solve the CV Hypothesis Testing problem)]\label{thm_upper_bound_hyp_testing}
Let \(T\in\mathbb{N}^+\), let \(z\ge1\), and let \(p,q\) be two distinct probability distributions over \(\{0,1\}^T\). The algorithm in Table~\ref{table_algo_hyp_testing} is a Gaussian protocol with squeezing bounded by \(z\) that solves the CV Hypothesis Testing problem between \(p\) and \(q\) with success probability at least \(2/3\), using sensing time at most
\bb\label{eq:sensing_time_ht_protocol}
    T_{\mathrm{sensing}}
    \le
    \frac{12\,T}
    {
        \left(1-e^{-z^2/4}\right)^{2T}
        \|p-q\|_1^2
    } .
\ee
In particular, if \(z=\omega(\sqrt{\log T})\), and \(\|p-q\|_1\) is at least inverse-polynomial in \(T\), then the sensing time scales polynomially with \(T\). Conversely, if \(\|p-q\|_1\) is exponentially small in \(T\), then the minimum sensing time required to distinguish \(p\) from \(q\) with success probability at least $2/3$ is exponential in \(T\), even in the idealized limit \(z\to\infty\), where the learner can read the sign pattern exactly and the problem reduces to ordinary classical hypothesis testing between \(p\) and \(q\).
\end{thm}

\begin{proof}
We first analyse the distribution of a single noisy readout. Fix one
coordinate $i$ in one block. Applying Eq.~\eqref{eq:pulse_bin_readout}
with $m_i=(-1)^{x_i}b_i$, the homodyne outcome, conditioned on $x_i$
and $b_i$, is distributed as
\[
    R_i
    \sim
    \mathcal{N}\!\left((-1)^{x_i}b_i,\frac{1}{2z^2}\right).
\]
The protocol sets \(\widetilde{x}_i=0\) if \(R_i\ge0\), and \(\widetilde{x}_i=1\) otherwise. Hence the probability of reading the wrong bit is the same for \(x_i=0\) and for \(x_i=1\), and is given by
\bb\label{eq:epsilon_ht_def}
    \varepsilon
    =
    \mathop{\mathbb{E}}_{b\sim \mathrm{Unif}([\frac12,\frac32])}
    \int_{b}^{\infty}
    \mathcal{N}\!\left(0,\frac{1}{2z^2}\right)(r)\,\mathrm{d}r .
\ee
Equivalently,
\[
    \widetilde{x}_i
    =
    x_i+e_i \pmod 2,
\]
where \(e_i\) is a Bernoulli random variable with
\[
    \Pr[e_i=1]=\varepsilon,
    \qquad
    \Pr[e_i=0]=1-\varepsilon.
\]
Thus one full block of measurements produces a sample from the product binary symmetric channel
\[
    P_\varepsilon^{\otimes T}(\widetilde{x}\mid x)
    =
    \prod_{i=1}^T
    P_\varepsilon(\widetilde{x}_i\mid x_i),
\]
where
\[
    P_\varepsilon(y\mid x)
    =
    \begin{cases}
        1-\varepsilon, & y=x,\\
        \varepsilon, & y\ne x.
    \end{cases}
\]
Consequently, if the hidden string is sampled from \(p\), the distribution of one observed string is
\bb\label{eq:Qp_ht_def}
    Q_p(y)
    =
    \mathop{\mathbb{E}}_{x\sim p}
    P_\varepsilon^{\otimes T}(y\mid x),
    \qquad y\in\{0,1\}^T,
\ee
and similarly
\bb\label{eq:Qq_ht_def}
    Q_q(y)
    =
    \mathop{\mathbb{E}}_{x\sim q}
    P_\varepsilon^{\otimes T}(y\mid x).
\ee
Since the protocol repeats the same block experiment \(n\) times independently, the full observed data are distributed either according to \(Q_p^{\otimes n}\) or according to \(Q_q^{\otimes n}\).

We now show that the choice of \(n\) in
Table~\ref{table_algo_hyp_testing} is sufficient. Let
\[
    \alpha_n
    \coloneqq
    \Pr_{Q_p^{\otimes n}}[\widetilde{h}=q],
    \qquad
    \beta_n
    \coloneqq
    \Pr_{Q_q^{\otimes n}}[\widetilde{h}=p]
\]
be the two conditional error probabilities, and define
\[
    P_{\mathrm{err}}
    \coloneqq
    \max\{\alpha_n,\beta_n\}.
\]
For the likelihood-ratio rule used by the algorithm,
\[
    \alpha_n+\beta_n
    =
    1-\frac12
    \left\|
        Q_p^{\otimes n}
        -
        Q_q^{\otimes n}
    \right\|_1 .
\]
Consequently,
\bb\label{eq:error_likelihood_ht}
    P_{\mathrm{err}}
    \le
    1-\frac12
    \left\|
        Q_p^{\otimes n}
        -
        Q_q^{\otimes n}
    \right\|_1 .
\ee
To lower bound the total variation distance after \(n\) repetitions, we use the classical fidelity
\[
    F(P,Q)\coloneqq \sum_y \sqrt{P(y)Q(y)}.
\]
The fidelity is multiplicative under tensor products, and the Fuchs--van de Graaf inequalities give
\[
    1-F(P,Q)
    \le
    \frac12\|P-Q\|_1
    \le
    \sqrt{1-F(P,Q)^2}.
\]
Therefore,
\bb\label{eq:tvd_tensor_ht}
\begin{aligned}
    \frac12
    \left\|
        Q_p^{\otimes n}
        -
        Q_q^{\otimes n}
    \right\|_1
    &\ge
    1-F(Q_p^{\otimes n},Q_q^{\otimes n})  \\
    &=
    1-F(Q_p,Q_q)^n  \\
    &\ge
    1-
    \left(
        1-\frac14\|Q_p-Q_q\|_1^2
    \right)^{n/2}  \\
    &\ge
    1-
    \exp\!\left(
        -\frac{n}{8}\|Q_p-Q_q\|_1^2
    \right).
\end{aligned}
\ee
Combining \eqref{eq:error_likelihood_ht} and \eqref{eq:tvd_tensor_ht}, we obtain
\bb\label{eq:error_bound_Q_ht}
    P_{\mathrm{err}}
    \le
    \exp\!\left(
        -\frac{n}{8}\|Q_p-Q_q\|_1^2
    \right).
\ee

It remains to lower bound \(\|Q_p-Q_q\|_1\) in terms of \(\|p-q\|_1\). We view \(p\) and \(q\) as vectors indexed by \(\{0,1\}^T\). Then
\[
    Q_p-Q_q
    =
    P_\varepsilon^{\otimes T}(p-q),
\]
where the one-bit channel matrix is
\[
    P_\varepsilon
    =
    \begin{pmatrix}
        1-\varepsilon & \varepsilon\\
        \varepsilon & 1-\varepsilon
    \end{pmatrix}.
\]
Since \(\varepsilon<1/2\), this matrix is invertible, with
\[
    P_\varepsilon^{-1}
    =
    \frac{1}{1-2\varepsilon}
    \begin{pmatrix}
        1-\varepsilon & -\varepsilon\\
        -\varepsilon & 1-\varepsilon
    \end{pmatrix}.
\]
The induced \(1\to1\) norm of \(P_\varepsilon^{-1}\) is
\[
    \left\|P_\varepsilon^{-1}\right\|_{1\to1}
    =
    \frac{1}{1-2\varepsilon}.
\]
Hence, using multiplicativity of the induced \(1\to1\) norm under tensor products,
\[
    \left\|
    (P_\varepsilon^{\otimes T})^{-1}
    \right\|_{1\to1}
    =
    \frac{1}{(1-2\varepsilon)^T}.
\]
Therefore,
\bb\label{eq:l1_lower_Q_ht}
    \|p-q\|_1
    =
    \left\|
    (P_\varepsilon^{\otimes T})^{-1}
    (Q_p-Q_q)
    \right\|_1
    \le
    \frac{1}{(1-2\varepsilon)^T}
    \|Q_p-Q_q\|_1,
\ee
and so
\bb\label{eq:Q_l1_lower_ht}
    \|Q_p-Q_q\|_1
    \ge
    (1-2\varepsilon)^T
    \|p-q\|_1 .
\ee

We now relate \(\varepsilon\) to the squeezing parameter. From \eqref{eq:epsilon_ht_def} and the Gaussian tail bound in Lemma~\ref{lemma_bounds_Q},
\bb\label{eq:eps_alpha_ht}
    2\varepsilon
    &=
    2\,
    \mathop{\mathbb{E}}_{b}
    \int_{\sqrt{2}zb}^{\infty}
    \mathcal{N}(0,1)(r)\,\mathrm{d}r
    \le
    \mathop{\mathbb{E}}_{b}
    e^{-z^2b^2}
    \le
    e^{-z^2/4}.
\ee
Thus \(1-2\varepsilon\ge1-e^{-z^2/4}\), and \eqref{eq:Q_l1_lower_ht} gives
\bb\label{eq:Q_l1_lower_alpha_ht}
    \|Q_p-Q_q\|_1
    \ge
    (1-e^{-z^2/4})^T
    \|p-q\|_1 .
\ee
Substituting this estimate into \eqref{eq:error_bound_Q_ht}, we find
\[
    P_{\mathrm{err}}
    \le
    \exp\!\left(
        -\frac{n}{8}
        (1-e^{-z^2/4})^{2T}
        \|p-q\|_1^2
    \right).
\]
To guarantee \(P_{\mathrm{err}}\le \frac13\), and hence success probability at least \(2/3\), it is enough to choose
\[
    n
    \ge
    \frac{8\log 3}
    {
    (1-e^{-z^2/4})^{2T}
    \|p-q\|_1^2
    } .
\]
The choice of \(n\) in Table~\ref{table_algo_hyp_testing} satisfies this condition. This proves the correctness claim.

Moreover, each repetition uses exactly \(T\) units of sensing time. Let
\[
    D
    \coloneqq
    \left(1-e^{-z^2/4}\right)^{2T}
    \|p-q\|_1^2.
\]
Since \(p\ne q\), we have \(D>0\), while $D\le \|p-q\|_1^2\le 4$. Hence \(9/D\ge 9/4\). Using
\(\lceil x\rceil\le \frac43x\) for every \(x\ge 9/4\), we obtain
\bb
    T_{\mathrm{sensing}}
    =
    nT
    =
    T\left\lceil\frac{9}{D}\right\rceil
    \le
    \frac{12T}{D}
    =
    \frac{12\,T}
    {
        \left(1-e^{-z^2/4}\right)^{2T}
        \|p-q\|_1^2
    }.
\ee
This proves Eq.~\eqref{eq:sensing_time_ht_protocol}.

If $z= \omega(\sqrt{\log T})$, then \(e^{-z^2/4}=O(1/T)\). In this regime, \((1-e^{-z^2/4})^{-2T}=O(1)\), and the sensing time scales as
\[
    T_{\mathrm{sensing}}
    =
    O\!\left(
        \frac{T}{\|p-q\|_1^2}
    \right).
\]
Thus, whenever \(\|p-q\|_1\) is at least inverse-polynomial in \(T\), the sensing time is polynomial in \(T\).

Let us finally justify the converse statement in the theorem. It is enough to prove it in the idealized limit \(z\to\infty\). Indeed, in that limit the learner can read the sign pattern in each block exactly. Thus the problem reduces to the following classical task: given \(N\) i.i.d.~samples from either \(p\) or \(q\), decide which hypothesis generated the samples.

Suppose that a classical test distinguishes \(p\) from \(q\) with success probability at least \(2/3\) using \(N\) samples. Then, we must have
\[
    \frac12
    \left\|
        p^{\otimes N}-q^{\otimes N}
    \right\|_1\ge \frac13.
\]
We now upper bound the left hand side in terms of $\frac12\|p-q\|_1$. Let
\[
    F(p,q)
    \coloneqq
    \sum_{x\in\{0,1\}^T}\sqrt{p(x)q(x)}
\]
be the classical fidelity. By the Fuchs--van de Graaf inequalities, $F(p,q)\ge 1-\frac12     \left\|         p-q   \right\|_1$. Moreover, fidelity is multiplicative under tensor products, so
\[
    F(p^{\otimes N},q^{\otimes N})=F(p,q)^N.
\]
Applying again the Fuchs--van de Graaf inequalities, this time to the \(N\)-sample distributions, gives
\[
\begin{aligned}
    \frac12     \left\|         p^{\otimes N}-q^{\otimes N}     \right\|_1
    &\le
    \sqrt{
        1-
        F(p^{\otimes N},q^{\otimes N})^2
    }  \\
    &=
    \sqrt{
        1-
        F(p,q)^{2N}
    }  \\
    &\le
    \sqrt{
        1-
        (1-\frac12     \left\|         p-q   \right\|_1)^{2N}
    } .
\end{aligned}
\]
Since \(\frac12     \left\|         p^{\otimes N}-q^{\otimes N}     \right\|_1\ge1/3\), it follows that
\[
    (1-\frac12     \left\|         p-q   \right\|_1)^{2N}
    \le
    \frac89 .
\]
Equivalently,
\[
    N
    \ge
    \frac{\log(9/8)}
    {2\log\!\left(\frac{1}{1-\frac12     \left\|         p-q   \right\|_1}\right)} .
\]
Therefore, if \(\|p-q\|_1\) is exponentially small in \(T\), the number of exact classical samples needed to distinguish \(p\) from \(q\) with success probability at least \(2/3\) is exponentially large in \(T\). Since the idealized exact-readout model is at least as informative as any finite-squeezing sensing protocol, the same exponential obstruction applies to the CV Hypothesis Testing problem for every value of the squeezing. 
\end{proof}

\end{document}